\documentclass[letterpaper,11pt]{article}
\usepackage{fullpage}
\usepackage[ansinew]{inputenc}
\usepackage{verbatim}
\usepackage{graphicx}
\usepackage{tabularx}
\usepackage{array}
\usepackage{delarray}
\usepackage{dcolumn}
\usepackage{float}
\usepackage{amsmath}
\usepackage{psfrag}
\usepackage{booktabs}
\usepackage{multirow,hhline}
\usepackage{amstext}
\usepackage{amssymb}
\usepackage{fancyhdr}
\usepackage{setspace}
\usepackage{amsthm}
\usepackage{makeidx}
\usepackage{tikz}
\usepackage{lipsum}
\usepackage{url}
\usepackage{adjustbox}
\usepackage{enumitem}
\usepackage{threeparttable}
\usepackage{adjustbox}
\usepackage{hyperref}
\usepackage{bookmark}

\usepackage{xr}
\RequirePackage[round,authoryear,longnamesfirst]{natbib}
\usetikzlibrary{arrows,calc}
\hypersetup{citecolor=true,colorlinks=true,allcolors=blue}
\makeindex
\allowdisplaybreaks
\numberwithin{equation}{section}
\newtheorem{lem}{Lemma}[section]
\newtheorem{thm}{Theorem}[section]

\newtheorem{ass}{Assumption}

\renewcommand{\citep}[1]{\citeauthor{#1}, \citeyear{#1}}

\newcommand{\diag}{\text{diag}}

\newcommand{\convP}{\stackrel{p}{\longrightarrow}}
\newcommand{\convD}{\rightsquigarrow}

\newcommand{\N}{\mathcal{N}}
\newcommand{\eps}{\varepsilon}

\renewcommand{\epsilon}{\varepsilon}
\DeclareMathOperator*{\argmax}{arg\,max}
\DeclareMathOperator*{\argmin}{arg\,min}

\newcommand{\hidefromtoc}{%
  \setcounter{oldtocdepth}{\value{tocdepth}}%
  \addtocontents{toc}{\protect\setcounter{tocdepth}{-10}}%
}
\newcommand{\unhidefromtoc}{%
  \addtocontents{toc}{\protect\setcounter{tocdepth}{\value{oldtocdepth}}}%
}
\newcounter{oldtocdepth}

\theoremstyle{definition}

\newtheorem{rem}{Remark}[section]

\makeatletter
\newcommand*{\rom}
[1]{\expandafter\@slowromancap\romannumeral #1@}
\makeatother

\title{A Dimension-Agnostic Bootstrap Anderson-Rubin Test For Instrumental Variable Regressions\thanks{\scriptsize{We are grateful to Marine Carrasco, Xiaohong Chen, Xu Cheng, Firmin Doko Tchatoka, Liyu Dou, Jean-Marie Dufour, Qingliang Fan, Jiti Gao, Kazuhiko Hayakawa, Hiroaki Kaido, Jia Li, Yifan Li, Ruixuan Liu, Whitney Newey, Ryo Okui, Donald Poskitt, Katsumi Shimotsu, Xiaoxia Shi, Kevin Song, Rami Tabri, Jun Yu, Kai Zhang, Takahsi Yamagata, and all the participants at the CUHK Econometrics Workshop, SMU Econometrics Workshop, the 2024 Asian Meeting of the Econometric Society in China, the 33rd Australia New Zealand Econometric Study Group Meeting, the Workshop on Recent Advances in Econometrics, the 18th and 19th International Symposium on Econometric Theory and Applications, the 59th Annual Meetings of the Canadian Economics Association, the 2025 World Congress of the Econometric Society, and 2025 Symposium for High Dimensional Econometrics and Machine Learning for their valuable comments. Zhang acknowledges the financial support from the NSFC under grant No. 72133002. This research is supported by the Ministry of Education, Singapore under its MOE Academic Research Fund Tier 2 (Project ID: MOE-000767-00). Any opinions, findings and conclusions or recommendations expressed in this material are those of the author(s) and do not reflect the views of the Ministry of Education, Singapore. All errors are our own.}}}
\author{Dennis Lim\thanks{Singapore Management University.\ E-mail~address: dennis.lim.2019@phdecons.smu.edu.sg.} \and Wenjie Wang\thanks{Division of Economics, School of Social Sciences, Nanyang Technological University.
		HSS-04-65, 14 Nanyang Drive, Singapore 637332. 
		E-mail address: wang.wj@ntu.edu.sg. The corresponding author.}   \and Yichong Zhang\thanks{
		Singapore Management University.\ E-mail~address: yczhang@smu.edu.sg.} 
}

\begin{document}
	\maketitle
	\begin{abstract}
		% The conventional and jackknife Anderson-Rubin (AR) tests and tests of overidentifying restrictions
		
		\footnotesize{Weak-identification-robust tests for instrumental variable (IV) regressions are typically developed separately depending on whether the number of IVs is treated as fixed or increasing with the sample size, forcing researchers to make a stance on the asymptotic behavior, which is often ambiguous in practice. This paper proposes a bootstrap-based, dimension-agnostic Anderson-Rubin (AR) test that achieves correct asymptotic size regardless of whether the number of IVs is fixed or diverging, and even accommodates cases where the number of IVs exceeds the sample size. By incorporating ridge regularization, our approach reduces the effective rank of the projection matrix and yields regimes where the limiting distribution of the AR statistic can be a weighted chi-squared, a normal, or a mixture of the two. Strong approximation results ensure that the bootstrap procedure remains uniformly valid across all regimes, while also delivering substantial power gains over existing methods by exploiting rank reduction.}
	\end{abstract}
	
	
	\vspace{2mm}
	
	\hidefromtoc
	
	%\newpage
	\section{Introduction}

Weak and numerous instruments remain persistent concerns in instrumental variable (IV) regressions across various fields. Surveys by \cite{Andrews-Stock-Sun2019} and \citeauthor{lee2021} (\citeyear{lee2021}) find that a considerable number of IV regressions in the \textit{American Economic Review} report first-stage F-statistics below 10. In addition, empirical studies often involve many instruments, such as the 180 IVs used by \cite{Angrist-Krueger1991} to examine the effect of schooling on wages. In ``judge design'' studies, the number of instruments (number of judges) is typically proportional to the sample size (\citep{MS22}).\footnote{E.g., see \cite{kling2006}, \cite{doyle2007}, \cite{dahl2014}, \cite{dobbie2018}, \cite{sampat2019}, \cite{agan2023}, \cite{frandsen2023}, \cite{chyn2024} and the references therein.} Similar patterns of many IVs occur in Fama-MacBeth regressions (\citep{fama1973}; \citep{shanken1992}), shift-share IVs (\citep{Goldsmith2020}), wind-direction IVs (\citep{deryugina2019mortality}; \citep{bondy2020crime}), granular IVs (\citep{gabaix2024}),  local average treatment effect estimation (\citep{blandhol2022}; \citep{boot2024}; \citep{sloczynski2024should}), and Mendelian randomization (\citep{davey2003}; \citep{davies2015}).

However, existing weak-identification-robust inference methods for IV regressions are either based on an asymptotic framework in which the number of instruments $K$ is treated as fixed\footnote{E.g., see \cite{Staiger-Stock1997}, \cite{Stock-Wright2000}, \citet{Kleibergen2002, Kleibergen2005}, \cite{Moreira2003}, \cite{Andrews-Cheng2012}, \cite{Andrews-Mikusheva2016}, \cite{Andrews2018}, \cite{Andrews-Guggenberger2019}, \cite{Moreira-Moreira2019}, among others.} or on an alternative one that allows $K_n$ to diverge to infinity with the sample size $n$.\footnote{E.g., see \cite{Andrews-Stock2007}, \cite{Newey-Windmeijer2009}, \cite{anatolyev2011}, \cite{crudu2021}, \cite{MS22}, \cite{matsushita2024}, \cite{LWZ2023}, \cite{DKM24}, among others.} These methods compare distinct test statistics with distinct critical values so that procedures formulated under the fixed-$K$ asymptotics generally do not have correct size control under the diverging-$K$ asymptotics and vice versa. An empirical researcher is, therefore, forced to take a stance on the asymptotic regime of the number of instruments to implement them, which can be ambiguous in many empirical applications. For example, when $K$ is moderate compared with $n$ (e.g., $K=10$ and $n=200$), it is unclear which test the researcher should use. Furthermore, as we will see below, a third asymptotic regime may arise with the use of regularization, making dimension-robust inference even more challenging.

Motivated by this issue, we propose a bootstrap-based, dimension-agnostic AR test. First, by deriving strong approximations for the proposed test statistic and its bootstrap counterpart (under both the null and alternative hypotheses), we show that the new bootstrap test has a correct asymptotic size, regardless of whether the number of IVs $K$ is fixed or diverging. Our proof, which relies on the Lindeberg swapping strategy, contributes a general result on the strong approximation for quadratic forms with independent and heteroskedastic errors. Additionally, our (conditional) strong approximation derivation for bootstrap statistics involving quadratic forms is novel and may be of independent interest. Second, by employing a ridge-regularized projection matrix, our AR test remains valid in high-dimensional cases where $K$ exceeds the sample size $n$. 
%making it robust across all asymptotic regimes considered in the literature. 
Third, the characterization of the errors in strong approximation offers a theoretically sound basis for selecting the ridge regularizer without taking a stance on the specific regime of $K$. Our choice of the regularizer also helps to reduce the effective rank of the projection matrix, which can potentially improve the power performance of the test. Fourth, we show that depending on the asymptotic behavior of both $K$ and $R_\lambda$ (the effective rank of the regularized projection matrix), the limit distribution of the test statistic can be (1) normal, (2) weighted chi-squared, or (3) a mixture of weighted chi-squared and normal distributions. Given the strong approximation result, our bootstrap inference remains uniformly valid regardless of the asymptotic regimes, and we further provide its power properties under each scenario. Fifth, the strong approximation and uniform inference results are all established when the number of control variables is allowed to diverge at the same rate or even faster than $\sqrt{n}$. Sixth, simulation experiments and an empirical application to the dataset of \cite{card2009} confirm the excellent size and power properties of our bootstrap test compared to alternative methods.    
 
 \textbf{Relation to the literature:} 
For weak-identification-robust inference based on the classical AR test, \cite{Andrews-Stock2007} showed its validity under many instruments, but %their IV model is homoskedastic and 
requires the number of instruments to diverge more slowly than the cube root of the sample size $n$ ($K^3/n \rightarrow 0$). \cite{Newey-Windmeijer2009} proposed a GMM-AR test under many (weak) moment conditions but imposed the same rate condition on $K$. \cite{anatolyev2011} constructed a modified AR test that allows $K$ to be proportional to $n$ but requires homoskedastic errors, and \cite{Kaffo-Wang2017} proposed a bootstrap version of their test.
For estimation with many instruments, \cite{carrasco2012}, \cite{carrasco2015, carrasco2016efficient}, \cite{hansen2014}, and \cite{carrasco2017} proposed regularization approaches for two-stage least squares, limited information maximum likelihood, and jackknife IV (\citep{Angrist1999}) estimators. Furthermore, \cite{carrasco2016} first proposed a ridge-regularized AR test that allows for $K$ being larger than $n$ with homoskedastic errors. \cite{Bun-Farbmacher-Poldermans2020} compared the centered and uncentered GMM-AR test and identified a missing degrees-of-freedom correction when $K/n \rightarrow 0$. Recently, \cite{crudu2021} and \cite{MS22} proposed jackknifed versions of the AR test under many instruments and general heteroskedasticity. 
\cite{DKM24} developed a ridge-regularized version of the jackknife AR test, which is further robust to the scenario where $K$ diverges faster than the sample size. 
However, the jackknife AR tests are based on standard normal critical values that require $K$ to diverge; thus, they may not have the correct size under fixed $K$. \citet[Section 2.3]{tuvaandorj2024} established the validity of a permutation AR test under heteroskedasticity and diverging $K$, requiring $K^3/n \rightarrow 0$. In contrast to the above methods, our test remains valid with heteroskedastic errors uniformly across a broad asymptotic regime for $K$, spanning from fixed to diverging faster than the sample size.

Furthermore, in their seminal paper, \cite{belloni2012} proposed a sup-score test that is weak-identification-robust and dimension-agnostic with heteroskedastic errors, and it allows $K$ to increase exponentially with the sample size. The sup-score test is based on the $\ell_\infty$-norm version of the AR statistic, while our test is based on its $\ell_2$-norm version. As discussed in \citet[Section 3.3]{Kock2026}, there is no ranking of power between $\ell_2$-normed and $\ell_\infty$-normed tests, so that neither test is dominant over the other. In particular, our ridge-regularized bootstrap AR test may offer more power when the first-stage signal is dense, whereas the sup-score test may have a power advantage when the signal is sparse. Therefore, the two tests could be considered complementary. 

Recently, \cite{N23} introduced a jackknife version of the \cite{Kleibergen2002}'s K test and combined it with the sup-score test, but his method relies on a sparse $\ell_1$-regularized estimation of $\rho(Z_i)$, the conditional correlation between the endogenous variable and the outcome error. Without the sparsity assumption, the estimation of $\rho(Z_i)$ may be inconsistent when the dimension of $Z_i$ is large. \cite{boot-ligtenberg2023} developed a dimension-robust AR test based on continuous updating, but relied on an invariance assumption. In contrast to the aforementioned approaches, our bootstrap inference procedure accommodates many instruments and heteroskedastic errors, yet does not rely on invariance or sparsity assumptions.

Our paper also relates to the literature on bootstrap inference for IV models. In this literature, it is found that, when implemented appropriately, bootstrap approaches can substantially improve inference accuracy, including the cases where IVs may be rather weak.\footnote{E.g., see \cite{Davidson-Mackinnon2008, Davidson-Mackinnon2010, Davidson-Mackinnon2014b}, \cite{Moreira-Porter-Suarez2009}, \cite{Wang-Kaffo2016}, \cite{Finlay-Magnusson2019}, \cite{roodman2019}, \cite{young2022}, \cite{Wang-Zhang2024, Wang-Zhang2026}, among others.}

However, no existing study has uniformly established the bootstrap validity with regard to the number of IVs. We fill this gap by deriving strong approximation results for both the test statistic and its bootstrap counterpart. The strong approximation for the AR statistic is related to the analysis of quadratic forms by \cite{HS01} and functions with low influences on product probability space by \cite{MOO10}. Additionally, our results of (conditional) strong approximation for bootstrap statistics with a quadratic form are, based on our best knowledge, new to the literature.

Our test also remains valid even when the number of control variables diverges at a rate of $\sqrt{n}$ or faster, provided it remains of a smaller order than $n$, regardless of whether $K$ is fixed or diverging. As pointed out by \cite{Chao2023} and \cite{mikusheva2024weak}, the presence of many controls can introduce additional bias in jackknife IV estimators and AR tests. This phenomenon, often referred to as the quadratic barrier (see \cite{CJM18}; \cite{LSMPW24}), poses a major challenge for inference. To address this, we design a debiasing procedure for the AR statistic following the construction in \cite{CJN18}, \cite{Chao2023}, and \cite{mikusheva2024weak}. Furthermore, to achieve valid bootstrap inference under many controls, we explicitly account for the impact of debiasing on the dispersion of our AR statistic by appropriately adjusting the form of its bootstrap counterpart.   

Lastly, \citet{Anatolyev-Solvsten2023} proposed an analytical dimension-agnostic $F$ test for linear regressions by analyzing the asymptotic behavior of quadratic forms under two distinct regimes: (1) a fixed number of restrictions, resulting in a weighted chi-squared limiting distribution, and (2) a growing number of restrictions, yielding a normal limiting distribution. Their $F$-test is, in principle, applicable to our setting by testing zero restrictions on the IV coefficients in a linear regression under the null, and it is more general in two respects: (1) it accommodates control variables whose dimension can be of the same order as the sample size $n$, and (2) it allows for testing general linear restrictions. However, our bootstrap inference offers several key advantages. First, although we require the number of controls to be of a smaller order than $n$, we allow the number of instruments $K$ to exceed $n$, a case not covered by their framework. %Moreover, our bootstrap procedure is simpler to implement and numerically more stable than their analytical approach, which relies on a leave-three-out variance estimator.\footnote{In our setting, where the number of controls is $o(n)$, their variance estimator could potentially be simplified.} 
Second, our use of ridge regularization reduces the effective rank of the projection matrix and gives rise to a third asymptotic regime, where $K$ diverges but a Lindeberg-type condition for asymptotic normality fails, resulting in a limiting distribution that is a mixture of weighted chi-squared and normal variables, akin to the regime analyzed in \citet[Sections 6 and 7]{KSS2020} and \cite{YGZ24}. This regime does not arise in \citet{Anatolyev-Solvsten2023} due to the absence of regularization. Analytical inference in this setting requires knowledge of the number of dominant eigenvalues, which can be a challenging task. In contrast, our bootstrap approach circumvents such difficulty by directly employing the strong approximation and remains uniformly valid across all three regimes. In our simulations, when the number of instruments is proportional to the sample size, the use of ridge regularization places the test statistic in the third asymptotic regime. Our bootstrap inference procedure has excellent size control even in this challenging setting, and further provides substantial power gains compared to alternative methods because of the rank reduction.
%the $F$-test of \citet{Anatolyev-Solvsten2023} and the jackknife AR tests of \citet{crudu2021}, \citet{MS22}, and \cite{DKM24}.}

\textbf{Structure of the paper:} Section \ref{section:setup} makes precise the model setup and provides the testing procedure for our dimension-robust AR test statistic. 
%It further motivates and introduces the robust critical-value for our test statistic. 
Sections \ref{section:strong_approx} and \ref{sec: strong_approx_boot} provide the strong approximation results under both null and alternative for our test statistic and its bootstrap counterpart, respectively.    
We derive the power properties of our test under the fixed-$K$ and diverging-$K$ asymptotics, respectively, in Section \ref{section: asy_power}. 
Section \ref{section: monte-carlo-simulation} presents the results of Monte Carlo simulations and Section \ref{empirical-application} applies our test to an empirical application. Proofs of the theorems are given in the Supplemental Appendix, along with additional lemmas and simulation results. 
%Specifically, this section demonstrates that our test consistently differentiates the null from the alternative under strong identification, for both fixed and diverging instruments. Furthermore, that our test have exact asymptotic size-control for both fixed and diverging instruments is also shown. As an additional result, we derive the exact distribution of a generic Jackknifed-AR statistic under fixed $K$ setting in this section. Note that the number of instruments is assumed to be less than the sample size in sections \ref{section:setup}--\ref{section:asymptotic_size} in order to simplify our discussion. Section \ref{section:rank_deficiency} relaxes this and allow the number of instruments to be possibly larger than the sample-size. In particular, this section discusses the case of instruments being rank-deficient, and includes high-dimensional instruments as a special case. Section \ref{subsec:simulation} provides simulation results for our power-curve based on calibrated data, which lends itself to our theory. Section \ref{empirical-application} provides an application of our theory to empirical data.  In Appendix \ref{two_estimators} we provide details on the two estimators satisfying \eqref{estimator_assumption}. In Appendix \ref{section:limit_problems} we discuss general limit problems under fixed and diverging instruments. Appendix \ref{section:appendix_rank_deficiency} provides more detail on the rank-deficiency procedure of Section \ref{section:rank_deficiency}.

\textbf{Notations:} We denote by $[n]$ the set $\{1,\cdots,n\}$, and use $||A||_{op}$ and $||A||_F$ to refer to the operator (spectral) and Frobenius norms of a matrix $A$, respectively.

\section{Setup and Testing Procedure} \label{section:setup}
\subsection{Setup}
Consider the linear instrumental variable regression
\begin{align}
&  \widetilde{Y}_i = \widetilde{X}_i \beta + W_i^{\top} \Gamma + \widetilde e_i \notag \\
&  \widetilde{X}_i = \widetilde{\Pi}_i + \widetilde{v}_i,
\label{eq:model}
\end{align}
where $\widetilde X_i$ denotes a scalar endogenous variable and $W_i \in \mathbb{R}^{d_w}$ denotes the exogenous control variables. In addition, we have $K$-dimensional instrumental variables (IVs) denoted as $\widetilde Z_i$, and $ \widetilde{\Pi}_i \equiv \mathbb{E}( \widetilde{X}_i | \widetilde{Z}_i, W_i)$. 
%and $\widetilde\pi_i \equiv \mathbb{E}( \widetilde{Y}_i - \widetilde{X}_i \beta - W_i \Gamma | \widetilde{Z}_i, W_i)$.
We stack $\widetilde Z_i^\top$ up and denote the resulting $n \times K_n$ matrix $\widetilde Z$. We define $\widetilde Y \in \mathbb R^n$, $\widetilde X \in \mathbb R^n$, $\widetilde \Pi \in \mathbb R^n$, $\widetilde v \in \mathbb R^n$, and $ W \in \mathbb R^{n \times d_w}$ in the same manner. Throughout the paper, we also allow $d_w$ to diverge to infinity but at rate that is slower than the sample size $n$, i.e., $d_w = o(n)$. We further require $W$ to be of full column rank so that its projection matrix $P_W = W(W^\top W)^{-1} W^\top$ is well defined.  We allow, but do not require, $K_n$ to increase with $n$. Specifically, the dimension of $Z$ can be fixed, grow proportional to, or even faster than $n$.

We focus on the model with a scalar endogenous variable for two reasons. First, in many empirical applications of IV regressions, there is only one endogenous variable (as can be seen from the surveys by \cite{Andrews-Stock-Sun2019} and \cite{lee2021}). Second, the strong approximation results derived in Sections \ref{section:strong_approx} and \ref{sec: strong_approx_boot} extend directly to the general case of full-vector inference with multiple endogenous variables. 
Additionally, for the dimension-robust subvector inference, one may use a projection approach (\citep{Dufour-Taamouti2005}) after implementing our test on the whole vector of endogenous variables.\footnote{Alternative subvector inference methods for IV regressions (e.g., see \cite{GKMC2012}, \cite{Andrews2017}, and \cite{GKM2019, GKM2021}) provide a power improvement over the projection approach under fixed $K$. However, whether they can be applied to the current setting is unclear.  
Also, \cite{Wang-Doko2018} and \cite{Wang2020} show that bootstrap tests based on the standard subvector AR statistic may not be robust to weak identification even under fixed $K$ and conditional homoskedasticity.}
%However, we also emphasize that our results do not speak to the subvector AR test, which is, overall, an underdeveloped area in the literature. See Remark XXX for further discussions. We refer to seminal works by XX for more discussion. The subvector weak-identification-robust inference is outside the scope of this paper.} 

To proceed, we first partial out the exogenous control variables $W$ from our IV regressions. Specifically, we stack up $(Y_i,X_i,e_i,\Pi_i,v_i)$ to $(Y,X,e,\Pi,v)$, which are defined as $Y = M_W \widetilde{Y}$, $X = M_W \widetilde{X}$, $\Pi = M_W \widetilde{\Pi}$, %$\pi = M_W \widetilde{\pi}$, 
$e = M_W \widetilde{e}$, and $v = M_W \widetilde{v}$, where $M_W = I_n - P_W$ and $I_n$ is an $n \times n$ identity matrix. In addition, we define  $Z = M_W \widetilde{Z}$. Then, \eqref{eq:model} can be rewritten as 
\begin{align}
      & Y_i = X_i \beta + e_i \notag  \\
  &  X_i = \Pi_i + v_i.
  \label{model_2}
\end{align}

Throughout our analysis, we treat $(Z,W)$ as fixed, which is equivalent to taking all expectations and probability measures conditionally on $(Z,W)$.

\subsection{Test Statistic}
Given that we allow $K_n$ to be greater than $n$, the matrix $Z^\top Z$ is not necessarily invertible. Therefore, we define $P_\lambda = Z (Z^\top Z + \lambda I_{K_n})^{-1}Z^\top$ as the ridge-regularized projection matrix of $Z$ with some ridge penalty $\lambda$ that will be chosen based on $Z$ and $W$ only. As we treat the instruments and control variables as non-random or conditioned upon, so are the ridge-regularizer $\lambda$ and matrix $P_\lambda$. We also allow $\lambda$ to depend on the sample size $n$. For notational simplicity, we suppress this dependence throughout. The $(i,j)$ element of $P_{\lambda}$ is denoted as $P_{\lambda,ij}$. Further denote $e_i(\beta_0) = Y_i - X_i \beta_0$. Then, our dimension-agnostic AR test statistic is written as 
\begin{align}
\widehat{Q}(\beta_0) & = \frac{\sum_{i \in [n]} \sum_{j \in [n], j \neq i}e_{i}(\beta_0) P_{\lambda,ij} e_{j}(\beta_0)}{\sqrt{R_\lambda}} - \frac{\sum_{i,j \in [n]^2} \kappa_{ij} e_j^2(\beta_0) A_{\lambda,ii}}{\sqrt{R_\lambda}},
\label{Q_hat_statistic_definition}
\end{align}
where $\kappa = (M_W \circ M_W)^{-1}$,\footnote{Here $\circ$ denotes the Hadamard product and $M_W \circ M_W$ is invertible as long as $d_w < n/2$ as shown by \cite{CJN18}.} $A_{\lambda,ii} = 2 P_{\lambda,ii} P_{W,ii} - B_{\lambda,ii}$,
$B_{\lambda,jk} = \sum_{i \in [n]} P_{W,ik} P_{W,ij} P_{\lambda,ii} = [P_W D_{\lambda} P_W]_{jk}$, 
\begin{align*}
D_{\lambda} = \diag(P_{\lambda,11}, \cdots, P_{\lambda,nn}) = \diag(P_{\lambda}), 
\end{align*}
\begin{align}\label{eq:K_lambda}
    R_\lambda = \sum_{i \in [n]} \sum_{j \in [n], j \neq i} \Xi_{\lambda,ij}^2,
\end{align}
and
\begin{align*}
\Xi_{\lambda,ij} =  \begin{cases}
     P_{\lambda,ij} + (P_{\lambda,ii} + P_{\lambda,jj}) P_{W,ij} - B_{\lambda,ij} & \quad i \neq j \\
     0 & \quad i = j
\end{cases},
\end{align*}
where we adopt the convention for the diagonal operator $\diag(\cdot)$ that, when the argument is a vector $a$, $\diag(a)$ denotes the diagonal matrix with $a$ on the diagonal, and when the argument is a square matrix $A$, $\diag(A)$ denotes the diagonal matrix formed from the diagonal elements of $A$. 
We can also regard $R_\lambda$ as the effective rank under the ridge regularization.

We note that under the null (i.e., $\beta = \beta_0$), the first quadratic term of $\widehat Q(\beta_0)$ in \eqref{Q_hat_statistic_definition} does not have an exact zero mean due to partialling out controls. This bias is not asymptotically negligible when the dimension of the controls ($d_w$) is of the order $\sqrt{n}$ or greater. The second term of $\widehat Q(\beta_0)$ in \eqref{Q_hat_statistic_definition} is used to correct such a bias. The similar construction of the inverse of the Hadamard product of the annihilator matrices was used by \cite{CJN18}, \cite{Chao2023}, and \cite{mikusheva2024weak}. 

The regularizer $\lambda$ is chosen as 
\begin{align}\label{eq:lambda}                 \lambda = \max\left\{ \theta \in [0,\overline \theta]: \left(\frac{\max_{i \in [n]}P_{\theta,ii}^2}{R_\theta}\right) \left(1+\sum_{i \in [n]} P_{W,ii}^2 \right)\leq c_1,  \frac{\max_{i \in [n]}\sum_{j \in [n], j \neq i}\Xi_{\theta,ij}^2}{R_\theta} \leq \frac{c_2}{\sqrt{n} } \right\},
\end{align}
where $P_{\theta} = Z (Z^\top Z + \theta I_{K_n})^{-1}Z^\top$, $P_{\theta, ij}$ is the $(i,j)$ entry of $P_{\theta}$,  $\overline \theta = ||Z^{\top}Z||_{op}$, $R_{\theta}$ is defined in \eqref{eq:K_lambda} with $\lambda$ replaced by $\theta$, while $c_1$ and $c_2$ are two positive constants chosen by the researcher such that $c_1$ is sufficiently small. We view $\frac{c}{0} = + \infty$ for any $c>0$. In practice, we use $c_1 = 0.1$ and $c_2 = 1$.\footnote{We have done extensive simulations and find that the results of our test are not sensitive to the specific choice of $c_1$ and $c_2$. The simulation results with alternative choices of $c_1$ and $c_2$ are reported in the Supplemental Appendix.} If there is no $\lambda$ that satisfies both inequalities in \eqref{eq:lambda}, then we choose 
\begin{align*}
    \lambda = \argmin_{\theta  \in [0,\bar \theta] }\left(\frac{\max_{i \in [n]}P_{\theta,ii}^2}{R_\theta}\right) \left(1+\sum_{i \in [n]} P_{W,ii}^2 \right).
\end{align*}

\begin{rem} Several remarks regarding the choice of the regularizer are in order. 
First, we select the regularizer $\lambda$ as the largest value over the interval $[0,\overline \theta]$ that both 
\begin{align*}
 \frac{\max_{i \in [n]}P_{\lambda,ii}^2}{R_\lambda}    \left(1+\sum_{i \in [n]} P_{W,ii}^2 \right) \quad \text{and} \quad \frac{\max_{i \in [n]}\sum_{j \in [n], j \neq i}\Xi_{\lambda,ij}^2}{R_\lambda}
\end{align*}
 remain small. These are two critical conditions for ensuring the validity of our strong approximation results for both the test statistic and the bootstrap critical value. We will discuss the theoretical properties of these two terms in detail below. Moreover, since the choice of $\lambda$ depends solely on the instruments and controls, which are treated as non-random or conditioned upon, it does not introduce any model selection bias. 
 
%    In the revised version, the regularizer $\lambda$ is chosen as 
%\begin{align} 
%\lambda = \max\begin{Bmatrix}
%\theta \in [0,\overline \theta]: \left(\frac{\max_{i \in [n]}P_{\theta,ii}^2}{R_\theta}\right) \left(1+\sum_{i \in [n]} P_{W,ii}^2 \right)\leq c_1,  \frac{\max_{i \in [n]}\sum_{j \in [n], j \neq i}\Xi_{\theta,ij}^2}{R_\theta} \leq \frac{c_2}{\sqrt{n} } 
%\end{Bmatrix},
%\end{align}where we recommend choosing $\overline \theta = ||Z^{\top}Z||_{op}$, $c_1 = 0.1$, and $c_2 =1$. 
Second, given that the conditions for strong approximation are satisfied, we choose the regularizer $\lambda$ as large as possible over $[0, \bar\theta]$. Such a choice is inspired by \cite{carrasco2012}, \cite{carrasco2015, carrasco2016efficient}, and \cite{carrasco2017}, who showed that their proposed regularized IV estimators can be more efficient than those without regularization by employing a sufficiently large value of the regularizer relative to the overall instrument strength (concentration parameter).\footnote{For example, see Proposition 1 of \cite{carrasco2012}, Proposition 2 of \cite{carrasco2015}, and Proposition 2 of \cite{carrasco2016efficient}, in which regularized IV estimators are shown to achieve the semiparametric efficiency bound under homoskedastic errors, given a sufficiently large value of the regularizer relative to the concentration parameter.}  
%achieve the semiparametric efficiency bound under homoskedastic errors when $n/(\lambda ||\Pi||^2_2) \rightarrow 0$.
%\footnote{This corresponds to $1/(\alpha \mu^2_n) \rightarrow 0$ in their papers ($\alpha$ and $\mu^2_n$ correspond to $\lambda/n$ and $||\Pi||^2_2$, respectively, in our setting).}

Third, our choice of the upper bound for $\lambda$ as $\overline \theta = ||Z^{\top}Z||_{op}$ is motivated by the fact that the ridge regularization transforms the eigenvalues of $Z^\top Z$. Specifically, consider the case where $K_n \leq n$ and the singular value decomposition of $Z$ as 
$Z = \mathcal U \mathcal S \mathcal V^\top,$ where $\mathcal U \in \Re^{n \times K_n}$ with $\mathcal U^\top \mathcal U = I_{K_n}$, $S = \diag(s_1,\cdots,s_{K_n})$ is a diagonal matrix of non-zero singular values in descending order, $\mathcal V \in \Re^{K_n \times K_n}$, and $\mathcal V^\top \mathcal V = I_{K_n}$. Then, the regularized projection matrix is given by
\begin{align*}
 P_\lambda = \mathcal U \diag\left(\frac{s_1^2}{s_1^2 + \lambda},\cdots,\frac{s_{K_n}^2}{s_{K_n}^2 + \lambda} \right) \mathcal U^\top.
\end{align*}
If for some $k \in [K_n]$, the ratio $s_k/s_1$ is close to zero, then choosing $\lambda$ on the order of $s_1^2 = ||Z^{\top}Z||_{op}$ will cause the $k$-th singular value of $P_\lambda$ (i.e., $\frac{s_k^2}{s_k^2 + \lambda}$) to be close to zero. Intuitively, a large $\lambda$ attenuates the contributions of directions associated with small singular values, effectively reducing the rank of $P_\lambda$ and helping to improve the power performance of our test. We will give more details on this point in Section \ref{section: asy_power} (e.g., see Remark \ref{rem: power}).

\end{rem}

\subsection{Bootstrap Critical Value}\label{sec:bootstrap CV}
To implement the dimension-agnostic test, we propose to use bootstrap critical values. Specifically, let $\{\eta_i\}_{i \in [n]}$ be an independent sequence of random variables with zero mean and unit variance that are generated independently from the samples. 
Our bootstrap AR test statistic is denoted as $\widehat{Q}^*(\beta_0)$ and defined as 
\begin{align}\label{eq:Qhat*}
\widehat{Q}^*(\beta_0) = \frac{\sum_{i \in [n]} \sum_{j \in [n], j \neq i} \eta_{i} e_{i}(\beta_0)\Xi_{\lambda,ij} \eta_{j} e_{j}(\beta_0)}{\sqrt{R_\lambda}}.
\end{align}
Then, the bootstrap critical value is denoted as $\widehat{\mathcal C}^*_{\alpha}(\beta_0)$ and defined as the $(1-\alpha)$-th percentile of $\widehat{Q}^*(\beta_0)$ conditional on data, where $\alpha$ is the nominal level of rejection under the null. We reject the null hypothesis of $\beta = \beta_0$ if $\widehat{Q}(\beta_0) > \widehat{\mathcal C}^*_{\alpha}(\beta_0).$

\begin{rem}\label{rem: boot}
Unlike the first term of $\widehat Q(\beta_0)$ defined in \eqref{Q_hat_statistic_definition}, we use $\Xi_\lambda$ instead of $P_\lambda$ to define the bootstrap AR statistic. Note that under the null, $e(\beta_0) = M_W \tilde e$, whose elements are not independent from each other. When the dimension of controls $d_w$ diverges at a rate $\sqrt{n}$ or higher, such a cross-sectional dependence is not asymptotically negligible. However, the bootstrap multipliers $\{\eta_i\}_{i \in [n]}$ are independent and, thus, unable to mimic the dependence. Instead, we explicitly account for this difference by adjusting the middle matrix $P_\lambda$ in the original statistic to $\Xi_\lambda$, so that valid bootstrap inference can be achieved under many controls. Additionally, we impose the null on the bootstrap data-generating process, following the recommendations in the literature of bootstrap for IV regressions or non-homoskedastic errors, such as \cite{Cameron2008}, 
\cite{Davidson-Mackinnon2010}, \cite{roodman2019}, and \cite{mackinnon2023fast}, among others.    
\end{rem}

\begin{rem}\label{rem:intuition}
As pointed out by \citet{anatolyev2011} and \cite{MS22}, when $K$ is fixed, no regularization is used, and the errors are homoskedastic, the jackknife AR test statistic admits the usual re-centered chi-squared approximation:
\begin{align*}
\frac{\widehat{Q}(\beta_0)}{c_n} \convD \frac{\chi^2_K - K}{\sqrt{2K}}
\end{align*}
for some normalization scalar $c_n$ computed under homoskedasticity. 

Furthermore, \citet{MS22} noted that this re-centered chi-squared distribution converges quickly to the standard normal distribution as $K$ increases. This suggests that critical values based on $\frac{\chi^2_K - K}{\sqrt{2K}}$ remain valid whether $K$ is fixed or diverging, making it a dimension-agnostic strong approximation for (the re-scaled) $\widehat{Q}(\beta_0)$ under homoskedasticity. In this paper, we extend this idea to the heteroskedastic setting by deriving a weighted re-centered chi-squared approximation for $\widehat{Q}(\beta_0)$ and establishing conditions under which a bootstrap critical value yields valid inference uniformly across different asymptotic regimes. In doing so, we also accommodate a diverging number of controls and ridge regularization, which allows the number of instruments $K$ to exceed the sample size and provides power improvement as well.
\end{rem}

\begin{rem}
Our proposed bootstrap test is AR-based. It is possible to extend our dimension-agnostic inference procedure to score-based Lagrangian Multiplier (LM) tests provided that the first-stage residual $\tilde v$ is consistently estimable. Given the consistency of residuals, we conjecture that our bootstrap inference remains valid for score-based statistics, including the cases where the effect of the endogenous variable $\tilde X$ may be heterogeneous and the structural equation \eqref{eq:model} is thus misspecified.\footnote{In such settings of heterogeneous treatment effects, especially when the number of instruments diverges with the sample size, researchers typically assume that the reduced-form models for both endogenous variables $\tilde Y$ and $\tilde X$ are linear; see, for example, \cite{kolesar2018}, \cite{EK18}, \cite{boot2024}, and \cite{yap2024inference}. In such cases, we require the consistency of reduced-form residuals for the bootstrap validity.} Specifically, this may require restricting the dimension of $(W,\tilde Z)$ to be of a smaller order of $n$,  imposing some sparsity conditions, and/or assuming that the reduced form regressions for $(\tilde Y, \tilde X)$  are approximately linear. One advantage of our AR-based inference procedure is that it imposes minimal assumptions on the first stage. For instance, we do not have any restriction on $\tilde\Pi$, aligned closely with the setting in \cite{MS22}.
\end{rem}

%Similarly, we define the bootstrap test statistic for overidentifying restrictions as 
%\begin{align}\label{eq:Qhat*}
%\widehat{Q}^*(\hat\beta) = \frac{\sum_{i \in [n]} \sum_{j \in [n], j \neq i} \eta_{i} e_{i}(\hat \beta)P_{\lambda,ij} \eta_{j} e_{j}(\hat \beta)}{\sqrt{R_\lambda}},
%\end{align}
%and reject the null hypothesis of $||\pi||^2_2 = 0$ if 
%\begin{align*}
%    \widehat{Q}(\hat \beta) > \widehat{\mathcal C}^*_{\alpha}(\hat %\beta), 
%\end{align*}
%where $\widehat{\mathcal C}^*_{\alpha}(\hat \beta)$ denotes the $(1-\alpha)$-th percentile of $\widehat{Q}^*(\hat \beta)$ conditional on data.

%%%%%%%%%%%%%%%%%%%%%%%%%%  strong approximation  %%%%%%%%%%%%%%%%%
\section{Strong Approximation of the Test Statistic} \label{section:strong_approx}
This section is concerned with the conditions under which the null distribution of the test statistic defined in \eqref{Q_hat_statistic_definition} can be approximated by its bootstrap counterpart, no matter whether the dimension $K_n$ of the IVs is fixed or diverging with the sample size. We make the following assumptions on the data-generating process (DGP) to establish this result. 

\begin{ass}
   \begin{enumerate}
   \item Let \eqref{eq:model} hold in which $W$ and $Z$ are treated as fixed (i.e., non-random or conditioned upon), $\{\tilde e_i,\tilde v_i\}_{i \in [n]}$ are independent, mean zero, but potentially heteroskedastic. 
       \item There exist constants $C \in (0,\infty)$ and $q > 6$ such that $\max_{i \in [n]} \mathbb  E (\tilde e_{i}^{2q} + X_i^{2q}) \leq C$. 
           \item Let $\tilde \sigma_i^2 = \mathbb E \tilde e_i^2$. Then, there exist constants $\infty > \bar c > \underline c>0$ such that 
           $$\bar c \geq \max_{i \in [n]} \tilde \sigma_i^2 \geq  \min_{i \in [n]} \tilde \sigma_i^2 \geq \underline c.$$ 
            \item The matrix $W^\top W$ is invertible and $\max_{i \in [n]} P_{W,ii} = o(1)$, where $P_{W,ii}$ denotes the i-th diagonal element of the projection matrix $P_W$. 
       \item Let $p_n = \max_{i \in [n]} \frac{  \sum_{j \in [n], j \neq i} \Xi_{\lambda,ij}^2}{R_\lambda}$ and ${p_n'} = \max_{i \in [n]} \frac{ P_{\lambda,ii}^2}{R_\lambda}$.       Then, we have $p_n n^{3/q} = o(1)$ and ${p_n'} (1+\sum_{i \in [n]} P_{W,ii}^2) = o(1)$.    
       % \item $R_\lambda \geq c>0$ for some constant $c>0$. 
       \item Let $\{\eta_i\}_{i \in [n]}$ be i.i.d. and independent of data, have mean zero, unit variance, and sub-Gaussian tail in the sense that $\inf\left\{u>0: \mathbb{E}\exp\left(\frac{|\eta|}{u}\right)^2\leq 2\right\} \leq C <\infty$ for some fixed constant $C \in (0,\infty)$. 
   \end{enumerate} \label{ass:reg}
\end{ass}

Assumptions \ref{ass:reg}.1--\ref{ass:reg}.3 are standard regularity conditions. Assumption \ref{ass:reg}.4 allows the dimension of control variables to diverge at a rate that is slower than the sample size, i.e., $d_w = o(n)$. The impact of partialling out $W$ from both $Y$ and $X$ becomes asymptotically negligible only when $d_w = o(\sqrt{n})$, reflecting a broader phenomenon commonly referred to as the \textit{quadratic barrier}. See, for example, \citet{CJM18} and \citet{LSMPW24} for further discussions. We overcome this barrier and establish bootstrap validity by carefully debiasing the AR statistic and further adjusting the middle matrix of the bootstrap quadratic form (as noted in Remark \ref{rem: boot}). 
To the best of our knowledge, this is the mildest rate condition regarding the number of controls established for bootstrap inference with high-dimensional IVs (without imposing a sparsity assumption).
We note that analytical inference remains feasible even when $d_w$ is proportional to $n$, as demonstrated in \citet{Anatolyev-Solvsten2023}. However, in such a high-dimensional control setting, our current bootstrap inference procedure may fail to control size. At present, it is unclear whether any valid resampling-based inference method exists in the proportion regime, let alone one that remains valid uniformly over the dimensions of both $Z$ and $W$. We leave this important question for future research. In the following, we provide further comparisons between analytical and bootstrap inference approaches in Remarks \ref{rem:AS} and \ref{rem:KSS}.

% Assumptions \ref{ass:reg}.1--\ref{ass:reg}.3 are standard regularity conditions. Assumptions \ref{ass:reg}.4 allows for the dimension of controls to diverge with the sample size. The effect of partialling out $W$ from both $Y$ and $X$ is asymptotically negligible only when $d_w = o(\sqrt{n})$, which is known as the \textit{quadratic barrier}. See, for example, \cite{CJM18} and \cite{LSMPW24} for more discussions. We breach this barrier by carefully debiasing and adjust the middle matrix for the quadratic form. It is possible to make analytical inference when $d_w$ is proportional to $n$, as illustrated by \cite{Anatolyev-Solvsten2023}. However, the current bootstrap procedure may be unable to control size in this case. It is also unclear whether there exist any valid bootstrap or resampling inference methods. This is left for future research. 

Assumption \ref{ass:reg}.5 requires that $p_n$ and ${p_n'}$ vanish sufficiently fast. Consider the case without ridge regularization (i.e., $\lambda=0$) and where the projection matrix is well-defined (i.e., $K_n < n$). If the diagonal elements of $P_\lambda$ (with $\lambda = 0$) are well-balanced in the sense that $P_{\lambda,ii} = K_n/n$, then we have 
\begin{align*}
    R_\lambda \geq C K_n \quad \text{and} \quad \max_{i \in [n]}    \sum_{j \in [n], j \neq i} \Xi_{\lambda,ij}^2 \leq C K_n/n.
\end{align*}
This implies $p_n = O(n^{-1})$ and ${p_n'} = O(n^{-1})$. Importantly, we note that these results hold  regardless of whether $K_n$ is fixed or increasing with $n$. If $P_W$ is also well-balanced such that $\max_{i \in [n]} P_{W,ii} \leq Cd_w/n$, then
\begin{align*}
    {p_n'} (1 + \sum_{i \in [n]}P_{W,ii}^2) \leq C(1/n +d_w^2/n^2) = o(1)
\end{align*}
as long as $d_w = o(n)$. In the minimum, even we only have ${p_n'} =o(1)$, if $d_w = O(\sqrt{n})$, then $\sum_{i \in [n]}P_{W,ii}^2 = O(1)$, which still guarantees that ${p_n'} (1 + \sum_{i \in [n]}P_{W,ii}^2)=o(1)$.

These calculations imply that our inference procedure remains valid even when the number of control variables diverges at the rate $\sqrt{n}$ or faster. Moreover, in high-dimensional settings where $K_n > n$, our ridge-regularized approach with the choice of $\lambda$ in \eqref{eq:lambda} ensures that Assumption \ref{ass:reg}.5 holds provided $q>6$.

Assumption~\ref{ass:reg}.5 is closely related to Assumptions~2 and~3 in \cite{DKM24}, but it accommodates a broader set of scenarios, which lead to different limiting behaviors of the test statistic. In particular, it allows for cases in which $K_n \to \infty$ while
\[
%\max_{i \in [n]} \sum_{j \in [n],\, j \neq i} \Xi_{\lambda,ij}^2
\sum_{i \in [n]} \sum_{j \in [n],\, j \neq i} \Xi_{\lambda,ij}^2
\]
remains bounded, as well as cases in which $K_n$ is bounded. These settings give rise to the second and third asymptotic regimes analyzed in Section~\ref{section: asy_power}, respectively.

Finally, Assumption \ref{ass:reg}.6 requires the bootstrap weights $\eta_i$ to have sub-Gaussian tails. In practice, we recommend using standard normal or Rademacher random variables, both of which satisfy this condition. We use Rademacher random variables in our simulations and the empirical application.\footnote{The recommendation follows those in the literature on wild/multiplier bootstrap methods. In particular, \cite{Davidson-Flachaire2008} show favorable finite-sample properties of the Rademacher weights, and \citet[section 3.4]{roodman2019} use the Rademacher weights as their default choice.}
\vspace{0.1in}

To proceed, we need to introduce some more notation. Define $\Delta =  \beta - \beta_0$, $\tilde \tau_i = \mathbb E(\tilde e_i \tilde v_i)$, $\tilde \varsigma_i^2 = \mathbb E \tilde v_i^2$, $\breve e_i(\beta_0) = \tilde e_i(\beta_0) + \Pi_i \Delta$, and $\tilde e_i(\beta_0) = \tilde e_i +  \tilde v_i \Delta $. Then, we denote  
\begin{align*}
& \tilde \sigma_{i}^{2}(\beta_0) = Var( \breve e_i(\beta_0)) = \mathbb E \tilde e_i^2(\beta_0)  = \tilde \sigma_i^2 + 2 \Delta \tilde \tau_i + \Delta^2\tilde \varsigma_i^2, \\
& \breve \sigma_{i}^{2}(\beta_0) = \mathbb E \breve e_i^2(\beta_0)  = \tilde \sigma_i^2 + 2 \Delta \tilde \tau_i + \Delta^2(\tilde \varsigma_i^2 + \Pi_{i}^2) = \tilde \sigma_{i}^{2}(\beta_0) + \Pi_i^2 \Delta^2.
\end{align*}
In addition, let 
\begin{align}
Q(\beta_0) & = \frac{\sum_{i \in [n]} \sum_{j \in [n], j \neq i} (g_i\tilde \sigma_i(\beta_0)+\Pi_i \Delta) \Xi_{\lambda,ij} (g_j\tilde \sigma_j(\beta_0)+\Pi_j \Delta) }{\sqrt{R_\lambda}}  \label{eq:Q} 
\end{align}
and 
\begin{align}
& Q^*(\beta_0) = \frac{\sum_{i \in [n]} \sum_{j \in [n], j \neq i} g_{i}\breve \sigma_{i}(\beta_0) \Xi_{\lambda,ij} g_{j} \breve \sigma_{j}(\beta_0)}{\sqrt{R_\lambda}}, \label{eq:Q^*}
\end{align}
where $\{g_{i}\}_{i \in [n]}$ are i.i.d. standard normal random variables that are generated independent of data. The following theorem shows that our proposed AR test statistic $\widehat Q(\beta_0)$ can be strongly approximated by $Q(\beta_0) + C(\Delta)$ in Kolmogorov distance, where 
\begin{align*}
C(\Delta)  =   \frac{\sum_{i \in [n]} \sum_{j \in [n], j \neq i} \Pi_i  \left(P_{\lambda,ij} - \Xi_{\lambda,ij}\right) \Pi_j \Delta^2 }{\sqrt{R_\lambda}} . 
\end{align*}
Furthermore, Theorem \ref{thm:Fhat-F} in the next section shows the bootstrap statistic $\widehat Q^*(\beta_0)$ can be strongly approximated by $Q^*(\beta_0)$ in Kolmogorov distance conditionally on data. 
% Both strong approximations hold under the null and alternative in \eqref{null_hypothesis}. 
Note that $Q^*(\beta_0)$ is equal to $Q(\beta_0)$ under the null hypothesis.\footnote{Under the null, we have $\Delta=0$ and $\tilde \sigma_i(\beta_0) = \breve \sigma_i(\beta_0)$.} 

\begin{thm}
    Suppose Assumption \ref{ass:reg} holds, and $||\Pi||_2^2 \Delta^2 / \min\left(R_\lambda^{1/2},R_\lambda^{2/3}\right)$ is bounded. Then, we have
\begin{align*}
  \sup_{y\in \Re}\left|\mathbb P(\widehat Q(\beta_0) \leq y)- \mathbb P(Q(\beta_0) + C(\Delta)\leq y)\right|  = o(1).
\end{align*}\label{thm:main_null}
\end{thm}

\begin{rem}
We note that $Q(\beta_0)$ is implicitly indexed by the sample size $n$, which explains why we call it a strong approximation rather than a limit of our AR statistic $\widehat Q(\beta_0)$. 
Second, as noted in Remark \ref{rem: boot}, the cross-sectional dependence between the elements of $e(\beta_0)$ is not asymptotically negligible when $d_w$ diverges at a rate $\sqrt{n}$ or higher. 
On the other hand, $\{g_{i}\}_{i \in [n]}$ in $Q(\beta_0)$ and $Q^*(\beta_0)$ are i.i.d. standard normal random variables. We account for this by adjusting $P_\lambda$ in the original statistic to $\Xi_\lambda$ in \eqref{eq:Q} and \eqref{eq:Q^*}. Third, we can see that 
\begin{align}\label{eq:noncentrality}
    \mathbb E Q(\beta_0) + C(\Delta) = \frac{\sum_{i \in [n]} \sum_{j \in [n], j \neq i} \Pi_i P_{\lambda,ij} \Pi_j \Delta^2 }{\sqrt{R_\lambda}}, 
\end{align}
which is the non-centrality parameter for the AR statistic under the alternative. 
\end{rem}

\begin{rem}\label{rem:AS}
The strong approximation $Q(\beta_0) + C(\Delta)$ encompasses three asymptotic regimes in a unified framework:
(1) when both $K$ and $R_\lambda$ are bounded, 
$Q(\beta_0) + C(\Delta)$ asymptotically follows a weighted non-central chi-squared distribution; 
(2) when both $K$ and $R_\lambda$ diverge so that a Lindeberg-type condition holds, it converges in distribution to a normal random variable; and (3) 
%when the number of IVs diverges but the Lindeberg-type condition fails, 
when $K$ diverges but $R_\lambda$ is bounded, it converges to a mixture of a weighted sum of non-central chi-squared distributions and a normal distribution. These three regimes are discussed separately by \citet[Sections 4, 5, and 6]{KSS2020} in the setting of estimation of variance components. For testing linear restrictions, \cite{Anatolyev-Solvsten2023} proposed an analytical inference procedure that is valid under regimes (1) and (2). However, in their setting, where ridge regularization is not employed, the third regime does not arise. 
A key advantage of our bootstrap procedure and the associated strong approximation results is that they are valid irrespective of the asymptotic regime, including the challenging case with a mixture of distributions. We will provide further details on the regimes in Section \ref{section: asy_power}.

% The strong approximation $Q(\beta_0) + C(\Delta)$ covers three asymptotic limits in a unified manner: (1) the number of IVs is fixed so that $Q(\beta_0) + C(\Delta)$ asymptotically behaves as a weighted non-central chi-squared random variable, (2) the number of IVs diverges and some Lindeberg-type condition holds so that $Q(\beta_0) + C(\Delta)$ asymptotically behaves like a normal random variable, and (3) the number of IVs diverges but the Lindeberg-type condition fails so that $Q(\beta_0) + C(\Delta)$ asymptotically behaves like a mixture of a weighted non-central chi-squared random variable and a normal random variable.  
\end{rem}

\begin{rem}
Theorem \ref{thm:main_null} is valid under both the null and alternative hypotheses. The nature of the alternatives depends on the magnitude of $||\Pi||_2^2 \Delta^2 / \min\left(R_\lambda^{1/2},R_\lambda^{2/3}\right)$. When  $R_\lambda$ is bounded, we have weak (strong) identification when the concentration parameter $||\Pi||_2^2$ is bounded (diverging). When $R_\lambda$ is diverging, as shown by \cite{MS22}, weak (strong) identification arises when the concentration parameter $||\Pi||_2^2/\sqrt{R_\lambda}$ is bounded (diverging). Under either regime with regard to $R_{\lambda}$, Theorem \ref{thm:main_null} accommodates (i) fixed alternatives under weak identification and (ii) local alternatives scaled by the square root of the concentration parameter under strong identification.
\end{rem}

\section{Strong Approximation of the Bootstrap Statistic}\label{sec: strong_approx_boot}
This section concerns the strong approximation of the bootstrap statistic defined in Section \ref{sec:bootstrap CV} in Kolmogorov distance conditionally on data. 
The approximation in Theorem \ref{thm:Fhat-F} is the same as that for the original statistic under the null hypothesis, as established in Theorem \ref{thm:main_null}, which directly implies that the proposed test with bootstrap critical values achieves a correct asymptotic size. Such a result holds no matter whether the dimension of IVs $K_n$ is fixed or diverging to infinity. 

\begin{thm}
Let $\mathcal D$ denote all observations in our sample. Suppose Assumption \ref{ass:reg} holds and $\Delta$ is bounded. Then, we have
\begin{align*}
\sup_{y \in \Re} |\mathbb P( \widehat Q^*(\beta_0) \leq y | \mathcal D) - \mathbb P(  Q^*(\beta_0) \leq y) | = o_P(1).
\end{align*}

\label{thm:Fhat-F}
\end{thm}
\begin{rem}\label{rem:41}
Theorem \ref{thm:Fhat-F} remains valid under both the null and alternative hypotheses. In contrast to Theorem \ref{thm:main_null}, it accommodates fixed alternatives even in the presence of strong identification (without requiring $||\Pi||_2^2 \Delta^2 / \min\left(R_\lambda^{1/2},R_\lambda^{2/3}\right)$ to be bounded). This distinction originates from the fact that the alternative hypothesis $\Delta$ affects $Q(\beta_0)$ and $Q^*(\beta_0)$ differently -- introducing non-centrality bias in the former and variance in the latter (e.g., see the non-centrality bias in \eqref{eq:noncentrality} and the definition of $Q^*(\beta_0)$ in \eqref{eq:Q^*}, respectively). The distinction also underpins the power of our dimension-agnostic AR test, which will be analyzed in detail in Section \ref{section: asy_power}.\footnote{Similar phenomenon of an inflated variance under the alternative hypothesis also occurs with the jackknife AR tests using analytical variance estimators that impose the null hypothesis (e.g., \cite{crudu2021}, \cite{MS22}, and \cite{DKM24}) so that the resulting tests can be robust to weak identification.}
\end{rem}

\begin{rem}
Theorem \ref{thm:Fhat-F} can be viewed as a general result of strong approximation for the bootstrap version of the quadratic forms. The proof extends the Lindeberg swapping strategy mentioned in the previous section. Indeed, compared with Theorem \ref{thm:main_null}, it is substantially more involved to establish Theorem \ref{thm:Fhat-F} because the second moment of the bootstrap statistic $\widehat Q^*(\beta_0)$ conditional on data is random and does not exactly match that of its strong approximation (i.e., $Q^*(\beta_0)$). We rely on the concentration inequalities for quadratic forms (i.e., Hanson-Wright inequality) and linear forms of martingale difference sequence to bound the approximation error in Kolmogorov distance due to the mismatch of the second moments. This technique seems new to the literature and may be of independent interest. 
\end{rem}

\begin{rem}
In addition, we observe from \eqref{eq:Q}-\eqref{eq:Q^*} that $Q(\beta_0)$ and $Q^*(\beta_0)$ have the same marginal distribution under the null hypothesis ($\Delta=0$). This means that, under the null, the bootstrap statistic closely approximates the test statistic in Kolmogorov distance when conditioned on the data, whether $K_n$ is fixed or diverging.  
This equivalence forms the basis for our bootstrap test to achieve the correct size. 
\end{rem}

% \begin{ass}\label{ass:den}
%  Denote $\mathcal C_{\alpha}(\beta_0) = \inf\{y \in \Re: 1-\alpha \leq F_{\beta_0}(y)\}$, where $F_{\beta_0}(y) = \mathbb P(Q(\beta_0) \leq y)$. Let the $\eps$-neighborhood around $\mathcal C_{\alpha}(\beta_0)$ be $\mathbb B(\mathcal C_{\alpha}(\beta_0), \eps)$. Then, under the null, the density of $Q(\beta_0)$ exists in the neighborhood $\mathbb B(\mathcal C_{\alpha}(\beta_0), \eps)$ and is denoted as $f_{n}(\cdot)$.  In addition, there exits an $\eps>0$ such that $\liminf_{n \rightarrow \infty}    \inf_{y \in \mathbb B(\mathcal C_{\alpha}(\beta_0), \eps)}f_{n}(y) \geq \underline c >0,$
% where $\underline c>0$ is a fixed constant.  
% \end{ass}

% \begin{rem}
% We discuss three asymptotic regimes for power analysis in Section \ref{section: asy_power}. In each regime, the limiting distribution of $Q(\beta_0)$ is either normal, weighted chi-squared, or a mixture of the two. This ensures that Assumption \ref{ass:den} holds automatically in all three regimes.
% \end{rem}

\begin{thm}\label{cor:size}
Suppose we are under the null hypothesis $\beta = \beta_0$ and Assumption \ref{ass:reg} holds. Then, we have
\begin{align*}
    \mathbb P(\widehat{Q}(\beta_0) > \widehat{\mathcal C}^*_{\alpha}(\beta_0)) \rightarrow \alpha,
\end{align*}
where $\widehat{\mathcal C}^*_{\alpha}(\beta_0)$ is the $(1-\alpha)$-th percentile of $\widehat{Q}^*(\beta_0)$ conditional on data.
\end{thm}

\section{Asymptotic Power}\label{section: asy_power}
In this section, we discuss the power of the bootstrap inference by focusing on three separate cases: (I) both $K$ and $R_\lambda$ diverge, (II) $K$ diverges but $R_\lambda $ is bounded, and (III) both $K$ and $R_\lambda$ are bounded.  
%Given that $R_\lambda$ is bounded from above by $\min(K_n,n)$, case (1) corresponds to the case where $K_n$ diverges, potentially even faster than $n$. Alternatively, when $K_n = K$ is fixed,  $R_\lambda$ must be bounded and, thus, belongs to case (2). 

%%%%%%%%%%%%%%%%%%%%%% Power for diverging K
\subsection{The Case with Diverging \texorpdfstring{$K$ and $R_\lambda$}{K and R-lambda}}\label{subsec: power_case1}

%Let $$\Psi(\beta_0) = \frac{2 \sum_{i \in [n]} \sum_{j \in [n], j \neq i}  \breve %\sigma^2_i(\beta_0) P_{\lambda,ij}^2 \breve \sigma_j^2(\beta_0) }{R_\lambda },$$ 
%\begin{align*}
%    \breve \sigma_i^2(\beta_0) = Var(\tilde e_i + \tilde v_i \Delta) = \tilde %\sigma_i^2 + 2 \Delta \tilde \tau_i + \Delta^2 \tilde \varsigma_i^2, 
%\end{align*}

To proceed, we let  
$\Psi(\beta_0) = 2 \sum_{i \in [n]} \sum_{j \in [n], j \neq i}  \tilde \sigma^2_i(\beta_0) \Xi_{\lambda,ij}^2 \tilde \sigma_j^2(\beta_0)/R_\lambda$, where $\tilde \sigma_i^2(\beta_0) = Var(\tilde e_i (\beta_0) ) = \tilde \sigma_i^2 + 2 \Delta \tilde \tau_i + \Delta^2 \tilde \varsigma_i^2$.

\begin{ass}\label{ass:alt_divK}
\begin{enumerate}
    \item $K \rightarrow \infty$, $R_{\lambda} \rightarrow \infty$,  and $||\Pi||_2^2 \Delta^2/\sqrt{R_{\lambda}}$ is bounded. 
    \item $\Delta$ and $\max_{i \in [n]}|\Pi_i|$ are bounded. 
    \item $\Psi^{-1/2}(\beta_0) \frac{ \sum_{i \in [n]} \sum_{j \in [n], j \neq i} \Pi_i P_{\lambda,ij} \Pi_j \Delta^2 }{\sqrt{R_\lambda} } \rightarrow \mu(\beta_0)$.  
\end{enumerate}    
\end{ass}

\begin{thm}\label{thm:power_divK}
Suppose Assumptions \ref{ass:reg} and \ref{ass:alt_divK} hold. Then, we have
\begin{align*}
  \mathbb P(\widehat Q(\beta_0) >  \widehat C^*_\alpha(\beta_0)) \rightarrow \mathbb P \left( \N(\mu(\beta_0),1) >  z_\alpha \right),    
\end{align*}
where $\N(\mu,1)$ is a normal random variable with mean $\mu$ and unit variance and $z_\alpha$ is the $(1-\alpha)$ quantile of a standard normal random variable. 
\end{thm}

\begin{rem}\label{rem:lindeberg}
Let us denote $(\varpi_1,\cdots,\varpi_n)$ as the eigenvalues of the matrix
\begin{align*}
\diag(\tilde \sigma_1(\beta_0),\cdots,\tilde \sigma_n(\beta_0))  \Xi_{\lambda}  \diag(\tilde \sigma_1(\beta_0),\cdots,\tilde \sigma_n(\beta_0)),
\end{align*}
ordered such that $|\varpi_1| \geq |\varpi_2| \geq \cdots \geq |\varpi_n|$. From the proof of Theorem \ref{thm:power_divK}, we note that under the null,
\begin{align*}
\widehat Q(\beta_0) = \frac{\sum_{i=1}^n (g_i^2 -1)\varpi_i}{\sqrt{R_{\lambda}}}  + o_P(1),
\end{align*}
where $\{g_i\}_{i \in [n]}$ is an i.i.d. sequence of standard normal random variables.  Furthermore, we have
\begin{align*}
\varpi_1 = O(1) \quad \text{and} \quad \sum_{i \in [n]} \varpi_i^2 \geq \underline c R_\lambda,
\end{align*}
for some constant $\underline c > 0$. This implies when $R_\lambda \to \infty$,
\begin{align*}
\frac{\varpi_1^2}{\sum_{i \in [n]} \varpi_i^2 } = o(1),
\end{align*}
which is a Lindeberg-type condition that guarantees asymptotic normality of the test statistic, as established in Theorem \ref{thm:power_divK}.
\end{rem}

\begin{rem}\label{rem: power}
When $d_w = o(\sqrt{n})$, Theorem \ref{thm:power_divK} holds if we replace $\Xi_{\lambda}$ by $P_{\lambda}$ in the definition of $\Psi(\beta_0)$ as the effect of partialling out controls is asymptotically negligible. If $K_n < n$ and we set $\lambda = 0$ so that $P_\lambda = P$ (i.e., without ridge regularization), the local power of our test is asymptotically equivalent to that of the jackknife AR tests proposed by \cite{crudu2021} and \cite{MS22}.\footnote{\cite{crudu2021} and \cite{MS22} proposed different variance estimators for the jackknife AR statistic. Under local alternatives characterized under Assumption \ref{ass:alt_divK}, the two variance estimators are asymptotically equivalent.} 
% On the other hand, simulations in Section \ref{section: monte-carlo-simulation} suggest that our test with the proposed regularizer in \eqref{eq:lambda} can achieve substantial power improvement over the jackknife AR tests under many IVs.\footnote{Additionally, it may be interesting to study if our regularizer can improve the power of other tests with quadratic forms such as those in Remark \ref{rem: many-reg}. We leave this line of investigation for future research.} 

In general, the regularizer $\lambda$ can affect the power through $\mu(\beta_0)$, which depends on $P_{\lambda}$ and $R_\lambda$. 
Specifically, following Remark \ref{rem:41}, we note that the alternative $\Delta$
affects the limiting distribution through (1) the non-centrality bias of the test statistic $\widehat Q(\beta_0)$, given by 
\begin{align}\label{eq: non-central}
\frac{ \sum_{i \in [n]} \sum_{j \in [n], j \neq i} \Pi_i P_{\lambda,ij} \Pi_j \Delta^2 }{\sqrt{R_\lambda} } \equiv \mathbb{C}_\lambda \Delta^2, 
\end{align}
where $\mathbb{C}_\lambda$ denotes the concentration parameter under the ridge regularization, and (2) the variance of the statistic, captured by $\Psi(\beta_0) = \tilde \sigma_i^2 + 2 \Delta \tilde \tau_i + \Delta^2 \tilde \varsigma_i^2$. Both components contribute to the mean $\mu(\beta_0)$ of the limiting distribution.

Furthermore, we note that \eqref{eq: non-central} also motivates our choice of the regularizer $\lambda$. 
In particular, as we restrict the upper bound $\bar\theta$ for the regularizer to be $||Z^\top Z||_{op}$, the ridge regularization $\lambda I_{K_n}$ will not dominate $Z^\top Z$. Then it is plausible that the numerator of the concentration parameter, i.e., $\sum_{i,j \in [n]^2, i \neq j}\Pi_i P_{\lambda, ij} \Pi_j$ does not change order for the range of $\lambda$ we consider. On the other hand, as $\lambda$ increases, the effective rank $R_\lambda$ decreases, which typically causes the non-centrality in \eqref{eq: non-central} to increase and thus lead to power improvement. Notice that it is possible for 
$\mathbb{C}_\lambda$ in \eqref{eq: non-central} to achieve a higher order of magnitude than the concentration parameter without ridge regularization 
\begin{align}\label{eq: non-central-no-ridge}
\frac{ \sum_{i \in [n]} \sum_{j \in [n], j \neq i} \Pi_i P_{ij} \Pi_j }{\sqrt{K}}, 
\end{align}
as long as $R_\lambda = o(K)$ (e.g., $K$ diverges but $R_\lambda$ is fixed). 
Such an advantage of regularization has been pointed out in previous studies, such as \cite{carrasco2015, carrasco2016efficient} and \cite{carrasco2017}.
In the next section, we further study in detail the case where $R_\lambda$ is bounded while $K$ diverges.
\end{rem}

\subsection{The Case with Diverging \texorpdfstring{$K$ but Bounded $R_\lambda$}{K but Bounded R-lambda}}\label{subsec:power_case2} 
Following Remark \ref{rem:lindeberg}, we now consider the case where $R_\lambda$ remains bounded, resulting in the failure of the Lindeberg-type condition for asymptotic normality.

\begin{ass}\label{ass:alt_divK2}
    \begin{enumerate}
        \item Suppose there exists a fixed positive integer $R$ such that 
        \begin{align*}
            \frac{\varpi_i}{ (\sum_{j \in [n]} \varpi_i^2)^{1/2} } \rightarrow r_i \neq 0, \quad \forall i = 1,\cdots,R, \quad \text{and} \quad \frac{\varpi_{R+1}^2}{ \sum_{i \in [n]} \varpi_i^2 } = o(1). 
        \end{align*}
        \item Denote $(\varpi_1^*,\cdots,\varpi_n^*)$ as the eigenvalues of the matrix
\begin{align*}
\diag(\breve \sigma_1(\beta_0),\cdots,\breve \sigma_n(\beta_0))  \Xi_{\lambda}  \diag(\breve \sigma_1(\beta_0),\cdots,\breve \sigma_n(\beta_0)),
\end{align*}
ordered such that $|\varpi_1^*| \geq |\varpi_2^*| \geq \cdots \geq |\varpi_n^*|$. Suppose there exists a fixed positive integer $R^*$ such that 
        \begin{align*}
            \frac{\varpi_{i}^*}{ \left(\sum_{j \in [n]} \varpi_i^{*2} \right)^{1/2} } \rightarrow r_i^* \neq 0, \quad \forall i = 1,\cdots,R^*, \quad \text{and} \quad \frac{\varpi_{R^*+1}^{*2}}{ \sum_{i \in [n]} \varpi_i^{*2} } = o(1). 
        \end{align*}
\item Suppose        $||\Pi||_2^2 \Delta^2/\sqrt{R_{\lambda}}$, $\Delta$, and $\max_{i \in [n]}|\Pi_i|$ are bounded. 
\item $\Psi^{-1/2}(\beta_0) \frac{ \sum_{i \in [n]} \sum_{j \in [n], j \neq i} \Pi_i P_{\lambda,ij} \Pi_j \Delta^2 }{\sqrt{R_\lambda} } \rightarrow \mu(\beta_0)$ and $\frac{\breve \Psi(\beta_0)}{\Psi(\beta_0)} \rightarrow \psi(\beta_0)>0$, where   
\[\breve \Psi(\beta_0) = 2 \sum_{i \in [n]} \sum_{j \in [n], j \neq i}  \breve \sigma^2_i(\beta_0) \Xi_{\lambda,ij}^2 \breve \sigma_j^2(\beta_0)/R_\lambda.\]
    \end{enumerate}
\end{ass}

\begin{thm}\label{thm:power_divK_fixedKlambda}
Suppose Assumptions \ref{ass:reg} and \ref{ass:alt_divK2} hold. Then,  we have
\begin{align*}
  \mathbb P(\widehat Q(\beta_0) >  \widehat C^*_\alpha(\beta_0)) \rightarrow \mathbb P \left(\chi(\{r_i\}_{i \in [R]}) + \mu(\beta_0) >   \psi^{1/2}(\beta_0) \mathcal C_\alpha (\{r_i^*\}_{i \in [R^*]}) \right), 
\end{align*}
where the random variable $\chi(\{r_i\}_{i \in [R]})$  has the distribution
$$\frac{\sum_{ i \in [R]} (g_i^2 -1) r_i}{\sqrt{2}} +  \left(1-\sum_{i \in [R]}r_i^2 \right)^{1/2} g_{R+1},$$
with $\{g_i\}_{i \in [R+1]}$ being i.i.d. standard normal random variables, and $\mathcal C_\alpha (\{r_i\}_{i \in [R]})$ is the $(1-\alpha)$-th quantile of $\chi(\{r_i\}_{i \in [R]})$.
\end{thm}

\begin{rem}\label{rem:KSS}
When the Lindeberg-type condition fails due to $R_\lambda$ being bounded, the limiting distribution of our test statistic becomes a mixture of a weighted sum of chi-squared random variables and a standard normal random variable. This is similar to the scenario described in \citet[Sections 6 and 7]{KSS2020} and \cite{YGZ24}. Analytical inference in this regime is difficult, as it requires estimating the number of dominant eigenvalues $R$ driving the asymptotic distribution or reporting (the union of) confidence intervals corresponding to consecutive values of $R$ (see, e.g., Section 7.2 of \cite{KSS2020}).
A key advantage of our bootstrap inference procedure is that it does not require prior knowledge of the number of dominant eigenvalues $R$ or associated weights $\{r_i\}_{i \in [R]}$, since it is valid regardless of the asymptotic regime. In our simulations in Section \ref{section: monte-carlo-simulation} with $K = 160$ and our choice of $\lambda$, we observe one dominant eigenvalue ($R= 1$) and $r_1 = \sqrt{0.948}$.
Furthermore, in the empirical application in Section \ref{empirical-application}, we find that $R_\lambda$ equals 2.015 and 1.550 for the specifications with 38 and 342 IVs, respectively, suggesting that this regime also applies to our application of \cite{card2009}'s dataset.

\end{rem}

\begin{rem}
As mentioned earlier, the alternative $\Delta$ affects the location and scale of the test statistic and the bootstrap critical value, represented by $\mu(\beta_0)$ and $\psi(\beta_0)$, respectively. When $R_\lambda$ diverges, the scale effect becomes asymptotically negligible, as indicated by $\psi(\beta_0) = 1$ in Theorem \ref{thm:power_divK}. However, when $R_\lambda$ is bounded, $\psi(\beta_0)$ may differ from one, and the scale effect remains relevant in the limiting distribution. 
\end{rem}

\begin{rem}
Furthermore, we note that the ridge-regularized concentration parameter $\mathbb{C}_\lambda$ in \eqref{rem: power} can achieve a higher divergence rate than that without regularization, given that $R_\lambda$ is bounded while $K$ diverges. In particular, as established by \citet[Theorems 1 and 4]{MS22}, $\sum_{i \in [n]}\sum_{j \in [n], j\neq i}\Pi_i P_{ij}\Pi_j /\sqrt{K} \rightarrow \infty$ is required for the jackknife AR test (without regularization) to be consistent. By contrast, with the help of regularization, our test only requires 
$\sum_{i \in [n]}\sum_{j \in [n], j\neq i}\Pi_i P_{\lambda,ij}\Pi_j \rightarrow \infty$ to be consistent if $R_\lambda$ is bounded.  
\end{rem}

\subsection{The Case with Bounded \texorpdfstring{$K$ and $R_\lambda$}{K and R-lambda}} \label{subsection:power_fixed_instruments}
In this section, we consider the power property of our bootstrap AR test in the asymptotic framework that the dimension of $Z$ (i.e., $K_n = K$) is fixed. To rigorously state the regularity conditions, we recall the singular value decomposition of $Z$ as
$Z = \mathcal U \mathcal S \mathcal V^\top,$ 
where $\mathcal U \in \Re^{n \times n}$, $\mathcal U^\top \mathcal U = I_n$, $\mathcal S = [ S_0, 0_{K, n-K}]^\top$, $S_0$ is a diagonal matrix of non-zero singular values, $0_{K, n-K} \in \Re^{K \times (n-K)}$ is a matrix of zeros, $\mathcal V \in \Re^{K \times K}$, and $\mathcal V^\top \mathcal V = I_K$. Denote $\mathcal U = [\mathcal U_1,\mathcal U_2]$ such that $\mathcal U_1 \in \Re^{n \times K}$, $\mathcal U_2 \in \Re^{n \times (n-K)}$, $\mathcal U_1^\top \mathcal U_1 = I_K$, $\mathcal U_1^\top \mathcal U_2 = 0_{K,n-K}$, and $\mathcal U_2^\top \mathcal U_2 = I_{n-K}$. Further denote $\Omega(\beta_0) \equiv \mathcal U_1^\top  \diag(\tilde \sigma_1^2(\beta_0),\cdots,\tilde \sigma_n^2(\beta_0) ) \mathcal U_1$ and the eigenvalue decomposition 
\begin{align*}
\lim_{n \rightarrow \infty}\frac{   \Omega^{1/2}(\beta_0) S_0 (S_0^2 + \lambda I_K)^{-1} S_0  \Omega^{1/2}(\beta_0) }{\sqrt{R_\lambda}} = \mathbb U \diag(\omega_1,\cdots,\omega_K) \mathbb U^\top.
\end{align*}
Last, denote $\nu(\beta_0) = \lim_{n \rightarrow \infty} \mathbb U^\top \Omega^{-1/2}(\beta_0) \Delta \mathcal U_1^\top \Pi$. 

\begin{ass}\label{ass:alt_fixK}
\begin{enumerate}
    \item Suppose the IVs $Z$ have a fixed dimension $K$, $\left(\max_{i \in [n]}P_{\lambda,ii}\right) d_w  = o(1)$, and $R_\lambda \geq c$ for some constant $c>0$. 
    \item $\max_{i \in [n]}||\mathcal{U}_{1,i}||_2 = o(1)$, where $\mathcal{U}_{1,i}^\top \in \Re^{1 \times K}$ is the $i$-th row of $\mathcal{U}_1$. 
    \item $||\Pi||_2^2 \Delta^2 /\sqrt{R_\lambda}$ is bounded.
\end{enumerate}
\end{ass}

The following theorem establishes our AR test's power property in the fixed $R_\lambda$ scenario. 
\begin{thm}\label{thm:power_fixK}
Suppose Assumptions  \ref{ass:reg} and  \ref{ass:alt_fixK} hold. Then, we have
\begin{align*}
  \mathbb P(\widehat Q(\beta_0) >  \widehat C^*_\alpha(\beta_0)) \rightarrow \mathbb P \left(\sum_{k\in [K]} \omega_k \chi_k^2(\nu_k^2(\beta_0)) > \mathcal C_{\omega}(1-\alpha) \right), 
\end{align*}
where $\omega = (\omega_1,\cdots,\omega_K)$, $\{ \chi_k^2(\nu_k^2(\beta_0))\}_{k \in [K]}$ is a sequence of independent non-central chi-squared random variables with one degree of freedom and noncentrality parameter $\nu_k^2(\beta_0)$, $\nu_k(\beta_0)$ is the $k$-th element of $\nu(\beta_0)$, $\mathcal C_{\omega}(1-\alpha)$ is the $(1-\alpha)$ quantile of a weighted chi-squared random variable $\sum_{k \in [K]} \omega_k \chi^2_k$, and $\{\chi^2_k\}_{k \in [K]}$ is a sequence of i.i.d. centered chi-squared random variables with one degree of freedom. 
\end{thm}

Next, we demonstrate that when $K$ is fixed, our dimension-agnostic AR test is (asymptotically) admissible within a specific class of tests, which includes the standard (heteroskedasticity-robust) AR test designed for fixed $K$. 
%This result implies that the power of our dimension-agnostic AR test is not uniformly dominated by the standard AR test. Conversely, as the standard AR test is also admissible in this class, it is not uniformly dominated by our test either.
Let 
\begin{align}\label{eq:mathcalGhat}
    \widehat {\mathcal G} (\beta_0) = \mathbb U \hat \Omega^{-1/2}(\beta_0) \mathcal {U}_1^\top e(\beta_0), 
\end{align}
and $\widehat {\mathcal G}_k (\beta_0)$ be the $k$-th element of $\widehat {\mathcal G} (\beta_0)$, where $\hat \Omega(\beta_0)$ is a consistent estimator of $\Omega(\beta_0)$. We observe that, in the scenario where $K$ is fixed, the standard AR test rejects if 
\begin{align}\label{eq:ARstd}
    \widehat {\mathcal G}^\top (\beta_0) \widehat {\mathcal G} (\beta_0) = \sum_{k\in [K]} \widehat {\mathcal G}_k^2 (\beta_0) > \mathcal C_{\iota_K}(1-\alpha),
\end{align}
where $\iota_K$ is a $K$-dimensional vector of ones and $\mathcal C_{\iota_K}(1-\alpha)$ is just the $(1-\alpha)$ quantile of the centered chi-squared random variable with $K$ degrees of freedom. On the other hand, our bootstrap AR test is asymptotically equivalent to a test that rejects if
\begin{align*}
 \sum_{k\in [K]} \omega_k \widehat {\mathcal G}_k^2 (\beta_0)  > \mathcal C_{\omega}(1-\alpha).  
\end{align*}
In addition, the proof of Theorem \ref{thm:power_fixK} shows 
\begin{align*}
\left( \widehat {\mathcal G}_1^2 (\beta_0),\cdots,\widehat {\mathcal G}_K^2 (\beta_0) \right) \convD \left( \chi_1^2(\nu_1^2(\beta_0)),\cdots,\chi_K^2(\nu_K^2(\beta_0)) \right). 
\end{align*}

We consider the class $\Phi_\alpha$ of tests $\phi(\cdot)$ defined as 
\begin{align*}
\Phi_{\alpha} = \begin{Bmatrix}
\phi(\cdot): \Re^K \mapsto [0,1], \quad \mathbb{E}\phi(\chi_1^2(\nu_1^2(\beta_0)),\cdots,\chi_K^2(\nu_K^2(\beta_0))) \leq \alpha, \\
\text{when $\nu_k^2(\beta_0) = 0$, $k=1,\cdots,K$,}\\
\text{the set of discontinuities of $\phi(\cdot)$ has zero Lebesgue measure}
\end{Bmatrix}.
\end{align*}
Both the standard and bootstrap AR tests control size, and thus, belong to this class.  The power of any test $\phi(\cdot) \in \Phi_\alpha$ is determined by $\nu(\beta_0) \in \Re^K$.

% We also note that 
% $
% \left( \widehat {\mathcal G}_1^2 (\beta_0),\cdots,\widehat {\mathcal G}_K^2 (\beta_0) \right) \convD \left( \chi_1^2(\nu_1^2(\beta_0)), \cdots, \chi_K^2(\nu_K^2(\beta_0)) \right),
% $
% where $\{ \chi_k^2(\nu_k^2(\beta_0)\}_{k \in [K]}$  is defined in Theorem \ref{thm:power_fixK}. This means the power of any test $\phi(\cdot) \in \Phi_\alpha$ is determined by $\nu(\beta_0) \in \Re^K$. 

\begin{thm}\label{cor:admissible}
Suppose Assumptions  \ref{ass:reg} and  \ref{ass:alt_fixK} hold. In addition,  let $\widehat {\mathcal G}(\beta_0)$ be defined in \eqref{eq:mathcalGhat}, $\hat \Omega(\beta_0) \convP \Omega(\beta_0)$, and $0<c \leq \lambda_{\min}\left(\Omega(\beta_0) \right) \leq  \lambda_{\max}\left(\Omega(\beta_0) \right) \leq C < \infty.$
Then, our bootstrap test $\phi_0 = 1\{ \widehat{Q}(\beta_0) > \widehat{\mathcal C}^*_{\alpha}(\beta_0)\}$ is asymptotically admissible w.r.t. $\Phi_\alpha$ in the sense that if there exists a test $\phi^* \in \Phi_\alpha$ such that for all values of $\nu(\beta_0) \in \Re^K$, 
\begin{align*}
  \lim_{n \rightarrow \infty}  \mathbb E \phi^*(\widehat {\mathcal G}_1^2 (\beta_0),\cdots,\widehat {\mathcal G}_K^2 (\beta_0)) \geq  \lim_{n \rightarrow \infty}  \mathbb E \phi_0, 
\end{align*}
then we must have 
\begin{align*}
  \lim_{n \rightarrow \infty}  \mathbb E \phi^*(\widehat {\mathcal G}_1^2 (\beta_0),\cdots,\widehat {\mathcal G}_K^2 (\beta_0)) =  \lim_{n \rightarrow \infty}  \mathbb E \phi_0, 
\end{align*}
for all $\nu(\beta_0) \in \Re^K$. 

\end{thm}

\begin{rem}
It is reasonable to assume there exists a consistent estimator $\hat \Omega(\beta_0)$ for $\Omega(\beta_0)$, which is a $K \times K$ matrix with $K$ fixed. 
\end{rem}

\begin{rem}
Because the standard AR test defined in \eqref{eq:ARstd} belongs to $\Phi_\alpha$, Theorem \ref{cor:admissible} implies our bootstrap test $\phi_0$ is not dominated by the standard AR test for all alternatives. In fact, the standard AR test is also admissible among the tests in $\Phi_\alpha$ so that it is not dominated by $\phi_0$ either. However, our bootstrap test is dimension-robust, while the standard AR test does not have the correct size under the regimes in Sections \ref{subsec: power_case1} and \ref{subsec:power_case2}.
\end{rem}

\begin{rem}
Under a fixed number of IVs and strong identification against local alternatives, the K test proposed by \cite{Kleibergen2002} exhibits better power properties than our test. However, similar to the standard AR test, the K test is not dimension-robust. Indeed, \cite{LWZ2023} propose a counterpart of the K test tailored to the setting with many weak instruments and heteroskedastic errors, although this test may be invalid when the number of IVs is fixed. Moreover, both the K test and its many-weak-IV counterpart suffer from power ditches and therefore have no power against certain fixed alternatives, even under strong identification, see, for example, Section~3.1 of \cite{A16} and Lemma~2.3 of \cite{LWZ2023}. Additionally, \cite{chernozhukov2009} establish admissibility of the weighted average power likelihood ratio tests derived by \cite{Andrews-Moreira-Stock2006} under a fixed $K$. In particular, they show that the class of tests covered by this admissibility result contains the standard AR test, so that there is no ``waste of degrees of freedom" in the conventional statistical sense.

% On the other hand, our bootstrap AR test is dimension-agnostic and does not have the non-monotone power issue.
\end{rem}

\begin{rem}
\cite{N23} proposed a dimension-robust version of the K test, which de-correlates the endogenous variable $X_i$ and outcome error $e_i$ conditionally on $Z_i$. This approach requires consistently estimating the conditional correlation $\rho(Z_i) = \mathbb E(X_i e_i|Z_i)$. However, when the dimension of $Z_i$ is large, in general, $\rho(Z_i)$ cannot be consistently estimated. Instead, \cite{N23} imposes a sparsity condition and estimates $\rho(Z_i)$ by an $\ell_1$-regularized regression. According to his simulations, the dimension-robust K test can also suffer from the power ditch issue due to the (null-imposed) decorrelation. Unlike \citeauthor{N23}'s (\citeyear{N23}) procedure, our test achieves robustness against the dimension of IVs without imposing any additional structure. Furthermore, if one is comfortable with imposing the sparsity assumption on $\rho(Z_i)$, then it is possible to combine our test and \citeauthor{N23}'s (\citeyear{N23}) K test (e.g., by constructing a dimension-robust version of the conditional linear combination test in \cite{LWZ2023}, which is efficient under strong identification and also solves the power ditch issue).

\end{rem}

%%%%%%%%%%%%%%%%%%%%%%%%%%%%% asymptotic power

%%%%%%%%%%%%%%% Simulations
\section{Monte Carlo Simulations} \label{section: monte-carlo-simulation}
This section investigates the finite sample size and power performance of existing tests and our proposed test. %under two different data-generating processes (DGPs). 
To begin, we explicitly define these tests and their corresponding critical values. In addition, following \cite{belloni2012} and \cite{DKM24}, upon obtaining a given instrument set $Z$, we standardize it by $\frac{1}{n} \sum_{i=1}^n Z_{ij}^2 = 1$, for
$j=1,...,K.$
Note that the tests described in section \ref{subsection:description_of_tests} below are based on the standardized $Z$. 
Throughout the simulations, we set the number of Monte Carlo and bootstrap replications equal to $5,000$ and $10,000$ respectively, and set the nominal level $\alpha=0.05$. 

\subsection{Description of Tests} \label{subsection:description_of_tests}
Specifically, we consider the following twelve tests: 

\begin{enumerate}
    \item[(1)] BS:  Our bootstrap test based on \eqref{Q_hat_statistic_definition} and \eqref{eq:Qhat*}, which rejects $H_0$ whenever 
   $\widehat{Q}(\beta_0) > \widehat{\mathcal C}^*_{\alpha}(\beta_0),$
and we let the upper bound defined in \eqref{eq:lambda} be $\overline{\theta} \equiv || Z^{\top}Z||_{op}$;
\item[(2)] JAR$_{\rm std}$: The jackknife AR test based on \cite{crudu2021}'s standard variance estimator for diverging $K$, which rejects $H_0$ whenever 
\begin{align*}
\frac{1}{\sqrt{\widehat{\Phi}^{std}(\beta_0)} \sqrt{K}} \sum_{i \in [n]} \sum_{j \in [n], j \neq i} P_{ij} e_i(\beta_0) e_j(\beta_0) > q_{1-\alpha} \left(\mathcal{N}(0,1)\right),
\end{align*}
where $\widehat{\Phi}^{std}(\beta_0) := \frac{2}{K}\sum_{i \in [n]}\sum_{j \neq i}P_{ij}^2e_i^2(\beta_0)e_j^2(\beta_0)$ and $P_{ij}$ denotes the $(i,j)$ element of $P:= Z(Z^{\top}Z)^{-1}Z^{\top}$;\footnote{Note that this statistic is slightly different from the one proposed by \cite{crudu2021}, in that they replace $P_{ij}$ by $C_{ij}$, where $C$ is defined in Section 3.2 of their paper.}

\item[(3)] AR$_{\rm Robust}$: \citet[Section 5.2]{mikusheva2024weak}'s Robust AR test, which is based on a cross-fit variance estimator for many instruments/controls and rejects $H_0$ whenever
\begin{align*}
\frac{1}{\sqrt{K \widehat{\Phi}_W}} (Y- \beta_0 X)^\top A (Y - \beta_0 X) > q_{1-\alpha} \left(\mathcal{N}(0,1)\right), 
\end{align*}
where $\widehat{\Phi}_W := \frac{2}{K}\sum_{i,j \in [n]} \frac{A_{ij}^2}{M_{ZW,ii} M_{ZW,jj} + M_{ZW,ij}^2} \widehat\sigma^2_i \widehat\sigma_j^2$, $\widehat\sigma_i^2 = \sum_{k \in [n]} M_{ZW,ik} (Y_i - \beta_0 X_i)(Y_k - \beta_0 X_k)$, $M_{ZW,ij}$ are the elements of $M_{ZW}$, a projection matrix orthogonal to both $Z$ and $W$, $A_{ij}$ are the elements of the matrix $A$ given by \cite{Chao2023}, which has the form $A = M_W (P - D_\zeta) M_W$, and $D_\zeta$ is a diagonal matrix with elements $\zeta_1,...,\zeta_{n}$ selected in such a way that $A$ has zero diagonals;%\footnote{In the simulations, the cross-fit variance estimator $\widehat{\Phi}_W$ can be negative at times. To ensure the AR$_{\rm Robust}$ test is well-defined, we set the variance estimator to be $\max\left(\widehat{\Phi}_W,\frac{1}{\sqrt{n \log(n)}}\right)$.}

\item[(4)] AR: The classical heteroskedasticity-robust AR test for fixed $K$, rejecting $H_0$ whenever
\begin{align*}
   J_n^{\top}(\beta_0) \widehat{\Omega}_n(\beta_0)^{-1} J_n(\beta_0) > q_{1-\alpha}(\chi^2_K), \;\; 
\end{align*}
where $J_n(\beta_0) := n^{-1/2} Z^{\top} e(\beta_0)$ and $\widehat{\Omega}_n(\beta_0):= n^{-1} Z^{\top} \{ diag(e_1^2(\beta_0),...,e_n^2(\beta_0) ) \} Z$;

\item[(5)] RJAR: The ridge-regularized jackknife AR test for diverging $K$ proposed by \cite{DKM24}, which rejects $H_0$ whenever
\begin{align*}
    \frac{1}{\sqrt{\widehat{\Phi}_{\gamma_n^*}(\beta_0)} \sqrt{r_n}} \sum_{i \in [n]} \sum_{j \in [n], j \neq i} P_{\gamma_n^*,ij} e_i(\beta_0) e_j(\beta_0) > q_{1-\alpha} \left( \mathcal{N}(0,1) \right),
\end{align*}
where $P_{\gamma_n^*,ij}$ denotes the $(i,j)$ element of $P_{\gamma_n^*}:= Z(Z^{\top}Z + \gamma_n^* I_K)^{-1}Z^{\top}$, $r_n := rank(Z)$, $\widehat{\Phi}_{\gamma_n^*}(\beta_0) := \frac{2}{r_n} \sum_{i \in [n]}\sum_{j \neq i}(P_{\gamma_n^*,ij})^2e_i^2(\beta_0)e_j^2(\beta_0)$,
$\gamma_n^* := \max \argmax_{\theta \in \Gamma_n} \sum_{i \in [n]} \sum_{j \neq i} P_{\theta,ij}^2$,
and $  \Gamma_n := \{  \gamma_n \in \mathbb{R}: \gamma_n \geq 0 \text{ if } r_n = K, \text{ and } \gamma_n \geq 1 \text{ if } r_n < K   \} $;

\item[(6)] $\rm BCCH_{1.1}$: \cite{belloni2012}'s sup-score test, which rejects $H_0$ whenever
\begin{align*}
    \max_{1 \leq j \leq K} \frac{ \left| \sum_{i \in [n]} e_i(\beta_0) Z_{ij} \right| }{ \sqrt{\sum_{i \in [n]} e_i^2(\beta_0) Z_{ij}^2 }} > c_{_{BCCH}} q_{1-\alpha/(2K)}(\mathcal{N}(0,1)),
\end{align*}
where $c_{BCCH}>0$ is a tuning constant and we let $c_{_{BCCH}} = 1.1$, following \citet[Section 2.3]{belloni2012}'s recommendation;

\item[(7)] $\rm BCCH_{1.01}$: This is the BCCH test in (6) with $c_{ BCCH} = 1.01$ (instead of $1.1$).

\item[(8)] CT: The ridge-regularized AR test proposed by \cite{carrasco2016}, which rejects $H_0$ whenever
\begin{align*}
    \frac{n e(\beta_0)^{\top}P_{0.05} e(\beta_0)}{e(\beta_0)^{\top}(I_n - P_{0.05})e(\beta_0)}  > \widehat{\mathcal C}^*_{\alpha, CT}(\beta_0),
\end{align*}
where $\widehat{\mathcal C}^*_{\alpha, CT}(\beta_0)$ denotes the bootstrap critical value discussed in Section 3 of their paper, and the choice of the fixed scalar 0.05 for the regularizer (which does not depend on $n$) follows that used in the simulations of \cite{DKM24}.\footnote{\cite{carrasco2016} show that under homoskedastic errors, their test statistic converges to an infinite sum of weighted $\chi^2_1$ distributions. For inference, they proposed a residual bootstrap procedure, which is based on the empirical distribution of residuals.}

\item[(9)] LM: The jackknife LM test for diverging $K$ proposed by \cite{matsushita2024}, which rejects $H_0$ whenever
\begin{align*}
    \frac{1}{\sqrt{\widehat\Psi(\beta_0)} \sqrt{K}} \sum_{i \in [n]} \sum_{j \neq i} P_{ij}X_i  e_j(\beta_0) > q_{1-\alpha}\left( \mathcal{N}(0,1)  \right),
\end{align*}
where $\widehat\Psi(\beta_0) := \frac{1}{K} \left( \sum_{i \in [n],j \neq i} P_{ij} X_j e_i^2(\beta_0)+\sum_{i \in [n],j \neq i} P_{ij}^2 X_i X_j e_i(\beta_0)e_j(\beta_0) \right);$
%\begin{align*}
%    \Psi(\beta_0) := \frac{\sum_{i \in [n],j \neq i} P_{ij} X_j e_i^2(\beta_0)+\sum_{i \in [n],j \neq i} P_{ij}^2 X_i X_j e_i(\beta_0)e_j(\beta_0)}{K};
%\end{align*}

\item[(10)] AS: The dimension-robust $F$ test proposed by \cite{Anatolyev-Solvsten2023}, which rejects $H_0$ whenever 
\begin{align*}
    F > \widehat{C}_{\alpha, AS},
\end{align*}
where $F$ and $\widehat C_{\alpha,AS}$ denote the $F$-test statistic and the critical value described in Sections 2.1 and 2.3, respectively, in \cite{Anatolyev-Solvsten2023};\footnote{The code for their test can be found at \url{https://github.com/mikkelsoelvsten/manyRegressors/blob/master/R/LOFtest.R}. Translating our model to that of \cite{Anatolyev-Solvsten2023}, our structural equation of \eqref{eq:model} can be given by
%\begin{align*}\label{eq:AS}
$\tilde Y - \tilde X \beta = W \Gamma + \tilde Z \Theta + \tilde e,$
%\end{align*}
where the AS test corresponds to testing $\Theta = 0_{K_n}$ under the null hypothesis $\beta = \beta_0$.
%and $W \in \mathbb{R}^{n \times d_w}$. 
In terms of the notation in Section 2 of their paper, $y_i = \widetilde{Y}_i - \widetilde X_i \beta_0$ and $x_i = (W_i^\top, Z_i^\top)^\top$ with $H^{AS}_0:
\textbf{R} \pmb{\beta}^{AS} = \textbf{q},$ where $\pmb{\beta}^{AS} =(\Gamma^\top,\Theta^\top)^\top$,
$\textbf{R} = \left[ \begin{array}{cc}
\mathbf{0}_{K \times d_w} & I_K
\end{array} \right]$,
and 
$\textbf{q} = \mathbf{0}_{K \times 1}$.
}

\item[(11)] Empirical: The bootstrap test based on our test statistic in \eqref{Q_hat_statistic_definition} but with its critical value generated by the empirical distribution of $e_i(\beta_0)$ instead, which rejects $H_0$ whenever 
   $$\widehat{Q}(\beta_0) > \widetilde{\mathcal C}^*_{\alpha}(\beta_0),$$
   where $\widetilde{\mathcal C}^*_{\alpha}(\beta_0)$ is the $(1-\alpha)$-th percentile of $\widetilde{Q}^*(\beta_0)$ conditional on data, 
   %$(1-\alpha)$-quantile of the collection $\{\widetilde{Q}_\ell(\beta_0)\}_{\ell \in [B]}$ for some large integer $B$,\footnote{We set $B=10,000$ in simulations} 
 $$\widetilde{Q}^*(\beta_0):= \frac{1}{\sqrt{R_\lambda}}
 \sum_{i \in [n]} \sum_{j \in [n], j \neq i} e^*_{i}(\beta_0)\Xi_{\lambda,ij} e^*_{j}(\beta_0),$$
       %\frac{\sum_{i \in [n]} \sum_{j \in [n], j \neq i}\widetilde{e}_{i,\ell}(\beta_0) P_{\lambda,ij} \widetilde{e}_{j,\ell}(\beta_0)}{\sqrt{R_\lambda}} - \frac{\sum_{i,j \in [n]^2} \kappa_{ij} \widetilde{e}_{j,\ell}^2(\beta_0) A_{\lambda,ii}}{\sqrt{R_\lambda}}
   and $\{e^*_{i}(\beta_0)\}_{i \in [n]}$ is drawn from the empirical distribution of $\{e_i(\beta_0)\}_{i \in [n]}$. We use the same regularizer as that for the BS test in (1).
   %We let the upper bound defined in \eqref{eq:lambda} be $\overline{\theta} \equiv || Z^{\top}Z||_{op}$;

\item[(12)] JK: The jackknife K test proposed by \cite{N23}, which rejects $H_0$ whenever
\begin{align*}
    JK(\beta_0) > q_{1-\alpha}(\chi^2_{1}),
\end{align*}
with $JK(\beta_0)$ defined in (2.5) of his paper.

\end{enumerate}

%%%%%%%%%%%%%%%%%%%%%%%%%%%%
\subsection{Simulations Based on \texorpdfstring{\cite{Haus2012}}{Hausman et al. (2012)}}
\label{subsec: simulation_hausman}
In this section, we consider a model based on the DGP given by \cite{Haus2012}, with a sample size $n=200$ and a heteroskedastic error structure. 
\begin{align*}
    &Y =  X \beta + W \mathbb{\Gamma} + D_{z1} e, \quad X = Z \pi + U_2, \quad  \beta_0 = 0, \quad \mathbb{\Gamma}=\left(\frac{1}{\sqrt{d_W}},...,\frac{1}{\sqrt{d_W}}\right)^\top \in \mathbb{R}^{d_W}, \quad   \\
    & \text{where} \;\; W = 
\begin{bmatrix}
1 & w_{12} & w_{13} & \cdots & w_{1,d_W} \\
1 & w_{22} & w_{23} & \cdots & w_{2,d_W} \\
\vdots & \vdots & \vdots & \ddots & \vdots \\
1 & w_{n2} & w_{n3} & \cdots & w_{n,d_W}
\end{bmatrix},
\quad 
w_{ij} \stackrel{\text{i.i.d.}}{\sim} \mathcal{N}(0,1) \;\; \text{for} \;\; j \ge 2, \quad d_W = 15, \quad  \\
& D_{z1} = \diag(\sqrt{1+z_{11}^2},\sqrt{1+z_{21}^2},...,\sqrt{1+z_{n1}^2}), \\
    & e_i = \rho U_{2i} + \sqrt{\frac{1-\rho^2}{\phi^2 + 0.86^4} } \left( \phi v_{1i} + 0.86 v_{2i} \right), \; v_{1i} \sim z_{1i}(\rm{Beta}(0.5,0.5) - 0.5), \; v_{2i} \sim \mathcal{N}(0,0.86^2),\\
    & z_{i1} \sim \mathcal{N}(0.5,1), \quad U_{2i} \sim \rm{Exponential}(0.2)-5, \quad \phi = 0.3, \quad \text{and} \quad \rho = 0.3.
\end{align*}
We assume that the errors are independent across $i$. 
%The diagonal matrix $D_{z_1}$ allows the error term $e$ to be dependent on $z_1$ but at the same time has its variance bounded away from zero. %in the event some elements of $z_1$ are close to zero. 
We vary the number of instruments $K \in \{2,10,40,160,300\}$ and $\beta \in [-2,2]$ to investigate the size and power properties of the eleven tests under both fixed and diverging $K$ settings.
Specifically, the $i$th instrument observation for $K \geq 10$ is given by
$$Z_i^{\top} = (z_{i1},z_{i1}^2,z_{i1} \mathbb{1}(z_{i1}<q_{25}), z_{i1} \mathbb{1}(q_{25} \leq z_{i1} < q_{50}), z_{i1} \mathbb{1}(q_{50} \leq z_{i1} < q_{75}) ,z_{i1}D_{i1},...,z_{i1}D_{i,K-5}),$$
where $q_{\alpha}$ is the $\alpha$-percentile of $\{ z_{i1}\}_{i \in [n]}$, $D_{ik} \in \{0,1\}$ is a dummy variable that is independent across $(i,k)$ with $\mathbb{P}(D_{ik} = 1) = 1/2$, so that $Z_i \in \mathbb{R}^K$. Furthermore, for the case with $K=2$, we let 
$$Z_i^\top := (z_{i1},z_{i1}^2).$$ 
We define $\mu^2 := n \pi^{\top}\pi,$
and consider $\mu^2 = 72$ for $K=2$, while $\mu^2 = 8$ for $K \geq 10$, following \cite{Haus2012}.\footnote{Specifically, for $K=2$, we let $\pi = \frac{0.6}{\sqrt{K}} \iota_K$; for $K \geq 10$ we let $\pi = \frac{0.2}{\sqrt{K}} \iota_K$. We allow $\mu^2$ to be larger for $K=2$ to demonstrate a non-negligible power; otherwise, all the tests would have a trivial power.} %We provide a discussion of the size and power results below.

\textbf{Regularizers:} Recall that $\overline{\theta} = \|Z^{\top} Z\|_{op}$. 
For our ridge regularization parameter $\lambda$, when 
$K = 2, 10, 40, 160,$ and $300$, we have $\lambda = 200, 1234,3665,13791,125589$,   respectively. 
For the regularization parameter used by the RJAR procedure, when 
$K \in \{2, 10, 40, 160\}$ we obtain $\gamma_n^* = 0$ under this DGP, whereas 
$\gamma_n^* = 41$ when $K = 300$.

\begin{table}[ht]
\caption{\textbf{Null Rejection Probabilities}}
\vspace{3mm}
\begin{center}
\resizebox{\textwidth}{!}{
\begin{tabular}{ c| c c c c c c c c c c c c} 
\hline 
& $\rm RJAR$ & $\rm JAR_{\rm std}$ & $\rm AR_{\rm Robust}$  & $\rm AR$ & $\rm AS$ 
& $\rm BCCH_{1.1}$ & $\rm BCCH_{1.01}$ & $\rm CT$  
& $\rm Empirical$ & $\rm LM$ & $\rm JK$ & $\rm BS$ \\  
& (5\%) & (5\%) & (5\%) & (5\%) & (5\%) 
& (5\%) & (5\%) & (5\%)  
& (5\%) & (5\%) & (5\%) & (5\%) \\ [0.5ex] 
\hline
$K=2$ 
& \bf{0.107} & \bf{0.107} & \bf{0.096} & 0.065 & 0.069 
& 0.030 & 0.051 & 0.069 
& \bf{0.156} & 0.064 & \bf{0.097} & 0.068 \\ 
\hline
$K=10$ 
& \bf{0.108} & \bf{0.108} & \bf{0.082} & 0.046 & 0.053 
& \bf{0.011} & 0.034 & 0.062 
& \bf{0.122} & 0.048 & \bf{0.102} & 0.060 \\ [1ex] 
\hline
$K=40$ 
& \bf{0.187} & \bf{0.187} & \bf{0.096} & \bf{0.014} & 0.049 
& \bf{0.006} & 0.021 & 0.064 
& \bf{0.123} & 0.066 & \bf{0.109} & 0.059 \\ [1ex] 
\hline
$K=160$ 
& \bf{0.784} & \bf{0.784} & \bf{0.134} & \bf{0.000} & \bf{0.006} 
& \bf{0.001} & \bf{0.008} & \bf{0.638} 
& \bf{0.123} & \bf{0.478} & 0.064 & 0.057 \\ [1ex] 
\hline
$K=300$ 
& \bf{0.995} & -- & -- & -- & -- 
& \bf{0.002} & \bf{0.008} & \bf{1.000} 
& \bf{0.122} & -- & 0.061 & 0.061 \\ [1ex] 
\hline
\end{tabular}
}
\\
\end{center}
{\footnotesize{\textbf{Note:} We set the nominal level $\alpha = 0.05$. We highlight values with more than 3\% size distortions (under- or over-rejections). We round to 3 decimal places.}}
\label{table:size}
\end{table}

\vspace{2mm}
\textbf{Size Properties:} Table \ref{table:size} reports the null rejection probabilities of the twelve tests across different values of $K$. We make several observations below. First, the $\rm RJAR$, $\rm JAR_{std}$, $\rm AR_{\rm Robust}$, $\rm CT$, $\rm Empirical$, $\rm LM$, and $\rm JK$ tests suffer from over-rejections under some or all values of $K$. 
Second, the classical $\rm AR$ test for fixed $K$ and the $\rm AS$ test control size for all values of $K$, but become conservative when $K$ is large. 
Similarly, we observe that the $\rm BCCH_{1.1}$ test is relatively conservative across different numbers of IVs, while $\rm BCCH_{1.01}$ also becomes conservative when $K$ is large. By contrast, our proposed dimension-robust $\rm BS$ test largely resolves the size-distortion issues for all values of $K$ considered.
%This could be due to the $\rm CT$ test being constructed for homoskedastic errors, while the $\rm BCCH$ test requires sparsity, which does not hold for large $K$. 
Overall, our $\rm BS$ test has the best size properties among these tests. 

We observe that ${\rm AR_{Robust}}$ is undefined at times due to the estimated variance $\widehat\Phi_W$ being negative. As in \cite{DKM24}, we take the conservative approach and do not reject the null whenever $\widehat\Phi_W<0$. 

Several factors contribute to the deviation of empirical sizes from the nominal level in existing procedures. When $K$ is small, the limiting distributions of JAR$_{\rm std}$, AR$_{\rm Robust}$, RJAR, CT, and LM are not normal, so the use of normal critical values can lead to size distortions. When $K$ is large, some observations may have very high leverage. For instance, when $K = 160$ and $n = 200$, we observe that $\max_{i \in [n]} P_{ii} \approx 0.9595$. High leverage can lead to numerical problems with AR-type statistics, such as JAR$_{\rm std}$ and AR$_{\rm Robust}$.%\footnote{For example, the cross-fit variance estimator of \cite{mikusheva2024weak} may become negative.} 

RJAR and our statistics instead rely on ridge regularization to mitigate the large-leverage problem. Regularization, however, may lead to a non-Gaussian limiting distribution (corresponding to the asymptotic regime studied by Section \ref{subsec:power_case2}) and invalidate analytical inference procedures. We therefore employ bootstrap critical values, which remain valid irrespective of asymptotic regimes. 

Furthermore, the classical $\rm AR$ test suffers from the many-IV bias. The $\rm JAR_{std}$, $\rm RJAR$, $\rm CT$, $\rm LM$, and $\rm JK$ tests are valid when the number of controls is fixed but may suffer from the many-control bias. The $\rm AR_{\rm Robust}$, $\rm BCCH$, and $\rm AS$ procedures are designed for data-rich environments. However, in our simulations $\rm AR_{\rm Robust}$ tends to over-reject, whereas $\rm AS$ tends to under-reject, due to the numerical issues discussed above. The $\rm BCCH$ procedure relies on Bonferroni-adjusted critical values, and is conservative with the default tuning constant $1.1$. Using a smaller constant instead, such as $1.01$, can mitigate the conservative size, which also leads to power improvement.

Additionally, we note that the normalization by the effective rank is less crucial for our size control, since our inference is based on the bootstrap. Specifically, as can be seen from \eqref{Q_hat_statistic_definition} and \eqref{eq:Qhat*}, any such normalization appears in both the test statistic and the bootstrap critical value, and therefore cancels out. However, this is not the case for analytical inference procedures that rely on critical values derived from the limiting distribution. This provides another advantage of our bootstrap-based inference.

The empirical bootstrap (``Empirical'') exhibits relatively large size distortions compared with our multiplier bootstrap, as it is not robust to heteroskedasticity. Let $\mathbb{E}^*$ denote the conditional expectation under the residual bootstrap. The conditional variance of $\widetilde{Q}^*(\beta_0)$ is
\begin{align}
\frac{2}{R_{\lambda}}
\sum_{i \in [n]} \sum_{j \in [n], j \neq i}
\mathbb{E}^*\!\left[
e_i^{*2}(\beta_0)\,
\Xi_{\lambda,ij}^2\,
e_j^{*2}(\beta_0)
\right]
=
2\left(
\frac{1}{n}\sum_{i \in [n]} e_i^2(\beta_0)
\right)^2,
\label{eq: boot_1}
\end{align}
where the equality follows from conditional independence of $e_i^*(\beta_0)$ and $e_j^*(\beta_0)$ for $i \neq j$ and from the bootstrap sampling scheme that draws $e_i^*(\beta_0)$ from the empirical distribution of $\{e_i(\beta_0)\}_{i \in [n]}$. In general, \eqref{eq: boot_1} does not approximate $\Psi(\beta_0)$, the variance of the statistic $\widehat{Q}(\beta_0)$, under heteroskedastic errors. By contrast, the multiplier bootstrap remains valid under heteroskedasticity across a wide range of $R_{\lambda}$ regimes.

\vspace{2mm}
\textbf{Power Properties:} Figures \ref{figure:Hausman_final_a_b} and \ref{figure:Hausman_final_c_d} report the power curves for $K = 10, 40, 160,$ and $300$ instruments, respectively. The power curve for $K = 2$ is reported in Figure \ref{figure:Hausman_K_1} in the Supplemental Appendix. Several remarks are in order. First, $\rm JAR_{\rm std}$ and $\rm RJAR$ exhibit identical power curves for $K \in \{2, 10, 40, 160\}$ because the regularization parameter selected by RJAR satisfies $\gamma_n^* = 0$ under the current DGP. Second, the power of $\rm AR_{\rm Robust}$ is relatively low when $K = 160$. Third, when the number of instruments is moderate to large, such as $K = 40$ or $160$, the classical $\rm AR$, $\rm CT$, $\rm LM$, $\rm AS$, and $\rm JK$ tests generally exhibit low power. Fourth, the $\rm JK$ test has relatively low power when $K = 10$ and $40$, but performs substantially better when $K = 160$ and $300$. Fifth, under the current DGP, $\rm BCCH_{1.1}$ exhibits relatively low power with the default tuning constant $c_{\rm BCCH} = 1.1$. In contrast, $\rm BCCH_{1.01}$ achieves substantially higher power, particularly when $K = 300$. 

Lastly, we emphasize that the power comparisons reported here are specific to this simulation design. Other procedures may exhibit superior power performance relative to our bootstrap AR test under alternative designs. For example, we expect that the sup-score-based BCCH procedure will have a power advantage when the first-stage signal is sparse. It is also possible to combine our AR test with a sup-score-based test to enhance power, in the spirit of \cite{fan_liao_yao15}.

\begin{figure}[!htbp]    \includegraphics[width=1\textwidth,height = 18cm]{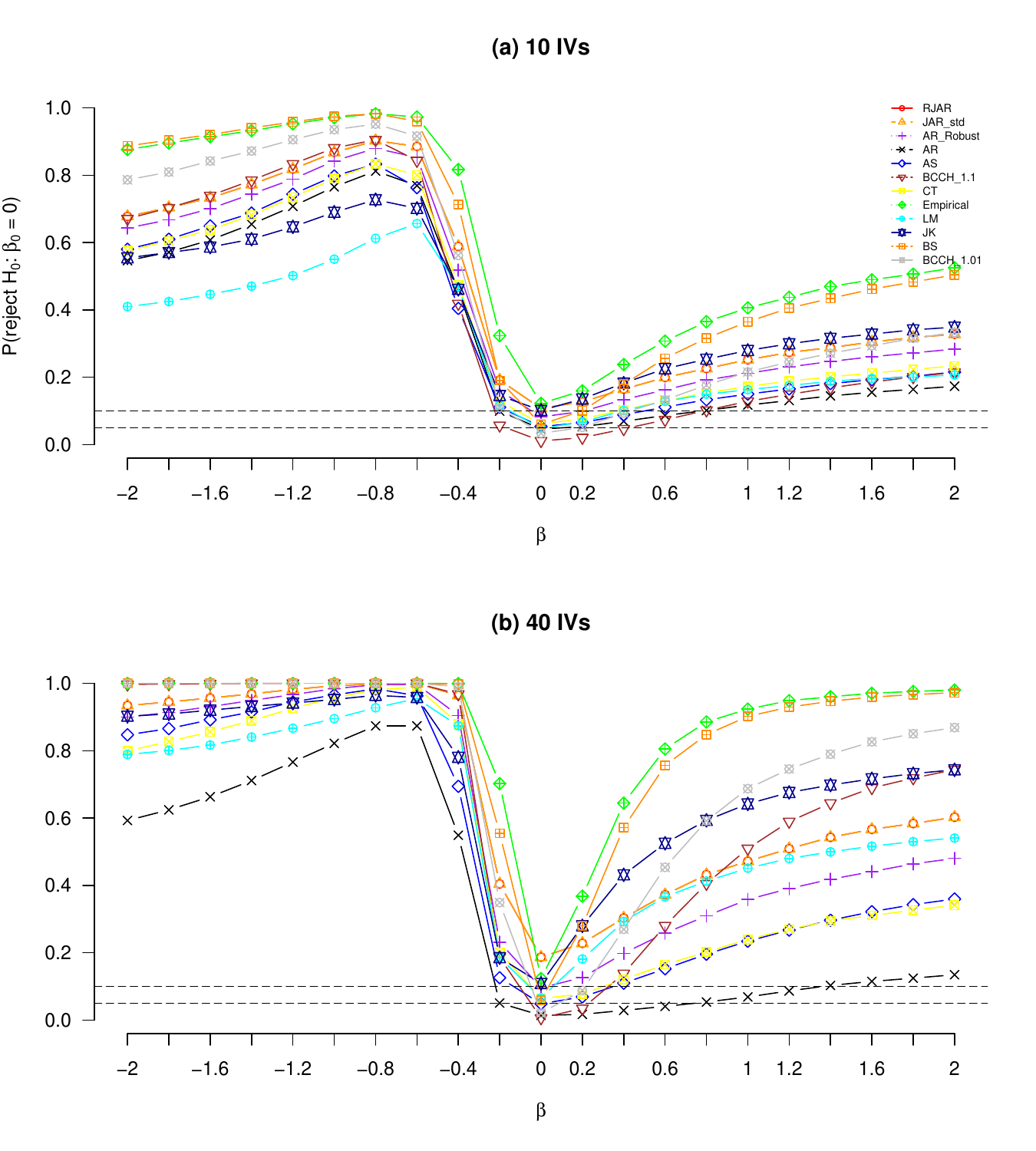}
    \caption{Power curves with $10$ and $40$ IVs}
     \footnotesize{\textbf{Note:} The red curve with a hollow circle represents $\rm RJAR$; the orange curve with an upward triangle represents $\rm JAR_{\rm std}$; the purple curve with a cross represents $\rm AR_{\rm Robust}$; the black curve with $\rm \times$ represents $\rm AR$; the blue curve with diamond represents $\rm AS$; the brown curve with inverted triangle represents $\rm BCCH_{1.1}$; the yellow curve with a filled square represents $\rm CT$; the green curve with a filled diamond represents $\rm Empirical$; the cyan curve with a filled circle represents $\rm LM$; the dark-blue curve with hexagram represents $\rm JK$; the dark-orange curve with the $+$ in the square-box represents $\rm BS$. The grey curve with a filled $\rm \times$ in the circle represents $\rm BCCH_{1.01}$. The horizontal dotted black lines represent the $5\%$ and $10\%$ levels.}
    \label{figure:Hausman_final_a_b}
\end{figure}

\begin{figure}[!htbp]    \includegraphics[width=1\textwidth,height = 18cm]{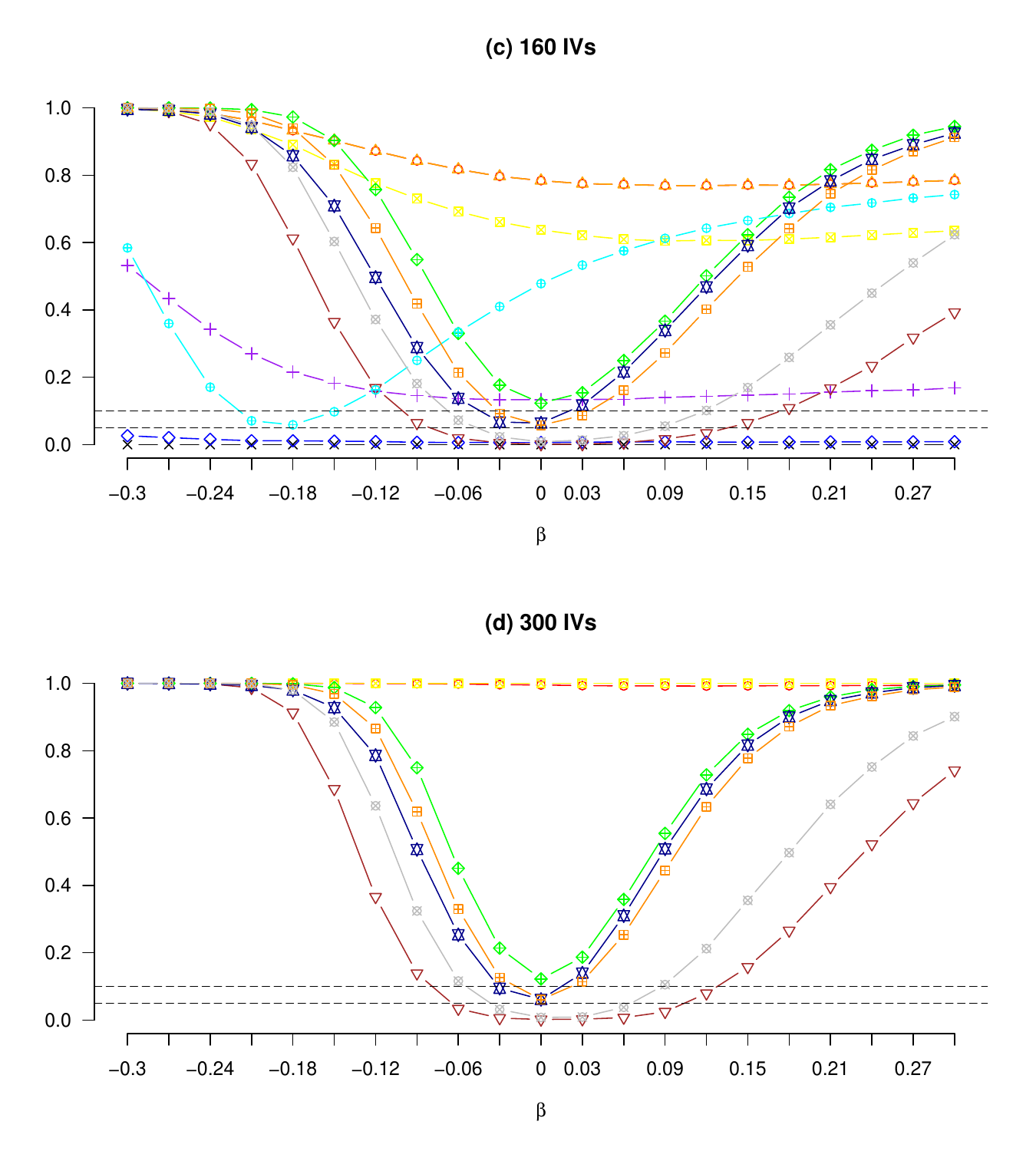}
    \caption{Power curves with $160$ and $300$ IVs}
    \footnotesize{
     \textbf{Note:} The red curve with a hollow circle represents $\rm RJAR$; the orange curve with an upward triangle represents $\rm JAR_{\rm std}$; the purple curve with a cross represents $\rm AR_{\rm Robust}$; the black curve with $\rm \times$ represents $\rm AR$; the blue curve with diamond represents $\rm AS$; the brown curve with inverted triangle represents $\rm BCCH_{1.1}$; the yellow curve with a filled square represents $\rm CT$; the green curve with a filled diamond represents $\rm Empirical$; the cyan curve with a filled circle represents $\rm LM$; the dark-blue curve with hexagram represents $\rm JK$; the dark-orange curve with the $+$ in the square-box represents $\rm BS$. The grey curve with a filled $\rm \times$ in the circle represents $\rm BCCH_{1.01}$. The horizontal dotted black lines represent the $5\%$ and $10\%$ levels.}
    \label{figure:Hausman_final_c_d}
\end{figure}

\section{Empirical Application}\label{empirical-application}
In this section, we consider an empirical application of IV regressions with underlying specifications based on \cite{card2009}, \cite{Goldsmith2020}, and \cite{DKM24}. Specifically, we consider a single cross-section of data in year $2000$ across $124$ locations (cities) by using the following model:
\begin{align*}
    Y_{ls} = \beta_s X_{ls} + \Gamma_s^{\top} W_l + e_{ls},
\end{align*}
where $\beta_s$ is the coefficient of interest and can be interpreted as the (negative) inverse elasticity of substitution between immigrants and natives in the relevant skill group $s$. In addition, $Y_{ls}$ denotes the difference between the residual log wages\footnote{As discussed in \cite{card2009}, residual log wages are log wages after controlling for education, age, gender, race, and ethnicity of the U.S. workforce.} for immigrant and native men in skill group $s \in \{h,c\}$ (high school or college equivalent) and location (city) $l = 1,...,124$. The vector of location-level controls is denoted as $W_l$; in this application we include the following controls: $(1)$ log of city size, $(2)$ college shares, $(3)$ manufacturing shares in both 
$(i)$ $1980$ and $(ii)$ $1990$, $(4)$ mean wage residuals for $(i)$ all natives and $(ii)$ all immigrants in $1980$, together with $(5)$ a constant (so that there are $9$ controls available for each city, i.e., $W_l \in \mathbb{R}^9$).\footnote{See Table 6 in \cite{card2009} for more details on the controls.} $X_{ls}$ denotes the log ratio of immigrant to native hours worked in skill
group $s$ of both men and women in the city $l$. 

The ratio $X_{ls}$ is potentially endogenous because unobserved city-specific factors may shift the relative demand for immigrant workers, leading to higher relative wages and higher relative employment shares, thereby confounding the estimation of the inverse elasticity of substitution. To overcome this issue, \cite{card2009} suggests using the ratio of the total number of immigrants from foreign country $m$ in city $l$ to the total number of immigrants from country $m$ as an instrument. The rationale for such a choice is that existing immigrant enclaves are likely to attract additional immigrant labor through social and cultural channels unrelated to labor market outcomes. To define the instruments, we can exploit settlement patterns at some initial period (possibly together with the arrival rate of immigrants from specific countries in subsequent periods) to determine the inflow of immigrants in each location. Specifically, we let $N_{lm,1980}$ be the number of immigrants from country $m = 1...,38$ settling into location $l$ in $1980$ and let $N_{l,1980}$ be the total number of immigrants in location $l$ in $1980$, respectively. In addition, we let $P_{l,2000}$ denote the population size of location $l$ in $2000$, including both immigrants and natives. 

To proceed, we consider four sets of potential instruments for $X_{ls}$. 
The definition of the first two sets of instruments follows from \citet[Section 5]{DKM24}. Specifically, we let the instruments for each location $l$ be given by $z_{l,1980} := \left\{ z_{lm,1980} \right\}_{m=1}^{38} \equiv \left\{ \frac{N_{lm,1980}}{N_{l,1980}} \times \frac{1}{P_{l,2000}} \right\}_{m=1}^{38} \in \mathbb{R}^{38 \times 1}$, so that our first set of instruments can be written as $Z_{38} := (z_{1,1980},...,z_{124,1980} )^{\top} \in \mathbb{R}^{124 \times 38}$. For the second set of instruments, we let each of the 38 IVs interact with the $9$ controls described above, so that $z_{l,1980} \in \mathbb{R}^{342 \times 1}$ (i.e., each $z_{lm,1980}$ is interacted with $9$ controls). Then, our second set of instruments is defined as $Z_{342} \in \mathbb{R}^{124 \times 342}$. 
Furthermore, given that our proposed bootstrap test is dimensional-robust, we consider a case with $K=1$ for our third and fourth sets of instruments for the high school and college skill groups, respectively. Specifically, following \cite{Goldsmith2020}, we construct Bartik-type instruments, which are given by $Z_{Bartik,s} = \{B_{ls}\}_{l=1}^{124} \in \mathbb{R}^{124 \times 1}$ for $s \in \{h,c\}$, where $B_{ls}:= \sum_{m=1}^{38}  z_{l,1980} \times g_{ms}$ and $g_{ms}$ is the number of immigrants from country $m$ in skill group $s$ arriving in US from $1990$ to $2000$. Note that while the Bartik IVs depend on the skill group $s$ (i.e., $Z_{Bartik,h}$ and $Z_{Bartik,c}$), the first and second sets of instruments (i.e., $Z_{38}$ and $Z_{342}$) do not depend on $s$.

The empirical results of our bootstrap test and those in \cite{DKM24} are given in Tables \ref{table:bartik}--\ref{table:Z_342}. 
% For the Bartik instruments (i.e., $K=1$), the result for $\rm JAR_{\rm cf}$ is not reported because the cross-fit variance estimator is negative. 
Table \ref{table:bartik} shows the $95\%$ confidence intervals (CIs) with the Bartik IVs. $Z_{Bartik,h}$ and $Z_{Bartik,c}$ are applied separately to their respective skill groups. The regularizers for methods $\rm RJAR$ and $\rm BS$ are $\gamma_n^* = 0$ and $\lambda = 0$, respectively, with ${p_n'} = 0.077$ and $p_n = 0.216$.\footnote{When $K=1$, ${p_n'}(\lambda)$ and $p_n(\lambda)$ are independent of $\lambda$, and we set $\lambda = 0$.} 
$R_\lambda$ for BS is equal to $0.457$ and $0.253$ for high-school and college workers, respectively.  
In addition, the number beneath each CI represents its relative length compared to the $\rm BS$ CI. For $K=1$, all CIs have similar lengths. Methods $\rm RJAR$, $\rm JAR_{std}$ and $\rm AR_{ Robust}$ have shorter CIs, but this is because these methods may not control size when $K$ is fixed and tend to over-reject under the null, as observed in our simulation studies. Among the CIs that are theoretically valid for small $K$ (i.e., $\rm AR$, $\rm BS$, $\rm BCCH_{1.1}$ and $\rm BCCH_{1.01}$), the $\rm BS$ CIs is the shortest across both skill groups.

Table \ref{table:Z_38} reports the $95\%$ CIs for high-school and college workers, respectively, with $K = 38$; the set of instruments used for both skill groups is $Z_{38}$. The regularizers for methods $\rm RJAR$ and $\rm BS$ are $\gamma_n^* = 0$ and $\lambda = 13.8$, respectively, with $p_n = 0.022$ and ${p_n'} = 0.089$. 
$R_\lambda$ for BS is equal to $2.015$ for 38 IVs. 
% We find that $\rm BS$ has the shortest CI for college workers, while $\rm JAR_{\rm cf}$ has the shortest CI for high-school workers. But based on our simulation studies, $\rm JAR_{\rm cf}$ may over-reject under the null, which can result in shorter CIs. 
We find that the $\rm BS$ confidence intervals are the shortest for college-educated workers, whereas $\rm BCCH_{1.01}$ produces the shortest intervals for high-school workers, with $\rm BS$ intervals coming in second.
% Furthermore, the $\rm BS$ CIs are shorter than their $\rm BCCH$ counterparts for both high-school and college workers. 

Table \ref{table:Z_342} shows the $95\%$ CIs with $K = 342$; the set of instruments used for both skill groups is $Z_{342}$. The regularizers for methods $\rm RJAR$ and $\rm BS$ are $\gamma_n^* = 5.3$ and $\lambda = 67.4$, respectively, with $p_n = 0.016$ and ${p_n'} = 0.089$. $R_\lambda$ for BS is equal to $1.550$ for 342 IVs. For both high-school and college workers, $\rm CT$ rejects all null hypotheses and thus results in empty confidence intervals, potentially due to heteroskedastic errors. $\rm BS$ again has the shortest confidence interval for college workers, and is of similar length with $\rm RJAR$ for high-school workers. In contrast, $\rm BCCH_{1.1}$ has relatively wider CIs for high-school workers compared to $\rm BS$ and $\rm RJAR$, likely due to the conservative tuning constant $1.1$ discussed above; its $\rm BCCH_{1.01}$ counterpart greatly reduces the length of the CI.

\begin{table}[H]
 \caption{$1$ IV} \centering
 \rule{0pt}{1ex} \\
 \resizebox{0.9\textwidth}{!}{
 \begin{tabular}{c c c c c c c c c c c c c c c c c c}
  \hline
  \multirow[c]{3}{*}{} & \multicolumn{6}{c}{\textbf{High-School Workers}} \\
 \hline \hline
 \multicolumn{1}{c}{$\rm RJAR$} &
 \multicolumn{1}{c}{$\rm JAR_{std}$} &
 \multicolumn{1}{c}{$\rm AR_{Robust}$} &
 \multicolumn{1}{c}{$\rm AR$} &
 \multicolumn{1}{c}{$\rm BS$} &
 \multicolumn{1}{c}{$\rm BCCH_{1.1}$} &
 \multicolumn{1}{c}{$\rm BCCH_{1.01}$} &
 \multicolumn{1}{c}{$\rm CT$} \\
 \hline
 \bf{[-0.040, -0.012]} &
 \bf{[-0.040, -0.012]} &
 \bf{[-0.041, -0.013]} &
 \bf{[-0.041, -0.010]} &
 \bf{[-0.041, -0.010]} &
 \bf{[-0.043, -0.008]} &
 \bf{[-0.042, -0.010]} &
 \bf{[-0.042, -0.010]} \\
 $(0.903)$ &
 $(0.903)$ &
 $(0.903)$ &
 $(1.000)$ &
 $(1.000)$ &
 $(1.129)$ &
 $(1.032)$ &
 $(1.032)$ \\
 \hline
 \rule{0pt}{1ex} \\
 \multirow[c]{3}{*}{} & \multicolumn{6}{c}{\textbf{College Workers}} \\
 \hline \hline
 \multicolumn{1}{c}{$\rm RJAR$} &
 \multicolumn{1}{c}{$\rm JAR_{std}$} &
 \multicolumn{1}{c}{$\rm AR_{Robust}$} &
 \multicolumn{1}{c}{$\rm AR$} &
 \multicolumn{1}{c}{$\rm BS$} &
 \multicolumn{1}{c}{$\rm BCCH_{1.1}$} &
 \multicolumn{1}{c}{$\rm BCCH_{1.01}$} &
 \multicolumn{1}{c}{$\rm CT$} \\
 \hline
 \bf{[-0.094, -0.043]} &
 \bf{[-0.094, -0.043]} &
 \bf{[-0.098, 0.015]} &
 \bf{[-0.097, -0.040]} &
 \bf{[-0.096, -0.040]} &
 \bf{[-0.101, -0.037]} &
 \bf{[-0.097, -0.040]} &
 \bf{[-0.097, -0.040]} \\
 $(0.911)$ &
 $(0.911)$ &
 $(2.018)$ &
 $(1.018)$ &
 $(1.000)$ &
 $(1.143)$ &
 $(1.018)$ &
 $(1.018)$ \\
 \hline
 \end{tabular}
 }

 \raggedright{{\small \textbf{Note:} $95\%$ confidence intervals with
 $Z_{Bartik,h}$ and $Z_{Bartik,c}$ as the instrument for high-school
 and college workers, respectively.}}
 \label{table:bartik}
\end{table}

\begin{table}[H]
\caption{$38$ IVs}
\centering
\rule{0pt}{1ex} \\
\resizebox{\textwidth}{!}{
\begin{tabular}{c c c c c c c c c c c c c c c c c c c c c c}
 \hline
 \multirow[c]{3}{*}{} & \multicolumn{7}{c}{\textbf{High-School Workers}} \\
\hline \hline
\multicolumn{1}{c}{$\rm RJAR$} &
\multicolumn{1}{c}{$\rm JAR_{std}$} &
\multicolumn{1}{c}{$\rm AR_{Robust}$} &
\multicolumn{1}{c}{$\rm AR$} &
\multicolumn{1}{c}{$\rm BS$} &
\multicolumn{1}{c}{$\rm BCCH_{1.1}$} &
\multicolumn{1}{c}{$\rm BCCH_{1.01}$} &
\multicolumn{1}{c}{$\rm CT$} \\
\hline
\bf{[-0.082, -0.015]} &
\bf{[-0.082, -0.015]} &
\bf{[-0.089, -0.012]} &
\bf{[-0.114, 0.007]} &
\bf{[-0.074, -0.013]} &
\bf{[-0.073, -0.003]} &
\bf{[-0.066, -0.007]} &
\bf{[-0.094, -0.007]} \\
$(1.098)$ &
$(1.098)$ &
$(1.262)$ &
$(1.984)$ &
$(1.000)$ &
$(1.148)$ &
$(0.967)$ &
$(1.426)$ \\
\hline
\rule{0pt}{1ex} \\
\multirow[c]{3}{*}{} & \multicolumn{7}{c}{\textbf{College Workers}} \\
\hline \hline
\multicolumn{1}{c}{$\rm RJAR$} &
\multicolumn{1}{c}{$\rm JAR_{std}$} &
\multicolumn{1}{c}{$\rm AR_{Robust}$} &
\multicolumn{1}{c}{$\rm AR$} &
\multicolumn{1}{c}{$\rm BS$} &
\multicolumn{1}{c}{$\rm BCCH_{1.1}$} &
\multicolumn{1}{c}{$\rm BCCH_{1.01}$} &
\multicolumn{1}{c}{$\rm CT$} \\
\hline
\bf{[-0.12, 0.01]} &
\bf{[-0.12, 0.01]} &
\bf{[-0.12, 0.018]} &
\bf{[-0.12, 0.028]} &
\bf{[-0.117, -0.030]} &
\bf{[-0.12, -0.015]} &
\bf{[-0.118, -0.024]} &
\bf{[-0.12, 0.021]} \\
$(1.494)$ &
$(1.494)$ &
$(1.586)$ &
$(1.701)$ &
$(1.000)$ &
$(1.207)$ &
$(1.080)$ &
$(1.621)$ \\
\hline
\end{tabular}
}

\raggedright{\textbf{Note:} $95\%$ confidence intervals with $Z_{38}$ as instruments.}
\label{table:Z_38}
\end{table}

\begin{table}[H]
\caption{$342$ IVs}
\centering
\rule{0pt}{1ex} \\
\resizebox{0.5\textwidth}{!}{
\begin{tabular}{c c c c c c c c c c c}
 \hline
 \multirow[c]{3}{*}{} & \multicolumn{4}{c}{\textbf{High-School Workers}} \\
\hline \hline
& \multicolumn{1}{c}{$\rm RJAR$} &
\multicolumn{1}{c}{$\rm BS$} &
\multicolumn{1}{c}{$\rm BCCH_{1.1}$} &
\multicolumn{1}{c}{$\rm BCCH_{1.01}$} &
\multicolumn{1}{c}{$\rm CT$} \\
\hline
& \bf{[-0.077, -0.008]} &
\bf{[-0.071, -0.013]} &
\bf{[-0.084, 0.004]} &
\bf{[-0.072, -0.001]} &
\bf{$\varnothing$} \\
& $(1.190)$ &
$(1.000)$ &
$(1.517)$ &
$(1.224)$ &
$(\varnothing)$ \\
\hline
\rule{0pt}{1ex} \\
\multirow[c]{3}{*}{} & \multicolumn{4}{c}{\textbf{College Workers}} \\
\hline \hline
& \multicolumn{1}{c}{$\rm RJAR$} &
\multicolumn{1}{c}{$\rm BS$} &
\multicolumn{1}{c}{$\rm BCCH_{1.1}$} &
\multicolumn{1}{c}{$\rm BCCH_{1.01}$} &
\multicolumn{1}{c}{$\rm CT$} \\
\hline
& \bf{[-0.111, 0.009]} &
\bf{[-0.118, -0.027]} &
\bf{[-0.12, -0.003]} &
\bf{[-0.12, -0.013]} &
\bf{$\varnothing$} \\
& $(1.319)$ &
$(1.000)$ &
$(1.286)$ &
$(1.176)$ &
$(\varnothing)$ \\
\hline
\end{tabular}
}

\raggedright{\textbf{Note:} $95\%$ confidence intervals with $Z_{342}$ as instruments.}
\label{table:Z_342}
\end{table}

					%%%%%%%%%%%%%%%%%%%%%%%%%%% Citations
					
					%\newpage
					
					% %%%%%%%%%%%%%%%%%%%%%%%%%%%%   Appendix
					
					% \vspace{0.5in}
					% %\newpage
					% \appendix
					% \noindent{\huge\textbf{Appendix}}
					
					% \unhidefromtoc % Let the Table of content start from here, in Appendix
					% %\tableofcontents
					% %\newpage
%%%%%%%%%%%%%%%%%%%%%%%%%%% Citations

%\newpage
\bibliographystyle{chicago}
\singlespacing
\bibliography{mybibliography}

%%%%%%%%%%%%%%%%%%%%%%%%%%%%   Appendix
\newpage
\appendix
\noindent{\huge\textbf{Appendix}}

\unhidefromtoc % Let the Table of content start from here, in Appendix
\tableofcontents

\newpage

\section{Proof of Theorem \ref{thm:main_null}}\label{sec:pf_main_null}
Recall that
\begin{align*}
    Q(\beta_0)  = \frac{\sum_{i \in [n]} \sum_{j \in [n], j \neq i} (g_i\tilde \sigma_i(\beta_0) + \Delta \Pi_i) \Xi_{\lambda,ij} (g_j \tilde \sigma_j(\beta_0) + \Delta \Pi_j)}{\sqrt{R_\lambda}},
\end{align*}
% \begin{align*}
%     Q(\beta_0) & = \frac{\sum_{i \in [n]} \sum_{j \in [n], j \neq i} g_i\breve \sigma_i(\beta_0) \Xi_{\lambda,ij} g_j\breve \sigma_j(\beta_0) }{\sqrt{R_\lambda}} + \frac{2\sum_{i \in [n]} \sum_{j \in [n], j \neq i} \Pi_i \Delta P_{\lambda,ij} g_j\breve \sigma_j(\beta_0) }{\sqrt{R_\lambda}} \\
% & + \frac{\sum_{i \in [n]} \sum_{j \in [n], j \neq i} \Pi_i  P_{\lambda,ij} \Pi_j \Delta^2 }{\sqrt{R_\lambda}},
% \end{align*}
% where $B_{\lambda,jk} = \sum_{i \in [n]} P_{W,ik} P_{W,ij} P_{\lambda,ii} = [P_W D P_W]_{jk}$, 
% \begin{align*}
% D = \diag(P_{\lambda,11}, \cdots, P_{\lambda,nn}) = \diag(P_{\lambda}), 
% \end{align*}
% and $A_{\lambda,ii} = 2 P_{\lambda,ii} P_{W,ii} - B_{\lambda,ii}$. 

% and 
% \begin{align}
% \widehat{Q}(\beta_0) = \widehat{Q}:= \frac{\sum_{i \in [n]} \sum_{j \in [n], j \neq i}e_i P_{\lambda,ij} e_j}{\sqrt{R_\lambda}},
% \end{align}
% where $e_i$ is just $e_i(\beta_0)$ under the null (i.e., $\Delta = 0$) and defined as 
% \begin{align}\label{eq:e}
%     e_i = \tilde e_i - W^\top_i \hat \gamma_e \quad \text{and} \quad \hat \gamma_e =  (W^\top W)^{-1} W^\top \tilde e.
% \end{align}
We further define 
\begin{align*}
\breve Q(\beta_0) & := \frac{\sum_{i \in [n]} \sum_{j \in [n], j \neq i} \breve e_i(\beta_0) \Xi_{\lambda,ij} \breve e_j(\beta_0) }{\sqrt{R_\lambda}}    
\end{align*}
% \begin{align*}
%     \breve Q(\beta_0) & := \frac{\sum_{i \in [n]} \sum_{j \in [n], j \neq i} \tilde e_i(\beta_0) \Xi_{\lambda,ij} \breve e_j(\beta_0) }{\sqrt{R_\lambda}} + \frac{2\sum_{i \in [n]} \sum_{j \in [n], j \neq i} \Pi_i \Delta P_{\lambda,ij} \breve e_j(\beta_0) }{\sqrt{R_\lambda}} \\
% & + \frac{\sum_{i \in [n]} \sum_{j \in [n], j \neq i} \Pi_i  P_{\lambda,ij} \Pi_j \Delta^2 }{\sqrt{R_\lambda}},
% \end{align*}
where $\breve e_i(\beta_0) = \tilde e_i(\beta_0)  + \Delta \Pi_i$, $\tilde e_i(\beta_0) = \tilde e_i + \Delta \tilde v_i$, $B_{\lambda,jk} = \sum_{i \in [n]} P_{W,ik} P_{W,ij} P_{\lambda,ii} = [P_W D_{\lambda} P_W]_{jk}$, and 
\begin{align*}
\Xi_{\lambda,ij} =     P_{\lambda,ij} + (P_{\lambda,ii} + P_{\lambda,jj}) P_{W,ij} - B_{\lambda,ij}.
\end{align*}

The proof is divided into three sub-steps. In the first step, we prove that 
\begin{align}\label{eq:Qhat-Qtilde_f}
| \widehat{Q}(\beta_0) - \breve Q(\beta_0) - C(\Delta)| = o_P(1),   
\end{align}
where $C(\Delta)$ is a deterministic function of $\Delta$ defined in \eqref{eq:CDelta}.

In the second step, we prove that 
\begin{align}\label{eq:Qtilde-Q_f}
     \sup_{y \in \Re} |\mathbb P(\breve Q(\beta_0) \leq y ) - \mathbb P(Q(\beta_0) \leq y )| = o_P(1).
\end{align}
In the last step, we combine \eqref{eq:Qhat-Qtilde_f} and \eqref{eq:Qtilde-Q_f} to derive the final result. 

\medskip

\noindent \textbf{\Large{Step 1: Show \eqref{eq:Qhat-Qtilde_f}}}

Recall $e_i(\beta_0) = \breve e_i(\beta_0) - W_i^{\top} \hat \gamma(\beta_0),$
where $\hat \gamma(\beta_0)= (W^\top W)^{-1} (W^\top \tilde e(\beta_0))$. This implies
\begin{align*}
    e_i(\beta_0)e_{j}(\beta_0) - \breve e_i(\beta_0) \breve e_{j}(\beta_0) &= (\breve e_i(\beta_0) - W_i^{\top} \hat \gamma(\beta_0))(\breve e_{j}(\beta_0) - W_{j}^{\top} \hat \gamma(\beta_0)) - \breve e_i(\beta_0) \breve e_{j}(\beta_0)\\
    & = - \breve e_i(\beta_0) W_{j}^{\top} \hat \gamma(\beta_0) - \breve e_{j}(\beta_0) W_i^{\top} \hat \gamma (\beta_0) + \hat \gamma^\top (\beta_0)W_i  W_{j}^{\top} \hat \gamma (\beta_0). 
\end{align*}

By Lemma \ref{lem:gamma}(4) and the fact that $\sum_{j \in [n]} P_{\lambda, ij} W_j^\top = 0$, we have
\begin{align}\label{eq:step1_1}
\widehat{Q}(\beta_0) & = \frac{\sum_{i \in [n]} \sum_{j \in [n], j \neq i}e_{i}(\beta_0) P_{\lambda,ij} e_{j}(\beta_0)}{\sqrt{R_\lambda}} - \frac{\sum_{i,j \in [n]^2} \kappa_{ij} e_j^2(\beta_0) A_{\lambda,ii}}{\sqrt{R_\lambda}} \notag \\
 & = \frac{\sum_{i \in [n]} \sum_{j \in [n], j \neq i} \breve e_i(\beta_0) P_{\lambda,ij} \breve e_j(\beta_0) }{\sqrt{R_\lambda}} - \frac{2\sum_{i \in [n]} \sum_{j \in [n], j \neq i}  \breve e_i(\beta_0)P_{\lambda,ij}  W_{j}^{\top} \hat \gamma (\beta_0)}{\sqrt{R_\lambda}} \notag \\
  & + \frac{\sum_{i \in [n]} \sum_{j \in [n], j \neq i}  (W_i^{\top} \hat \gamma (\beta_0)) P_{\lambda,ij} W_{j}^{\top} \hat \gamma(\beta_0) }{\sqrt{R_\lambda}} - \frac{\sum_{i \in [n]} \tilde \sigma_i^2(\beta_0) A_{\lambda,ii}}{\sqrt{R_\lambda}} + o_P(1) \notag  \\
  & =  \frac{\sum_{i \in [n]} \sum_{j \in [n], j \neq i} \breve e_i(\beta_0) P_{\lambda,ij} \breve e_j(\beta_0) }{\sqrt{R_\lambda}} +  \frac{2\sum_{i \in [n]}  \breve e_i(\beta_0)P_{\lambda,ii}  W_{i}^{\top} \hat \gamma(\beta_0)}{\sqrt{R_\lambda}} \notag \\
  & - \frac{\sum_{i \in [n]}   (W_i^{\top} \hat \gamma (\beta_0))^2P_{\lambda,ii} }{\sqrt{R_\lambda}} - \frac{\sum_{i \in [n]} \tilde \sigma_i^2(\beta_0) A_{\lambda,ii}}{\sqrt{R_\lambda}} + o_P(1). 
\end{align}
We note that
\begin{align*}
W_{i}^{\top} \hat \gamma(\beta_0) = \sum_{j \in [n]} P_{W,ij} \tilde e_j(\beta_0),
\end{align*}
and thus,
\begin{align*}
\frac{\sum_{i \in [n]}  \breve e_i(\beta_0)P_{\lambda,ii}  W_{i}^{\top} \hat \gamma(\beta_0)}{\sqrt{R_\lambda}} & = \frac{\sum_{i,j \in [n]^2}  \Pi_i \Delta P_{\lambda,ii} P_{W,ij} \tilde e_j(\beta_0)}{\sqrt{R_\lambda}} + \frac{\sum_{i \in [n]}  \tilde e_i^2(\beta_0) P_{\lambda,ii} P_{W,ii}}{\sqrt{R_\lambda}} \\
& + \frac{\sum_{i,j \in [n]^2,i \neq j}  \tilde e_i(\beta_0) P_{\lambda,ii} P_{W,ij} \tilde e_j(\beta_0)}{\sqrt{R_\lambda}},
\end{align*}
where 
\begin{align*}
    Var \left(\frac{\sum_{i,j \in [n]^2}  \Pi_i \Delta P_{\lambda,ii} P_{W,ij} \tilde e_j(\beta_0)}{\sqrt{R_\lambda}} \right) & \lesssim \frac{ \sum_{j \in [n]} \left( \sum_{i \in [n]}  \Pi_i \Delta P_{\lambda,ii} P_{W,ij} \right)^2 }{ R_\lambda} \\
    & = \frac{ \sum_{i,k \in [n]^2} \Pi_i \Delta P_{\lambda,ii} P_{W,ik} \Pi_k \Delta P_{\lambda,kk} }{ R_\lambda} \\ 
    & \lesssim \frac{ \sum_{i \in [n]} \Pi_i^2 \Delta^2 P_{\lambda,ii}^2}{ R_\lambda} \\
    & \lesssim \frac{\max_{i \in [n]} P_{\lambda,ii}^2 }{ \sqrt{R_\lambda}} \frac{||\Pi||_2^2 \Delta^2 }{ \sqrt{R_\lambda}} \\
    & \lesssim {p_n'}^{1/2}  \frac{||\Pi||_2^2 \Delta^2 }{ \sqrt{R_\lambda}}  = o(1)
\end{align*}
and 
\begin{align*}
Var \left( \frac{\sum_{i \in [n]}  \tilde e_i^2(\beta_0) P_{\lambda,ii} P_{W,ii}}{\sqrt{R_\lambda}} \right) & \lesssim \frac{ \sum_{i \in [n]} P_{\lambda, ii}^2 P_{W,ii}^2}{ R_{\lambda}} \lesssim {p_n'} (\sum_{i \in [n]} P_{W,ii}^2)= o(1).
\end{align*}

This implies 
\begin{align}\label{eq:step1_2}
& \frac{\sum_{i \in [n]}  \breve e_i(\beta_0)P_{\lambda,ii}  W_{i}^{\top} \hat \gamma(\beta_0)}{\sqrt{R_\lambda}}  \notag  \\
& = \mathbb E \left(\frac{\sum_{i \in [n]}  \tilde e_i^2(\beta_0) P_{\lambda,ii} P_{W,ii}}{\sqrt{R_\lambda}} \right) + \frac{\sum_{i,j \in [n]^2,i \neq j}  \tilde e_i(\beta_0) P_{\lambda,ii} P_{W,ij} \tilde e_j(\beta_0)}{\sqrt{R_\lambda}} + o_P(1) \notag  \\
& = \frac{\sum_{i \in [n]} P_{\lambda, ii} P_{W,ii} \tilde \sigma_i^2(\beta_0) }{\sqrt{R_{\lambda}}} + \frac{\sum_{i,j \in [n]^2,i \neq j}  \tilde e_i(\beta_0) P_{\lambda,ii} P_{W,ij} \tilde e_j(\beta_0)}{\sqrt{R_\lambda}}+ o_P(1).
\end{align}

In addition, we have
\begin{align*}
\frac{\sum_{i \in [n]}   (W_i^{\top} \hat \gamma (\beta_0))^2P_{\lambda,ii} }{\sqrt{R_\lambda}} & = \frac{\sum_{i \in [n]} ( \sum_{j \in [n]} P_{W,ij} \tilde e_j(\beta_0)  )^2 P_{\lambda,ii}  }{\sqrt{R_\lambda}} \\
& = \frac{\sum_{i,j,k \in [n]^3}  \tilde e_k(\beta_0) P_{W,ik} P_{W,ij} P_{\lambda,ii} \tilde e_j(\beta_0) }{\sqrt{R_\lambda}} \\
& = \frac{\sum_{j,k \in [n]^2, j \neq k}  \tilde e_j(\beta_0) B_{\lambda,jk} \tilde e_k(\beta_0) } {\sqrt{R_\lambda}}+ \frac{\sum_{j \in [n]}  \tilde e_j^2(\beta_0) B_{\lambda,jj}}{\sqrt{R_\lambda}}
\end{align*}
and 
\begin{align*}
    Var \left( \frac{\sum_{j \in [n]}  \tilde e_j^2(\beta_0) B_{\lambda,jj}}{\sqrt{R_\lambda}}\right) & \lesssim \frac{\sum_{j \in [n]}  B_{\lambda,jj}^2 }{R_\lambda} \\
    & = \frac{ \sum_{j \in [n]} (\sum_{i \in [n]} P_{W,ij}^2 P_{\lambda,ii})^2}{R_\lambda} \\
    & \lesssim \frac{ \left(\max_{i \in [n]} P_{\lambda,ii}^2\right) \left( \sum_{j \in [n]} P_{W,jj}^2 \right)}{ R_{\lambda}} \\
    & \lesssim {p_n'} (\sum_{j \in [n]} P_{W,jj}^2) = o(1),
\end{align*}
which implies 
\begin{align}\label{eq:step1_3}
\frac{\sum_{i \in [n]}   (W_i^{\top} \hat \gamma (\beta_0))^2P_{\lambda,ii} }{\sqrt{R_\lambda}} & = \frac{\sum_{j,k \in [n]^2, j \neq k}  \tilde e_j(\beta_0) B_{\lambda,jk} \tilde e_k(\beta_0) } {\sqrt{R_\lambda}}+ \frac{\sum_{j \in [n]}  \tilde \sigma_j^2(\beta_0) B_{\lambda,jj}}{\sqrt{R_\lambda}} + o_P(1).
\end{align}

Combining \eqref{eq:step1_1}, \eqref{eq:step1_2}, and \eqref{eq:step1_3}, we have
\begin{align*}
\widehat{Q}(\beta_0) & =  \frac{\sum_{i \in [n]} \sum_{j \in [n], j \neq i} \breve e_i(\beta_0) P_{\lambda,ij} \breve e_j(\beta_0) }{\sqrt{R_\lambda}} +  \frac{2\sum_{i \in [n]}  \breve e_i(\beta_0)P_{\lambda,ii}  W_{i}^{\top} \hat \gamma(\beta_0)}{\sqrt{R_\lambda}}  - \frac{\sum_{i \in [n]}   (W_i^{\top} \hat \gamma (\beta_0))^2P_{\lambda,ii} }{\sqrt{R_\lambda}} \notag  \\
  & - \frac{\sum_{i \in [n]} \tilde \sigma_i^2(\beta_0) A_{\lambda,ii}}{\sqrt{R_\lambda}} + o_P(1) \\
& = \frac{\sum_{i \in [n]} \sum_{j \in [n], j \neq i} \tilde e_i(\beta_0) \Xi_{\lambda,ij} \tilde e_j(\beta_0) }{\sqrt{R_\lambda}} + \frac{2\sum_{i \in [n]} \sum_{j \in [n], j \neq i} \Pi_i \Delta P_{\lambda,ij} \tilde e_j(\beta_0) }{\sqrt{R_\lambda}} \\
& + \frac{\sum_{i \in [n]} \sum_{j \in [n], j \neq i} \Pi_i  P_{\lambda,ij} \Pi_j \Delta^2 }{\sqrt{R_\lambda}} + o_P(1) \\
& = \breve Q(\beta_0) + \frac{2\sum_{i \in [n]} \sum_{j \in [n], j \neq i} \Pi_i \Delta \left(P_{\lambda,ij} - \Xi_{\lambda,ij}\right) \tilde e_j(\beta_0) }{\sqrt{R_\lambda}} \\
& + \frac{\sum_{i \in [n]} \sum_{j \in [n], j \neq i} \Pi_i  \left(P_{\lambda,ij} - \Xi_{\lambda,ij}\right) \Pi_j \Delta^2 }{\sqrt{R_\lambda}} + o_P(1)\\
& = \breve Q(\beta_0)  + C(\Delta) + o_P(1),
% & = \frac{\sum_{i \in [n]} \sum_{j \in [n], j \neq i} \breve e_i(\beta_0) P_{\lambda,ij} \tilde e_j(\beta_0) }{\sqrt{R_\lambda}} + \frac{2\sum_{i,j \in [n]^2,i \neq j}  \tilde e_i(\beta_0) P_{\lambda,ii} P_{W,ij} \breve e_j(\beta_0)}{\sqrt{R_\lambda}} - \frac{\sum_{j,k \in [n]^2, j \neq k}  \breve e_j(\beta_0) B_{\lambda,jk} \breve e_k(\beta_0) } {\sqrt{R_\lambda}}  + o_P(1),
\end{align*}
where the last line is by Lemma \ref{lem:gamma}(5) and 
\begin{align}\label{eq:CDelta}
C(\Delta)  =   \frac{\sum_{i \in [n]} \sum_{j \in [n], j \neq i} \Pi_i  \left(P_{\lambda,ij} - \Xi_{\lambda,ij}\right) \Pi_j \Delta^2 }{\sqrt{R_\lambda}} . 
\end{align}

\medskip

\noindent \textbf{\Large{Step 2: Show \eqref{eq:Qtilde-Q_f}}}

For a set $A_y = (-\infty,y)$, define 
\begin{align*}
h_{n,y}(x):= \max \left(0, 1- \frac{d(x,A_y^{3 \delta_n})}{\delta_n} \right) \quad \text{and} \quad f_{n,y}(x) := \mathbb{E} h_{n,y}(x + h_n \mathcal{N}),   
\end{align*}
where $A_y^{3\delta_n}$ is the $3\delta_n$-enlargement of $A_y$, $\mathcal{N}$ has a standard normal distribution, $\delta_n:= C_h h_n$ for some $C_h > 1$, and $h_n = p_n^{1/(7-\zeta)}$ for an arbitrary constant $\zeta \in (0,1)$. 

Applying \cite[Theorem 10.18]{P01} with $\eps$, $\sigma$, $\delta$, $A$, $f(\cdot)$, $g(\cdot)$ in the theorem replaced by $B$, $h_n$, $\delta_n$, $A_y$, $f_{n,y}(\cdot)$, and $g_{n,y}(\cdot)$ in our notation, respectively,\footnote{Theorem 10.18 in \cite{P01} was also employed by \cite{CCK14} in their analysis.} we have $f_{n,y}(\cdot)$ is twice-continuously differentiable such that for all $x,y,v$, and for $\delta_n > h_n$, 
\begin{align*}
    (1-B)1\{x \in A_y\} \leq f_{n,y}(x) \leq B + (1-B)1\{x \in A_y^{3\delta_n}\}
\end{align*}
and 
\begin{align*}
    \left| f_{n,y}(x+v) - f_{n,y}(x) - v \partial f_{n,y}(x) - \frac{1}{2} v^2 \partial^2 f_{n,y}(x)  \right| \leq C_f |v|^3,
\end{align*}
where 
$B = \left(\frac{1+ a}{e^a } \right)^{1/2}$, $1+a = \delta_n^2/h_n^2$, and $C_f = (h_n^2 \delta_n)^{-1}$. Because we set $\delta_n = C_h h_n$, $\delta_n> h_n$ is equivalent to $C_h > 1$. In addition, 
\begin{align*}
1+a = \delta_n^2/h_n^2 = C_h^2, 
\end{align*}
which implies 
\begin{align*}
a= C_h^2 -1 \quad \text{and} \quad B =     \left(\frac{C_h^2}{\exp(C_h^2 -1)} \right)^{1/2}. 
\end{align*}
To highlight the dependence of $B$ on $C_h$, we rewrite it as $B(C_h)$. Therefore, under our notation, \citet[Theorem 10.18]{P01} implies for $C_h > 1$ and $\delta_n = C_h h_n$, 
\begin{align}
    \left| f_{n,y}(x+v) - f_{n,y}(x) - v \partial f_{n,y}(x) - \frac{1}{2} v^2 \partial^2 f_{n,y}(x)  \right| \leq \frac{|v|^3}{\delta_n h_n^2 },
    \label{eq:char_1_A}
\end{align}
\begin{align}
    (1-B(C_h)) 1\{x \in A_y \} \leq f_{n,y}(x) \leq B(C_h) + (1-B(C_h)) 1\{x \in A_y^{3 \delta_n} \},
    \label{eq:char_2_A}
\end{align}
where $B(C_h) := \left( \frac{C_h^2}{\exp(C_h^2 - 1)} \right)^{1/2}$ and 
\begin{align}\label{eq:f''_A}   
\partial^2 f_{n,y}(x) = h_n^{-2} \mathbb E g_{n,y}(x+h_n \N)(\N^2 -1). 
\end{align}

By \eqref{eq:char_2_A}, we have
\begin{align*}
    \mathbb P(\breve Q(\beta_0) \leq y ) - \mathbb P(Q(\beta_0) \leq y ) & \leq (1-B(C_h))^{-1} \mathbb E( f_{n,y}(\breve Q(\beta_0) ) ) - \mathbb P(Q(\beta_0) \leq y ) \notag \\
    & \leq (1-B(C_h))^{-1} \left| \mathbb E( f_{n,y}(\breve Q(\beta_0)) ) - \mathbb E( f_{n,y}(Q(\beta_0)) )\right| \notag  \\
    & + \frac{B(C_h)}{1-B(C_h)} + \mathbb P(Q(\beta_0) \leq y+3\delta_n ) - \mathbb P(Q(\beta_0) \leq y ) \notag \\
    & \leq (1-B(C_h))^{-1} \sup_{y \in \Re} \left| \mathbb E( f_{n,y}(\breve Q(\beta_0)) ) - \mathbb E( f_{n,y}(Q(\beta_0)) )\right| \notag \\
    & + \frac{B(C_h)}{1-B(C_h)} + \sup_{y \in \Re } \mathbb P(|Q(\beta_0)-y| \leq 3\delta_n ).
\end{align*}

Similarly, we have 
\begin{align*}
     \mathbb P(Q(\beta_0) \leq y ) - \mathbb P(\breve Q(\beta_0) \leq y )     & \leq (1-B(C_h))^{-1} \sup_{y \in \Re} \left| \mathbb E( f_{n,y}(\breve Q(\beta_0)) ) - \mathbb E( f_{n,y}(Q(\beta_0)) )\right| \notag \\
    & + \frac{B(C_h)}{1-B(C_h)} + \sup_{y \in \Re } \mathbb P(|Q(\beta_0)-y| \leq 3\delta_n ).
\end{align*}
which implies 
\begin{align*}
   \sup_{y \in \Re} \left\vert \mathbb P(\breve Q(\beta_0) \leq y ) - \mathbb P(Q(\beta_0) \leq y ) \right\vert &  \leq (1-B(C_h))^{-1} \sup_{y \in \Re} \left| \mathbb E( f_{n,y}(\breve Q(\beta_0)) ) - \mathbb E( f_{n,y}(Q(\beta_0)) )\right| \notag \\
    & + \frac{B(C_h)}{1-B(C_h)} + \sup_{y \in \Re } \mathbb P(|Q(\beta_0)-y| \leq 3\delta_n ).
\end{align*}

For any $\eps>0$, we choose $C_h$ to be sufficiently large so that $B(C_h)/(1-B(C_h)) = \eps$, or equivalently, $B(C_h) = \eps/(1+\eps)$. This is possible because $B(u)$ is a monotone decreasing function on $u>1$ and $\lim_{u \rightarrow \infty} B(u) = 0$. 

Throughout, we omit the dependence of $C_h$ on $\eps$ for notation simplicity. Then, we have
\begin{align}\label{eq:Qtilde-Q}
   \sup_{y \in \Re} \left\vert \mathbb P(\breve Q(\beta_0) \leq y ) - \mathbb P(Q(\beta_0) \leq y ) \right\vert &  \leq (1+\eps) \sup_{y \in \Re} \left| \mathbb E( f_{n,y}(\breve Q(\beta_0)) ) - \mathbb E( f_{n,y}(Q(\beta_0)) )\right| \notag \\
    & + \eps + \sup_{y \in \Re } \mathbb P(|Q(\beta_0)-y| \leq 3\delta_n ).  
\end{align}

Next, we first bound $\sup_{y \in \Re }\left| \mathbb E( f_{n,y}(\breve Q (\beta_0)) ) - \mathbb E( f_{n,y}(Q (\beta_0)) )\right| $. Let 
\begin{align}\label{eq:G}
    &\mathcal{G}_n(\{a_i\}_{i \in [n]}):= \frac{\sum_{i \in [n]} \sum_{j \in [n], j \neq i} \{a_i  \Xi_{\lambda,ij} a_j \} }{\sqrt{R_\lambda}}, 
\end{align}
and $\breve g_i(\beta_0) = g_i \tilde \sigma_i (\beta_0) + \Delta \Pi_i $, where $\{g_i\}_{i \in [n]}$ are i.i.d. standard normal random variables. Then, we can rewrite $\breve Q (\beta_0)$ and $Q (\beta_0)$ as $\breve Q (\beta_0) = \mathcal{G}_n(\{\breve e_i(\beta_0) \}_{i \in [n]})$ and $Q(\beta_0) = \mathcal{G}_n(\{\breve g_i (\beta_0)\}_{i \in [n]})$, respectively. 

For each $k \in [n]$, define
\begin{align*}
  & s_{k} := \frac{\sum_{i <k} \sum_{j <k, j \neq i}   \breve e_i (\beta_0) \Xi_{\lambda,ij} \breve e_j (\beta_0)  }{\sqrt{R_\lambda}} + \frac{\sum_{i > k} \sum_{j >k, j \neq i}  \breve g_i (\beta_0) \Xi_{\lambda,ij} \breve g_j (\beta_0)  }{\sqrt{R_\lambda}} \\
  & + \frac{2\sum_{i <k} \sum_{j >k}   \breve e_i (\beta_0) \Xi_{\lambda,ij} \breve g_j (\beta_0)  }{\sqrt{R_\lambda}},
  \\
  & \mathcal{S}_k := 2 \breve e_k (\beta_0) \left(\frac{ \sum_{i <k} \Xi_{\lambda,ki} \breve e_i (\beta_0) + \sum_{i >k} \Xi_{\lambda,ki} \breve g_i (\beta_0)}{\sqrt{R_\lambda}}\right), \\
  & \breve{\mathcal{S}}_k :=  2 \breve g_k (\beta_0) \left(\frac{ \sum_{i <k} \Xi_{\lambda,ki} \breve e_i (\beta_0) + \sum_{i >k} \Xi_{\lambda,ki} \breve g_i (\beta_0)}{\sqrt{R_\lambda}}\right), 
\end{align*}
so that 
\begin{align*}
    & \mathcal{G}_n(\breve e_1 (\beta_0),..., \breve e_k (\beta_0) , \breve g_{k+1}(\beta_0),\cdots, \breve g_n(\beta_0)) = \mathcal{S}_k + s_k, \quad \text{and} \\
    & \mathcal{G}_n(\breve e_1 (\beta_0),\cdots,\breve e_{k-1} (\beta_0), \breve g_k(\beta_0),\cdots, \breve g_n (\beta_0)) = \breve{\mathcal{S}}_k + s_k.
\end{align*}
By letting  $\mathcal{I}_k$ be the $\sigma$-field generated by $\{\breve g_i(\beta_0), \breve e_i(\beta_0) \}_{i < k} \cup \{ \breve g_i (\beta_0), \breve e_i(\beta_0) \}_{i >k}$, we have $s_k \in \mathcal I_k$, 
\begin{align*}
   &  \mathbb{E}(\mathcal{S}_k| \mathcal{I}_k) = \mathbb{E}(\breve{\mathcal{S}}_k| \mathcal{I}_k), \\
       & \mathbb{E} (\mathcal{S}_k^2 | \mathcal{I}_k ) =  \mathbb{E} (\breve{\mathcal{S}}_k^2 | \mathcal{I}_{k} )= 4\breve \sigma_k^2 (\beta_0) \left[\frac{\sum_{i < k} \Xi_{\lambda,ki} \breve e_i(\beta_0) + \sum_{i >k} \Xi_{\lambda,ki} \breve g_i(\beta_0)}{\sqrt{R_\lambda}}  \right]^2.
\end{align*}
This implies 
\begin{align*}
    \mathbb E \breve{\mathcal{S}}_k \partial f_{n,y}(s_k) = \mathbb E \mathcal{S}_k \partial f_{n,y}(s_k) \quad \text{and} \quad \mathbb E \frac{ \breve{\mathcal{S}}_k^2}{2} \partial^2 f_{n,y}(s_k)  = \mathbb E \frac{ \mathcal{S}_k^2}{2} \partial^2 f_{n,y}(s_k) .
\end{align*}

By telescoping, we have
\begin{align}
    &\left| \mathbb{E} (f_{n,y}(\breve Q (\beta_0)) )- \mathbb{E} (f_{n,y}(Q (\beta_0)))  \right| \notag \\
    & = \left| \sum_{k \in [n]}  \begin{pmatrix}
&         \mathbb{E}\left[f_{n,y}(\mathcal{G}_n(\breve e_1 (\beta_0),...,\breve e_k (\beta_0), \breve g_{k+1} (\beta_0),\cdots, \breve g_n (\beta_0))) \right] \\
& - \mathbb{E}\left[ f_{n,y}(\mathcal{G}_n(\breve e_1 (\beta_0),...,\breve e_{k-1} (\beta_0),\breve g_k (\beta_0),..., \breve g_n (\beta_0))) \right]
    \end{pmatrix} \right| \notag \\
    & \leq  \sum_{k \in [n]}  \left|  \mathbb{E}\left[f_{n,y}(\mathcal{S}_k + s_k) \right] - \mathbb{E}\left[ f_{n,y}(\breve{\mathcal{S}}_k + s_k) \right] \right| \notag \\
    & \leq \sum_{k \in [n]}  \left|  \mathbb{E}\left[f_{n,y}(\mathcal{S}_k + s_k) \right] - \mathbb{E} f_{n,y}(s_k) - \mathbb E \mathcal{S}_k \partial f_{n,y}(s_k) - \mathbb E \frac{\mathcal{S}_k^2}{2} \partial^2 f_{n,y}(s_k)  \right| \notag \\
    & + \sum_{k \in [n]}  \left|  \mathbb{E}\left[f_{n,y}(\breve{\mathcal{S}}_k + s_k) \right] - \mathbb{E} f_{n,y}(s_k) - \mathbb E \breve{\mathcal{S}}_k \partial f_{n,y}(s_k) - \mathbb E \frac{ \breve{\mathcal{S}}_k^2}{2} \partial^2 f_{n,y}(s_k)  \right| \notag \\
    & \leq \sum_{k\in [n]}\frac{\mathbb E(|\mathcal{S}_k|^3 +|\breve{\mathcal{S}}_k|^3 )}{ C_h h_n^3},
    \label{eq:telescope_1}
\end{align}
where we define $\mathcal{G}_n(\breve g_1 (\beta_0),...,\breve g_n (\beta_0), \breve g_{n+1} (\beta_0)) \equiv \mathcal{G}_n(\breve g_1 (\beta_0),...,
\breve g_n (\beta_0))$ and $\mathcal{G}_n(\breve g_{0} (\beta_0),\breve e_1 (\beta_0),...,\breve e_n (\beta_0)) \equiv \mathcal{G}_n(\breve e_1 (\beta_0),...,\breve e_n (\beta_0))$. As the RHS of \eqref{eq:telescope_1} does not depend on $y$, we have
\begin{align}\label{eq:telescope_2}
 \sup_{y \in \Re }   \left| \mathbb{E} (f_{n,y}(\breve Q(\beta_0) ) )- \mathbb{E} (f_{n,y}(Q (\beta_0)))  \right| \leq   \sum_{k\in [n]}\frac{\mathbb E(|\mathcal{S}_k|^3 +|\breve{\mathcal{S}}_k|^3 )}{ C_h h_n^3}. 
\end{align}
Recall $\tilde e_i(\beta_0) = \tilde e_i + \tilde v_i \Delta $, $\tilde g_i(\beta_0) = g_i \tilde \sigma_i(\beta_0)$, and
\begin{align*}
    \theta_{k,i} = \begin{cases}
        \Xi_{\lambda, ki} \tilde e_i(\beta_0)  \quad i < k\\
        \Xi_{\lambda, ki} \tilde g_i (\beta_0)  \quad i > k.
    \end{cases}
\end{align*}

Then, we have
\begin{align*}
    \breve e_i(\beta_0) = \tilde e_i(\beta_0) + \Delta \Pi_i, \quad     \breve g_i(\beta_0) = \tilde g_i(\beta_0) + \Delta \Pi_i, \quad
    \mathcal{S}_k = 2 \breve e_k (\beta_0)\frac{\sum_{i \in [n], i \neq k} (\theta_{k,i} + \Xi_{\lambda,ki} \Delta \Pi_i)}{\sqrt{R_\lambda}}
\end{align*}
and
\begin{align}\label{eq:Delta_k}
    \mathbb E (|\mathcal{S}_k|^3 ) & \lesssim \frac{1}{R_\lambda^{3/2}} \mathbb E\left|\sum_{i \in [n],i \neq k} (\theta_{k,i} + \Xi_{\lambda,ki} \Delta \Pi_i )\right|^3  \notag \\
    & \lesssim \frac{1}{R_\lambda^{3/2}} \mathbb E\left|\sum_{i \in [n],i \neq k} \theta_{k,i}\right|^3 + \frac{1}{R_\lambda^{3/2}} \left|\sum_{i \in [n],i \neq k} \Xi_{\lambda,ki} \Delta \Pi_i \right|^3.
\end{align}

Note that $\{\theta_{k,i}\}_{i \in [n],i \neq k}$ is a sequence of independent mean zero random variables.  Then, by Marcinkiewicz-Zygmund inequality, we have 
\begin{align}\label{eq:strong_approx_inter_2_1}
   \mathbb{E} \left(\left|\sum_{i \in [n], i \neq k} \theta_{k,i} \right|^3  \right) 
 & \leq C \mathbb{E} \left[  (\sum_{i \in [n], i \neq k} \theta_{k,i}^2)^{3/2}   \right] \leq C \left[\mathbb{E} \left( (\sum_{i \in [n], i \neq k} \theta_{k,i}^2)^2  \right) \right]^{3/4} \notag\\ 
 & \leq C \left[\sum_{i \in [n],i \neq k} \sum_{j \in [n],j \neq k} \Xi_{\lambda,ik}^2\Xi_{\lambda,jk}^2  \right]^{3/4} \notag \\
 & \leq C \left[\sum_{i \in [n],i \neq k}  \Xi_{\lambda,ik}^2 \right]^{3/2}. 
\end{align}

This implies 
\begin{align*}
      \sum_{k\in [n]}\frac{\mathbb{E} \left(\left|\sum_{i \in [n], i \neq k} \theta_{k,i} \right|^3  \right) }{R_\lambda^{3/2}} & \lesssim    \sum_{k\in [n]}\frac{\left[\sum_{i \in [n], i \neq k} \Xi_{\lambda,ik}^2 \right]^{3/2}}{ R_\lambda^{3/2}} \\
      & \lesssim \frac{\max_{k \in [n]} \left[\sum_{i \in [n], i \neq k} \Xi_{\lambda,ik}^2 \right]^{1/2}}{R_\lambda^{1/2}} = O\left(p_n^{1/2}\right).
\end{align*}
% where the equality is by Lemma \ref{lem:gamma}(6). 
% where the equality is due to Lemma \ref{lem:qN}. 

For the second term on the RHS of \eqref{eq:Delta_k}, we have
\begin{align*}
    \max_{k \in [n]} \left|\sum_{i \in [n],i \neq k} \Xi_{\lambda,ki} \Delta \Pi_i \right|/R_{\lambda}^{1/2} & \leq     \max_{k \in [n]} \left[\sum_{i \in [n],i \neq k} \Xi_{\lambda,ki}^2 \right]^{1/2} R_{\lambda}^{-1/2} |\Delta| \left\Vert \Pi \right\Vert_2 \\
    & \leq p_n^{1/2} |\Delta| \left\Vert \Pi \right\Vert_2. 
\end{align*}
% where the second inequality is due to Lemma \ref{lem:qN}. 

 Therefore, we have
\begin{align*}
  & \sum_{k\in [n]}  \frac{\left|\sum_{i \in [n],i \neq k} \Xi_{\lambda,ki} \Delta \Pi_i \right|^3}{R_\lambda^{3/2}}  \leq  p_n^{1/2} |\Delta| \left\Vert \Pi \right\Vert_2 \sum_{k\in [n]}   \frac{\left|\sum_{i \in [n],i \neq k} \Xi_{\lambda,ki} \Delta \Pi_i \right|^2}{R_\lambda} \\
  & = p_n^{1/2} |\Delta| \left\Vert \Pi \right\Vert_2 \frac{\Pi^\top \Xi_\lambda^2 \Pi \Delta^2 }{R_\lambda} 
\leq p_n^{1/2} \left( \frac{\left\Vert \Pi \right\Vert_2^2 \Delta^2}{ R_\lambda^{2/3}} \right)^{3/2}  = O(p_n^{1/2}), 
\end{align*}
which implies 
\begin{align*}
      \sum_{k\in [n]}\frac{\mathbb E(|\mathcal{S}_k|^3 )}{ C_h h_n^3}  =  O\left(\frac{p_n^{1/2}}{ h_n^3}\right).
\end{align*}
Similarly, we have
\begin{align*}
          \sum_{k\in [n]}\frac{\mathbb E(|\breve{\mathcal{S}}_k|^3 )}{ C_h h_n^3} =   O\left(\frac{p_n^{1/2}}{ h_n^3}\right), 
\end{align*}
and thus, 
\begin{align}\label{eq:f_11}
\sup_{y \in \Re }\left| \mathbb E( f_{n,y}(\breve Q (\beta_0)) ) - \mathbb E( f_{n,y}(Q (\beta_0)) )\right|  = O\left(\frac{p_n^{1/2}}{ h_n^3}\right).
\end{align}

In addition, by Lemma \ref{lem:anti-concentration}, we have
\begin{align}\label{eq:term3_1}
 \sup_{y \in \Re } \mathbb P(|Q(\beta_0)-y| \leq 3\delta_n ) \leq C_\zeta (3C_h)^{(1-\zeta)/2} h_n^{(1-\zeta)/2}
\end{align}
for any $\zeta \in (0,1)$ and $C_{\zeta} \in (0,\infty)$ that only depends on $\zeta$ and $\underline c$ in Assumption \ref{ass:reg}.3.

Then, combining \eqref{eq:Qtilde-Q}, \eqref{eq:f_11}, and \eqref{eq:term3_1}, we have
\begin{align*}
& \sup_{y \in \Re} \left\vert \mathbb P(\breve Q(\beta_0) \leq y ) - \mathbb P(Q(\beta_0) \leq y ) \right\vert \\
&  \leq (1+\eps) \sup_{y \in \Re} \left| \mathbb E( f_{n,y}(\breve Q(\beta_0)) ) - \mathbb E( f_{n,y}(Q(\beta_0)) )\right|  + \eps + \sup_{y \in \Re } \mathbb P(|Q(\beta_0)-y| \leq 3\delta_n ) \\  
& \leq O\left( \frac{p_n^{1/2}} {h_n^3}\right) + \eps + C_\zeta (3C_h)^{(1-\zeta)/2} h_n^{(1-\zeta)/2}. 
\end{align*}

By letting $n \rightarrow \infty$, we have
\begin{align*}
& \limsup_{n \rightarrow \infty}   \sup_{y \in \Re} \left\vert \mathbb P(\breve Q(\beta_0) \leq y ) - \mathbb P(Q(\beta_0) \leq y ) \right\vert \leq  \eps. 
\end{align*}
Because $\eps$ is arbitrary, we have
\begin{align}\label{eq:step2_final_1}
    \sup_{y \in \Re }  \left| \mathbb P(Q (\beta_0) \leq y )-\mathbb P(\breve Q (\beta_0) \leq y ) \right| = o(1).
\end{align}

% Note that 
% \begin{align*}
%     e_i = M_{W,i} \tilde e_i = \tilde e_i - W_i^{\top} \hat \gamma, 
% \end{align*}
% where $M_{W,i}$ is the $i$th row of $M_W$ and $\hat \gamma (\beta_0)= (W^\top W)^{-1}(W^\top \tilde e)$. 

% By Lemma \ref{lem:Qhat-Qtilde}

\medskip

\noindent \textbf{\Large{Step 3: Concluding the Proof}}  

For any sufficiently small $\eps>0$, we have
\begin{align*}
   &  \mathbb P(\widehat{Q} (\beta_0) \leq y ) - \mathbb P(Q (\beta_0) + C(\Delta)\leq y ) \\
   & \leq \mathbb P(\widehat{Q} (\beta_0) \leq y, |\widehat{Q} (\beta_0) - \breve Q (\beta_0) - C(\Delta)| \leq \eps ) - \mathbb P(Q (\beta_0) + C(\Delta)\leq y ) +  \mathbb P( |\widehat{Q} (\beta_0) - \breve Q (\beta_0) - C(\Delta)| \geq \eps ) \\
   & \leq  \mathbb P(\breve Q (\beta_0) + C(\Delta) \leq y + \eps ) - \mathbb P(Q (\beta_0)+C(\Delta) \leq y ) +  \mathbb P( |\widehat{Q} (\beta_0) - \breve Q (\beta_0) - C(\Delta)| \geq \eps ) \\
   & \leq \sup_{y \in \Re}\left| \mathbb P(\breve Q (\beta_0) \leq y ) - \mathbb P(Q (\beta_0)\leq y ) \right| + \sup_{y \in \Re} \mathbb P(|Q (\beta_0) -y|\leq \eps)+ \sup_{y \in \Re}\mathbb P( |\widehat{Q} (\beta_0) - \breve Q (\beta_0)-C(\Delta)| \geq \eps ).
\end{align*}
Similarly, we can show that 
\begin{align*}
   &  \mathbb P(\widehat{Q} (\beta_0) > y ) - \mathbb P(Q (\beta_0) + C(\Delta)> y ) \\
   & \leq \mathbb P(\widehat{Q} (\beta_0) > y, |\widehat{Q} (\beta_0) - \breve Q (\beta_0)-C(\Delta)| \leq \eps ) - \mathbb P(Q (\beta_0) + C(\Delta)> y ) +  \mathbb P( |\widehat{Q} (\beta_0) - \breve Q (\beta_0)-C(\Delta)| \geq \eps ) \\
   & \leq  \mathbb P(\breve Q (\beta_0) + C(\Delta)> y - \eps ) - \mathbb P(Q (\beta_0) + C(\Delta) > y ) +  \mathbb P( |\widehat{Q} (\beta_0) - \breve Q (\beta_0)-C(\Delta)| \geq \eps ) \\
   & \leq \sup_{y \in \Re}\left| \mathbb P(\breve Q (\beta_0) \leq y ) - \mathbb P(Q (\beta_0)\leq y ) \right| + \sup_{y \in \Re} \mathbb P(|Q (\beta_0) -y|\leq \eps)+ \sup_{y \in \Re}\mathbb P( |\widehat{Q} (\beta_0) - \breve Q (\beta_0)-C(\Delta)| \geq \eps ),
   % & = C_\zeta \eps^{(1-\zeta)/2} + o_P(1), 
\end{align*}
or equivalently, 
\begin{align*}
& \mathbb P(Q (\beta_0) + C(\Delta) \leq y ) -    \mathbb P(\widehat{Q} (\beta_0) \leq y ) \\
& \leq \sup_{y \in \Re}\left| \mathbb P(\breve Q (\beta_0) \leq y ) - \mathbb P(Q (\beta_0)\leq y ) \right| + \sup_{y \in \Re} \mathbb P(|Q (\beta_0) -y|\leq \eps)+ \sup_{y \in \Re}\mathbb P( |\widehat{Q} (\beta_0) - \breve Q (\beta_0)-C(\Delta)| \geq \eps ).
\end{align*}

Combining the two results, we have 
\begin{align*}
& \sup_{y \in \Re}    \left| \mathbb P(Q (\beta_0) + C(\Delta)\leq y ) - \mathbb P(\widehat{Q} (\beta_0) \leq y ) \right| \\
& \leq \sup_{y \in \Re}\left| \mathbb P(\breve Q (\beta_0) \leq y ) - \mathbb P(Q (\beta_0)\leq y ) \right| + \sup_{y \in \Re} \mathbb P(|Q (\beta_0) -y|\leq \eps)+ \sup_{y \in \Re}\mathbb P( |\widehat{Q} (\beta_0) - \breve Q (\beta_0)-C(\Delta)| \geq \eps ) \\
&\leq C_\zeta \eps^{(1-\zeta)/2} + o(1), 
\end{align*}
where the last inequality is by Lemma \ref{lem:anti-concentration} and the above two steps. 

As $\eps$ is arbitrary, by letting $\eps$ shrink to zero, we obtain the desired result that, 
\begin{align}\label{eq:Qhat-Q_f}
 \sup_{y\in \Re}\left|\mathbb P(\widehat Q(\beta_0) \leq y)- \mathbb P(Q(\beta_0) + C(\Delta) \leq y)\right| = o(1).    
\end{align}

\section{Proof of Theorem \ref{thm:Fhat-F}}\label{sec:pf_Fhat-F}
Throughout this section, we rely on the following notation: 
$M_n = n^{1/q}$, $h_n = (p_n n^{3/q})^{1/(7-\zeta)}$, where $\zeta$ is an arbitrary constant that belongs to the interval $(0,1)$, $\delta_n = C_h h_n$ for some constant $C_h$ that is fixed and defined later, and
$$t_n = (M_n^2 h_n^{-4} p_n \log(n))^{1/2} + p_n M_n^2 h_n^{-2} \log(n).$$
By Assumption \ref{ass:reg}.5, we have
\begin{align*}
    t_n =  \left[\left(p_n n^{\frac{2-2 \zeta}{q(3-\zeta)}}\right)^{\frac{3-\zeta}{7-\zeta}} \log (n)\right]^{1/2} +  \left(p_n n^{\frac{8-2 \zeta}{q(5-\zeta)}}\right)^{\frac{5-\zeta}{7-\zeta}} \log (n) = o(1).
\end{align*}

The constants $(c,C)$ below are independent of $n$ but may take different values in different contexts.  We also note that, in this section, we do not require the null hypothesis to hold. 
We aim to bound the Kolmogorov distance between $\widehat{Q}^*(\beta_0)$ and $Q^*(\beta_0)$ given data $\mathcal D$, and the definitions of $\widehat{Q}^*(\beta_0)$ and $Q^*(\beta_0)$ can be found in \eqref{eq:Qhat*} and \eqref{eq:Q^*}, respectively. Further, define
\begin{align}\label{eq:Q*}
   \breve Q^*(\beta_0)  = \frac{\sum_{i \in [n]} \sum_{j \in [n], j \neq i} \eta_i \breve e_i(\beta_0)\Xi_{\lambda,ij} \eta_j \breve e_{j}(\beta_0)}{\sqrt{R_\lambda}},
\end{align}
where $\{\eta_i\}_{i \in [n]}$ is the same as those in the definition of $\widehat{Q}^*(\beta_0)$.
Then, we have
\begin{align}\label{eq:Fhat-F}
 & \sup_{y \in \Re} |\mathbb P( \widehat Q^*(\beta_0) \leq y | \mathcal D) - \mathbb P(  Q^*(\beta_0) \leq y) | \notag \\
   & \leq \sup_{y \in \Re}|\mathbb P(\widehat{Q}^*(\beta_0) \leq y | \mathcal D) - \mathbb P(\breve Q^*(\beta_0)\leq y | \mathcal D)| + \sup_{y \in \Re}|\mathbb P(\breve Q^*(\beta_0) \leq y | \mathcal D) - \mathbb P(Q^*(\beta_0) \leq y )|,
\end{align}
where we use the fact that $Q^*(\beta_0)$ is independent of data $\mathcal D$ by construction. 

\medskip

\noindent \textbf{\Large{Step 1: Bound $\sup_{y \in \Re}|\mathbb P(\widehat{Q}^*(\beta_0) \leq y | \mathcal D) - \mathbb P(\breve Q^*(\beta_0)\leq y | \mathcal D)|$}}

Recall $e_i(\beta_0) = \breve e_i(\beta_0) - W_i^{\top} \hat \gamma(\beta_0),$ 
where $\breve e_i(\beta_0) = \tilde e_i + (\Pi_i + \tilde v_i)\Delta = \tilde e_i(\beta_0) + \Pi_i \Delta$ and  $\hat \gamma(\beta_0)= (W^\top W)^{-1} (W^\top \tilde e(\beta_0))$. This implies
\begin{align*}
    e_i(\beta_0)e_{j}(\beta_0) - \breve e_i(\beta_0) \breve e_{j}(\beta_0) &= (\breve e_i(\beta_0) - W_i^{\top} \hat \gamma(\beta_0))(\breve e_{j}(\beta_0) - W_{j}^{\top} \hat \gamma(\beta_0) ) - \breve e_i(\beta_0) \breve e_{j}(\beta_0)\\
    & = - \breve e_i(\beta_0) W_{j}^{\top} \hat \gamma (\beta_0)- \breve e_{j}(\beta_0) W_i^{\top} \hat \gamma (\beta_0) + \hat \gamma^\top (\beta_0) W_i  W_{j}^{\top} \hat \gamma (\beta_0),  
\end{align*}
and thus, 
\begin{align}\label{eq:Qhat*-Q*}
  \widehat{Q}^*(\beta_0) - \breve Q^*(\beta_0) & = - \frac{2\sum_{i \in [n]} \sum_{j \in [n], j \neq i} \eta_i \breve e_i(\beta_0)\Xi_{\lambda,ij} \eta_j W_{j}^{\top} \hat \gamma (\beta_0)}{\sqrt{R_\lambda}} \notag \\
  & + \frac{\sum_{i \in [n]} \sum_{j \in [n], j \neq i} \eta_i (W_i^{\top} \hat \gamma(\beta_0))\Xi_{\lambda,ij} \eta_j W_{j}^{\top} \hat \gamma (\beta_0)}{\sqrt{R_\lambda}}.
\end{align}
For the first term on the RHS of \eqref{eq:Qhat*-Q*}, we have 
\begin{align}\label{eq:var1}
&    \mathbb E \left[ \left(\frac{\sum_{i \in [n]} \sum_{j \in [n], j \neq i} \eta_i \breve e_i(\beta_0)\Xi_{\lambda,ij} \eta_j W_{j}^{\top} \hat \gamma (\beta_0)}{\sqrt{R_\lambda}} \right)^2  \mid \mathcal D \right] \notag \\
& = \frac{2\sum_{i \in [n]} \sum_{j \in [n], j \neq i} \breve e^2_i(\beta_0) \Xi^2_{\lambda,ij} ( W_{j}^{\top} \hat \gamma (\beta_0))^2}{R_\lambda} \notag \\
& = \frac{2\sum_{i \in [n]} \sum_{j \in [n], j \neq i} \breve e^2_i(\beta_0) \Xi^2_{\lambda,ij} ( \sum_{k \in [n]} P_{W,jk} \tilde e_k(\beta_0))^2}{R_\lambda}
\end{align}
and
\begin{align*}
&    \mathbb E \frac{\sum_{i \in [n]} \sum_{j \in [n], j \neq i} \breve e^2_i(\beta_0) \Xi^2_{\lambda,ij} ( \sum_{k \in [n]}P_{W,jk} \tilde e_k(\beta_0))^2}{R_\lambda} \\
& =  \mathbb E \frac{\sum_{i,j,k \in [n]^3, j \neq i} \breve e^2_i(\beta_0) \Xi^2_{\lambda,ij} P_{W,jk}^2 \tilde e_k(\beta_0)^2}{R_\lambda} \\
& \lesssim  \frac{\sum_{i,j,k \in [n]^3, j \neq i}  \Xi^2_{\lambda,ij} P_{W,jk}^2 }{R_\lambda} \\
& \lesssim \max_{i \in [n]} P_{W,ii},
\end{align*}
which implies 
\begin{align}\label{eq:Qhat*-Q*1}
   \mathbb E \left[ \left(\frac{\sum_{i \in [n]} \sum_{j \in [n], j \neq i} \eta_i \breve e_i(\beta_0)\Xi_{\lambda,ij} \eta_j W_{j}^{\top} \hat \gamma (\beta_0)}{\sqrt{R_\lambda}} \right)^2  \mid \mathcal D \right] = O_P\left(\max_{i \in [n]} P_{W,ii} \right).
\end{align}

For the second term on the RHS of \eqref{eq:Qhat*-Q*}, we note that 
\begin{align*}
& Var \left(    \frac{\sum_{i \in [n]} \sum_{j \in [n], j \neq i} \eta_i (W_i^{\top} \hat \gamma(\beta_0))\Xi_{\lambda,ij} \eta_j W_{j}^{\top} \hat \gamma (\beta_0)}{\sqrt{R_\lambda}} \mid \mathcal D \right) \\
&=  \frac{2\sum_{i \in [n]} \sum_{j \in [n], j \neq i} (\sum_{l\in [n]} P_{W,il} \tilde e_l(\beta_0))^2 \Xi_{\lambda,ij}^2 (\sum_{k \in [n]} P_{W,jk} \tilde e_k(\beta_0))^2}{R_\lambda} 
\end{align*}
and 
\begin{align*}
 & \mathbb E \frac{\sum_{i \in [n]} \sum_{j \in [n], j \neq i} (\sum_{l\in [n]} P_{W,il} \tilde e_l(\beta_0))^2 \Xi_{\lambda,ij}^2 (\sum_{k \in [n]} P_{W,jk} \tilde e_k(\beta_0))^2}{R_\lambda}   \\
 & \lesssim \mathbb E \frac{\sum_{i,j,k,l \in [n]^4, j \neq i} P_{W,il}^2 \tilde e_l^2(\beta_0) \Xi_{\lambda,ij}^2 P_{W,jk}^2 \tilde e_k^2(\beta_0)}{R_\lambda}  \\
 & + \mathbb E \frac{\sum_{i,j,k,l \in [n]^4, j \neq i} P_{W,il}P_{W,jl} \tilde e_l^2(\beta_0) \Xi_{\lambda,ij}^2 P_{W,ik} P_{W,jk} \tilde e_k^2(\beta_0)}{R_\lambda} \\
 & \lesssim \frac{ \sum_{i,j \in [n]^2, i \neq j}(P_{W,ii} P_{W,jj}+P_{W,ij}^2 ) \Xi_{\lambda,ij}^2}{R_\lambda} \\
 & \lesssim \left(\max_{i \in [n]}P_{W,ii}^2 \right),
\end{align*}
where we use the fact that 
\begin{align*}
 P_{W,ij}^2 = \left(\sum_{l \in [n]} P_{W,il} P_{W,jl}\right)^2 \lesssim  \left( \sum_{l \in [n]}   P_{W,il}^2 \right) \left( \sum_{l \in [n]}   P_{W,jl}^2 \right) \lesssim P_{W,ii} P_{W,jj}. 
\end{align*}

Therefore, for any sequence $\eps_n \downarrow 0$, we have
\begin{align*}
  & \mathbb P( |\widehat{Q}^*(\beta_0) - \breve Q^*(\beta_0)| \geq \eps_n \mid \mathcal D) \\
  & \leq \frac{\mathbb E \left[\left(\widehat{Q}^*(\beta_0) - \breve Q^*(\beta_0)\right)^2 \mid \mathcal D \right]}{ \eps_n^2} \\
     & \leq \frac{2 \mathbb E\left[ \left( \frac{2\sum_{i \in [n]} \sum_{j \in [n], j \neq i} \eta_i \breve e_i(\beta_0)P_{\lambda,ij} \eta_j W_{j}^{\top} \hat \gamma(\beta_0)}{\sqrt{R_\lambda}}\right)^2 \mid \mathcal D \right]}{\eps_n^2} \\
     & +  \frac{2 \mathbb E\left[ \left( \frac{\sum_{i \in [n]} \sum_{j \in [n], j \neq i} \eta_i (W_i^{\top} \hat \gamma (\beta_0))P_{\lambda,ij} \eta_j W_{j}^{\top} \hat \gamma (\beta_0)}{\sqrt{R_\lambda}}\right)^2 \mid \mathcal D\right]}{\eps_n^2}  = O_P\left( \frac{\max_{i \in [n]}P_{W,ii}}{ \eps_n^2} \right).
\end{align*}

In addition, we have
\begin{align*}
   &  \mathbb P(\widehat{Q}^*(\beta_0) \leq y | \mathcal D) - \mathbb P(\breve Q^*(\beta_0)\leq y | \mathcal D) \\
   & \leq \mathbb P(\widehat{Q}^*(\beta_0) \leq y, |\widehat{Q}^*(\beta_0) - \breve Q^*(\beta_0)| \leq \eps_n | \mathcal D) - \mathbb P(\breve Q^*(\beta_0)\leq y | \mathcal D) +  \mathbb P( |\widehat{Q}^*(\beta_0) - \breve Q^*(\beta_0)| \geq \eps_n \mid \mathcal D) \\
   & \leq  \mathbb P(\breve Q^*(\beta_0) \leq y + \eps_n | \mathcal D) - \mathbb P(Q^*(\beta_0)\leq y | \mathcal D) +  \mathbb P( |\widehat{Q}^*(\beta_0) - Q^*(\beta_0)| \geq \eps_n \mid \mathcal D). 
\end{align*}
In the same manner, we have
\begin{align*}
       &  \mathbb P(\widehat{Q}^*(\beta_0) > y | \mathcal D) - \mathbb P(\breve Q^*(\beta_0) > y | \mathcal D) \\
   & \leq \mathbb P(\widehat{Q}^*(\beta_0) > y, |\widehat{Q}^*(\beta_0) - \breve Q^*(\beta_0)| \leq \eps_n | \mathcal D) - \mathbb P(\breve Q^*(\beta_0)> y | \mathcal D) +  \mathbb P( |\widehat{Q}^*(\beta_0) - \breve Q^*(\beta_0)| \geq \eps_n \mid \mathcal D) \\
   & \leq  \mathbb P(\breve Q^*(\beta_0) > y - \eps_n | \mathcal D) - \mathbb P(\breve Q^*(\beta_0)> y | \mathcal D) +  \mathbb P( |\widehat{Q}^*(\beta_0) - \breve Q^*(\beta_0)| \geq \eps_n \mid \mathcal D), 
\end{align*} 
which implies 
\begin{align*}
       & \mathbb P(\breve Q^*(\beta_0) \leq y | \mathcal D) - \mathbb P(\widehat{Q}^*(\beta_0) \leq y | \mathcal D)  \\
    & \leq  \mathbb P(\breve Q^*(\beta_0) \leq y | \mathcal D) - \mathbb P(\breve Q^*(\beta_0) \leq y - \eps_n | \mathcal D)  +  \mathbb P( |\widehat{Q}^*(\beta_0) - \breve Q^*(\beta_0)| \geq \eps_n \mid \mathcal D).
\end{align*} 

Combining the above two bounds, we have
\begin{align*}
    & |\mathbb P(\widehat{Q}^*(\beta_0) \leq y | \mathcal D) - \mathbb P(\breve Q^*(\beta_0)\leq y | \mathcal D)| \\
    & \leq \mathbb P( |\widehat{Q}^*(\beta_0) - \breve Q^*(\beta_0)| \geq \eps_n \mid \mathcal D) +  |\mathbb P(\breve Q^*(\beta_0) \leq y + \eps_n | \mathcal D) - \mathbb P(\breve Q^*(\beta_0)\leq y -\eps_n| \mathcal D)| \\
    & \leq \sup_{y \in \Re} 2|\mathbb P(\breve Q^*(\beta_0) \leq y | \mathcal D) - \mathbb P(Q^*(\beta_0) \leq y )| +   \mathbb P(|Q^*(\beta_0) - y|\leq \eps_n \mid \mathcal D) + \mathbb P( |\widehat{Q}^*(\beta_0) - \breve Q^*(\beta_0)| \geq 2\eps_n \mid \mathcal D).
\end{align*}

Taking $\sup_{y \in \Re}$ on both sides, we have
\begin{align*}
& \sup_{y \in \Re}   |\mathbb P(\widehat{Q}^*(\beta_0) \leq y | \mathcal D ) - \mathbb P(\breve Q^*(\beta_0) \leq y | \mathcal D)| \\
& \leq \sup_{y \in \Re} 2|\mathbb P(\breve Q^*(\beta_0) \leq y | \mathcal D) - \mathbb P(Q^*(\beta_0) \leq y )| + \sup_{y \in \Re} \mathbb P(|Q^*(\beta_0) - y|\leq \eps_n \mid \mathcal D) \\
& + \mathbb P( |\widehat{Q}^*(\beta_0) - \breve Q^*(\beta_0)| \geq 2\eps_n \mid \mathcal D) \\
& \lesssim  \sup_{y \in \Re} |\mathbb P(\breve Q^*(\beta_0) \leq y | \mathcal D) - \mathbb P(Q^*(\beta_0) \leq y )| + \eps_n^{(1-\zeta)/2} + O_P\left( \frac{\max_{i \in [n]}P_{W,ii}}{ \eps_n^2} \right) ,
\end{align*}
where $\zeta$ is an arbitrary constant in $(0,1)$ and the last inequality is by Lemma \ref{lem:anti-concentration}. 

By choosing 
$\eps_n =\left(\max_{i \in [n]} P_{W,ii}\right)^{\frac{2}{5-\zeta}}$,
we have
\begin{align}\label{eq:Fhat-F2}
& \sup_{y \in \Re}|\mathbb P(\widehat{Q}^*(\beta_0) \leq y | \mathcal D) - \mathbb P(\breve Q^*(\beta_0)\leq y | \mathcal D)| \notag \\
& \lesssim  \sup_{y \in \Re} 2|\mathbb P(\breve Q^*(\beta_0) \leq y | \mathcal D) - \mathbb P(Q^*(\beta_0) \leq y )| + O_P\left( \left(\max_{i \in [n]} P_{W,ii}\right)^{\frac{1-\zeta}{5-\zeta}} \right) \notag \\
& =  \sup_{y \in \Re} 2|\mathbb P(\breve Q^*(\beta_0) \leq y | \mathcal D) - \mathbb P(Q^*(\beta_0) \leq y )| + o_P\left(1\right).
\end{align}

\medskip

\noindent \textbf{\Large{Step 2: Bound $\sup_{y \in \Re} |\mathbb P(\breve Q^*(\beta_0) \leq y | \mathcal D) - \mathbb P(Q^*(\beta_0) \leq y )|$.}} 

For a set $A_y = (-\infty,y)$, we define $g_{n,y}(x):= \max \left(0, 1- \frac{d(x,A_y^{3 \delta_n})}{\delta_n} \right)$ and $f_{n,y}(x) := \mathbb{E} g_{n,y}(x + h_n \mathcal{N})$, where $A_y^{3\delta_n}$ is the $3\delta_n$-enlargement of $A_y$, the random variable $\mathcal{N}$ has a standard normal distribution, $\delta_n:= C_h h_n$ for some $C_h > 1$ to be determined later, and $h_n = (p_n n^{3/q})^{\frac{1}{7-\zeta}} = o(1)$. 

Applying \cite[Theorem 10.18]{P01} with $\eps$, $\sigma$, $\delta$, $A$, $f(\cdot)$, $g(\cdot)$ in the theorem replaced by $B$, $h_n$, $\delta_n$, $A_y$, $f_{n,y}(\cdot)$, and $g_{n,y}(\cdot)$ in our notation, respectively,\footnote{Theorem 10.18 in \cite{P01} was also employed by \cite{CCK14} in their analysis.} we have $f_{n,y}(\cdot)$ is twice-continuously differentiable such that for all $x,y,v$, and for $\delta_n > h_n$, 
\begin{align*}
    (1-B)1\{x \in A_y\} \leq f_{n,y}(x) \leq B + (1-B)1\{x \in A_y^{3\delta_n}\}
\end{align*}
and 
\begin{align*}
    \left| f_{n,y}(x+v) - f_{n,y}(x) - v \partial f_{n,y}(x) - \frac{1}{2} v^2 \partial^2 f_{n,y}(x)  \right| \leq C_f |v|^3,
\end{align*}
where 
$B = \left(\frac{1+ a}{e^a } \right)^{1/2}$, $1+a = \delta_n^2/h_n^2$, and $C_f = (h_n^2 \delta_n)^{-1}$. Because we set $\delta_n = C_h h_n$, $\delta_n> h_n$ is equivalent to $C_h > 1$. In addition, 
\begin{align*}
1+a = \delta_n^2/h_n^2 = C_h^2, 
\end{align*}
which implies 
\begin{align*}
a= C_h^2 -1 \quad \text{and} \quad B =     \left(\frac{C_h^2}{\exp(C_h^2 -1)} \right)^{1/2}. 
\end{align*}
To highlight the dependence of $B$ on $C_h$, we rewrite it as $B(C_h)$. Therefore, under our notation, \citet[Theorem 10.18]{P01} implies for $C_h > 1$ and $\delta_n = C_h h_n$, 
\begin{align}
    \left| f_{n,y}(x+v) - f_{n,y}(x) - v \partial f_{n,y}(x) - \frac{1}{2} v^2 \partial^2 f_{n,y}(x)  \right| \leq \frac{|v|^3}{\delta_n h_n^2 } = \frac{|v|^3}{C_h h_n^3 },
    \label{eq:char_1}
\end{align}
\begin{align}
    (1-B(C_h)) 1\{x \in A_y \} \leq f_{n,y}(x) \leq B(C_h) + (1-B(C_h)) 1\{x \in A_y^{3 \delta_n} \},
    \label{eq:char_2}
\end{align}
where $B(C_h) := \left( \frac{C_h^2}{\exp(C_h^2 - 1)} \right)^{1/2}$ and 
\begin{align}\label{eq:f''}   
\partial^2 f_{n,y}(x) = h_n^{-2} \mathbb E g_{n,y}(x+h_n \N)(\N^2 -1). 
\end{align}

By \eqref{eq:char_2}, we have
\begin{align*}
    \mathbb P(\breve Q^*(\beta_0) \leq y | \mathcal D) - \mathbb P(Q^*(\beta_0) \leq y ) & \leq (1-B(C_h))^{-1} \mathbb E( f_{n,y}(\breve Q^*(\beta_0)) | \mathcal D) - \mathbb P(Q^*(\beta_0) \leq y ) \notag \\
    & \leq (1-B(C_h))^{-1} \left| \mathbb E( f_{n,y}(\breve Q^*(\beta_0)) | \mathcal D) - \mathbb E( f_{n,y}(Q^*(\beta_0)) )\right| \notag  \\
    & + \frac{B(C_h)}{1-B(C_h)} + \mathbb P(Q^*(\beta_0) \leq y+3\delta_n ) - \mathbb P(Q^*(\beta_0) \leq y ) \notag \\
    & \leq (1-B(C_h))^{-1} \left| \mathbb E( f_{n,y}(\breve Q^*(\beta_0)) | \mathcal D) - \mathbb E( f_{n,y}(Q^*(\beta_0)) )\right| \notag \\
    & + \frac{B(C_h)}{1-B(C_h)} + \sup_{y \in \Re} \mathbb P(|Q^*(\beta_0)-y| \leq 3\delta_n),
\end{align*}
where we use the fact that $Q^*(\beta_0)$ is independent of $\mathcal D$. Similarly, we have
\begin{align*}
 \mathbb P(Q^*(\beta_0) \leq y ) - \mathbb P(\breve Q^*(\beta_0) & \leq y | \mathcal D)  \leq (1-B(C_h))^{-1} \left| \mathbb E( f_{n,y}(\breve Q^*(\beta_0)) | \mathcal D) - \mathbb E( f_{n,y}(Q^*(\beta_0)) )\right| \notag \\
    & + \frac{B(C_h)}{1-B(C_h)} + \sup_{y \in \Re} \mathbb P(|Q^*(\beta_0)-y| \leq 3\delta_n),
\end{align*}
which implies 
\begin{align*}
  \sup_{y \in \Re} \left\vert \mathbb P(\breve Q^*(\beta_0) \leq y | \mathcal D) - \mathbb P(Q^*(\beta_0) \leq y ) \right\vert  
    & \leq (1-B(C_h))^{-1} \sup_{y \in \Re} \left| \mathbb E( f_{n,y}(\breve Q^*(\beta_0)) | \mathcal D) - \mathbb E( f_{n,y}(Q^*(\beta_0)) )\right| \notag \\
    & + \frac{B(C_h)}{1-B(C_h)} + \sup_{y \in \Re} \mathbb P(|Q^*(\beta_0)-y| \leq 3\delta_n).
\end{align*}

For any $1>\eps>0$, we choose $C_h$ to be sufficiently large so that $B(C_h)/(1-B(C_h)) = \eps$, or equivalently, $B(C_h) = \eps/(1+\eps)$. This is possible because $B(u)$ is a monotone decreasing function of $u>1$ and $\lim_{u \rightarrow \infty} B(u) = 0$. 

For the rest of the proof, we omit the dependence of $C_h$ on $\eps$ for notation simplicity. Then, we have
\begin{align}\label{eq:Qstar-Q}
  \sup_{y \in \Re} \left\vert \mathbb P(\breve Q^*(\beta_0) \leq y | \mathcal D) - \mathbb P(Q^*(\beta_0) \leq y ) \right\vert  
    & \leq (1+\eps) \sup_{y \in \Re} \left| \mathbb E( f_{n,y}(\breve Q^*(\beta_0)) | \mathcal D) - \mathbb E( f_{n,y}(Q^*(\beta_0)) )\right| \notag \\
    & + \eps + \sup_{y \in \Re} \mathbb P(|Q^*(\beta_0)-y| \leq 3\delta_n). 
\end{align}

Next, we aim to bound $\sup_{y \in \Re}\left| \mathbb E( f_{n,y}(\breve Q^*(\beta_0)) | \mathcal D) - \mathbb E( f_{n,y}(Q^*(\beta_0)) )\right|$ on the RHS of \eqref{eq:Qstar-Q}. Define
\begin{align*}
    &\mathcal{G}_n(\{a_i\}_{i \in [n]}):= \frac{\sum_{i \in [n]} \sum_{j \in [n], j \neq i} \{a_i  \Xi_{\lambda,ij} a_j \} }{\sqrt{R_\lambda}}.
\end{align*}
We further define $\breve \eta_i = \eta_i \breve e_i(\beta_0)$ and $\breve g_i = g_i \breve \sigma_i(\beta_0)$, where $\breve e_i(\beta_0) = \tilde e_i + \Delta(\Pi_i + \tilde v_i)$, $\breve \sigma_i^2(\beta_0) = \mathbb E \breve e_i^2(\beta_0)$, $\{\eta_i\}_{i\in [n]}$ is an i.i.d. sequence of random variables with zero mean and unit variance as defined in Assumption \ref{ass:reg}, and $\{g_i\}_{i \in [n]}$ is an i.i.d. sequence of standard normal random variables. 

Under these definitions, we can rewrite $\breve Q^*(\beta_0)$ and $Q^*(\beta_0)$ as $\breve Q^*(\beta_0) = \mathcal{G}_n(\{\breve \eta_i\}_{i \in [n]})$ and $Q^*(\beta_0) = \mathcal{G}_n(\{\breve g_i\}_{i \in [n]})$, respectively. 

For each $k \in [n]$, define
\begin{align*}
  & s_k := \frac{\sum_{i <k} \sum_{j <k, j \neq i}  \{ \breve \eta_{i} \Xi_{\lambda,ij} \breve \eta_{j} \} }{\sqrt{R_\lambda}} + \frac{\sum_{i > k} \sum_{j >k, j \neq i}  \{\breve g_{i} \Xi_{\lambda,ij} \breve g_{j} \} }{\sqrt{R_\lambda}} \\
  & + \frac{2\sum_{i <k} \sum_{j >k}  \{ \breve \eta_{i} \Xi_{\lambda,ij} \breve g_{j} \} }{\sqrt{R_\lambda}}
  \\
  & \mathcal{S}_k := 2 \breve \eta_k \left(\frac{ \sum_{i <k}  \Xi_{\lambda,ki} \breve \eta_{i}+ \sum_{i >k}  \Xi_{\lambda,ki} \breve g_{i}}{\sqrt{R_\lambda}}\right) \\
  & \breve{\mathcal{S}}_k :=  2 \breve g_k\left(\frac{ \sum_{i <k}  \Xi_{\lambda,ki} \breve \eta_{i}+ \sum_{i >k}  \Xi_{\lambda,ki} \breve g_{i}}{\sqrt{R_\lambda}}\right) 
\end{align*}
so that $\mathcal{G}_n(\breve \eta_{1},...,\breve \eta_k, \breve g_{k+1},\cdots, \breve g_{n}) = \mathcal{S}_k + s_k$ and $\mathcal{G}_n(\breve \eta_{1},\cdots,\breve \eta_{k-1}, \breve g_{k},\cdots, \breve g_{n}) = \breve{\mathcal{S}}_k + s_k$. By letting  $\mathcal{I}_k$ be the $\sigma$-field generated by $\{\breve g_{i}, \breve \eta_{i} \}_{i < k} \cup \{ \breve g_{i}, \breve \eta_{i} \}_{i >k}$, we have
\begin{align*}
   &  \mathbb{E}(\mathcal{S}_k| \mathcal{I}_k,\mathcal D) = \mathbb{E}(\breve{\mathcal{S}}_k| \mathcal{I}_k, \mathcal D) \\
   & \mathbb{E} (\mathcal{S}_k^2 | \mathcal{I}_k,\mathcal D ) = 4\breve e_k^2(\beta_0)\left[\frac{\sum_{i <k-1}  \Xi_{\lambda,ki} \breve \eta_{i} + \sum_{i >k}  \Xi_{\lambda,ki} \breve g_{i}}{\sqrt{R_\lambda}}  \right]^2, \\
   & \mathbb{E} (\breve{\mathcal{S}}_k^2 | \mathcal{I}_{k},\mathcal D ) = 4\breve \sigma^2_k(\beta_0)\left[\frac{\sum_{i <k-1}  \Xi_{\lambda,ki} \breve \eta_{i} + \sum_{i >k}  \Xi_{\lambda,ki} \breve g_{i}}{\sqrt{R_\lambda}}  \right]^2.
\end{align*}

By telescoping, we have
\begin{align}
    &\left| \mathbb{E} (f_{n,y}(\breve Q^*(\beta_0) ) | \mathcal D)- \mathbb{E} (f_{n,y}(Q^*(\beta_0))| \mathcal D)  \right| \notag \\
    & = \left| \sum_{k \in [n]}  \mathbb{E}\left[f_{n,y}(\mathcal{G}_n(\breve \eta_{1},\cdots,\breve \eta_k, \breve g_{k+1},\cdots, \breve g_{n}))| \mathcal D \right] - \mathbb{E}\left[ f_{n,y}(\mathcal{G}_n(\breve \eta_{1},\cdots,\breve \eta_{k-1},\breve g_k,\cdots, \breve g_{n})) | \mathcal D\right] \right|,
    \label{eq:telescope}
\end{align}
where we define $\mathcal{G}_n(\breve g_{1},\cdots,\breve g_{n}, \breve \eta_{n+1}) \equiv \mathcal{G}_n(\breve g_{1},\cdots,
\breve g_{n})$ and $\mathcal{G}_n(\breve g_{0},\breve \eta_{1},\cdots,\breve \eta_{n}) \equiv \mathcal{G}_n(\breve \eta_{1},\cdots,\breve \eta_{n})$. 

Then, by letting $x = s_k$, $v = \mathcal{S}_k$ and $\breve{\mathcal{S}}_k$ in \eqref{eq:char_1}, we have
\begin{align}
& \bigg\vert  \mathbb{E}(f_{n,y}(\mathcal{G}_n(\breve \eta_{1},\cdots,\breve \eta_k, \breve g_{k+1},\cdots, \breve g_{n}))|\mathcal D) - \mathbb{E}(f_{n,y}(\mathcal{G}_n(\breve \eta_{1},\cdots,\breve \eta_{k-1},\breve g_k,\cdots, \breve g_{n}))|\mathcal D) \notag \\
& - \frac{1}{2} \sum_{k \in [n]} \mathbb E\left( 2 \partial^2 f_{n,y}(s_k) \left[\frac{\sum_{i <k}  \Xi_{\lambda,ki} \breve \eta_{i} + \sum_{i >k}  \Xi_{\lambda,ki} \breve g_{i}}{\sqrt{R_\lambda}}  \right]^2 \mid \mathcal D  \right)(\breve e_k^2(\beta_0) - \breve \sigma_k^2(\beta_0)) \bigg\vert \notag \\
& \leq \frac{\mathbb E(|\mathcal{S}_k|^3 + |\widetilde{ \mathcal{S}}_k|^3|\mathcal D)}{ C_h h_n^3}.
    \label{eq:strong_approx_helper}
\end{align}

Define 
\begin{align}\label{eq:Hky}
H_{k,y} =  \mathbb E\left( \partial^2 f_{n,y}(s_k) \left[\frac{\sum_{i <k}  \Xi_{\lambda,ki} \breve \eta_{i} + \sum_{i >k}  \Xi_{\lambda,ki} \breve g_{i}}{\sqrt{R_\lambda}}  \right]^2 \mid \mathcal D  \right)    
\end{align}
and $\mathcal E_k$ be the sigma field generated by $\breve e_{1}(\beta_0),\cdots,\breve e_k(\beta_0)$. Then, we have $H_{k,y} \in \mathcal E_{k-1}$ and 
\begin{align}\label{eq:f1}
 & \sup_{y \in \Re}  \left| \mathbb{E}(f_{n,y}(\breve Q^*(\beta_0) )|\mathcal D) - \mathbb{E}(f_{n,y}(Q^*(\beta_0)) |\mathcal D)  \right| \notag \\
 & \leq \sup_{y \in \Re} \left| \sum_{k \in [n]} H_{k,y} (\breve e_k^2(\beta_0) - \breve \sigma_k^2(\beta_0))\right| +  \sum_{k\in [n]}\frac{\mathbb E(|\mathcal{S}_k|^3 +|\breve{\mathcal{S}}_k|^3 |\mathcal D)}{ C_h h_n^3}.
\end{align}
In the following, we aim to bound the two terms on the RHS of the \eqref{eq:f1}. 

\medskip

\noindent \textbf{\Large{Step 2.1: Bound $\sup_{y \in \Re}\left| \sum_{k \in [n]} H_{k,y} (\breve e_k^2(\beta_0) - \breve \sigma_k^2(\beta_0))\right|$}}

We note that  $\{H_{k,y} (\breve e_k^2(\beta_0) - \breve \sigma_k^2(\beta_0))\}_{i \in [n]}$ is a martingale difference sequence (MDS) w.r.t. the filtration $\{\mathcal E_k\}_{k \in [n]}$. For some sufficiently large constant $C_1>0$, define
\begin{align*}
H_{k,y,\leq} = H_{k,y}1\{ \max_{i \in [k-1]} \breve e_{i}^2(\beta_0) \leq C_1M_n \}    ,
\end{align*}
$\breve e_{k,\leq}^2(\beta_0) = \breve e_k^2(\beta_0)1\{\breve e_k^2(\beta_0) \leq C_1 M_n\}$, $\breve e_{k,>}^2(\beta_0) = \breve e_k^2(\beta_0) - \breve e_{k,\leq}^2(\beta_0)$, $\breve \sigma^2_{k,\leq} (\beta_0) = \mathbb E\left(\breve e_{k,\leq}^2(\beta_0)  \right)$ and $\breve \sigma^2_{k,>} (\beta_0)= \mathbb E\left(\breve e_{k,>}^2(\beta_0)  \right)$. Then, for the sequence $t_n$ defined at the beginning of the section, we have
\begin{align}\label{eq:etilde^2-sigmatilde^2}
   & \mathbb P\left(\sup_{y \in \Re} \left| \sum_{k \in [n]} H_{k,y} (\breve e_k^2(\beta_0) - \breve \sigma_k^2(\beta_0))\right| \geq 4 C_1^3 t_n \right) \notag \\
   & \leq  \mathbb P\left(\sup_{y \in \Re} \left| \sum_{k \in [n]} H_{k,y,\leq} (\breve e_k^2(\beta_0) - \breve \sigma_k^2(\beta_0))\right| \geq 3 C_1^3 t_n \right) \notag \\
   & + \mathbb P\left(\sup_{y \in \Re} \left| \sum_{k \in [n]} (H_{k,y}-H_{k,y,\leq}) (\breve e_k^2(\beta_0) - \breve \sigma_k^2(\beta_0))\right| \geq C_1^3 t_n \right)  \notag \\
   & \leq  \mathbb P\left(\sup_{y \in \Re} \left| \sum_{k \in [n]} H_{k,y,\leq} (\breve e_{k,\leq}^2(\beta_0) -  \breve \sigma_{k,\leq}^2(\beta_0))\right| \geq C_1^3 t_n \right) \notag \\
    & + \mathbb P\left(\sup_{y \in \Re} \left\vert \sum_{k \in [n]} H_{k,y,\leq} \breve e_{k,>}^2(\beta_0) \right\vert > C_1^3 t_n \right) \notag \\
    & + \mathbb P\left(\sup_{y \in \Re} \left\vert  \sum_{k \in [n]} H_{k,y,\leq}  \breve \sigma_{k,>}^2(\beta_0) \right\vert \geq C_1^3 t_n \right) \notag \\
    & + \mathbb P\left(\sup_{y \in \Re} \left| \sum_{k \in [n]} (H_{k,y}-H_{k,y,\leq}) (\breve e_k^2(\beta_0) - \breve \sigma_k^2(\beta_0))\right| > C_1^3 t_n \right) \notag \\
    & \leq  \mathbb P\left(\sup_{y \in \Re} \left| \sum_{k \in [n]} H_{k,y,\leq} (\breve e_{k,\leq}^2(\beta_0) - \breve \sigma_{k,\leq}^2(\beta_0))\right| \geq C_1^3 t_n \right) \notag \\
    & + \mathbb P\left(\sup_{y \in \Re} \left\vert \sum_{k \in [n]} H_{k,y,\leq}  \breve \sigma_{k,>}^2(\beta_0) \right\vert  \geq C_1^3 t_n \right) +   2\mathbb P(\max_{i \in [n]} \breve e_i^2(\beta_0) \geq C_1 M_n ) \notag \\
    & \leq  \mathbb P\left(\sup_{y \in \Re} \left| \sum_{k \in [n]} H_{k,y,\leq} (\breve e_{k,\leq}^2(\beta_0) - \breve \sigma_{k,\leq}^2(\beta_0))\right| \geq C_1^3 t_n \right) \notag \\
    & + 1\{ C \geq C_1^{q+1} t_n M_n^{q-2} h_n^2\} +  \frac{C n}{C_1^q M_n^q}  \notag  \\
    & = \mathbb P\left(\sup_{y \in \Re} \left| \sum_{k \in [n]} H_{k,y,\leq} (\breve e_{k,\leq}^2(\beta_0) - \breve \sigma_{k,\leq}^2(\beta_0))\right| \geq C_1^3 t_n \right) +  \frac{C n}{C_1^q M_n^q} \notag \\
    & = \mathbb P\left(\sup_{y \in \Re} \left| \sum_{k \in [n]} H_{k,y,\leq} (\breve e_{k,\leq}^2(\beta_0) - \breve \sigma_{k,\leq}^2(\beta_0))\right| \geq C_1^3 t_n \right) +  \frac{C}{C_1^q}, 
\end{align}
where the second last inequality holds because if $\max_{i \in [n]} \breve e_i^2(\beta_0) \leq C_1M_n $, then 
\begin{align*}
& \sup_{y \in \Re}  \sum_{k \in [n]} H_{k,y,\leq} \breve e_{k,>}^2(\beta_0)  = 0 \quad \text{and}\\
&    \sup_{y \in \Re} \left| \sum_{k \in [n]} (H_{k,y}-H_{k,y,\leq}) (\breve e_k^2(\beta_0) - \breve \sigma_k^2(\beta_0))\right| = 0,
\end{align*}
the last inequality holds by \eqref{eq:f''}, 
\begin{align*}
  \breve \sigma_{k,>}^2(\beta_0)  = \mathbb E   \breve e_k^{2}(\beta_0)1\{\breve e_k^{2}(\beta_0) > C_1M_n \} \leq \mathbb E  \frac{ \breve e_k^{2q}(\beta_0)}{(C_1M_n)^{q-1} } \leq \frac{C}{C_1^{q-1} M_n^{q-1}},
\end{align*}
\begin{align*}
  & \sup_{y \in \Re} \left\vert \sum_{k \in [n]} H_{k,y,\leq}  \breve \sigma_{k,>}^2(\beta_0) \right\vert  \\
  & \leq \sup_{y \in \Re} \sum_{k \in [n]} \left\vert  H_{k,y,\leq} \right\vert \frac{C}{C_1^{q-1} M_n^{q-1}} \\
   & \leq C \sum_{k \in [n]} h_n^{-2} \mathbb E \left( \left[\frac{\sum_{i <k-1}  \Xi_{\lambda,ki} \breve \eta_{i} + \sum_{i >k}  \Xi_{\lambda,ki} \breve g_{i}}{\sqrt{R_\lambda}}  \right]^2 \mid \mathcal D  \right)1\{ \max_{i \in [k-1]} \breve e_{i}^2(\beta_0) \leq C_1 M_n \} \frac{C}{C_1^{q-1} M_n^{q-1}}  \notag \\
   & \leq C \sum_{k \in [n]} \left( \frac{\sum_{i <k-1} \Xi_{\lambda,ki}^2 \breve e_{i}^2(\beta_0) + \sum_{i >k} \Xi_{\lambda,ki}^2 \breve \sigma_{i}^2(\beta_0) }{R_\lambda h_n^2} \right) 1\{ \max_{i \in [k-1]} \breve e_{i}^2(\beta_0) \leq C_1 M_n \} \frac{1}{C_1^{q-1} M_n^{q-1}}\notag \\
   & \leq C C_1^{2-q} M_n^{2-q} h_n^{-2},
\end{align*}
and
\begin{align*}
   \mathbb P(\max_{i \in [n]} \breve e_i^2(\beta_0) \geq C_1 M_n )  \leq \mathbb P(\max_{i \in [n]} \breve e_i^{2q}(\beta_0) \geq C_1^q M_n^q )  
    \leq \frac{n \mathbb E \breve e_i^{2q}(\beta_0)}{C_1^q M_n^q} \leq  \frac{C n}{C_1^q M_n^q},
\end{align*}
and the second last equality on the RHS of \eqref{eq:etilde^2-sigmatilde^2} holds because $p_n \geq 1/n$ and 
\begin{align*}
    t_n M_n^{q-2} h_n^2 \geq p_n n \log (n) \geq \log(n) \rightarrow \infty.
\end{align*}

For any $\eps>0$, we can choose $C_1 \geq (C/\eps)^{1/q}$ where the constant $C$ is the one on the RHS of \eqref{eq:etilde^2-sigmatilde^2} so that 
\begin{align}\label{eq:etilde^2-sigmatilde^2'}
&   \mathbb P\left(\sup_{y \in \Re} \left| \sum_{k \in [n]} H_{k,y} (\breve e_k^2(\beta_0) - \breve \sigma_k^2(\beta_0))\right| \geq 4 C_1^3 t_n \right) \notag \\
& \leq \mathbb P\left(\sup_{y \in \Re} \left| \sum_{k \in [n]} H_{k,y,\leq} (\breve e_{k,\leq}^2(\beta_0) - \breve \sigma_{k,\leq}^2(\beta_0))\right| \geq C_1^3 t_n \right) +  \eps. 
\end{align}

To bound the first term on the RHS of \eqref{eq:etilde^2-sigmatilde^2'}, we partition the real line $\Re$ into $\{|y| \leq T_n\}$ and $\{|y| > T_n \}$, where $T_n = C_1^2 \log(n) M_n$. Then, 
\begin{align}\label{eq:H_I+II}
   & \mathbb P\left(\sup_{y \in \Re} \left| \sum_{k \in [n]} H_{k,y,\leq} (\breve e_{k,\leq}^2(\beta_0) - \breve \sigma_{k,\leq}^2(\beta_0))\right| \geq C_1^3 t_n \right) \notag \\
   & \leq \mathbb P\left(\sup_{|y| > T_n} \left| \sum_{k \in [n]} H_{k,y,\leq} (\breve e_{k,\leq}^2(\beta_0) - \breve \sigma_{k,\leq}^2(\beta_0))\right| \geq C_1^3 t_n \right) \notag \\
   & + \mathbb P\left(\sup_{|y| \leq T_n} \left| \sum_{k \in [n]} H_{k,y,\leq} (\breve e_{k,\leq}^2(\beta_0) - \breve \sigma_{k,\leq}^2(\beta_0))\right| \geq C_1^3 t_n \right) \notag \\
   & := I+II. 
\end{align}

\textbf{Bound Term $I$ on the RHS of \eqref{eq:H_I+II}.} Recall the definitions of $H_{k,y}$ in \eqref{eq:Hky} and $H_{k,y,\leq}$, in which 
$g_{n,y}(x)= \max \left(0, 1- \frac{d(x,A_y^{3 \delta_n})}{\delta_n} \right)$ and $f_{n,y}(x) = \mathbb{E} g_{n,y}(x + h_n \mathcal{N})$. Also recall the definition of $ \partial^2 f_{n,y}(x)$ in \eqref{eq:f''}. 

When $ y < - T_n$, we have
\begin{align*}
|  \partial^2 f_{n,y}(x)| & \leq h_n^{-2} \left(2\cdot1\{|x| \geq T_n/2\} + \mathbb E g_{n,y}(x+h_n \N)(\N^2 -1) 1\{|x| < T_n/2\}\right) \\
& \leq h_n^{-2} \left(2 \cdot 1\{|x| \geq T_n/2\} +  \mathbb E 1\{x+h_n \N \leq y + 3\delta_n\}(\N^2 +1) 1\{|x| < T_n/2\}\right) \\
& \leq h_n^{-2} \left(2\cdot1\{|x| \geq T_n/2\} +  \mathbb E 1\{h_n \N \leq -T_n/2 + 3\delta_n\}(\N^2 +1) 1\{|x| < T_n/2\}\right) \\
& \leq C h_n^{-2} \left(1\{|x| \geq T_n/2\} + \exp\left(- \frac{(T_n/2 - 3\delta_n)^2}{4h_n^2} \right)\right),
\end{align*}
where the first inequality uses the fact that $|g_{n,y}(x)| \leq 1$ and the last inequality uses the facts that $T_n/2 > 3\delta_n = 3C_h h_n$\footnote{This is because $T_n$ diverges to infinity while $h_n = o(1)$.} and 
\begin{align*}
 & \mathbb E 1\{h_n \N \leq -T_n/2 + 3\delta_n\}(\N^2 +1) \\
 & = \int_{-\infty}^{(-T_n/2 + 3\delta_n)/h_n} (u^2 + 1)\frac{1}{\sqrt{2\pi}}\exp(-u^2/2)du \\
 & \leq \int_{-\infty}^{(-T_n/2 + 3\delta_n)/h_n} (u^2 + 1)\frac{1}{\sqrt{2\pi}}\exp(-u^2/4)du \exp\left(-\frac{(T_n/2 - 3\delta_n)^2}{4h_n^2}\right) \\
 & \leq C\exp\left(-\frac{(T_n/2 - 3\delta_n)^2}{4h_n^2}\right).
\end{align*}
Similarly, when $y > T_n$, we have
\begin{align*}
|  \partial^2 f_{n,y}(x)| & \leq h_n^{-2} \left(2\cdot1\{|x| \geq T_n/2\} + \mathbb E g_{n,y}(x+h_n \N)(\N^2 -1) 1\{|x| < T_n/2\}\right) \\
& = h_n^{-2} \left(2\cdot1\{|x| \geq T_n/2\} + \mathbb E \left[1-g_{n,y}(x+h_n \N)\right](\N^2 -1) 1\{|x| < T_n/2\}\right) \\
& \leq h_n^{-2} \left(2\cdot1\{|x| \geq T_n/2\} +  \mathbb E 1\{x+h_n \N > y \}(\N^2 +1) 1\{|x| < T_n/2\}\right) \\
& \leq h_n^{-2} \left(2\cdot1\{|x| \geq T_n/2\} +  \mathbb E 1\{h_n \N \geq T_n/2\}(\N^2 +1) 1\{|x| < T_n/2\}\right) \\
& \leq C h_n^{-2} \left(1\{|x| \geq T_n/2\} + \exp\left(- \frac{T_n^2}{16h_n^2} \right)\right),
\end{align*}
where we use the fact that 
\begin{align*}
 \mathbb E g_{n,y}(x+h_n \N)(\N^2 -1) =  \mathbb E \left[1-g_{n,y}(x+h_n \N)\right](\N^2 -1)
\end{align*}
and 
\begin{align*}
    |1-g_{n,y}(x+h_n \N)| \leq 1\{x+h_n \N \geq y\}. 
\end{align*}
Therefore, we have
\begin{align*}
   \sup_{ |y| > T_n}  |  \partial^2 f_{n,y}(x)| \leq C h_n^{-2} \left(1\{|x| \geq T_n/2\} + \exp\left(- \frac{(T_n/2 - 3\delta_n)^2}{4h_n^2} \right)\right). 
\end{align*}
Denote $\mathbb I_{k-1} = 1\{ \max_{i \in [k-1]} \breve e_{i}^2(\beta_0) \leq C_1 M_n \}$ with $\mathbb I_{0} = 1$, we have
\begin{align}\label{eq:H}
& \sup_{|y| > T_n} |H_{k,y,\leq}|  \notag \\
& \leq   Ch_n^{-2} \mathbb E \left(1\{|s_k| \geq T_n/2\}\left[\frac{\sum_{i <k}  \Xi_{\lambda,ki} \breve \eta_{i} + \sum_{i >k}  \Xi_{\lambda,ki} \breve g_{i}}{\sqrt{R_\lambda}}  \right]^2 \mid \mathcal D  \right) \mathbb I_{k-1}  \notag \\
& + Ch_n^{-2} \exp\left(- \frac{(T_n/2 - 3\delta_n)^2}{4h_n^2} \right)\mathbb E\left( \left[\frac{\sum_{i <k}  \Xi_{\lambda,ki} \breve \eta_{i} + \sum_{i >k}  \Xi_{\lambda,ki} \breve g_{i}}{\sqrt{R_\lambda}}  \right]^2 \mid \mathcal D  \right) \mathbb I_{k-1} \notag \\
& \leq  Ch_n^{-2} \left[\mathbb P(|s_k| \geq T_n/2 \mid \mathcal D) \right]^{1/3}\left[ \mathbb E \left(\left\vert \frac{\sum_{i <k}  \Xi_{\lambda,ki} \breve \eta_{i} + \sum_{i >k}  \Xi_{\lambda,ki} \breve g_{i}}{\sqrt{R_\lambda}}  \right\vert^3 \mid \mathcal D  \right) \right]^{2/3} \mathbb I_{k-1} \notag \\
& + Ch_n^{-2} \exp\left(- \frac{(T_n/2 - 3\delta_n)^2}{4h_n^2} \right)\mathbb E\left( \left[\frac{\sum_{i <k}  \Xi_{\lambda,ki} \breve \eta_{i} + \sum_{i >k}  \Xi_{\lambda,ki} \breve g_{i}}{\sqrt{R_\lambda}}  \right]^2 \mid \mathcal D  \right) \mathbb I_{k-1} \notag \\
& \leq \frac{C C_1 M_n (\sum_{i \in [n], i \neq k} \Xi_{\lambda,ik}^2) }{h_n^2R_\lambda}\left\{ \left[\mathbb P(|s_k| \geq T_n/2 \mid \mathcal D) \right]^{1/3} + \exp\left(- \frac{(T_n/2 - 3\delta_n)^2}{4h_n^2} \right)\right\}\mathbb I_{k-1}, 
\end{align}
where the second inequality is by the H\"{o}lder's inequality and the third inequality is by \eqref{eq:theta^3} proved below.

Define $ \Xi_{\lambda}$ as an $n \times n$ matrix so that its  $(i,j)$th entry is just $\Xi_{\lambda,ij}$ if $i \neq j$ and its diagonal elements take value zero. 
In addition, let 
\begin{align*}
    & \Lambda_k = \diag( \breve e_{1}(\beta_0),\cdots,  \breve e_k(\beta_0),  \breve \sigma_{k+1}(\beta_0),\cdots,\breve \sigma_{n}(\beta_0) ), \\ 
& v_k = (\eta_{1},\cdots,\eta_{k-1},0,g_{k+1},\cdots, g_{n})^\top, \quad \text{and} \quad \mathcal A_k =  \Lambda_k  \Xi_{\lambda}  \Lambda_k.
\end{align*}

With these definitions, we have 
\begin{align*}
    s_k = v_k^\top \mathcal A_k v_k 
\end{align*}
and when $\mathbb I_{k-1} = 1$, 
\begin{align*}
||\mathcal A_k||_F^2    \leq C_1^2 M_n^2 \frac{\sum_{j \in [n]}\sum_{i \in [n], i \neq j} \Xi_{\lambda,ij}^2}{R_\lambda} =  C_1^2 M_n^2 . 
\end{align*}

Then, by the Hanson-Wright inequality (\citet[Theorem 6.2.1]{V18}) with $T_n = C_1^2 \log(n) M_n$ for some sufficiently large $C_1>0$, when  $\mathbb I_{k-1} = 1$, there exists a sufficiently large constant $C'>0$ such that  
\begin{align}\label{eq:HW}
    \mathbb P(|s_k| \geq T_n/2 \mid \mathcal D) & \leq 2 \exp\left[- c\min\left(\frac{T_n^2}{4 C ||\mathcal A_k||_F^2},\frac{T_n}{2C ||\mathcal A_k||_{op}} \right) \right] \notag \\
    & \leq 2 \exp\left[- c\min\left(\frac{T_n^2}{4 C ||\mathcal A_k||_F^2},\frac{T_n}{2C ||\mathcal A_k||_{F}} \right) \right] \notag \\
    & \leq 2 \exp\left[- c\min\left(\frac{C_1^4 M_n^2 \log^2 (n)}{4 C C_1^2 M_n^2 },\frac{C_1^2 M_n \log (n) }{2C C_1 M_n } \right) \right] \notag \\
      & = 2 n^{- cC_1}.
\end{align}
In addition, we have

\begin{align*}
\exp\left(- \frac{(T_n/2 - 3\delta_n)^2}{4h_n^2} \right) \leq \exp(- c C_1^2 \log^2(n)) \leq n^{-C_1}    
\end{align*}
for some fixed but sufficiently large $C_1$ and all sufficiently large $n$'s. 

Therefore, by  \eqref{eq:H}, we have 
\begin{align*}
&    \sup_{|y| >T_n} |H_{k,y,\leq}| \leq \frac{C_1 M_n (\sum_{i \in [n], i \neq k} \Xi_{\lambda,ik}^2) }{h_n^2R_\lambda n^{cC_1}} \leq \frac{C_1 p_n n^{1/q}}{ (p_n n^{3/q})^{2/(7-\zeta)} n^{cC_1}}
\end{align*}
and
\begin{align*}
 \sup_{|y| >T_n} \left| \sum_{k \in [n]} H_{k,y,\leq} (\breve e_{k,\leq}^2(\beta_0) - \breve \sigma_{k,\leq}^2(\beta_0))\right|  & \leq \sum_{k \in [n]}  \sup_{|y| >T_n} \left| H_{k,y,\leq} \right| C_1 M_n \\
& \leq \frac{C C_1^2 p_n n^{1+2/q}}{ (p_n n^{3/q})^{2/(7-\zeta)} n^{cC_1}} \leq \frac{C C_1^2 n^{1 + \frac{2(4-\zeta)}{q(7-\zeta)} }}{n^{cC_1}}.    
\end{align*}

By choosing a sufficiently large but fixed $C_1$, we have 
\begin{align*}
C_1^2 t_n & = C_1^2 \left[(M_n^2 h_n^{-4} p_n \log(n))^{1/2} + p_n M_n^2 h_n^{-2} \log(n)\right] \\
& > C_1^2 p_n M_n^2 h_n^{-2} \log(n) \\
& = C_1^2 \left(p_n n^{\frac{8-2 \zeta}{q(5-\zeta)}}\right)^{\frac{5-\zeta}{7-\zeta}} \log (n) \\
& \geq C_1^2 \left(n^{ \frac{8-2 \zeta}{q(5-\zeta)} -1 }\right)^{\frac{5-\zeta}{7-\zeta}} \log (n) \\
& \geq \frac{C C_1^2 n^{1 + \frac{2(4-\zeta)}{q(7-\zeta)} }}{n^{cC_1}},
\end{align*}
where we use the fact that $p_n \geq 1/n$. This implies, for some sufficiently large $C_1$, we have 
\begin{align}\label{eq:I_final}
    \text{$I$ on the RHS of \eqref{eq:H_I+II}} = 0. 
\end{align}

\textbf{Bound Term $II$ on the RHS of \eqref{eq:H_I+II}.} We can cover $[-T_n,T_n]$ by small intervals with center $y_l$ and length $\ell_n = \min(h_n^3t_n^2 /M_n^2,\delta_n)$. The total number of such small intervals needed to cover $[-T_n,T_n]$ is $L_n = 2\lceil T_n/\ell_n \rceil$, which grows in a polynomial rate in $n$ in the sense that $L_n = O(n^C)$ for some constant $C>0$.  Then, we have 
\begin{align}\label{eq:II}
 & \text{$II$ on the RHS of \eqref{eq:H_I+II}} \notag \\
 & \leq \mathbb P\left(\sup_{|y-y'| \leq \ell_n} \left| \sum_{k \in [n]} (H_{k,y,\leq} - H_{k,y',\leq}) (\breve e_{k,\leq}^2(\beta_0) - \breve \sigma_{k,\leq}^2(\beta_0))\right| \geq C_1^3 t_n/2 \right) \notag \\
&+ \mathbb P\left(\max_{l \in [L_n]} \left| \sum_{k \in [n]} H_{k,y_l,\leq} (\breve e_{k,\leq}^2(\beta_0) - \breve \sigma_{k,\leq}^2(\beta_0))\right| \geq C_1^3 t_n/2 \right) \notag \\
& : = II_1 + II_2.
\end{align}
To bound $II_1$ on the RHS of \eqref{eq:II}, we first note that, for any $(y_1,y_2)$ such that $y_1 \leq y_2\leq y_1 + \delta_n$, we have
\begin{align*}
  &  0 \leq g_{n,y_2}(a) - g_{n,y_1}(a) \\
  & = \frac{x - (y_1 + 3\delta_n)}{\delta_n} 1\{a \in (y_1+ 3\delta_n,y_2 + 3\delta_n)\} + \frac{y_2-y_1}{\delta_n}1\{x \in (y_2 + 3\delta_n,y_1 + 4\delta_n)\} \\
    & + \frac{y_2 + 4\delta_n - x}{\delta_n}1\{a \in (y_1 + 4\delta_n,y_2 + 4\delta_n)\}, 
\end{align*}
which implies, for $a = x + h_n \mathcal N$, $y_1' = y_1 + 3\delta_n - x$, and $y_2' = y_2 + 3\delta_n - x$,
\begin{align*}
& |        \partial^2 f_{n,y_1}(x) -     \partial^2 f_{n,y_2}(x)| \\
& \leq  h_n^{-2} \mathbb E |g_{n,y_2}(x+h_n \N) - g_{n,y_1}(x+h_n \N)|(\N^2 + 1) \\
& \leq \frac{1}{h_n^2 \delta_n}  \mathbb E \left[(h_n\N - y_1')1\{h_n\N \in (y_1',y_2')\}\right](\N^2+1) \\
& + \frac{1}{h_n^2 \delta_n}  \mathbb E \left[(y_2'-y_1')1\{h_n\N \in (y_2',y_1'+ \delta_n)\}\right](\N^2+1) \\
& + \frac{1}{h_n^2 \delta_n}  \mathbb E \left[(y_2'+\delta_n - h_n \N)1\{h_n\N \in (y_1'+\delta_n,y_2'+\delta_n')\}\right](\N^2+1) \\
& \leq \frac{C(y_2-y_1)}{h_n^3},
\end{align*}
where we use the fact that $\exp(-u^2/2)(u^2+1)$ is bounded. This implies 
\begin{align*}
& \sup_{|y-y'| \leq \ell_n}    |H_{k,y,\leq} - H_{k,y',\leq}| \\
& \leq C \sup_{|y-y'| \leq \ell_n}\mathbb E\left( \left| \partial^2 f_{n,y}(s_k) - \partial^2 f_{n,y'}(s_k) \right|\left[\frac{\sum_{i <k}  \Xi_{\lambda,ki} \breve \eta_{i} + \sum_{i >k}  \Xi_{\lambda,ki} \breve g_{i}}{\sqrt{R_\lambda}}  \right]^2  \mid \mathcal D \right)\mathbb I_{k-1} \\
& \leq C \frac{\ell_n}{h_n^3} \left[\frac{\sum_{i <k} \Xi_{\lambda,ki}^2 \tilde e^2_{i}(\beta_0) + \sum_{i >k} \Xi_{\lambda,ki}^2 \tilde \sigma_{i}^2(\beta_0) }{R_\lambda}  \right]\mathbb I_{k-1} \\
& \leq \frac{C C_1 \ell_n M_n}{h_n^3} \left[\frac{\sum_{i \in [n], i \neq k} \Xi_{\lambda,ki}^2}{R_\lambda}  \right],
\end{align*}
and thus, 
\begin{align*}
    & \sup_{|y-y'| \leq \ell_n} \left| \sum_{k \in [n]} (H_{k,y,\leq} - H_{k,y',\leq}) (\breve e_{k,\leq}^2(\beta_0) - \breve \sigma_{k,\leq}^2(\beta_0))\right| \\
    & \leq   \frac{C C_1 \ell_n M_n}{h_n^3} \sum_{k  \in [n]}\left[\frac{\sum_{i \in [n], i \neq k} 
    \Xi_{\lambda,ki}^2}{R_\lambda} \left|(\breve e_{k,\leq}^2(\beta_0) - \breve \sigma_{k,\leq}^2(\beta_0)) \right| \right]  \\
    & \leq \frac{C C_1^2 \ell_n M_n^2}{h_n^3} \leq C C_1^2 t_n^2,
\end{align*}
where the last inequality is by the definition of $\ell_n$.

Because $t_n \rightarrow 0$, we have 
\begin{align}\label{eq:II_1_final}
    & \text{$II_1$ on the RHS of \eqref{eq:II}} = 0. 
\end{align}

Last, we turn to $II_2$ on the RHS of \eqref{eq:II}. we note that, for any $l \in [L_n]$,  $H_{k,y_l,\leq} \in \mathcal E_{k-1}$, where $\mathcal E_{k-1}$ is the sigma field generated by $\breve e_{1}(\beta_0),\cdots,\breve e_{k-1}(\beta_0)$. Therefore, we have
$$\{H_{k,y_l,\leq}(\breve e_{k,\leq}^2(\beta_0) - \breve \sigma_{k,\leq}^2(\beta_0)), \mathcal E_k\}_{k \in [n]}$$ 
forms a martingale difference sequence. In addition, we have 
\begin{align*}
  & \max_{k \in [n]}  \left|H_{k,y_l,\leq}(\breve e_{k,\leq}^2(\beta_0) - \breve \sigma_{k,\leq}^2(\beta_0)) \right| \\
  & \leq \max_{k \in [n]}\left( \frac{\sum_{i <k-1} \Xi_{\lambda,ki}^2 \breve  e_{i}^2 (\beta_0) + \sum_{i >k} \Xi_{\lambda,ki}^2  \breve \sigma_{i}^2 (\beta_0) }{R_\lambda h_n^2} \right)1\{ \max_{i \in [k-1]} \breve e_{i}^2(\beta_0) \leq C_1 M_n \} 2 C_1M_n \\
  & \leq 2C_1^2 p_n M_n^2 h_n^{-2}
\end{align*}
and 
\begin{align*}
   V & \equiv \sum_{k \in [n]} \mathbb E \left[\left( H_{k,y_l,\leq}(\breve e_{k,\leq}^2(\beta_0) - \breve \sigma_{k,\leq}^2(\beta_0))\right)^2 \mid \mathcal E_{k-1}\right] \\
    & \leq C \sum_{k \in [n]}  H_{k,y,\leq}^2 \\
    & \leq C  \sum_{k \in [n]} \left( \frac{\sum_{i <k-1} \Xi_{\lambda,ki}^2 \breve e_{i}^2(\beta_0)  + \sum_{i >k} \Xi_{\lambda,ki}^2 \breve \sigma_{i}^2(\beta_0)  }{R_\lambda h_n^2} \right)^2 1\{ \max_{i \in [k-1]} \breve e_{i}^2(\beta_0) \leq C_1 M_n \} \\
    & \leq C C_1^2 M_n^2 h_n^{-4} p_n,
\end{align*}
where we use the fact that when $\max_{i \in [k-1]} \breve e_{i}^2(\beta_0) \leq C_1M_n$,
\begin{align*}
& \sum_{k \in [n]}    \left( \sum_{i <k-1} \Xi_{\lambda,ki}^2 \tilde e_{i}^2 (\beta_0) + \sum_{i >k} \Xi_{\lambda,ki}^2 \tilde \sigma_{i}^2(\beta_0)  \right)^2  \\
& \leq C C_1^2 M_n^2 \sum_{k \in [n]}    \left( \sum_{i \in [n], i \neq k} \Xi_{\lambda,ki}^2 \right)^2  \leq C C_1^2 M_n^2 \left(\max_{i \in [n]} \sum_{k \neq i} \Xi_{\lambda,ki}^2 \right) R_\lambda.
\end{align*}

Therefore, by Freedman's inequality (also known as Bernstein's inequality for the martingale difference sequence, \citet[Theorem 1.6]{F75}), we have
\begin{align}\label{eq:II_2_final}
& \text{$II_2$ on the RHS of \eqref{eq:II}} \notag \\
  & \leq \sum_{l \in [L_n]} \mathbb P\left( \left| \sum_{k \in [n]} H_{k,y_l,\leq} (\breve e_{k,\leq}^2(\beta_0) - \breve \sigma_{k,\leq}^2(\beta_0))\right| \geq C_1^3 t_n, V \leq C C_1^2 M_n^2 h_n^{-4} p_n  \right) \notag \\
  & \leq 2 \exp\left( \log(L_n)- \frac{C_1^6 t_n^2}{2C C_1^2 M_n^2 h_n^{-4} p_n + 4p_n C_1^5 M_n^2 h_n^{-2} t_n/3 } \right)  \leq n^{-c}
\end{align}
for some constant $c>0$. 

Combining \eqref{eq:H_I+II}, \eqref{eq:I_final}, \eqref{eq:II_1_final}, and \eqref{eq:II_2_final}, for a sufficiently large but fixed $C_1$, we have
\begin{align*}
\mathbb P\left(\sup_{y \in \Re} \left| \sum_{k \in [n]} H_{k,y,\leq} (\breve e_{k,\leq}^2(\beta_0) - \breve \sigma_{k,\leq}^2(\beta_0))\right| \geq t_n \right) \leq n^{-c}    
\end{align*}
for some constant $c>0$.

Therefore, for a sufficiently large $n$ such that $n^{-c} \leq \eps$, following \eqref{eq:etilde^2-sigmatilde^2'}, we have 
\begin{align*}
\mathbb P\left(\sup_{y \in \Re} \left| \sum_{k \in [n]} H_{k,y} (\breve e_k^2(\beta_0) - \breve \sigma_k^2(\beta_0))\right| \geq 4C_1^3 t_n \right) \leq 2\eps.
\end{align*}
This implies that 
\begin{align}\label{eq:f2}
    \sup_{y \in \Re} \left| \sum_{k \in [n]} H_{k,y} (\breve e_k^2(\beta_0) - \breve \sigma_k^2(\beta_0))\right| = O_P(t_n).
\end{align}

\medskip

\noindent \textbf{\Large{Step 2.2: Bound on $\sum_{k\in [n]}\frac{\mathbb E(|\mathcal{S}_k|^3 +|\breve{\mathcal{S}}_k|^3 |\mathcal D)}{ h_n^3}$}}

Recall $\breve \eta_i = \eta_i \breve e_i(\beta_0)$ and $\breve g_i = g_i \tilde \sigma_i(\beta_0)$, where $\breve e_i(\beta_0) = \tilde e_i + \Delta(\Pi_i + \tilde v_i)$, $\breve \sigma_i^2(\beta_0) = \mathbb E \breve e_i^2(\beta_0)$, $\{\eta_i\}_{i\in [n]}$ is an i.i.d. sequence of random variables with zero mean and unit variance, and $\{g_i\}_{i \in [n]}$ is an i.i.d. sequence of standard normal random variables. Let 
\begin{align*}
     \theta_{k,i} = \begin{cases}
        \Xi_{\lambda, ki} \breve \eta_{i}  \quad i < k\\
        \Xi_{\lambda, ki} \breve g_{i}  \quad i > k.
    \end{cases}
\end{align*}
Then, we have
\begin{align*}
    \mathcal{S}_k = 2 \breve \eta_k \frac{\sum_{i \in [n], i \neq k}  \theta_{k,i}}{\sqrt{R_\lambda}} \quad \text{and} \quad     \mathbb E (|\mathcal{S}_k|^3 \mid \mathcal D) \leq \frac{C |\breve e_k(\beta_0)|^3}{R_\lambda^{3/2}} \mathbb E\left( \left|\sum_{i \in [n],i \neq k}  \theta_{k,i} \right|^3 \mid \mathcal D \right).
\end{align*}
Conditionally on data ($\mathcal D$), $\{ \theta_{k,i}\}_{i \in [n],i \neq k}$ is a sequence of independent mean zero random variables. By Marcinkiewicz-Zygmund inequality, on  $\{ \max_{i \in [n]} \breve e_i^2(\beta_0) \leq C_1 M_n\}$, we have 
\begin{align} \label{eq:theta^3}
 \mathbb{E} \left(\left|\sum_{i \in [n], i \neq k}  \theta_{k,i} \right|^3 \mid \mathcal D\right) & \leq  \mathbb{E} \left[ \left( \sum_{i \in [n], i \neq k}  \theta_{k,i}^2 \right)^{3/2} \mid \mathcal D\right]  \notag \\
 & \leq \mathbb{E} \left( (\sum_{i \in [n], i \neq k}  \theta_{k,i}^2)^2 \mid \mathcal D \right)^{3/4} \notag  \\
 & =  \left[\sum_{ i \in [n], i \neq k} \sum_{j \in [n], j \neq k} \mathbb E( \theta_{k,i}^2\theta_{k,j}^2 | \mathcal D)\right]^{3/4}  \notag \\
    & \leq \left[C_1^2M_n^2 (\sum_{i \in [n], i \neq k} \Xi_{\lambda,ik}^2)^2\right]^{3/4}. 
\end{align}
% \begin{align*}
%     \mathbb{E} \left( (\sum_{i \in [n], i \neq k}  \theta_{k,i}^2)^2 \mid \mathcal D \right)  & =  \sum_{ i \in [n], i \neq k} \sum_{j \in [k], j \neq k} \mathbb E( \theta_{k,i}^2\theta_{k,j}^2 | \mathcal D) \\
%     \leq & M_n (\sum_{i \in [n], i \neq k} P_{\lambda,ik}^2)^2 + M_n^2 (\sum_{i \in [n], i \neq k}P_{\lambda,i_{i},k,l}^4) \\
%     & \leq M_n (\sum_{i \in [n], i \neq k} P_{\lambda,ik}^2)^2 + M_n^2 (\sum_{i \in [n], i \neq k}P_{\lambda,ik}^2) .
% \end{align*}

This implies, on $\{ \max_{i \in [n]} \breve e_i^2 (\beta_0) \leq C_1 M_n\}$,  
\begin{align*}
\sum_{k \in [n]}      \mathbb E (|\mathcal{S}_k|^3 \mid \mathcal D) & \leq \sum_{k \in [n]} \frac{C C_1^{3/2}|\breve e_k(\beta_0)|^3}{R_\lambda^{3/2}} \left[M_n^2 (\sum_{i \in [n], i \neq k} \Xi_{\lambda,ik}^2)^2 \right]^{3/4} \\
& \leq CC_1^{3/2} p_n^{1/2} M_n^{3/2} \sum_{k \in [n]} \frac{(\sum_{i \in [n], i \neq k} \Xi_{\lambda,ik}^2) |\breve e_k(\beta_0)|^3}{R_\lambda} \\
& \leq CC_1^{3/2}p_n^{1/2} M_n^{3/2} \left(C + \left|\sum_{k \in [n]} \frac{(\sum_{i \in [n], i \neq k} \Xi_{\lambda,ik}^2) (|\breve e_k(\beta_0)|^3 - \mathbb E |\breve e_k(\beta_0)|^3)}{R_\lambda} \right| \right).
\end{align*}
In addition, because 
\begin{align*}
 & Var\left( \sum_{k \in [n]} \frac{(\sum_{i \in [n], i \neq k} \Xi_{\lambda,ik}^2) (|\breve e_k(\beta_0)|^3 - \mathbb E |\breve e_k(\beta_0)|^3)}{R_\lambda} \right) \\
  & \leq C \frac{\sum_{k \in [n]} (\sum_{i \in [n], i \neq k} \Xi_{\lambda,ik}^2)^2  }{K^2_\lambda } \leq C p_n.
\end{align*}
Therefore, for any $\eps'>0$, there exist sufficiently large constants $\tilde C>0$ and $C_1>0$ such that 
\begin{align*}
&   \mathbb P \left(  \sum_{k \in [n]}      \mathbb E (|\mathcal{S}_k|^3 \mid \mathcal D) > \tilde C ( p_n^{1/2} n^{3/(2q)})  \right) \\
& \leq \mathbb P \left(  \sum_{k \in [n]}      \mathbb E (|\mathcal{S}_k|^3 \mid \mathcal D) > \tilde C ( p_n^{1/2} n^{3/(2q)}), \max_{i \in [n]} \breve e_i^2 (\beta_0) \leq C_1 M_n  \right) + \mathbb P \left(  \max_{i \in [n]} \breve e_i^2 (\beta_0) > C_1 M_n \right) \\
& \leq \mathbb P \left(  CC_1^{3/2} \left(C + \left|\sum_{k \in [n]} \frac{(\sum_{i \in [n], i \neq k} \Xi_{\lambda,ik}^2) (|\breve e_k(\beta_0)|^3 - \mathbb E |\breve e_k(\beta_0)|^3)}{R_\lambda} \right| \right) > \tilde C  \right) + \frac{C}{C_1^q} \\
& \leq \frac{Var\left( \sum_{k \in [n]} \frac{(\sum_{i \in [n], i \neq k} \Xi_{\lambda,ik}^2) (|\breve e_k(\beta_0)|^3 - \mathbb E |\breve e_k(\beta_0)|^3)}{R_\lambda} \right) }{ C \tilde C} + \frac{C}{C_1^q} \\
& \leq \frac{C P_n}{ \tilde C}  + \frac{C}{C_1^q} \leq \eps'.
\end{align*}
This implies 
\begin{align*}
 \sum_{k \in [n]}      \mathbb E (|\mathcal{S}_k|^3 \mid \mathcal D) = O_P( p_n^{1/2} n^{3/(2q)}).
\end{align*}
Similarly, we have 
\begin{align}
&  \sum_{k \in [n]}      \mathbb E (|\breve{\mathcal{S}}_k|^3 \mid \mathcal D) = O_P( p_n^{1/2} n^{3/(2q)}), \quad \text{and thus,}  \notag \\
&     \sum_{k\in [n]}\frac{\mathbb E(|\mathcal{S}_k|^3 +|\breve{\mathcal{S}}_k|^3 |\mathcal D)}{ h_n^3} =  O_P\left( \frac{p_n^{1/2} n^{3/(2q)}}{ h_n^3}\right) = O_P\left( (p_n n^{3/q})^{\frac{1-\zeta}{2(7 - \zeta)}}\right). \label{eq:f3}
\end{align}

\medskip

\noindent \textbf{ \Large{Step 2.3: Concluding Step 2 }}

By Lemma \ref{lem:anti-concentration} in the Supplemental Appendix, we have
\begin{align}\label{eq:term3}
 \sup_{y \in \Re} \mathbb P(|Q^*(\beta_0)-y| \leq 3\delta_n ) \leq  C_\zeta 3^{(1-\zeta)/2} C_h^{(1-\zeta)/2} h_n^{(1-\zeta)/2}
\end{align}
for any $\zeta \in (0,1)$ and $C_{\zeta} \in (0,\infty)$ that only depends on $\zeta$ and $\underline c$ in Assumption \ref{ass:reg}.3. 

Then, combining \eqref{eq:Qstar-Q} and \eqref{eq:f1}, for $\eps$ used in \eqref{eq:Qstar-Q}, we have
\begin{align*}
& \mathbb P \left( \sup_{y \in \Re} \left\vert  \mathbb P(\breve Q^*(\beta_0) \leq y | \mathcal D) - \mathbb P(Q^*(\beta_0) \leq y ) \right\vert > 4 \eps \right) \\
& \leq \mathbb P \left( \begin{pmatrix}
   & (1+\eps) \sup_{y \in \Re} \left| \mathbb E( f_{n,y}(\breve Q^*(\beta_0)) | \mathcal D) - \mathbb E( f_{n,y}(Q^*(\beta_0)) )\right| \notag \\
    & + \eps + \sup_{y \in \Re} \mathbb P(|Q^*(\beta_0)-y| \leq 3\delta_n)
\end{pmatrix}    > 4\eps \right) \\
& \leq \mathbb P \left(\sup_{y \in \Re}\left| \sum_{k \in [n]} H_{k,y} (\breve e_k^2(\beta_0) - \breve \sigma_k^2(\beta_0))\right| > \frac{\eps}{1+\eps} \right) + \mathbb P \left( \sum_{k\in [n]}\frac{\mathbb E(|\mathcal{S}_k|^3 +|\breve{\mathcal{S}}_k|^3 |\mathcal D)}{ h_n^3} > \frac{C_h\eps}{1+\eps} \right)  \\
& + 1\{ \sup_{y \in \Re} \mathbb P(|Q^*(\beta_0)-y| \leq 3\delta_n) > \eps \}.
\end{align*}
Taking $\limsup_{n \rightarrow \infty}$, we have
\begin{align*}
& \limsup_{n \rightarrow \infty}     \mathbb P \left( \sup_{y \in \Re} \left(  \mathbb P(\breve Q^*(\beta_0) \leq y | \mathcal D) - \mathbb P(Q^*(\beta_0) \leq y ) \right) > 4 \eps \right) \\
& \leq \limsup_{n \rightarrow \infty} \mathbb P \left(\sup_{y \in \Re}\left| \sum_{k \in [n]} H_{k,y} (\breve e_k^2(\beta_0) - \breve \sigma_k^2(\beta_0))\right| > \frac{\eps}{1+\eps} \right) \\
& + \limsup_{n \rightarrow \infty} \mathbb P \left( \sum_{k\in [n]}\frac{\mathbb E(|\mathcal{S}_k|^3 +|\breve{\mathcal{S}}_k|^3 |\mathcal D)}{  h_n^3} > \frac{C_h \eps}{1+\eps} \right) \\
& = 0,
\end{align*}
where the first inequality holds by \eqref{eq:term3} and that $h_n = o(1)$ so that for sufficiently large $n$,  
\begin{align*}
\sup_{y \in \Re} \mathbb P(|Q^*(\beta_0)-y|  \leq 3\delta_n) \leq C_\zeta 3^{(1-\zeta)/2} C_h^{(1-\zeta)/2} h_n^{(1-\zeta)/2} < \eps, 
\end{align*}
and the equality is by \eqref{eq:f2} and \eqref{eq:f3}. This implies 
\begin{align}\label{eq:step2_final*}
    \sup_{y \in \Re} \left\vert   \mathbb P(\breve Q^*(\beta_0) \leq y | \mathcal D) - \mathbb P(Q^*(\beta_0) \leq y ) \right\vert = o_P(1).
\end{align}

\medskip

\noindent \textbf{\Large{Step 3: Concluding the Entire Proof}}

Combining \eqref{eq:step2_final*} with \eqref{eq:Fhat-F2}, we have 
\begin{align}\label{eq:step1_final}
\sup_{y \in \Re}|\mathbb P(\widehat{Q}^*(\beta_0) \leq y | \mathcal D) - \mathbb P(\breve Q^*(\beta_0)\leq y | \mathcal D)| = o_P(1).
\end{align}
Then, combining \eqref{eq:Fhat-F} with \eqref{eq:step2_final*} and \eqref{eq:step1_final}, we have the desired result that 
\begin{align*}
 \sup_{y \in \Re}|\mathbb P(\widehat{Q}^*(\beta_0) \leq y | \mathcal D) - \mathbb P(Q^*(\beta_0)\leq y)| = o_P(1). 
\end{align*}

\section{Proof of Theorem \ref{cor:size}}\label{sec:pf_size}
Denote $C_\alpha(\beta_0)$ as the $1-\alpha$ quantile of $Q(\beta_0)$, 
\begin{align*}
\eta_n =     \sup_{y \in \Re}\left|\mathbb P ( \widehat Q^*(\beta_0) \leq y| \mathcal D) - \mathbb P( Q(\beta_0) \leq y)\right|, \quad \text{and} \quad \eta_n' = \left|\mathbb P \left(\widehat Q(\beta_0) \leq y\right) - \mathbb P( Q(\beta_0) \leq y) \right|. 
\end{align*}

By Theorems \ref{thm:main_null} and \ref{thm:Fhat-F}, we have $\eta_n' = o(1)$ and $\eta_n = o_P(1)$. 

Then, for any $\eps \in (0,\alpha)$, on the event that $\mathcal E_n = \{\eta_n \leq \eps\}$, we have
\begin{align*}
\mathbb P ( \widehat Q^*(\beta_0) \leq C_{\alpha-\eps}(\beta_0)| \mathcal D)     \geq \mathbb P( Q(\beta_0) \leq C_{\alpha-\eps}(\beta_0)) - \eta_n  \geq 1-\alpha + \eps - \eta_n \geq 1-\alpha
\end{align*}
and 
\begin{align*}
\mathbb P ( \widehat Q^*(\beta_0) \leq C_{\alpha+\eps}(\beta_0)| \mathcal D)     \leq \mathbb P( Q(\beta_0) \leq C_{\alpha+\eps}(\beta_0)) + \eta_n  \leq 1-\alpha - \eps + \eta_n     \leq 1-\alpha.
\end{align*}

This implies, on $\mathcal E_n$, we have 
\begin{align*}
\mathbb P\left( C_{\alpha+\eps}(\beta_0) \leq \widehat{\mathcal C}^*_{\alpha}(\beta_0)   \leq C_{\alpha-\eps}(\beta_0) \right) \leq \mathbb P(\mathcal E_n) = 1-o(1). 
\end{align*}
Therefore, we have
\begin{align*}
\mathbb P\left( \widehat Q(\beta_0) > \widehat{\mathcal C}^*_{\alpha}(\beta_0) \right) & \leq  \mathbb P\left( Q(\beta_0) > \widehat{\mathcal C}^*_{\alpha}(\beta_0) \right) + \eta_n'  \\
& \leq  \mathbb P\left( Q(\beta_0) > C_{\alpha+\eps}(\beta_0) \right) + \eta_n' + \mathbb P(\mathcal E_n^c) \\
& \leq \alpha + \eps + o(1)
\end{align*}
and 
\begin{align*}
\mathbb P\left( \widehat Q(\beta_0) \leq \widehat{\mathcal C}^*_{\alpha}(\beta_0) \right) & \leq  \mathbb P\left( Q(\beta_0) \leq \widehat{\mathcal C}^*_{\alpha}(\beta_0) \right) + \eta_n'  \\
& \geq  \mathbb P\left( Q(\beta_0) \leq C_{\alpha-\eps}(\beta_0) \right) + \eta_n' + \mathbb P(\mathcal E_n^c) \\
& \leq 1-\alpha +\eps + o(1),
\end{align*}
which is equivalent to 
\begin{align*}
\mathbb P\left( \widehat Q(\beta_0) > \widehat{\mathcal C}^*_{\alpha}(\beta_0) \right)   \geq \alpha-\eps - o(1). 
\end{align*}

This implies 
\begin{align*}
    \left\vert \mathbb P\left( \widehat Q(\beta_0) > \widehat{\mathcal C}^*_{\alpha}(\beta_0) \right)  - \alpha \right\vert \leq \eps + o(1).  
\end{align*}
As $\eps$ is arbitrary, we have established the desired result.

\section{Proof of Theorem \ref{thm:power_divK}}\label{sec: proof_power_divK}
By Theorem \ref{thm:main_null}, we have
\begin{align*}
& \sup_{y \in \Re} \left|\mathbb P \left( \widehat Q(\beta_0) \leq y\right) - \mathbb P \left(\frac{\sum_{i \in [n]} \sum_{j \in [n], j \neq i} (g_i \tilde \sigma_i(\beta_0) + \Delta \Pi_i) \Xi_{\lambda,ij} (g_j \tilde \sigma_j(\beta_0) + \Delta \Pi_j) }{\sqrt{R_\lambda} } + C(\Delta) \leq y\right)\right| \\
&= o_P(1), 
\end{align*}
 where $\{g_i\}_{i \in [n]}$ is a sequence of i.i.d. standard normal random variables. Note that 
 \begin{align}\label{eq:Q(beta0)}
  & \frac{\sum_{i \in [n]} \sum_{j \in [n], j \neq i} (g_i \tilde \sigma_i(\beta_0) + \Delta \Pi_i) \Xi_{\lambda,ij} (g_j \tilde \sigma_j(\beta_0) + \Delta \Pi_j) }{\sqrt{R_\lambda} } + C(\Delta) \notag \\
  &  = \frac{\sum_{i \in [n]} \sum_{j \in [n], j \neq i} g_i \tilde \sigma_i(\beta_0)  \Xi_{\lambda,ij} g_j \tilde \sigma_j(\beta_0) }{\sqrt{R_\lambda} } \notag \\
  & + \frac{\sum_{i \in [n]} 2 g_i \tilde \sigma_i(\beta_0) \left(\sum_{j \in [n], j \neq i}  \Xi_{\lambda,ij} \Delta \Pi_j\right) }{\sqrt{R_\lambda} } +  \frac{ \sum_{i \in [n]} \sum_{j \in [n], j \neq i} \Pi_i P_{\lambda,ij} \Pi_j \Delta^2 }{\sqrt{R_\lambda} }. 
 \end{align}
 
 To analyze the first term on the RHS of \eqref{eq:Q(beta0)}, we denote $(\varpi_1,\cdots,\varpi_n)$ as the eigenvalues of matrix 
 \begin{align*}
     \diag(\tilde \sigma_1(\beta_0),\cdots,\tilde \sigma_n(\beta_0)) \Xi_{\lambda} \diag(\tilde \sigma_1(\beta_0),\cdots,\tilde \sigma_n(\beta_0)).
 \end{align*}
Then, we have
\begin{align*}
     \frac{\sum_{i \in [n]} \sum_{j \in [n], j \neq i} g_i \tilde \sigma_i(\beta_0) \Xi_{\lambda,ij} g_j \tilde \sigma_j(\beta_0) }{\sqrt{R_\lambda} } \stackrel{d}{=} \sum_{i=1}^n g_i^2 \varpi_i/\sqrt{R_{\lambda}} = \sum_{i=1}^n (g_i^2 -1)\varpi_i/\sqrt{R_{\lambda}},
\end{align*}
where the second equality is by the fact that 
\begin{align*}
\sum_{i=1}^n \varpi_i &= \operatorname{tr} \left( \operatorname{diag}(\tilde \sigma_1^2 (\beta_0), \dots, \tilde \sigma_n^2 (\beta_0)) \Xi_\lambda \right) = 0.
\end{align*}

Let $\Psi(\beta_0) = Var(\sum_{i=1}^n g_i^2 \varpi_i/\sqrt{R_{\lambda}}) = \frac{2 \sum_{i \in [n]} \sum_{j \in [n], j \neq i}  \tilde \sigma^2_i(\beta_0) \Xi_{\lambda,ij}^2 \tilde \sigma_j^2(\beta_0) }{R_\lambda }$. Then, we have
\begin{align*}
    \sum_{i\in [n]} \mathbb E \frac{\left((g_i^2 -1)\varpi_i/\sqrt{R_{\lambda}} \right)^4}{\Psi^2(\beta_0) } & \lesssim \frac{\max_i \varpi_i^2}{\sum_{i \in [n]} \sum_{j \in [n], j \neq i}  \tilde \sigma^2_i(\beta_0) \Xi_{\lambda,ij}^2 \tilde \sigma_j^2(\beta_0) } \\
    & \lesssim \frac{\max_i \varpi_i^2}{ R_{\lambda}} = o(1),
\end{align*}
where the last inequality holds because  
\begin{align*}
\max_i \varpi_i^2 \leq \left\Vert     \diag(\tilde \sigma_1(\beta_0),\cdots,\tilde \sigma_n(\beta_0)) \Xi_{\lambda} \diag(\tilde \sigma_1(\beta_0),\cdots,\tilde \sigma_n(\beta_0)) \right\Vert_{op}^2 \leq C
\end{align*}
for some constant $C<\infty$.  This verifies the Lyapunov's condition.

Therefore, by CLT, we have 
\begin{align*}
  \Psi^{-1/2}(\beta_0)   \frac{\sum_{i \in [n]} \sum_{j \in [n], j \neq i} g_i \tilde \sigma_i(\beta_0) \Xi_{\lambda,ij} g_j \tilde \sigma_j(\beta_0) }{\sqrt{R_\lambda} }  \convD \N(0,1).
\end{align*}

For the second term on the RHS of \eqref{eq:Q(beta0)}, we note that 
\begin{align*}
& Var\left(     \frac{\sum_{i \in [n]} 2 g_i \tilde \sigma_i(\beta_0) \left(\sum_{j \in [n], j \neq i}  \Xi_{\lambda,ij} \Delta \Pi_j\right) }{\sqrt{R_\lambda} } \right) \lesssim \frac{ \sum_{i \in [n]}\left(\sum_{j \in [n], j \neq i}  \Xi_{\lambda,ij} \Delta \Pi_j\right)^2   }{R_\lambda} \\
& \lesssim \frac{\Delta^2 \Pi^\top  \Xi_\lambda^2 \Pi}{R_\lambda} 
 \lesssim \frac{1}{\sqrt{R_{\lambda} }} \frac{||\Pi||_2^2 \Delta^2}{\sqrt{R_{\lambda}}} = o(1). 
\end{align*}

Last, we have $\Psi(\beta_0) \geq c$ and  $\Psi^{-1/2}(\beta_0) \frac{ \sum_{i \in [n]} \sum_{j \in [n], j \neq i} \Pi_i P_{\lambda,ij} \Pi_j \Delta^2 }{\sqrt{R_\lambda} } \rightarrow \mu(\beta_0)$. Therefore, we have 
\begin{align*}
\Psi^{-1/2}(\beta_0)   \left[ \frac{\sum_{i \in [n]} \sum_{j \in [n], j \neq i} (g_i \tilde \sigma_i(\beta_0) + \Delta \Pi_i) \Xi_{\lambda,ij} (g_j \tilde \sigma_j(\beta_0) + \Delta \Pi_j) }{\sqrt{R_\lambda} } + C(\Delta) \right] \convD \N(\mu(\beta_0),1), 
\end{align*}
which, combined with Theorem \ref{thm:main_null}, implies 
\begin{align*}
   \Psi^{-1/2}(\beta_0) \left[\widehat Q(\beta_0) + C(\Delta) \right] \convD \N(\mu(\beta_0),1). 
\end{align*}

Similar to the analysis of the first term on the RHS of \eqref{eq:Q(beta0)}, we can show that 
\begin{align*}
  \tilde \Psi^{-1/2}(\beta_0)  Q^*(\beta_0) \convD \N(0,1), 
\end{align*}
where 
\begin{align*}
  \breve \Psi (\beta_0) =  \frac{2 \sum_{i \in [n]} \sum_{j \in [n], j \neq i}  (\tilde \sigma^2_i(\beta_0) + \Delta^2 \Pi_i^2) \Xi_{\lambda,ij}^2 (\tilde \sigma^2_j(\beta_0) + \Delta^2 \Pi_j^2) }{R_\lambda }.
\end{align*}
In addition, we have $\Psi(\beta_0) \geq c$ for some constant $c>0$ and 
\begin{align*}
\left\vert \frac{   \breve \Psi(\beta_0) - \Psi(\beta_0) }{\Psi(\beta_0) } \right\vert & \lesssim \frac{ \sum_{i \in [n]} \sum_{j \in [n], j \neq i}  \Delta^2 \Pi_i^2 \Xi_{\lambda,ij}^2 }{R_\lambda } + \frac{\sum_{i \in [n]} \sum_{j \in [n], j \neq i}   \Delta^2 \Pi_i^2 \Xi_{\lambda,ij}^2  \Delta^2 \Pi_j^2 }{R_\lambda } \\
& \lesssim \frac{ \sum_{i \in [n]} \sum_{j \in [n], j \neq i}  \Delta^2 \Pi_i^2 \Xi_{\lambda,ij}^2 }{R_\lambda }  \lesssim \frac{1}{\sqrt{R_{\lambda} }} \frac{||\Pi||_2^2 \Delta^2}{\sqrt{R_{\lambda}}} = o(1), 
\end{align*}
where the second inequality is by the fact that $|\Delta|$ and $|\Pi_j|$ are assumed to be bounded and the last inequality is by the fact that 
\begin{align*}
 \sum_{j \in [n], j \neq i}\Xi_{\lambda,ij}^2 \lesssim  \sum_{j \in [n], j \neq i} (P_{\lambda,ij}^2 + B_{\lambda,ij}^2)  \lesssim P_{\lambda,ii} + P_{W,ii} \lesssim 1.    
\end{align*}
This implies 
\begin{align}\label{eq:normal*}
 \Psi^{-1/2}(\beta_0)  Q^*(\beta_0) \convD \N(0,1). 
\end{align}

Next, we consider the limit of the bootstrap critical value. Recall that $\widehat{\mathcal C}^*_{\alpha}(\beta_0) = \inf\{y \in \Re: 1-\alpha \leq \hat F_{\beta_0}^*(y)\}$, where 
\begin{align*}
    \hat F_{\beta_0}^*(y) = \mathbb P ( \widehat Q^*(\beta_0) \leq y| \mathcal D).
\end{align*}
and $F_{\beta_0}^*(y) = \mathbb P( Q^*(\beta_0) \leq y)$ and $\mathcal C_{\alpha}^*(\beta_0) = \inf\{y \in \Re: 1-\alpha \leq F_{\beta_0}^*(y)\}$. 

Further denote 
\begin{align*}
\eta_n =     \sup_{y \in \Re}\left|\hat F_{\beta_0}^*(y) - F_{\beta_0}^*(y)\right|.
\end{align*}

 % \quad \text{and} \quad \eta_n' = \left|\mathbb P \left(\widehat Q(\beta_0) \leq y\right) - F_{\beta_0}(y) \right|. 

By Theorem \ref{thm:Fhat-F}, we have $\eta_n = o_p(1)$. Then, for any $y_0$ and any $\eps>0$ such that $1-\alpha \leq \hat F_{\beta_0}^*(y_0)$ and $\eta_n \leq \eps$, we have
\begin{align*}
 1- \alpha  \leq \hat F_{\beta_0}^*(y_0) \leq F_{\beta_0}^*(y_0) + \sup_{y \in \Re}|\hat F_{\beta_0}^*(y) - F_{\beta_0}^*(y)| \leq F_{\beta_0}^*(y_0) + \eps.
\end{align*}
Therefore, when $\eta_n \leq \eps$, we have
\begin{align*}
 \mathcal C_{\alpha+\eps}^*(\beta_0) \leq \widehat{\mathcal C}^*_{\alpha}(\beta_0). 
\end{align*}
% In addition, by definition, we have
% \begin{align*}
%     \widehat{\mathcal C}^*_{\alpha}(\beta_0) \leq y \quad \text{if and only if} \quad \alpha \leq \hat F_{\beta_0}^{*}(y).
% \end{align*}
% Then, we have
% \begin{align*}
% \mathbb P \left( \widehat Q(\beta_0) \geq \widehat{\mathcal C}^*_{\alpha}(\beta_0)\right) & \leq \mathbb P \left(\widehat Q(\beta_0) \geq \widehat{\mathcal C}^*_{\alpha}(\beta_0), \eta_n \leq \eps\right)  + \mathbb P(\eta_n >\eps) \\
% & \leq \mathbb P \left( \widehat Q(\beta_0) \geq \mathcal C_{\alpha+\eps}(\beta_0)\right) + \mathbb P(\eta_n >\eps) \\
% & \leq \mathbb P \left( Q(\beta_0) \geq \mathcal C_{\alpha+\eps}(\beta_0)\right) + \eta_n' + \mathbb P(\eta_n >\eps) \\
% & = \alpha + \eps + \eta_n' + \mathbb P(\eta_n >\eps),
% \end{align*}
% where the last inequality holds by Assumption \ref{ass:den}.  

Similarly, for any $y_0$ such that $1-(\alpha-\eps) \leq F_{\beta_0}^*(y_0)$, we have
\begin{align*}
 1-   (\alpha-\eps) \leq F_{\beta_0}^*(y_0) \leq \hat F_{\beta_0}^*(y_0) + \sup_{y \in \Re}|\hat F_{\beta_0}^*(y) - F_{\beta_0}^*(y)| \leq \hat F_{\beta_0}^*(y_0) + \eta_n,
\end{align*}
which implies, when $\eta_n \leq \eps$, 
\begin{align*}
\widehat{\mathcal C}^*_{\alpha}(\beta_0) \leq  \mathcal C_{\alpha-\eps}^*(\beta_0). 
\end{align*}
Therefore, for any $\eps>0$, we have
\begin{align}\label{eq:etan}
    \{\eta_n \leq \eps\} \subset \left\{ \mathcal C_{\alpha+\eps}^*(\beta_0) \leq  \widehat{\mathcal C}^*_{\alpha}(\beta_0) \leq  \mathcal C_{\alpha-\eps}^*(\beta_0) \right\}.
\end{align}

In addition, by \eqref{eq:normal*}, we have
\begin{align}\label{eq:normal2}
\Psi^{-1/2}(\beta_0) \mathcal C_{\alpha+\eps}^*(\beta_0) \convP z_{\alpha+\eps}  \quad \text{and} \quad \Psi^{-1/2}(\beta_0) \mathcal C_{\alpha-\eps}^*(\beta_0) \convP z_{\alpha-\eps}.    
\end{align}

Denote $f_N(\cdot)$ as the standard normal PDF. Then, for any $\eps'>0$, we can choose a sufficiently small $\eps$ such that $0< \eps \leq \min(\alpha/2,f_N(z_{\alpha/2})\eps')$ which implies 
\begin{align} \label{eq:zeps}
    & |z_{\alpha-\eps} - z_\alpha| \leq \eps/f_N(z_{\alpha-\eps}) \leq  \eps/f_N(z_{\alpha/2}) \leq \eps' \quad \text{and} \notag \\
    & |z_{\alpha+\eps} - z_\alpha| \leq \eps/f_N(z_{\alpha}) \leq  \eps/f_N(z_{\alpha/2}) \leq \eps'.
\end{align}

Then, we have
\begin{align*}
&    \mathbb P \left(  \left| \Psi^{-1/2}(\beta_0) \widehat{\mathcal C}^*_{\alpha}(\beta_0) - z_\alpha  \right| > 2\eps' \right) \\
& \leq     \mathbb P \left(  \left| \Psi^{-1/2}(\beta_0) \widehat{\mathcal C}^*_{\alpha}(\beta_0) - z_\alpha  \right| > 2\eps', \mathcal C_{\alpha+\eps}^*(\beta_0) \leq  \widehat{\mathcal C}^*_{\alpha}(\beta_0) \leq  \mathcal C_{\alpha-\eps}^*(\beta_0) \right)  + \mathbb P \left( \eta_n > \eps \right) \\
    & \leq    \mathbb P \left(  \left| \Psi^{-1/2}(\beta_0) \mathcal C_{\alpha+\eps}^*(\beta_0) - z_\alpha  \right| > 2\eps'\right)  + \mathbb P \left(  \left| \Psi^{-1/2}(\beta_0) \mathcal C_{\alpha-\eps}^*(\beta_0) - z_\alpha  \right| > 2\eps'\right)  + \mathbb P \left( \eta_n > \eps \right) \\
    & \leq    \mathbb P \left(   \left| z_{\alpha+\eps} - z_{\alpha}  \right| + \left| \Psi^{-1/2}(\beta_0) \mathcal C_{\alpha+\eps}^*(\beta_0) - z_{\alpha+\eps}  \right| > 2\eps'\right) \\ 
    & + \mathbb P \left(  \left| z_{\alpha-\eps} - z_{\alpha}  \right| + \left| \Psi^{-1/2}(\beta_0) \mathcal C_{\alpha-\eps}^*(\beta_0) - z_{\alpha-\eps}  \right| > 2\eps'\right) + \mathbb P \left( \eta_n > \eps \right) \\
        & \leq   \mathbb P \left(  \left| \Psi^{-1/2}(\beta_0) \mathcal C_{\alpha+\eps}^*(\beta_0) - z_{\alpha+\eps}  \right| > \eps'\right) + \mathbb P \left(  \left| \Psi^{-1/2}(\beta_0) \mathcal C_{\alpha-\eps}^*(\beta_0) - z_{\alpha-\eps}  \right| > \eps'\right) + \mathbb P \left( \eta_n > \eps \right),
\end{align*}
where the first inequality is by \eqref{eq:etan} and the last equality is by \eqref{eq:zeps}. Taking $\limsup_{n \rightarrow \infty}$ on both sides of the above display, we have
\begin{align*}
 & \limsup_{n \rightarrow \infty}    \mathbb P \left(  \left| \Psi^{-1/2}(\beta_0) \widehat{\mathcal C}^*_{\alpha}(\beta_0) - z_\alpha  \right| > 2\eps' \right) \\
 & \leq \limsup_{n \rightarrow \infty}  \mathbb P \left(  \left| \Psi^{-1/2}(\beta_0) \mathcal C_{\alpha+\eps}^*(\beta_0) - z_{\alpha+\eps}  \right| > \eps'\right) \\
 & + \limsup_{n \rightarrow \infty}  \mathbb P \left(  \left| \Psi^{-1/2}(\beta_0) \mathcal C_{\alpha-\eps}^*(\beta_0) - z_{\alpha-\eps}  \right| > \eps'\right) +  \lim_{n \rightarrow \infty}  \mathbb P \left( \eta_n > \eps \right) = 0, 
\end{align*}
where the equality holds by \eqref{eq:normal2} and the fact that $\eta_n = o_P(1)$. This implies 
\begin{align*}
 \Psi^{-1/2}(\beta_0) \widehat{\mathcal C}^*_{\alpha}(\beta_0) \convP z_\alpha,     
\end{align*} 
and thus, 
\begin{align*}
    \mathbb P(\widehat Q(\beta_0) >  \widehat C^*_\alpha(\beta_0)) \rightarrow \mathbb P \left( \N(\mu(\beta_0),1) >  z_\alpha \right). 
\end{align*}

\section{Proof of Theorem \ref{thm:power_divK_fixedKlambda}}\label{sec:pf_thm_power_divK_fixedKlambda}
By Theorem \ref{thm:main_null}, we have
\begin{align*}
 \mathbb P(\widehat Q(\beta_0) >  \widehat C^*_\alpha(\beta_0)) & =  \mathbb P( Q(\beta_0) + C(\Delta)>  \widehat C^*_\alpha(\beta_0)) + o(1) \\
 & =  \mathbb P(\left( \Psi(\beta_0) \right)^{-1/2} \left( Q(\beta_0) + C(\Delta)\right) > \left( \Psi(\beta_0) \right)^{-1/2} \widehat C^*_\alpha(\beta_0)) + o(1) 
\end{align*}

Following the argument in the proof of Theorem \ref{thm:power_divK}, we have
\begin{align*}
    Q(\beta_0) + C(\Delta) & = \frac{\sum_{i \in [n]} \sum_{j \in [n], j \neq i} (g_i \tilde \sigma_i(\beta_0) + \Delta \Pi_i) \Xi_{\lambda,ij} (g_j \tilde \sigma_j(\beta_0) + \Delta \Pi_j) }{\sqrt{R_\lambda} } + C(\Delta) \\
    & = \frac{\sum_{i=1}^n (g_i^2 -1)\varpi_i}{\sqrt{R_{\lambda}}} + \frac{ \sum_{i \in [n]} \sum_{j \in [n], j \neq i} \Pi_i P_{\lambda,ij} \Pi_j \Delta^2 }{\sqrt{R_\lambda} } + o_P(1) \\
    & = \frac{\sum_{i \in [R]} (g_i^2 -1)\varpi_i}{\sqrt{R_{\lambda}}} + \frac{\sum_{i=R+1}^n (g_i^2 -1)\varpi_i}{\sqrt{R_{\lambda}}} + \frac{ \sum_{i \in [n]} \sum_{j \in [n], j \neq i} \Pi_i P_{\lambda,ij} \Pi_j \Delta^2 }{\sqrt{R_\lambda} } + o_P(1). 
\end{align*}
Recall $\Psi(\beta_0) = \frac{2 \sum_{i \in [n]} \sum_{j \in [n], j \neq i}  \tilde \sigma^2_i(\beta_0) \Xi_{\lambda,ij}^2 \tilde \sigma_j^2(\beta_0) }{R_\lambda }$, which implies 
$ \sum_{i \in [n]} \varpi_i^2 = \Psi(\beta_0) R_\lambda/2$. Then, we have
\begin{align*}
\left( \Psi(\beta_0) \right)^{-1/2} \frac{\sum_{i \in [R]} (g_i^2 -1)\varpi_i}{\sqrt{R_{\lambda}}} \convD \sum_{i \in [R]} (g_i^2 -1) r_i/\sqrt{2}.
\end{align*}

In addition, the rest of the eigenvalues satisfy the Lindeberg-type condition. Following the same argument in the proof of Theorem \ref{thm:power_divK}, we have
\begin{align*}
\left( \Psi(\beta_0) \right)^{-1/2} \frac{\sum_{i=R+1}^n (g_i^2 -1)\varpi_i}{\sqrt{R_{\lambda}}}  \convD N\left(0, (1-\sum_{i \in [R]}r_i^2)\right). 
\end{align*}

Because $\{g_i\}_{i \in [R]}$ is independent of $\{g_i\}_{i > R}$, we have
\begin{align*}
\left( \Psi(\beta_0) \right)^{-1/2} \left(Q(\beta_0) + C(\Delta)\right) \convD  \chi(\{r_i\}_{i \in [R]}) + \mu(\beta_0). 
\end{align*}

Similarly, we can show that 
\begin{align*}
\left(\breve \Psi(\beta_0)  \right)^{-1/2} Q^*(\beta_0) \convD   \chi(\{r_i^*\}_{i \in [R^*]}),
\end{align*}
where 
\begin{align*}
    \breve \Psi(\beta_0) = \frac{2 \sum_{i \in [n]} \sum_{j \in [n], j \neq i}  \breve \sigma^2_i(\beta_0) \Xi_{\lambda,ij}^2 \breve \sigma_j^2(\beta_0) }{R_\lambda }. 
\end{align*}
This implies 
\begin{align*}
 \left( \Psi(\beta_0) \right)^{-1/2}   Q^*(\beta_0) \convD  \psi^{1/2}(\beta_0) \chi(\{r_i^*\}_{i \in [R^*]})
\end{align*}

The distribution of $\psi^{1/2}(\beta_0) \chi(\{r_i^*\}_{i \in [R^*]})$ is continuous and has a density bounded away from zero in a neighborhood of the relevant quantile. Then, we can follow the same argument in the proof of Theorem \ref{cor:size} and show that, for any $\eps>0$, with probability approaching one, 
\begin{align*}
\psi^{1/2}(\beta_0) \mathcal C_{\alpha + \eps}(\{r_i^*\}_{i \in [R^*]}) \leq \Psi^{-1/2}(\beta_0) \widehat C^*_\alpha(\beta_0) \leq \psi^{1/2}(\beta_0) \mathcal C_{\alpha - \eps}(\{r_i^*\}_{i \in [R^*]}).
\end{align*}
This implies 
$\Psi^{-1/2}(\beta_0) \widehat C^*_\alpha(\beta_0) \convP  \psi^{1/2}(\beta_0) \mathcal C_{\alpha}(\{r_i^*\}_{i \in [R^*]}) $, and thus,

\begin{align*}
\mathbb P(\widehat Q(\beta_0) >  \widehat C^*_\alpha(\beta_0)) \rightarrow  \mathbb P \left(\chi(\{r_i\}_{i \in [R]}) + \mu(\beta_0) >   \psi^{1/2}(\beta_0) \mathcal C_\alpha (\{r_i^*\}_{i \in [R^*]}) \right).
\end{align*}

\section{Proof of Theorem \ref{thm:power_fixK}}\label{sec:pf_thm_power_fixK}
By Theorem \ref{thm:main_null}, we have
\begin{align*}
 & \sup_{y \in \Re} \left|\mathbb P \left( \widehat Q(\beta_0) \leq y\right) - \mathbb P \left(\frac{\sum_{i \in [n]} \sum_{j \in [n], j \neq i} (g_i \tilde \sigma_i(\beta_0) + \Delta \Pi_i) \Xi_{\lambda,ij} (g_j \tilde \sigma_j(\beta_0) + \Delta \Pi_j) }{\sqrt{R_\lambda} } + C(\Delta) \leq y\right)\right| \\
 & = o_P(1), 
\end{align*}
 where $\{g_i\}_{i \in [n]}$ is a sequence of i.i.d. standard normal random variables.  

In addition, we have
\begin{align*}
& \frac{\sum_{i \in [n]} \sum_{j \in [n], j \neq i} (g_i \tilde \sigma_i(\beta_0) + \Delta \Pi_i) \Xi_{\lambda,ij} (g_j \tilde \sigma_j(\beta_0) + \Delta \Pi_j) }{\sqrt{R_\lambda} } + C(\Delta) \\
& = \frac{\sum_{i \in [n]} \sum_{j \in [n], j \neq i} (g_i \tilde \sigma_i(\beta_0)) \Xi_{\lambda,ij} (g_j \tilde \sigma_j(\beta_0)) }{\sqrt{R_\lambda} } + \frac{2\sum_{i \in [n]} \sum_{j \in [n], j \neq i} (g_i \tilde \sigma_i(\beta_0)) \Xi_{\lambda,ij} \Pi_j \Delta}{\sqrt{R_\lambda} } \\
& + \frac{\sum_{i \in [n]} \sum_{j \in [n], j \neq i} \Pi_i P_{\lambda,ij} \Pi_j \Delta^2 }{\sqrt{R_\lambda}}\\
& = \frac{\sum_{i \in [n]} \sum_{j \in [n], j \neq i} (g_i \tilde \sigma_i(\beta_0)) P_{\lambda,ij} (g_j \tilde \sigma_j(\beta_0)) }{\sqrt{R_\lambda} } + \frac{\sum_{i \in [n]} \sum_{j \in [n], j \neq i} (g_i \tilde \sigma_i(\beta_0)) (\Xi_{\lambda,ij} - P_{\lambda,ij}) (g_j \tilde \sigma_j(\beta_0)) }{\sqrt{R_\lambda} } \\
& + \frac{2\sum_{i \in [n]} \sum_{j \in [n], j \neq i} (g_i \tilde \sigma_i(\beta_0)) \Xi_{\lambda,ij} \Pi_j \Delta}{\sqrt{R_\lambda} } + \frac{\sum_{i \in [n]} \sum_{j \in [n], j \neq i} \Pi_i P_{\lambda,ij} \Pi_j \Delta^2 }{\sqrt{R_\lambda}}\\
& = \frac{\sum_{i \in [n]} \sum_{j \in [n], j \neq i} (g_i \tilde \sigma_i(\beta_0) + \Delta \Pi_i) P_{\lambda,ij} (g_j \tilde \sigma_j(\beta_0) + \Delta \Pi_j) }{\sqrt{R_\lambda} } \\
& +  \frac{\sum_{i \in [n]} \sum_{j \in [n], j \neq i} (g_i \tilde \sigma_i(\beta_0)) (\Xi_{\lambda,ij} - P_{\lambda,ij}) (g_j \tilde \sigma_j(\beta_0)) }{\sqrt{R_\lambda} } \\
& +  \frac{2\sum_{i \in [n]} \sum_{j \in [n], j \neq i} (g_i \tilde \sigma_i(\beta_0)) (\Xi_{\lambda,ij} - P_{\lambda,ij} ) \Pi_j \Delta}{\sqrt{R_\lambda} } \\
& = \frac{\sum_{i \in [n]} \sum_{j \in [n], j \neq i} (g_i \tilde \sigma_i(\beta_0) + \Delta \Pi_i) P_{\lambda,ij} (g_j \tilde \sigma_j(\beta_0) + \Delta \Pi_j) }{\sqrt{R_\lambda} } + o_P(1), 
\end{align*}
where the last equality is by the facts that
\begin{align}\label{eq:power_fixed1}
& Var\left[  \frac{\sum_{i \in [n]} \sum_{j \in [n], j \neq i} (g_i \tilde \sigma_i(\beta_0)) (\Xi_{\lambda,ij} - P_{\lambda,ij}) (g_j \tilde \sigma_j(\beta_0)) }{\sqrt{R_\lambda} }  \right] \notag  \\
&  \lesssim \frac{\sum_{i \in [n]} \sum_{j \in [n], j \neq i} (\Xi_{\lambda,ij} - P_{\lambda,ij})^2 }{R_\lambda} \notag \\
& \lesssim  \frac{ \left(\max_{i \in [n]}P_{\lambda,ii} \right) \sum_{i \in [n]} \sum_{j \in [n], j \neq i} P_{W,ij}^2 }{R_\lambda} + \frac{  ||B_{\lambda}||_F^2 }{R_\lambda} \notag \\
& \lesssim \frac{ \left(\max_{i \in [n]}P_{\lambda,ii} \right) ||P_W||_F^2 }{R_\lambda} \notag \\
& \lesssim \frac{ \left(\max_{i \in [n]}P_{\lambda,ii} \right) d_w }{R_\lambda}  = o(1)
\end{align}
and 
\begin{align*}
& Var\left[\frac{\sum_{i \in [n]} \sum_{j \in [n], j \neq i} (g_i \tilde \sigma_i(\beta_0)) (\Xi_{\lambda,ij} - P_{\lambda,ij} ) \Pi_j \Delta}{\sqrt{R_\lambda} } \right]    \\
& \lesssim \frac{\sum_{i \in [n]} \left( \sum_{j \in [n], j \neq i} (\Xi_{\lambda,ij} - P_{\lambda,ij} ) \Pi_j \Delta\right)^2  }{R_\lambda} \\
& \lesssim \frac{\sum_{i \in [n]} \left( \sum_{j \in [n], j \neq i} (\Xi_{\lambda,ij} - P_{\lambda,ij} )^2 \right) ||\Pi||_2^2 \Delta^2}{R_\lambda}  = o(1). 
\end{align*}

Next, we have
\begin{align*}
& \frac{\sum_{i \in [n]} \sum_{j \in [n], j \neq i} (g_i \tilde \sigma_i(\beta_0) + \Delta \Pi_i) P_{\lambda,ij} (g_j \tilde \sigma_j(\beta_0) + \Delta \Pi_j) }{\sqrt{R_\lambda} } \\
& = \frac{\sum_{i \in [n]} \sum_{j \in [n]} (g_i \tilde \sigma_i(\beta_0) + \Delta \Pi_i) P_{\lambda,ij} (g_j \tilde \sigma_j(\beta_0) + \Delta \Pi_j) }{\sqrt{R_\lambda} } - \frac{\sum_{i \in [n]} (g_i \tilde \sigma_i(\beta_0) + \Delta \Pi_i)^2 P_{\lambda,ii}}{\sqrt{R_\lambda} } \\
    & = \frac{\sum_{i \in [n]} \sum_{j \in [n]} (g_i \tilde \sigma_i(\beta_0) + \Delta \Pi_i) P_{\lambda,ij} (g_j \tilde \sigma_j(\beta_0) + \Delta \Pi_j) }{\sqrt{R_\lambda} } - \frac{\sum_{i \in [n]} \breve \sigma_i^2(\beta_0) P_{\lambda,ii}}{\sqrt{R_\lambda} } + o_P(1),
\end{align*}
where the second equality follows from the fact that $\breve \sigma_i^2(\beta_0) = \mathbb E (g_i \tilde \sigma_i(\beta_0) + \Delta \Pi_i)^2$ and 
\begin{align}\label{eq:power_fixed2}
Var \left( \frac{\sum_{i \in [n]} (g_i \tilde \sigma_i(\beta_0) + \Delta \Pi_i)^2 P_{\lambda,ii} }{\sqrt{R_\lambda} } \right) \lesssim \frac{\sum_{i \in [n]} P_{\lambda, ii}^2}{ R_{\lambda}} \lesssim \left(\max_{\lambda, ii} P_{\lambda,ii}\right) K = o(1),
\end{align}
where the last inequality holds by the fact that $\sum_{i \in [n]} P_{\lambda,ii} = tr(P_{\lambda}) \leq \min(K,n) = K.$ 

In addition, consider the singular value decomposition of $Z$ as 
$Z = \mathcal U \mathcal S \mathcal V^\top,$
where $\mathcal U \in \Re^{n \times n}$, $\mathcal U^\top \mathcal U = I_n$, $\mathcal S = [ S_0, 0_{K, n-K} ]^\top$, $S_0$ is a diagonal matrix of non-zero singular values, $0_{K,n-K} \in \Re^{K \times (n-K)}$ is a matrix of zeros, $\mathcal V \in \Re^{K \times K}$, and $\mathcal V^\top \mathcal V = I_K$. Further denote $\mathcal U = [\mathcal U_1,\mathcal U_2]$ such that $\mathcal U_1 \in \Re^{n \times K}$, $\mathcal U_2 \in \Re^{n \times (n-K)}$, $\mathcal U_1^\top \mathcal U_1 = I_K$, $\mathcal U_1^\top \mathcal U_2 = 0_{K,n-K}$, and $\mathcal U_2^\top \mathcal U_2 = I_{n-K}$. 

Then, we have
\begin{align*}
    P_\lambda & = \mathcal U \mathcal S \mathcal V^\top (\mathcal V (\mathcal S^\top \mathcal S + \lambda I_K) \mathcal V^\top)^{-1} \mathcal V \mathcal S^\top \mathcal U^\top \\
    & =  \mathcal U \mathcal S (S_0^2 + \lambda I_K)^{-1} \mathcal S^\top \mathcal U^\top \\
    & =  \mathcal U \begin{pmatrix}
S_0 (S_0^2 + \lambda I_K)^{-1} S_0 & 0_{K,n-K} \\
0_{n-K,K} & 0_{n-K,n-K}
    \end{pmatrix} \mathcal U^\top \\
        & =  \mathcal U_1 
S_0 (S_0^2 + \lambda I_K)^{-1} S_0 \mathcal U_1^\top.
\end{align*}

Denote $\overline g = (g_1,\cdots,g_n)^\top$,  
$\Omega(\beta_0) = \mathcal U_1^\top  \diag(\tilde \sigma_1^2(\beta_0),\cdots,\tilde \sigma_n^2(\beta_0) ) \mathcal U_1,$
and 
$$\tilde \nu(\beta_0) = \lim_{n \rightarrow \infty} \Omega^{-1/2}(\beta_0) \Delta \mathcal U_1^\top \Pi.$$
Then, we have
\begin{align*}
 \mathcal U_1^\top   \begin{pmatrix}
        g_1 \tilde \sigma_1(\beta_0) + \Delta \Pi_1 \\
        \vdots\\
        g_n \tilde \sigma_n(\beta_0) + \Delta \Pi_n
    \end{pmatrix} & = \mathcal U_1^\top  \left( \Delta \Pi +  \diag(\tilde \sigma_1(\beta_0),\cdots,\tilde \sigma_n(\beta_0) ) \overline g \right) \\
    & = \Omega^{1/2}(\beta_0) (\tilde \nu(\beta_0) +  \tilde {\mathcal G}),
\end{align*}
where 
 $\tilde {\mathcal G} = \Omega^{-1/2}(\beta_0)  \mathcal U_1^\top \diag(\tilde \sigma_1(\beta_0),\cdots,\tilde \sigma_n(\beta_0) ) \overline g$,
and $\tilde {\mathcal G}$ follows a $K$-dimensional standard normal distribution. 

Further, consider the eigenvalue decomposition 
\begin{align*}
A = \lim_{n \rightarrow \infty} \frac{   \left(\Omega^{1/2}(\beta_0) S_0 (S_0^2 + \lambda I_K)^{-1} S_0  \Omega^{1/2}(\beta_0)\right) }{\sqrt{R_\lambda}} = \mathbb U \diag(\omega_1,\cdots,\omega_K) \mathbb U^\top, 
\end{align*}
where $\mathbb U \in \Re^{K \times K}$, $\mathbb U^\top \mathbb U = I_K$, $\{\omega_k\}_{k \in [K]}$ are $K$ non-negative eigenvalues. Let $\mathcal G = \mathbb U^\top \tilde {\mathcal G} $, $\nu(\beta_0) = \mathbb U^\top \tilde \nu(\beta_0)$, $\nu_k(\beta_0)$ be the $k$-th element of $\nu(\beta_0)$, and $\mathcal G_k$ be the $k$-th element of $\mathcal G$ so that they are i.i.d. standard normal random variables.  Then, we have
\begin{align*}
& \frac{\sum_{i \in [n]} \sum_{j \in [n]} (g_i \tilde \sigma_i(\beta_0) + \Delta \Pi_i) P_{\lambda,ij} (g_j \tilde \sigma_j(\beta_0) + \Delta \Pi_j) }{\sqrt{R_\lambda}} \\
& = \frac{1}{\sqrt{R_\lambda}} \begin{pmatrix}
        g_1 \tilde \sigma_1(\beta_0) + \Delta \Pi_1 \\
        \vdots\\
        g_n \tilde \sigma_n(\beta_0) + \Delta \Pi_n
    \end{pmatrix}^\top  \mathcal U_1 
S_0 (S_0^2 + \lambda I_K)^{-1} S_0 \mathcal U_1^\top   \begin{pmatrix}
        g_1 \tilde \sigma_1(\beta_0) + \Delta \Pi_1 \\
        \vdots\\
        g_n \tilde \sigma_n(\beta_0) + \Delta \Pi_n
    \end{pmatrix} \\
& = (\tilde \nu(\beta_0) + \tilde {\mathcal G})^\top \frac{   \left(\Omega^{1/2}(\beta_0) S_0 (S_0^2 + \lambda I_K)^{-1} S_0  \Omega^{1/2} (\beta_0)\right) }{\sqrt{R_\lambda}} (\tilde \nu(\beta_0) +  \tilde {\mathcal G}) \\
& \convP (\tilde \nu(\beta_0) + \tilde {\mathcal G})^\top  \mathbb U \diag(\omega_1,\cdots,\omega_K) \mathbb U^\top (\tilde \nu(\beta_0) +  \tilde {\mathcal G}) \\
& = \sum_{k\in [K]} \omega_k (\nu_k(\beta_0) + \mathcal G_k)^2  = \sum_{k\in [K]} \omega_k \chi_k^2(\nu_k^2(\beta_0)),
\end{align*}
where $\chi^2_k(\nu_k^2(\beta_0)) = (\nu_k(\beta_0) + \mathcal G_k)^2$ is a sequence of independent chi-squared random variable with one degree of freedom and noncentrality parameter $\nu_k^2(\beta_0)$. 

Similarly, we have 
\begin{align*}
    Q^*(\beta_0) & =   \frac{\sum_{i \in [n]} \sum_{j \in [n], j \neq i} g_{i}\breve \sigma_{i}(\beta_0) P_{\lambda,ij} g_{j} \breve \sigma_{j}(\beta_0)}{\sqrt{R_\lambda}} + o_P(1) \\
    & = \frac{\sum_{i \in [n]} \sum_{j \in [n]} g_{i}\breve \sigma_{i}(\beta_0) P_{\lambda,ij} g_{j} \breve \sigma_{j}(\beta_0)}{\sqrt{R_\lambda}} - \frac{\sum_{i \in [n]} \breve \sigma_{i}^2(\beta_0) P_{\lambda,ii}}{\sqrt{R_\lambda}} + o_P(1),
\end{align*}
where the first and second equalities are by the same arguments as \eqref{eq:power_fixed1} and \eqref{eq:power_fixed2}, respectively.  

Let
\begin{align*}
   \breve \Omega(\beta_0)  = \mathcal U_1^\top  \diag(\breve \sigma_1^2(\beta_0),\cdots,\breve \sigma_n^2(\beta_0) ) \mathcal U_1.
\end{align*}
We have
\begin{align*}
& \left\Vert \frac{   \left(\Omega^{1/2}(\beta_0) S_0 (S_0^2 + \lambda I_K)^{-1} S_0  \Omega^{1/2}(\beta_0)\right) }{\sqrt{R_\lambda}} -  \frac{   \left(\breve \Omega^{1/2}(\beta_0) S_0 (S_0^2 + \lambda I_K)^{-1} S_0 \breve \Omega^{1/2}(\beta_0)\right) }{\sqrt{R_\lambda}} \right\Vert_{op}  \\
    & \lesssim \frac{\left \Vert \breve \Omega^{1/2}(\beta_0) - \Omega^{1/2}(\beta_0) \right\Vert_{op} }{\sqrt{R_\lambda}} \lesssim \frac{ \left \Vert \breve \Omega(\beta_0) - \Omega(\beta_0) \right\Vert_{op} }{\sqrt{R_\lambda}}  \lesssim \left\Vert \frac{1}{ \sqrt{R_\lambda}} \sum_{i \in [n]} \mathcal U_{1,i}\mathcal U_{1,i}^\top \Pi_i^2 \Delta^2 \right\Vert_{op} \\
   & \lesssim \left(\max_{i \in [n]} ||\mathcal{U}_{1,i}||_2^2\right) \frac{\Pi^\top \Pi \Delta^2}{\sqrt{R_\lambda}} = o(1),
\end{align*}
where the second inequality is by the fact that $||A^{1/2}-B^{1/2}||_S \leq ||A-B||_S^{1/2}$ for symmetric and positive semidefinite matrices $A$ and $B$ (see \citet[Theorem X.1.1]{bhatia2013} with $f(u) = u^{1/2}$), and the last equality is by the fact that $\Pi^\top \Pi \Delta^2/\sqrt{R_{\lambda}} = O(1)$ and $\max_{i \in [n]}||\mathcal{U}_{1,i}||_2 = o(1)$. 

Therefore, we have
\begin{align*}
Q^*(\beta_0)  \convD \sum_{k \in [K]} \omega_k \chi^2_k, \quad \widehat C^*_\alpha(\beta_0)) \convP \mathcal C_{\omega}(1-\alpha) , 
\end{align*}
and
\begin{align*}
  \mathbb P(\widehat Q(\beta_0) >  \widehat C^*_\alpha(\beta_0)) \rightarrow \mathbb P \left(\sum_{k\in [K]} \omega_k \chi_k^2(\nu_k^2(\beta_0)) > \mathcal C_{\omega}(1-\alpha) \right), 
\end{align*}
where $\mathcal C_{\omega}(1-\alpha)$ is the $(1-\alpha)$ quantile of $\sum_{k \in [K]} \omega_k \chi^2_k$.

\section{Proof of Theorem \ref{cor:admissible}}\label{sec: proof_admissible}
We follow the same notation in above section. We have
\begin{align*}
\widehat {\mathcal G} (\beta_0) & = \mathbb U \hat \Omega^{-1/2}(\beta_0) \mathcal {U}_1^\top e(\beta_0) \\
& = \mathbb U \Omega^{-1/2}(\beta_0) \mathcal {U}_1^\top e(\beta_0) + o_P(1) \\
& = \mathbb U \Omega^{-1/2}(\beta_0) \mathcal {U}_1^\top \left[\tilde e(\beta_0) + \Pi \Delta - P_W \tilde e(\beta_0) \right] + o_P(1) \\
& = \mathbb U \left[ \Omega^{-1/2}(\beta_0) \mathcal {U}_1^\top \tilde e(\beta_0) +  \tilde \nu(\beta_0)\right]  + o_P(1) \\
& \convD  \mathbb U  (\tilde {\mathcal G} + \tilde \nu(\beta_0)) \stackrel{d}{=} \mathcal G + \nu(\beta_0), 
\end{align*}
where $\tilde {\mathcal G}$, $\mathcal G$, $\tilde \nu(\beta_0)$, and $\nu(\beta_0)$ are defined in the proof of Theorem \ref{thm:power_fixK} above, the second equality is by the consistency of $\hat \Omega(\beta_0)$, the third equality is by the definition of $e(\beta_0)$, the fourth equality is by $\mathcal U_1^\top W = 0$, the convergence in distribution is by standard CLT induced by the fact that $\max_{i \in [n]} ||\mathcal U_{1,i}||_2 = o(1)$.  This implies 
\begin{align*}
\left\{     \widehat {\mathcal G}_k^2 (\beta_0) \right\}_{k \in [K]} \convD \left\{ \chi_k^2(\nu_k^2(\beta_0)  \right\}_{k \in [K]}, 
\end{align*}
and thus, 
\begin{align*}
& \left(\phi^*(\widehat {\mathcal G}_1^2 (\beta_0),\cdots,\widehat {\mathcal G}_K^2 (\beta_0)),\phi_0 \right) \\
& \convD \left( \phi^*(\chi_1^2(\nu_1^2(\beta_0)),\cdots,\chi_K^2(\nu_K^2(\beta_0)) ), 1\{ \sum_{k\in [K]} \omega_k \chi_k^2(\nu_k^2(\beta_0)) > \mathcal C_{\omega}(1-\alpha)\} \right).
\end{align*}
Given that both $\phi^*(\widehat {\mathcal G}_1^2 (\beta_0),\cdots,\widehat {\mathcal G}_K^2 (\beta_0))$ and $\phi_0$ are bounded, we have
\begin{align*}
&    \left(\mathbb E \phi^*(\widehat {\mathcal G}_1^2 (\beta_0),\cdots,\widehat {\mathcal G}_K^2 (\beta_0)), \mathbb E\phi_0 \right) \\
& \rightarrow \left( \mathbb E \phi^*(\chi_1^2(\nu_1^2(\beta_0)),\cdots,\chi_K^2(\nu_K^2(\beta_0)) ), \mathbb P\left( \sum_{k\in [K]} \omega_k \chi_k^2(\nu_k^2(\beta_0)) > \mathcal C_{\omega}(1-\alpha)\right) \right)
\end{align*}

In addition, we note that the acceptance region of test $1\{\sum_{k\in [K]} \omega_k \chi_k^2(\nu_k^2(\beta_0)) > \mathcal C_{\omega}(1-\alpha)\}$ is $\mathbb A = \{\mathcal X_1, \cdots, \mathcal X_K: \sum_{k\in [K]} \omega_k \mathcal X_k^2 \leq C_{\omega}(1-\alpha)  \}$, which is closed, convex, and monotone decreasing in the sense that if $(\mathcal X_1,\cdots,\mathcal X_K) \in \mathbb {A}$ and $0\leq \mathcal X_1' \leq \mathcal X_1, \cdots, 0 \leq \mathcal X_K' \leq  \mathcal X_K$, then $(\mathcal X_1',\cdots,\mathcal X_K') \in \mathbb{A}$. Then, the desired result follows from \citet[Theorem 1]{A16}, which is a direct consequence of results in \cite{MS76} and \cite{KP78}.

\section{An Anti-Concentration Inequality}\label{sec: anti}
\begin{lem}
Suppose Assumption \ref{ass:reg} holds. Then, for any $t >0$ and any $\zeta \in (0,1)$, there exists a constant $C_{\zeta}>0$ that only depends on $\underline c$ in Assumption \ref{ass:reg}.3 and $\zeta$ such that 
\begin{align*}
    \sup_{y \in \Re}    \mathbb P( |Q(\beta_0) - y| \leq t  ) 
    & \leq C_{\zeta}t^{(1-\zeta)/2}
\end{align*}
and 
\begin{align*}
    \sup_{y \in \Re}    \mathbb P( |Q^*(\beta_0) - y| \leq t) 
    & \leq C_{\zeta}t^{(1-\zeta)/2}.
\end{align*} 
    \label{lem:anti-concentration}
\end{lem}

\begin{proof}
Recall 
\begin{align*}
 Q(\beta_0)  = \frac{\sum_{i \in [n]} \sum_{j \in [n], j \neq i} (g_i\tilde \sigma_i(\beta_0) + \Delta \Pi_i) \Xi_{\lambda,ij} (g_j \tilde \sigma_j(\beta_0) + \Delta \Pi_j)}{\sqrt{R_\lambda}},
\end{align*}
where  $\{g_i\}_{i \in [n]}$ is a sequence of i.i.d. standard normal random variables.

Further define $\mathbb A = \Lambda(\beta_0)^{1/2} \Xi_{\lambda} \Lambda(\beta_0)^{1/2}/\sqrt{R_\lambda}$, where $\Lambda(\beta_0) = \diag(\tilde \sigma_{1}^2(\beta_0),\cdots, \tilde \sigma_n^2(\beta_0))$ and $ \tilde P_{\lambda}$ is a $n \times n$ matrix so that 
\begin{align*}
    \Xi_{\lambda,ij} =    P_{\lambda,ij} + (P_{\lambda,ii} + P_{\lambda,jj}) P_{W,ij} - B_{\lambda,ij} \quad \text{if } i \neq j \quad \text{and} \quad   \Xi_{\lambda,ij} = 0 \quad \text{if } i = j.
\end{align*}

Because $\Xi_{\lambda}$ is symmetric, we have 
\begin{align*}
    Q(\beta_0) = \overline g^\top(\beta_0) \mathbb A  \overline g(\beta_0) \stackrel{d}{=} \sum_{i \in [n]} \omega_i \chi_i^2\left( \nu_i^2 \right),
\end{align*}
where $\overline g(\beta_0) = (g_1 + \nu_1,\cdots,g_n+ \nu_n)^\top$, $\nu_i = \Delta \Pi_i/\tilde \sigma_i(\beta_0)$, $\omega_1,\cdots,\omega_n$ are the $n$ eigenvalues of $\mathbb A$, and $\chi_{1}^2(\nu_1^2),\cdots,\chi_n^2(\nu_n^2)$ are $n$ i.i.d. chi-squared random variables with one degree of freedom and noncentrality parameters $\nu_i^2$. In addition, for any $z \geq 0$ and $t>0$, we have
\begin{align*}
     \mathbb P( |\chi^2(\nu^2) - z| \leq t) & = \mathbb P ( \max(0,z-t) \leq (g+ \nu)^2 \leq z+ t) \\
     & = \mathbb P ( \sqrt{\max(0,z-t)} \leq g+ \nu \leq  \sqrt{ z+ t} ) \\
     & + \mathbb P ( -\sqrt{ z+ t}  \leq g+ \nu \leq - \sqrt{\max(0,z-t)} ) \\
     & \leq \frac{2}{\sqrt{2\pi^2}} (\sqrt{ z + t} - \sqrt{\max(0,z-t)} ) \\
     &  = \frac{2}{\sqrt{2\pi^2}} \frac{( z+ t ) - \max(0,z-t) }{\sqrt{ z+ t} + \sqrt{\max(0,z-t)}}   \leq \frac{2\sqrt{2} \sqrt{t}}{\pi}, 
 \end{align*}
where $g$ is a standard normal variable, the first inequality is by the fact that standard normal PDF is bounded by $1/\sqrt{2\pi^2}$ and the second inequality is by the fact that when $t,z\geq 0$, we have
\begin{align*}
  ( z+ t ) - \max(0,z-t)  \leq 2t \quad \text{and} \quad  \sqrt{ z+ t} + \sqrt{\max(0,z-t)} \geq \sqrt{t}.   
\end{align*}
Taking $\sup_{z \in \Re^+}$ on both sides, we have
\begin{align*}
 \sup_{z \in \Re^+} \mathbb P( |\chi^2(\nu^2) - z| \leq t) \leq \frac{2\sqrt{2} \sqrt{t}}{\pi}.
\end{align*}

In addition, when $z \leq 0$, we have
\begin{align*}
 \mathbb P( |\chi^2(\nu^2) - z| \leq t) =  \mathbb P( 0 \leq \chi^2(\nu^2)  \leq t + z) \leq \mathbb P( |\chi^2(\nu^2)| \leq t) \leq \sup_{z \in \Re^+} \mathbb P( |\chi^2(\nu^2) - z| \leq t).
\end{align*}
This implies
\begin{align*}
 \sup_{z \in \Re} \mathbb P( |\chi^2(\nu^2) - z| \leq t) \leq \frac{2\sqrt{2} \sqrt{t}}{\pi}.
\end{align*}
which verifies the condition in \citet[Theorem 1.5]{RV15}. Then, by \citet[Theorem 1.5]{RV15} with their $A$, $X$, $p$, $t$ being $(\omega_1,\cdots,\omega_n)$, $(\chi_{1}^2(\nu_1^2),\cdots,\chi_n^2(\nu_n^2))$, $ \frac{2\sqrt{2} \sqrt{t}}{\pi}$, and $t$, respectively. Then, for any $t>0$ and $\zeta \in (0,1)$, we have
\begin{align*}
    & \sup_{z \in \Re}    \mathbb P( | Q(\beta_0) - z| \leq t ||\mathbb A||_F) \\
    & = \sup_{z \in \Re}    \mathbb P( |\sum_{i \in [n]} \omega_i \chi_i^2(\nu_i^2) - z| \leq t ||\mathbb A||_F ) 
 \leq C_{\zeta}t^{(1-\zeta)/2},
\end{align*}
where we use the fact that $r(A)=1$ in \cite{RV15} and $||A||_{HS}$ in \citeauthor{RV15}{}'s (\citeyear{RV15}) notation is just $\sqrt{\sum_{i=1}^n \omega_i^2} = ||\mathbb A||_F$ in our notation. By Assumption \ref{ass:reg}.3, we have
\begin{align*}
    ||\mathbb A||_F^2 = \frac{\sum_{i \in [n]} \sum_{j \in [n], j \neq i} \Xi_{\lambda,ij}^2 \tilde \sigma_i^2(\beta_0) \tilde \sigma_j^2(\beta_0)}{R_\lambda} \geq \underline c^2>0. 
\end{align*}
Therefore, for any $t >0$, we have
\begin{align*}
    \sup_{z \in \Re}    \mathbb P( | Q(\beta_0) - z| \leq t  ) & = \sup_{z \in \Re}    \mathbb P\left( | Q(\beta_0) - z| \leq \frac{t}{||\mathbb A||_F} ||\mathbb A||_F \right)  \\
    & \leq \sup_{z \in \Re}    \mathbb P\left( |Q(\beta_0) - z| \leq \frac{t}{\underline c} ||\mathbb A||_F \right)  \leq C_{\zeta} \left(\frac{t}{\underline c}\right)^{(1-\zeta)/2}.
\end{align*}
Then, the desired result holds if we take on both sides of the above display. 

For the second result, we note that 
\begin{align*}
Q^*(\beta_0)  = \frac{\sum_{i \in [n]} \sum_{j \in [n], j \neq i} (g_i \breve \sigma_i(\beta_0)) \Xi_{\lambda,ij} (g_j \breve \sigma_j(\beta_0))}{\sqrt{R_\lambda}}. 
\end{align*}
Then, we can derive the result following the same argument above with $\nu_i$ and $\tilde \sigma_i(\beta_0)$ replaced by $0$ and $\breve \sigma_i(\beta_0)$. 
\end{proof}

\section{Additional Lemmas}\label{sec: add_lemmas}
\begin{lem} \label{lem:gamma}
    Suppose Assumption \ref{ass:reg} holds and $||\Pi||_2^2 \Delta^2 /\min\left(R_\lambda^{1/2},R_\lambda^{2/3}\right)$ is bounded. Recall $\kappa = (M_W \circ M_W)^{-1}$, 
\begin{align*}
   & A_{\lambda,ii} = 2P_{W,ii} P_{\lambda,ii} - B_{\lambda,ii}, \quad  B_{\lambda,ii} = [P_W D_{\lambda} P_W]_{ii}, \quad \text{and} 
   & \Xi_{\lambda,ij} =     P_{\lambda,ij} + (P_{\lambda,ii} + P_{\lambda,jj}) P_{W,ij} - B_{\lambda,ij},
\end{align*}
where $D_{\lambda} = \diag(P_{\lambda,11}, \cdots, P_{\lambda,nn})$.     Then, we have 
\begin{enumerate}
\item[(1)] $\max_{i \in [n]} A_{\lambda,ii} = o(1)$;
\item[(2)] $\max_{i \in [n]} A_{\lambda,ii}^2/\sqrt{R_\lambda} = o(1) $;
\item[(3)] $\sum_{i \in [n]} A_{\lambda,ii}^2/R_\lambda = o(1)$; 
\item[(4)] 
\begin{align*}
\frac{\sum_{i,j \in [n]^2} \kappa_{ij} e_j^2(\beta_0) A_{\lambda,ii}}{\sqrt{R_\lambda}} = \frac{\sum_{i \in [n]} \tilde \sigma_i^2(\beta_0) A_{\lambda,ii}}{\sqrt{R_\lambda}} + o_P(1);
\end{align*}
\item[(5)] \begin{align*}
        \frac{2\sum_{i \in [n]} \sum_{j \in [n], j \neq i} \Pi_i \Delta \left(P_{\lambda,ij} - \Xi_{\lambda,ij}\right) \tilde e_j(\beta_0) }{\sqrt{R_\lambda}}  = o_P(1).
    \end{align*}
% \item[(6)] 
% \begin{align*}
% \frac{\max_{i \in [n]}    \sum_{j \in [n],j \neq i} \Xi_{\lambda,ij}^2}{R_{\lambda}} \lesssim p_n + {p_n'}. 
% \end{align*}
% \item[(6)] \begin{align*}
%     \frac{\sum_{i \in [n]} \sum_{j \in [n], j \neq i} \Pi_i  \left(P_{\lambda,ij} - \Xi_{\lambda,ij}\right) \Pi_j \Delta^2 }{\sqrt{R_\lambda}} = o_P(1).
% \end{align*}
\end{enumerate}

\end{lem}
\begin{proof}
For the first claim, we have 
\begin{align*}
\max_{j \in [n]}    A_{\lambda, jj} \lesssim \max_{j \in [n]} P_{W,jj} P_{\lambda,jj} + \max_{j \in [n]} \sum_{i \in [n]} P_{W,ij}^2 P_{\lambda,ii} \lesssim \max_{j \in [n]} P_{W,jj} = o(1).  
\end{align*}

For the second claim, we have
\begin{align*}
\max_{i \in [n]} A_{\lambda,ii}^2/\sqrt{R_\lambda} & \lesssim \max_{i \in [n]} \frac{P_{W,ii}^2 P_{\lambda,ii}^2 }{\sqrt{R_\lambda} } +  \max_{i \in [n]} \frac{\sum_{j \in [n]} P_{W,ij}^2 P_{\lambda,jj} }{\sqrt{R_\lambda}} \\
& \lesssim \left( \max_{j \in [n]} \frac{P_{\lambda,jj}}{\sqrt{R_\lambda} } \right) \max_{i \in [n]} P_{W,ii} = o(1). 
\end{align*}

For the third claim, we have
\begin{align*}
 \frac{\sum_{i \in [n]} A_{\lambda,ii}^2 }{R_\lambda}    & \lesssim \frac{\sum_{i \in [n]} P_{W,ii}^2 P_{\lambda,ii}^2 }{R_\lambda} + \frac{\sum_{i \in [n]} B_{\lambda,ii}^2 }{R_\lambda}\\
    & \lesssim {p_n'} \left( \sum_{i \in [n]} P_{W,ii}^2 \right) + \frac{\sum_{i \in [n]} (\sum_{j \in [n]}P_{W,ij}^2 P_{\lambda,jj})^2 }{R_\lambda} \\
    & \lesssim {p_n'} \left( \sum_{i \in [n]} P_{W,ii}^2 \right) = o(1). 
\end{align*}

For the fourth claim, we note that 
\begin{align*}
    e_j(\beta_0) = \sum_{k\in [n]}M_{W,jk} \tilde e_k(\beta_0) + \Pi_j \Delta, 
\end{align*}
and thus, 
\begin{align*}
\mathbb E \left( \frac{\sum_{i,j \in [n]^2} \kappa_{ij} e_j^2(\beta_0) A_{\lambda,ii}}{\sqrt{R_\lambda}} \right) & =  \mathbb E \left( \frac{\sum_{i,j \in [n]^2} \kappa_{ij}  (\sum_{k\in [n]}M_{W,jk} \tilde e_k(\beta_0))^2 A_{\lambda,ii}}{\sqrt{R_\lambda}} \right) \\
& + \left( \frac{\sum_{i,j \in [n]^2} \kappa_{ij}   \Pi_j^2 \Delta^2 A_{\lambda,ii}}{\sqrt{R_\lambda}} \right)\\
& =  \mathbb E \left( \frac{\sum_{i,j,k \in [n]^3} \kappa_{ij}  M_{W,jk}^2 \tilde \sigma_k^2(\beta_0) A_{\lambda,ii}}{\sqrt{R_\lambda}} \right) + \left( \frac{\sum_{i,j \in [n]^2} \kappa_{ij}   \Pi_j^2 \Delta^2 A_{\lambda,ii}}{\sqrt{R_\lambda}} \right) \\
& =  \mathbb E \left( \frac{\sum_{i \in [n]}  \tilde \sigma_i^2(\beta_0) A_{\lambda,ii}}{\sqrt{R_\lambda}} \right) + \left( \frac{\sum_{i,j \in [n]^2} \kappa_{ij}   \Pi_j^2 \Delta^2 A_{\lambda,ii}}{\sqrt{R_\lambda}} \right),
\end{align*}
where last inequality is by construction that 
\begin{align*}
    \sum_{j \in [n]} \kappa_{ij} M_{W,jk}^2 = 1\{i=k\}. 
\end{align*}

In addition, by Theorem 1 of \cite{varah75}, we have
\begin{align}\label{eq:kappa}
\max_{j \in [n]} \sum_{i \in [n]}|\kappa_{ij}| \leq 1/(1/2 - \max_{i \in [n]} P_{W,ii}) \lesssim 1, 
\end{align}
which implies 
\begin{align*}
\left( \frac{\sum_{i,j \in [n]^2} \kappa_{ij}   \Pi_j^2 \Delta^2 A_{\lambda,ii}}{\sqrt{R_\lambda}} \right)  \lesssim (\max_{i \in [n]} A_{\lambda,ii}) \left( \max_{j \in [n]} \sum_{i \in [n]}|\kappa_{ij}| \right) ||\Pi||_2^2 \Delta^2/\sqrt{R_\lambda} = o(1).
\end{align*}
Therefore, we have
\begin{align*}
\mathbb E \left( \frac{\sum_{i,j \in [n]^2} \kappa_{ij} e_j^2(\beta_0) A_{\lambda,ii}}{\sqrt{R_\lambda}} \right) = \mathbb E \left( \frac{\sum_{i \in [n]}  \tilde \sigma_i^2(\beta_0) A_{\lambda,ii}}{\sqrt{R_\lambda}} \right) + o(1). 
\end{align*}

In addition, we have
\begin{align*}
Var \left( \frac{\sum_{i,j \in [n]^2} \kappa_{ij} e_j^2(\beta_0) A_{\lambda,ii}}{\sqrt{R_\lambda}} \right) & \lesssim \underbrace{ Var \left( \frac{ \sum_{j \in [n]}\left(\sum_{i \in [n]} \kappa_{ij} A_{\lambda,ii} \right) (\sum_{k\in [n]}M_{W,jk} \tilde e_k(\beta_0))^2 }{\sqrt{R_\lambda}} \right) }_{R_1}\\
& +  \underbrace{ Var \left( \frac{ \sum_{j \in [n]}\left(\sum_{i \in [n]} \kappa_{ij} A_{\lambda,ii} \right) (\sum_{k\in [n]}M_{W,jk} \tilde e_k(\beta_0))\Pi_j \Delta }{\sqrt{R_\lambda}} \right), }_{R_2}
\end{align*}
where 
\begin{align*}
    R_1 & \lesssim  \frac{ \sum_{k,l \in [n]^2} \left[\sum_{j \in [n]}\left(\sum_{i \in [n]} \kappa_{ij} A_{\lambda,ii} \right) (M_{W,jk}M_{W,jl})\right]^2 }{R_\lambda} \\
    & = \frac{ \sum_{k,l, i,i',j,j' \in [n]^6} \left[\kappa_{ij} A_{\lambda,ii} M_{W,jk}M_{W,jl} \right]\left[\kappa_{i'j'} A_{\lambda,i'i'} M_{W,j'k}M_{W,j'l} \right]  }{R_\lambda} \\
    & = \frac{ \sum_{ i,i',j,j' \in [n]^4} \kappa_{ij} A_{\lambda,ii} M_{W,jj'}^2 \kappa_{i'j'} A_{\lambda,i'i'} }{R_\lambda} \\
    & \lesssim \frac{\sum_{j \in [n]} \left(\sum_{i \in [n]} \kappa_{ij} A_{\lambda,ii}\right)^2}{R_\lambda}\\
    & = \frac{\sum_{i,k \in [n]^2} A_{\lambda,ii} \left(\sum_{j \in [n]} \kappa_{ij}\kappa_{kj} \right)  A_{\lambda,kk}}{R_\lambda}\\
    & \lesssim \frac{\sum_{i \in [n]} A_{\lambda,ii}^2 }{R_\lambda} = o(1),
\end{align*}
where the second inequality is by  $||M_W \circ M_W||_{op} \leq 1$ and the last inequality is due to the fact that by Section 3 in the Appendix of \citet{CJN18}, we have
\begin{align*}
   || \kappa^2 ||_{op} & = \left[\lambda_{\min} \left( M_W \circ M_W \right)\right]^{-2} \\
   &\lesssim \frac{1}{2 \min_{i \in [n]}(M_{W,ii}(M_{W,ii} - 1/2))} \\
   & =  \frac{1}{2 \left[ (1-\max_{i\in [n]}P_{W,ii})(1/2 -\max_{i\in [n]}P_{W,ii} )\right]} \lesssim 1. 
\end{align*}

Next, we have
\begin{align*}
    R_2 & \lesssim \sum_{k\in [n]} \frac{ (\sum_{i,j \in [n]^2}\kappa_{ij} A_{\lambda,ii} M_{W,jk} \Pi_j \Delta)^2 }{R_\lambda} \\
    & = \frac{ \sum_{i,j,i',j' \in [n]^4}\kappa_{ij} A_{\lambda,ii} \Pi_j M_{W,jj'}  \kappa_{i'j'} A_{\lambda,i'i'} \Pi_{j'} \Delta^2 }{R_\lambda} \\
    & \lesssim  \frac{ \sum_{j \in [n]}(\sum_{i \in [n]}\kappa_{ij} A_{\lambda,ii})^2  \Pi_j^2 \Delta^2 }{R_\lambda} \\
    & \lesssim \frac{\max_{i \in [n]} A_{\lambda,ii}^2}{ \sqrt{R_\lambda }} \frac{ ||\Pi||_2^2 \Delta^2 }{\sqrt{R_\lambda }} = o(1),
\end{align*}
where we use \eqref{eq:kappa}. This leads to the desired result. 

For the fifth claim, we have
\begin{align*}
& Var\left(   \frac{\sum_{i \in [n]} \sum_{j \in [n], j \neq i} \Pi_i \Delta \left(P_{\lambda,ij} - \Xi_{\lambda,ij}\right) \tilde e_j(\beta_0) }{\sqrt{R_\lambda}} \right) \\
& = Var\left(   \frac{\sum_{i \in [n]} \sum_{j \in [n], j \neq i} \Pi_i \Delta \left( (P_{\lambda,ii} + P_{\lambda,jj}) P_{W,ij} - B_{\lambda,ij}\right) \tilde e_j(\beta_0) }{\sqrt{R_\lambda}} \right) \\
& \lesssim \sum_{j \in [n]} \frac{ \left[\sum_{ i \in [n], i \neq j} \Pi_i \left( (P_{\lambda,ii} + P_{\lambda,jj}) P_{W,ij} - B_{\lambda,ij}\right) \right]^2 \Delta^2 }{R_\lambda}\\
& \lesssim \sum_{j \in [n]} \frac{ \left[\sum_{ i \in [n]} \Pi_i \left( (P_{\lambda,ii} + P_{\lambda,jj}) P_{W,ij} - B_{\lambda,ij}\right) \right]^2 \Delta^2 }{R_\lambda} \\
& + \sum_{j \in [n]} \frac{ \Pi_j^2 \left( (P_{\lambda,jj} + P_{\lambda,jj}) P_{W,jj} - B_{\lambda,jj}\right)^2 \Delta^2 }{R_\lambda} \\
& \lesssim \sum_{j \in [n]} \frac{ \left(\sum_{ i \in [n]} \Pi_i P_{\lambda,ii} P_{W,ij} \right)^2 \Delta^2 }{R_\lambda} + \sum_{j \in [n]} \frac{ \left(\sum_{ i \in [n]} \Pi_i P_{\lambda,jj} P_{W,ij}\right)^2 \Delta^2 }{R_\lambda} \\
& + \sum_{j \in [n]} \frac{ \left(\sum_{ i \in [n]} \Pi_i B_{\lambda,ij}\right)^2 \Delta^2 }{R_\lambda} + {p_n'}^{1/2}\sum_{j \in [n]} \frac{ \Pi_j^2 \Delta^2 }{\sqrt{R_\lambda}} \\
& \lesssim \sum_{i,k \in [n]^2} \frac{  \Pi_i P_{\lambda,ii} P_{W,ik} \Pi_k P_{\lambda,kk} \Delta^2 }{R_\lambda} + \sum_{i,k \in [n]^2} \frac{ \Pi_i \left(\sum_{j \in [n]} P_{\lambda,jj}^2 P_{W,ij}P_{W,kj} \right) \Pi_k  \Delta^2 }{R_\lambda} \\
& + \sum_{i,k \in [n]^2} \frac{ \Pi_i \left(\sum_{j \in [n]} B_{\lambda,ij}B_{\lambda,kj} \right) \Pi_k  \Delta^2 }{R_\lambda} + o(1) \\
& = \sum_{i,k \in [n]^2} \frac{  \Pi_i [D_{\lambda} P_{W} D_{\lambda} ]_{i,k} \Pi_k \Delta^2 }{R_\lambda} + \sum_{i,k \in [n]^2} \frac{ \Pi_i [P_{W} D_{\lambda}^2 P_W ]_{i,k} \Pi_k  \Delta^2 }{R_\lambda} \\
& + \sum_{i,k \in [n]^2} \frac{ \Pi_i [B_{\lambda}B_{\lambda}]_{i,k} \Pi_k  \Delta^2 }{R_\lambda} + o(1) \\
& \lesssim \frac{ ||D_{\lambda} P_{W} D_{\lambda} ||_{op} +  ||P_{W} D_{\lambda}^2 P_W ||_{op} + ||B_{\lambda}B_{\lambda}||_{op}}{ \sqrt{R_{\lambda}}} \frac{||\Pi||_2^2 \Delta^2}{ \sqrt{R_{\lambda}}}  + o(1)\\
& \lesssim \frac{{p_n'}^{1/2}||\Pi||_2^2 \Delta^2}{ \sqrt{R_{\lambda}}} + o(1) = o(1),
\end{align*}
where we repeated use the fact that 
\begin{align*}
    \max_{i \in [n]} B_{\lambda,ii} \leq ||B_\lambda||_{op} \leq ||D_\lambda||_{op} = \max_{i \in [n]} P_{\lambda,ii}. 
\end{align*}

This leads to the desired result.

\end{proof}

\section{Additional Simulation Results for Section \ref{subsec: simulation_hausman}}
\label{section:additional_simulation_results}

%%%%%% Hausman's DGP
\subsection{Simulations under \texorpdfstring{$K=2$}{K=2}}

Figure \ref{figure:Hausman_K_1} shows the power curves for the eleven tests under $K=2$ for the DGP based on \citet{Haus2012}.

\begin{figure}[H]    
\begin{center}
    \includegraphics[width=\textwidth,height = 11.5cm]{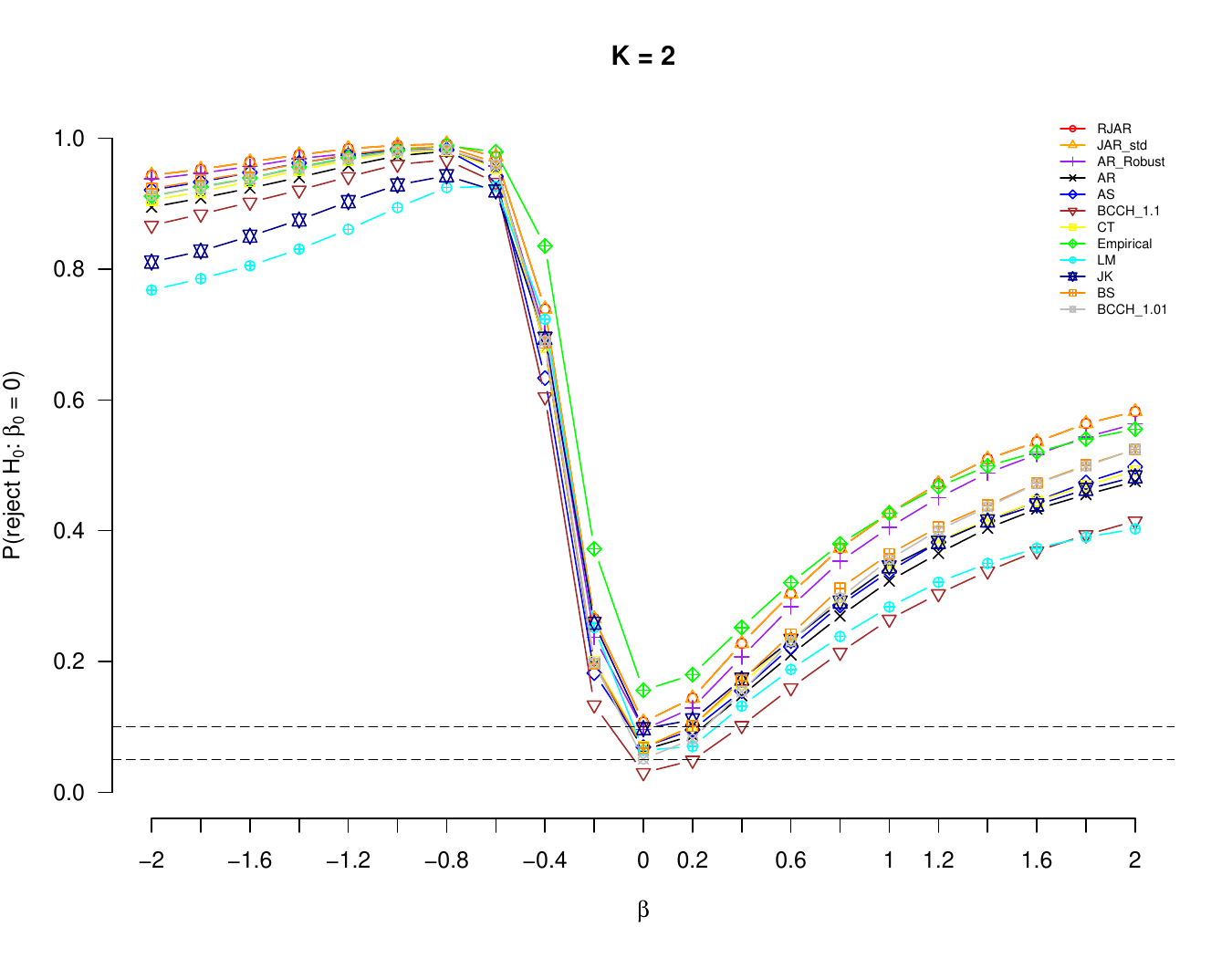}
    \caption{Power curves for $2$ IVs, $\mu^2 = 72$. }
    \raggedright{\footnotesize{\textbf{Note:} The red curve with a hollow circle represents $\rm RJAR$; the orange curve with an upward triangle represents $\rm JAR_{\rm std}$; the purple curve with a cross represents $\rm AR_{\rm Robust}$; the black curve with $\rm \times $ represents $\rm AR$; the blue curve with diamond represents $\rm AS$; the brown curve with inverted triangle represents $\rm BCCH_{1.1}$; the yellow curve with a filled square represents $\rm CT$; the green curve with a filled diamond represents $\rm Empirical$; the cyan curve with a filled circle represents $\rm LM$; the dark-blue curve with hexagram represents $\rm JK$; the dark-orange curve with the $+$ in the square-box represents $\rm BS$. The grey curve with fill $\rm \times$ in its circle is $\rm BCCH_{1.01}$.
     The horizontal dotted black lines represent the $5\%$ and $10\%$ levels.}}
    \label{figure:Hausman_K_1}
\end{center}
\end{figure}

\subsection{Simulations for varying \texorpdfstring{$c_1,c_2$}{c1, c2}}\label{sec: simu_alter_c1_c2}
Our bootstrap test requires specifying $c_1$ and $c_2$; in the main text, we suggested using $(c_1,c_2) = (0.1,1.1)$. In this section, we examine the sensitivity of our test to these choices through simulations. Specifically, we vary $c_1 \in \{0.05, 0.1, 0.2\}$ and $c_2 \in \{0.5, 1, 2\}$, yielding $3 \times 3 = 9$ combinations. The corresponding power curves are reported in Figures~\ref{Hausman_K_10_c1_0.05_c2_0.5}--\ref{Hausman_K_10_c1_0.2_c2_2}. The results show that the performance of our proposed test is robust to the choice of $c_1$ and $c_2$.

\begin{figure}[H]   \includegraphics[width=1\textwidth,height = 12cm]{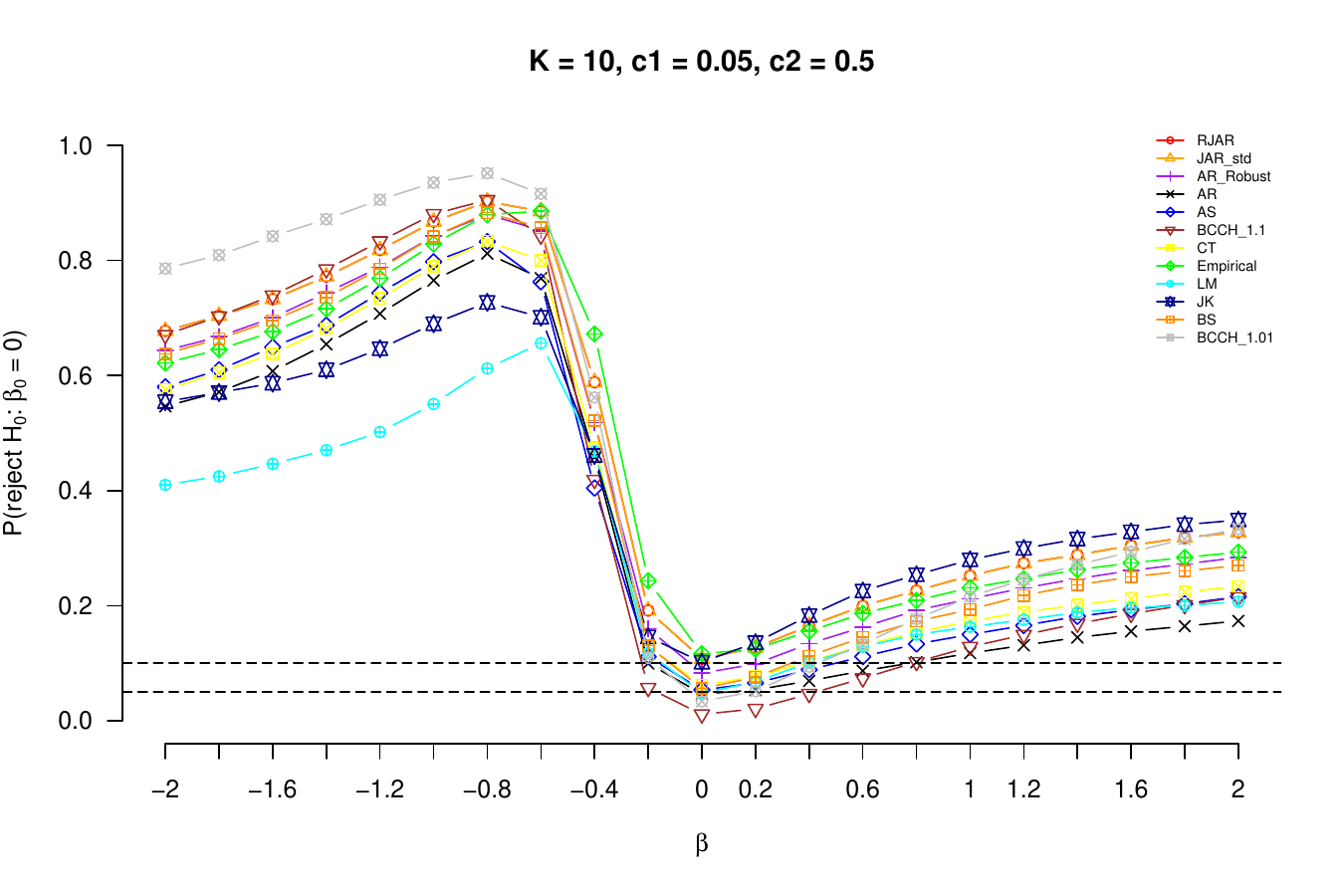}
    \caption{Plot with $(c_1,c_2) = (0.05,0.5)$ and $K = 10$}
    \textbf{Note:} We run 5{,}000 replications with 200 observations and 15 controls, using \citeauthor{Haus2012}{}'s DGP as in our main paper. The red circle line labeled `\texttt{RJAR}' is the test by \citet{DKM24}. The orange upward triangle labeled `\texttt{JAR\_standard}' is the test by \citet{crudu2021}. The purple cross labeled `\texttt{AR\_Robust}' is the test by \citet{mikusheva2024weak}. The green \texttt{x} labeled `\texttt{AR\_fixed}' is the classical AR test as given in the main paper. The blue diamond labeled `\texttt{AS}' is the test by \citet{Anatolyev-Solvsten2023}. The brown downward triangle labeled `\texttt{BCCH\_{1.1}}' is the test by \citet{belloni2012}. The yellow box labeled `\texttt{CT}' is the test by \citet{carrasco2016}. The green star labeled `\texttt{empirical}' is the bootstrap test using the empirical distribution of residuals. The cyan circle labeled `\texttt{LM}' is the test by \citet{matsushita2024}. The darkblue hexagram labeled `\texttt{JK}' is the test by \citet{N23}. The orange box labeled `\texttt{BS}' is our bootstrap test given in the main text. The grey box with a filled cross labeled `\texttt{ BCCH\_{1.01}}' is the test by \citet{belloni2012} with critical values $c_{BCCH} = 1.01$.
     \label{Hausman_K_10_c1_0.05_c2_0.5}
\end{figure}

\begin{figure}[H]   \includegraphics[width=1\textwidth,height = 8cm]{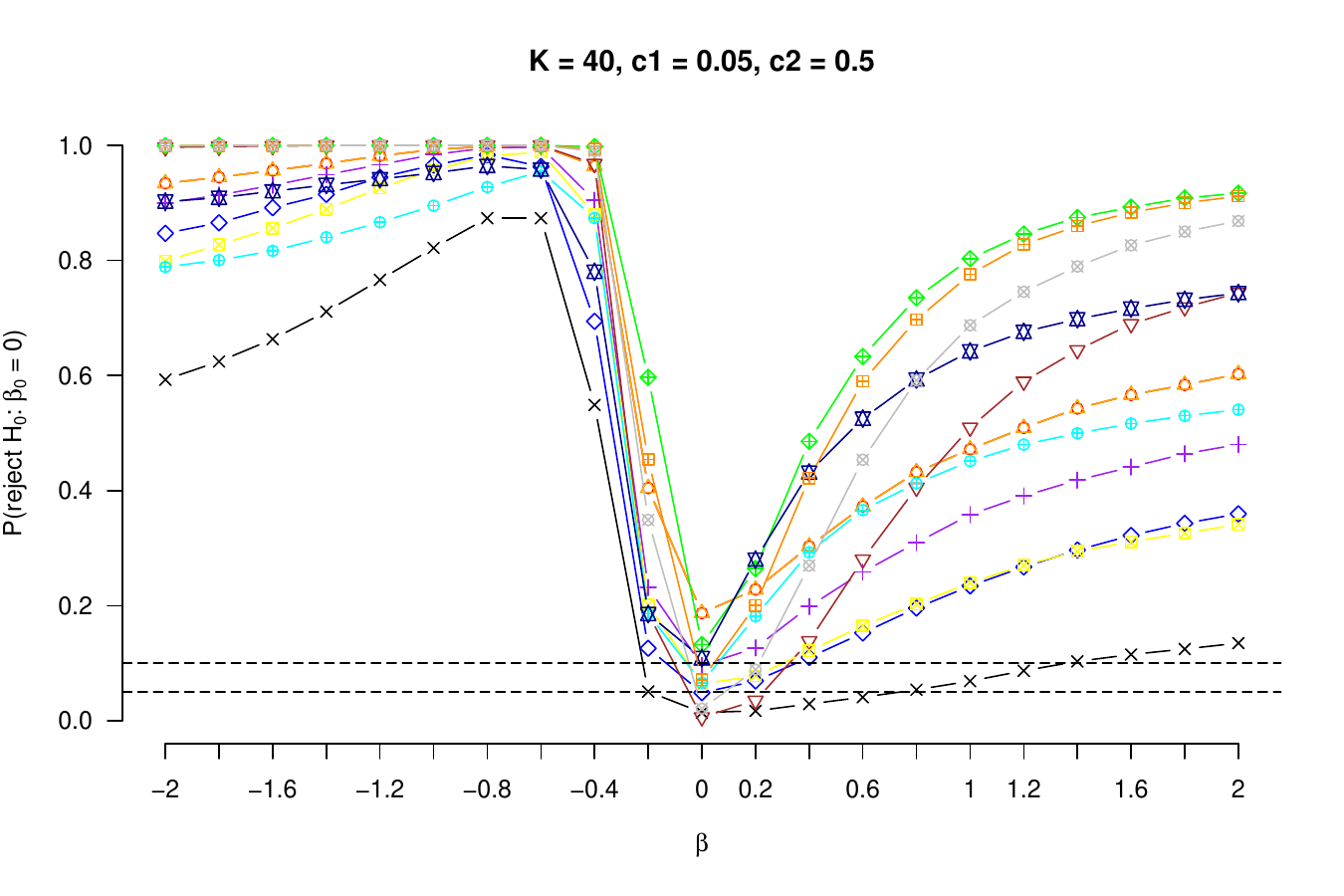}
    \caption{Plot with $(c_1,c_2) = (0.05,0.5)$ and $K = 40$}
    \textbf{Note:} We run 5,000 replications with 200 observations and 15 controls, using \citeauthor{Haus2012}{}'s DGP as in our main paper. 
\end{figure}

\begin{figure}[H]   \includegraphics[width=1\textwidth,height = 8cm]{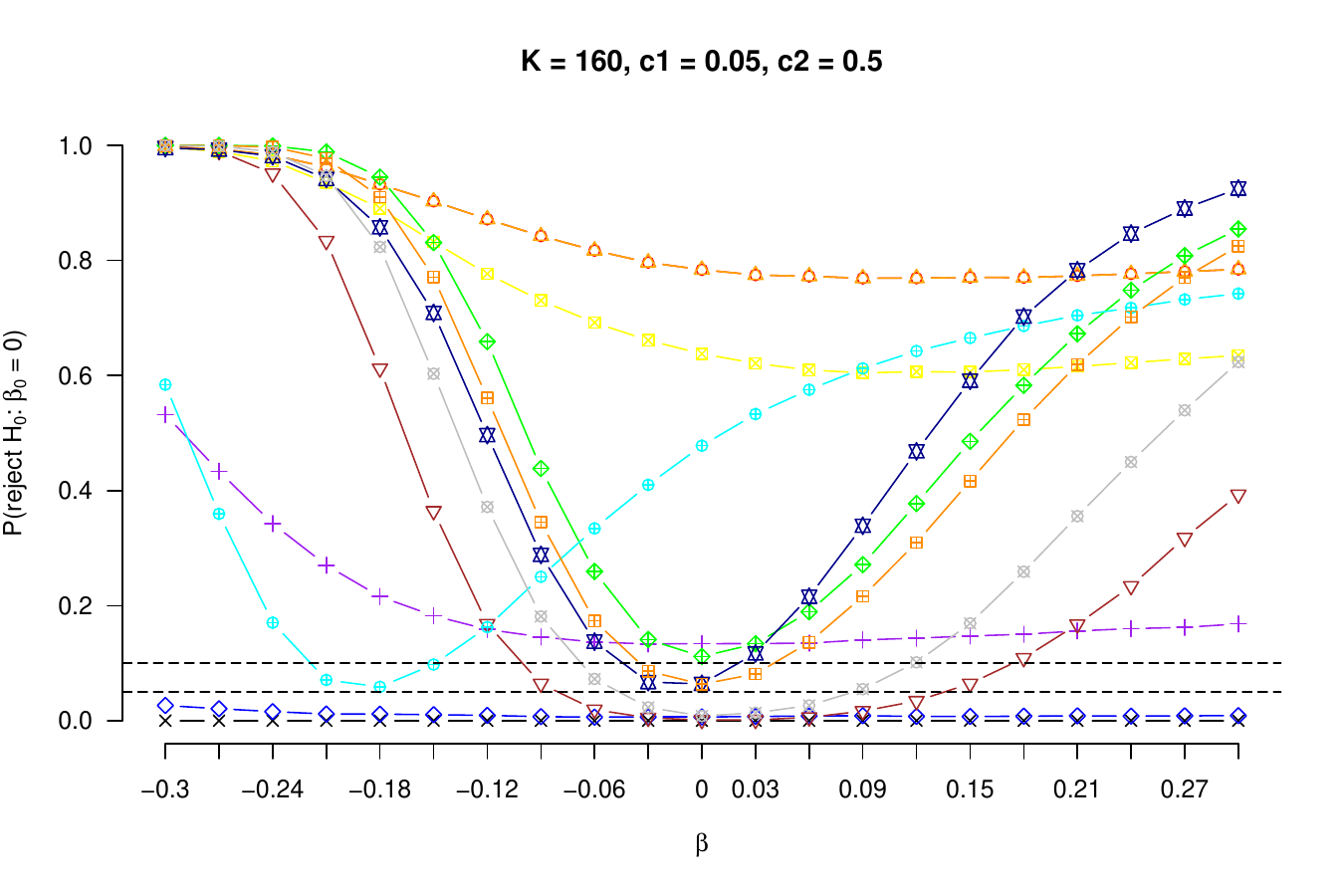}
    \caption{Plot with $(c_1,c_2) = (0.05,0.5)$ and $K = 160$}
   \textbf{Note:} We run 5,000 replications with 200 observations and 15 controls, using \citeauthor{Haus2012}{}'s DGP as in our main paper. 
\end{figure}

%%%
\begin{figure}[H]   \includegraphics[width=1\textwidth,height = 8cm]{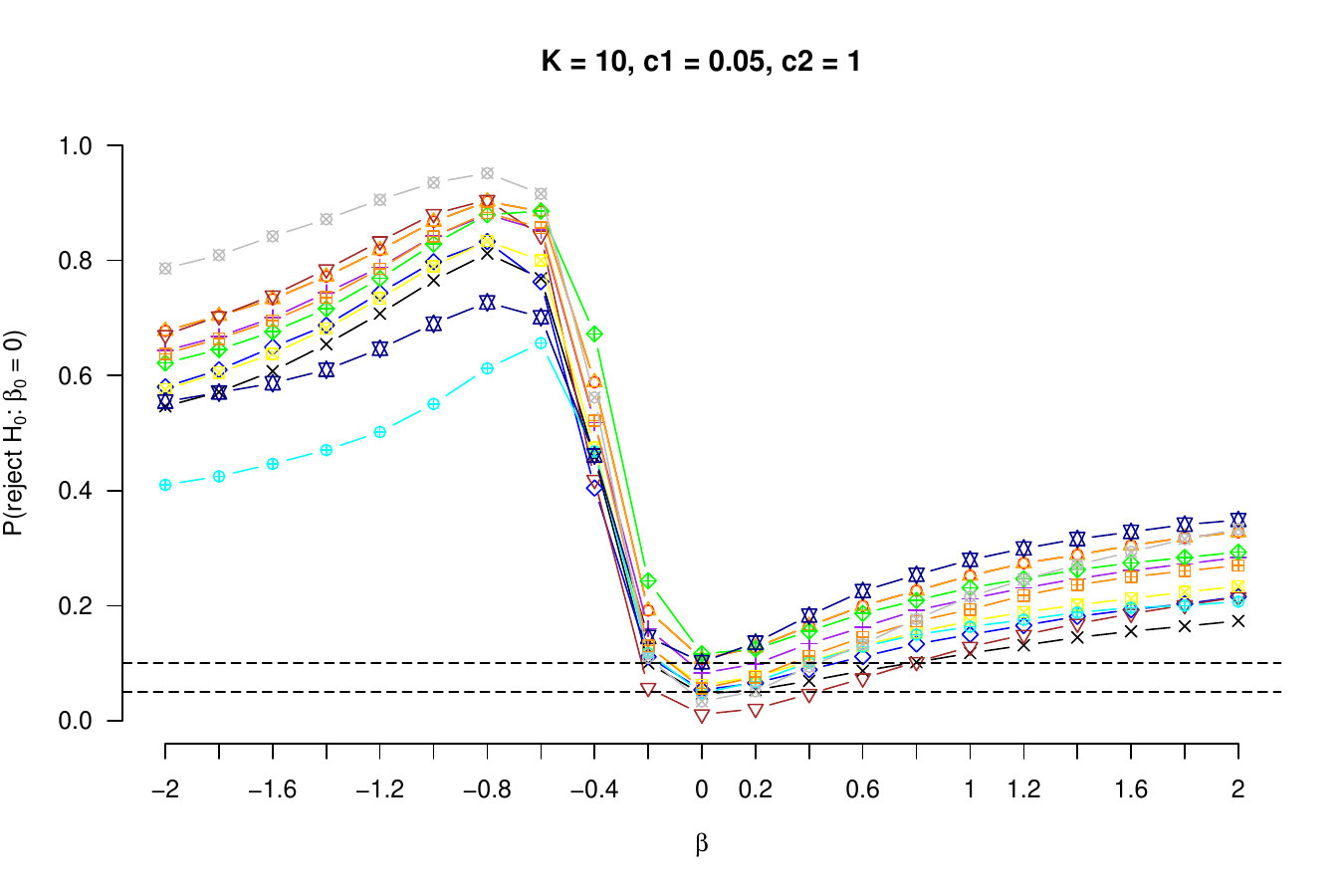}
    \caption{Plot with $(c_1,c_2) = (0.05,1)$ and $K = 10$}
     \textbf{Note:} We run 5,000 replications with 200 observations and 15 controls, using \citeauthor{Haus2012}{}'s DGP as in our main paper. 
\end{figure}

\begin{figure}[H]   \includegraphics[width=1\textwidth,height = 8cm]{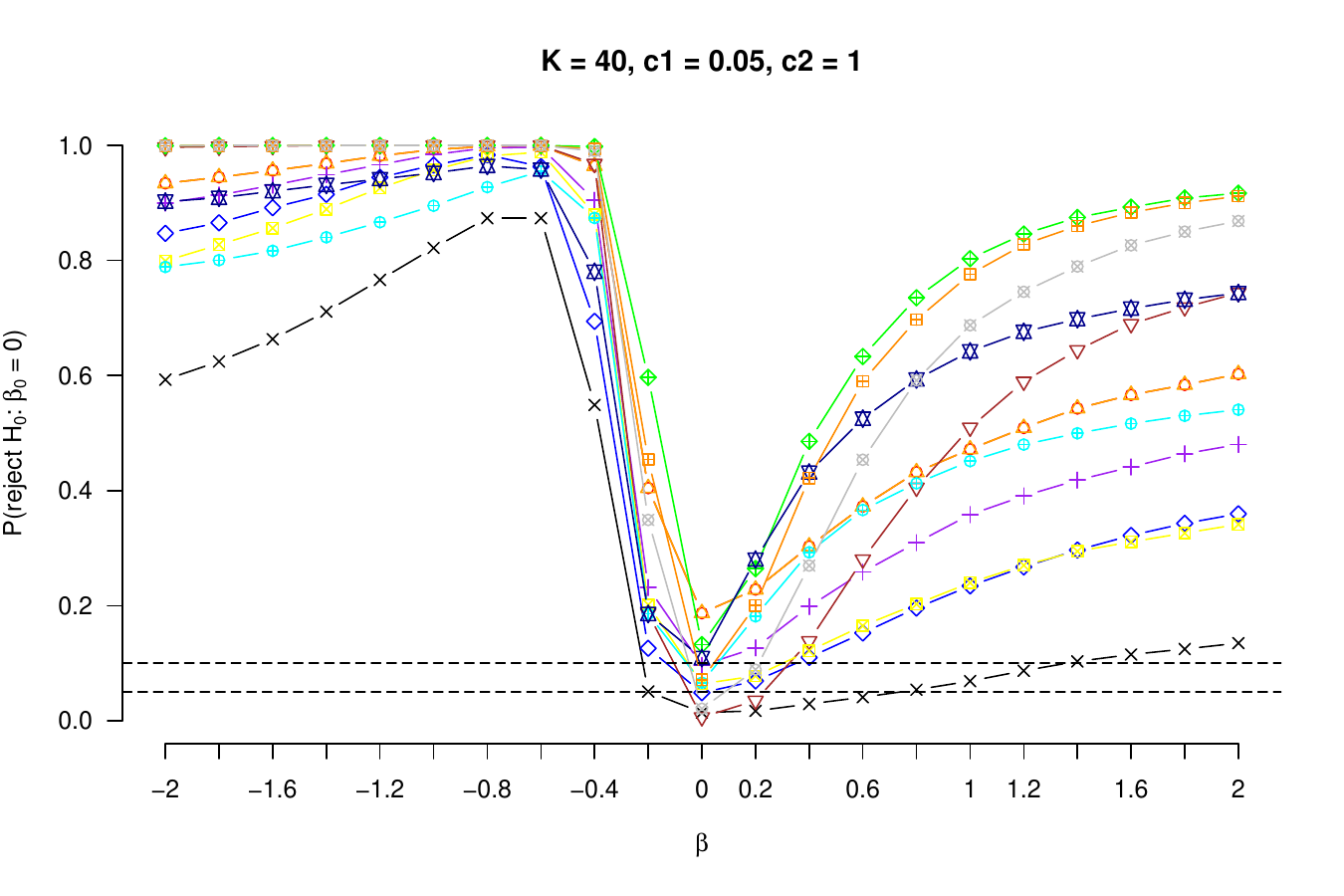}
    \caption{Plot with $(c_1,c_2) = (0.05,1)$ and $K = 40$}
    \textbf{Note:} We run 5,000 replications with 200 observations and 15 controls, using \citeauthor{Haus2012}{}'s DGP as in our main paper. 
\end{figure}

\begin{figure}[H]   \includegraphics[width=1\textwidth,height = 8cm]{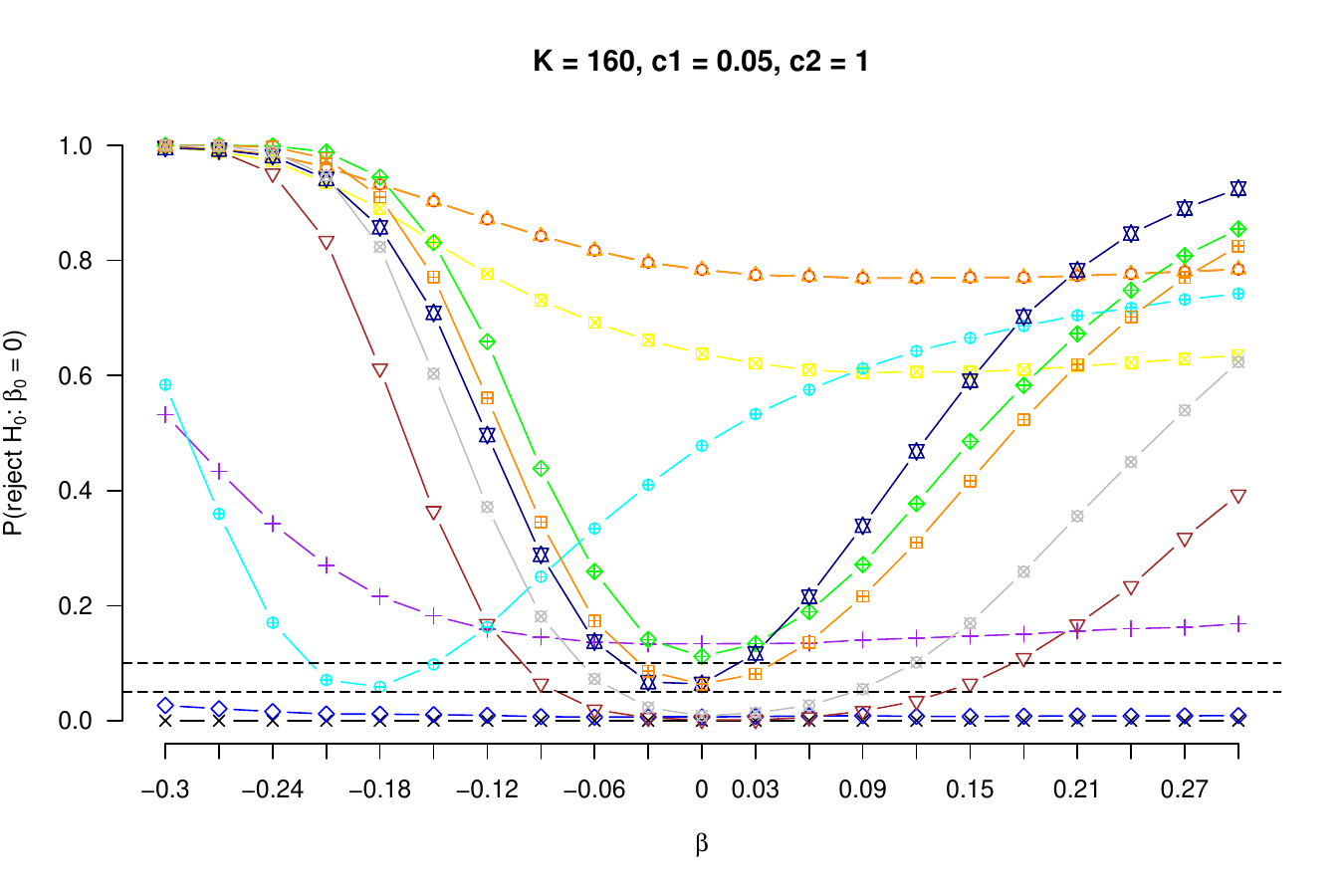}
    \caption{Plot with $(c_1,c_2) = (0.05,1)$ and $K = 160$}
     \textbf{Note:} We run 5,000 replications with 200 observations and 15 controls, using \citeauthor{Haus2012}{}'s DGP as in our main paper.   
\end{figure}

%%%
\begin{figure}[H]   \includegraphics[width=1\textwidth,height = 8cm]{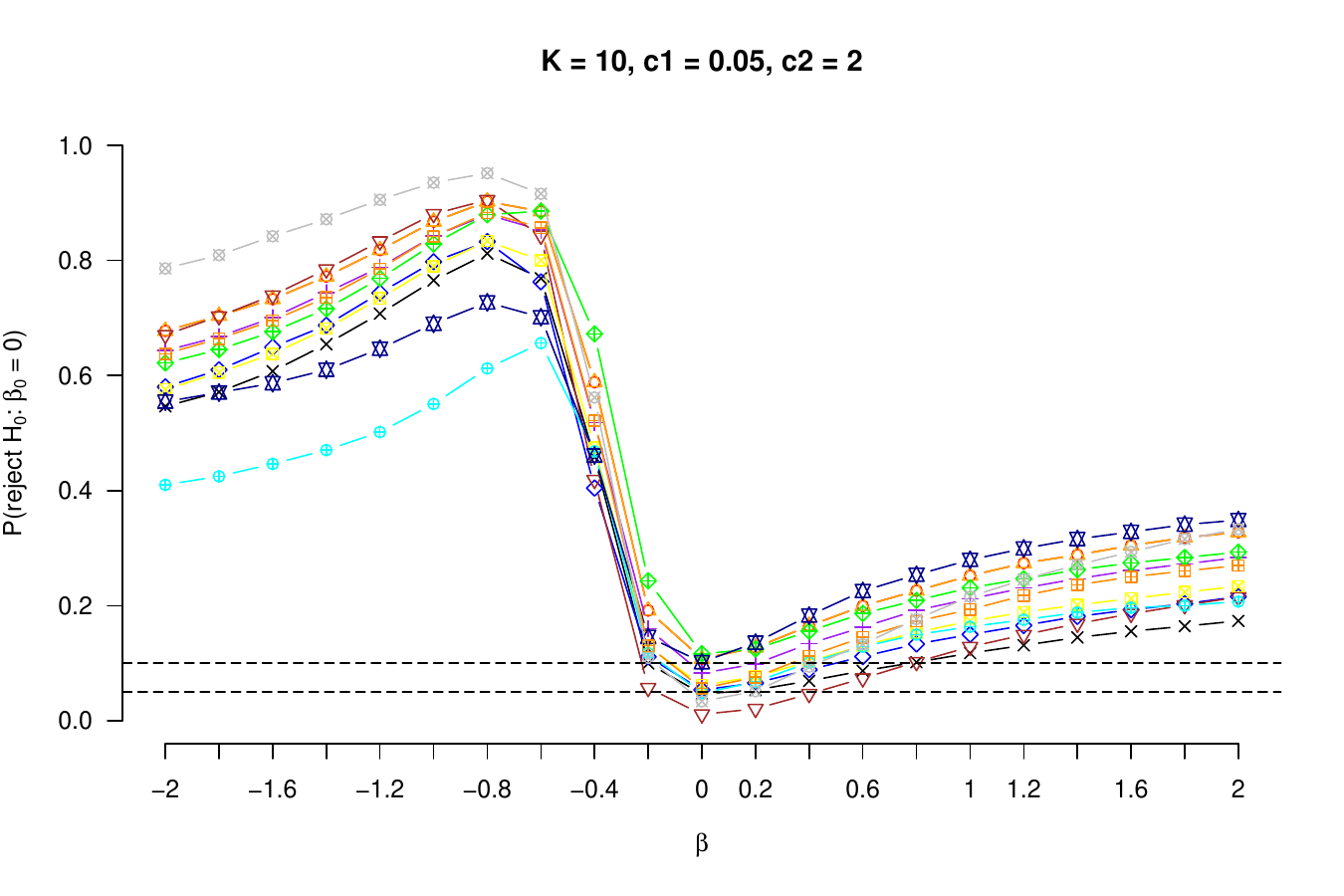}
    \caption{Plot with $(c_1,c_2) = (0.05,2)$ and $K = 10$}
     \textbf{Note:} We run 5,000 replications with 200 observations and 15 controls, using \citeauthor{Haus2012}{}'s DGP as in our main paper.  
\end{figure}

\begin{figure}[H]   \includegraphics[width=1\textwidth,height = 8cm]{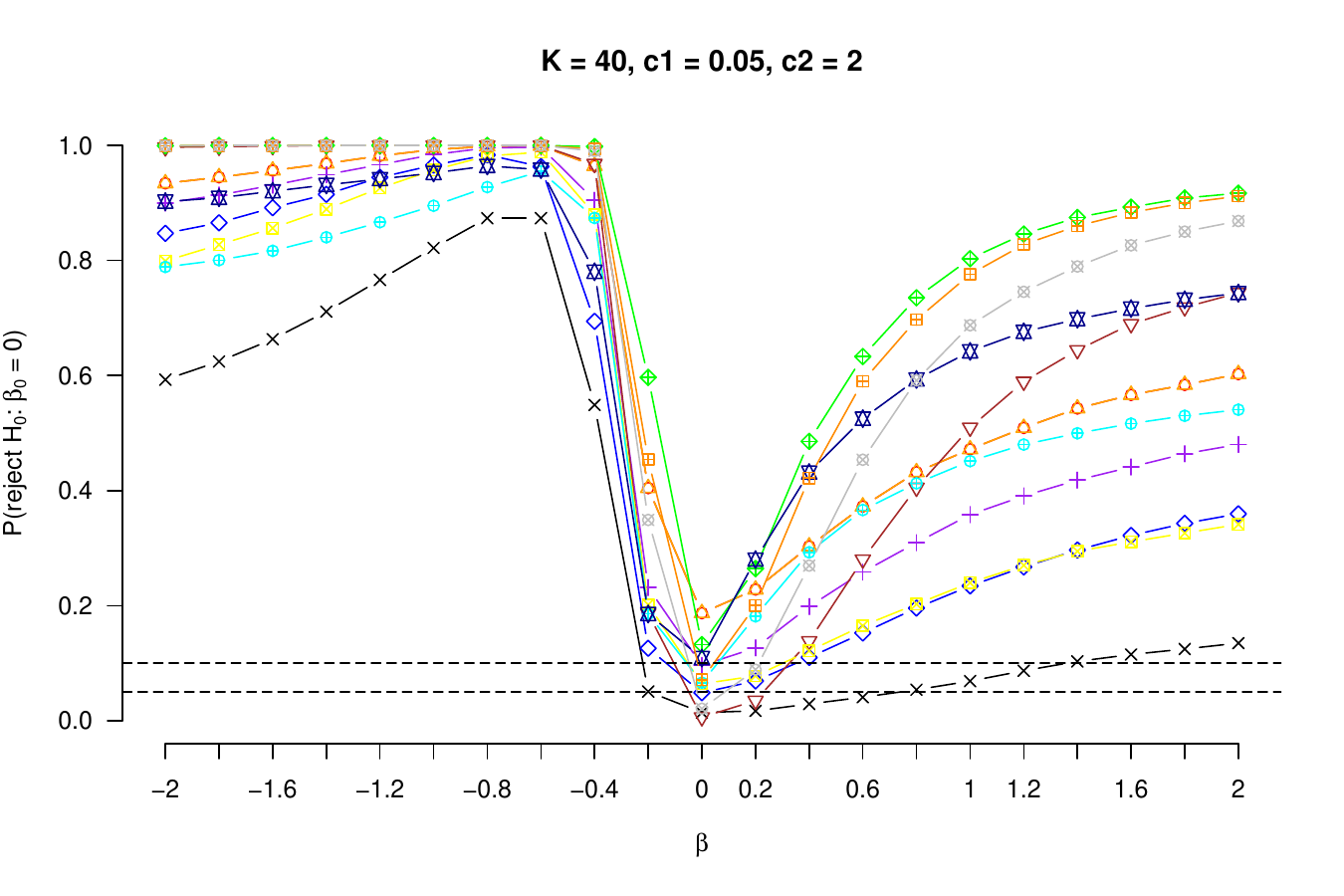}
    \caption{Plot with $(c_1,c_2) = (0.05,2)$ and $K = 40$}
     \textbf{Note:} We run 5,000 replications with 200 observations and 15 controls, using \citeauthor{Haus2012}{}'s DGP as in our main paper. 
\end{figure}

\begin{figure}[H]   \includegraphics[width=1\textwidth,height = 8cm]{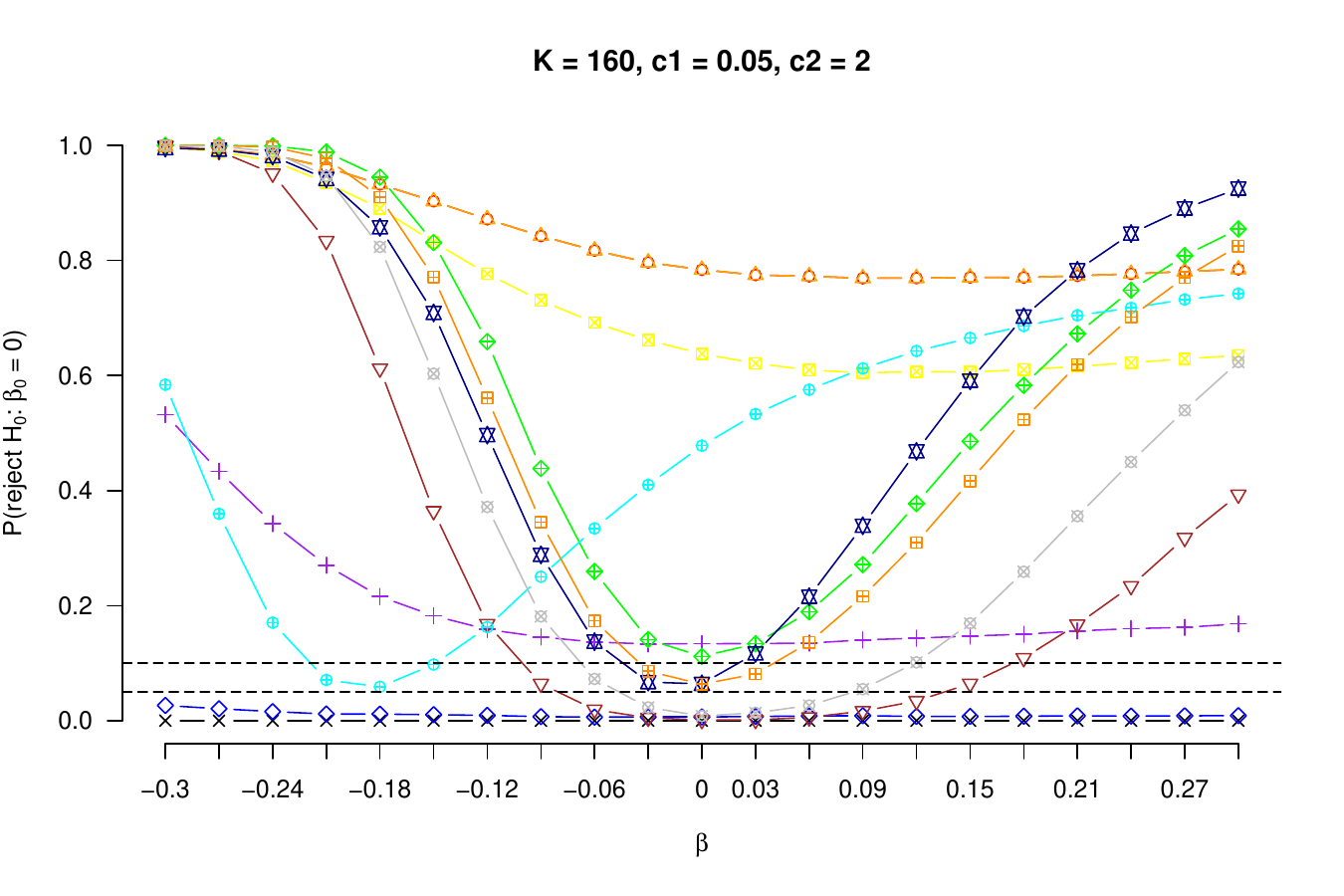}
    \caption{Plot with $(c_1,c_2) = (0.05,2)$ and $K = 160$}
     \textbf{Note:} We run 5,000 replications with 200 observations and 15 controls, using \citeauthor{Haus2012}{}'s DGP as in our main paper. 
\end{figure}

%%%
\begin{figure}[H]   \includegraphics[width=1\textwidth,height = 8cm]{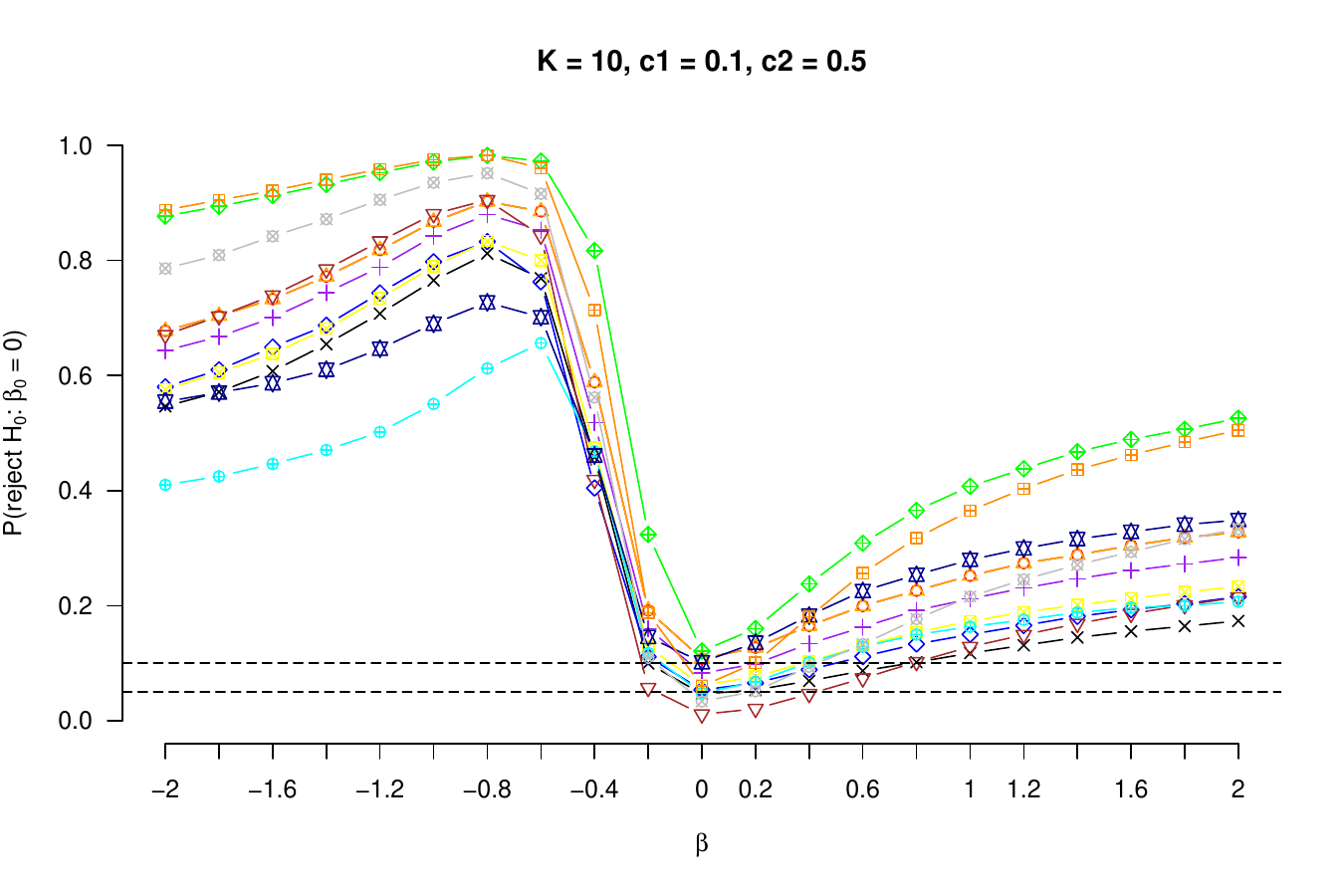}
    \caption{Plot with $(c_1,c_2) = (0.1,0.5)$ and $K = 10$}
     \textbf{Note:} We run 5,000 replications with 200 observations and 15 controls, using \citeauthor{Haus2012}{}'s DGP as in our main paper.  
\end{figure}

\begin{figure}[H]   \includegraphics[width=1\textwidth,height = 8cm]{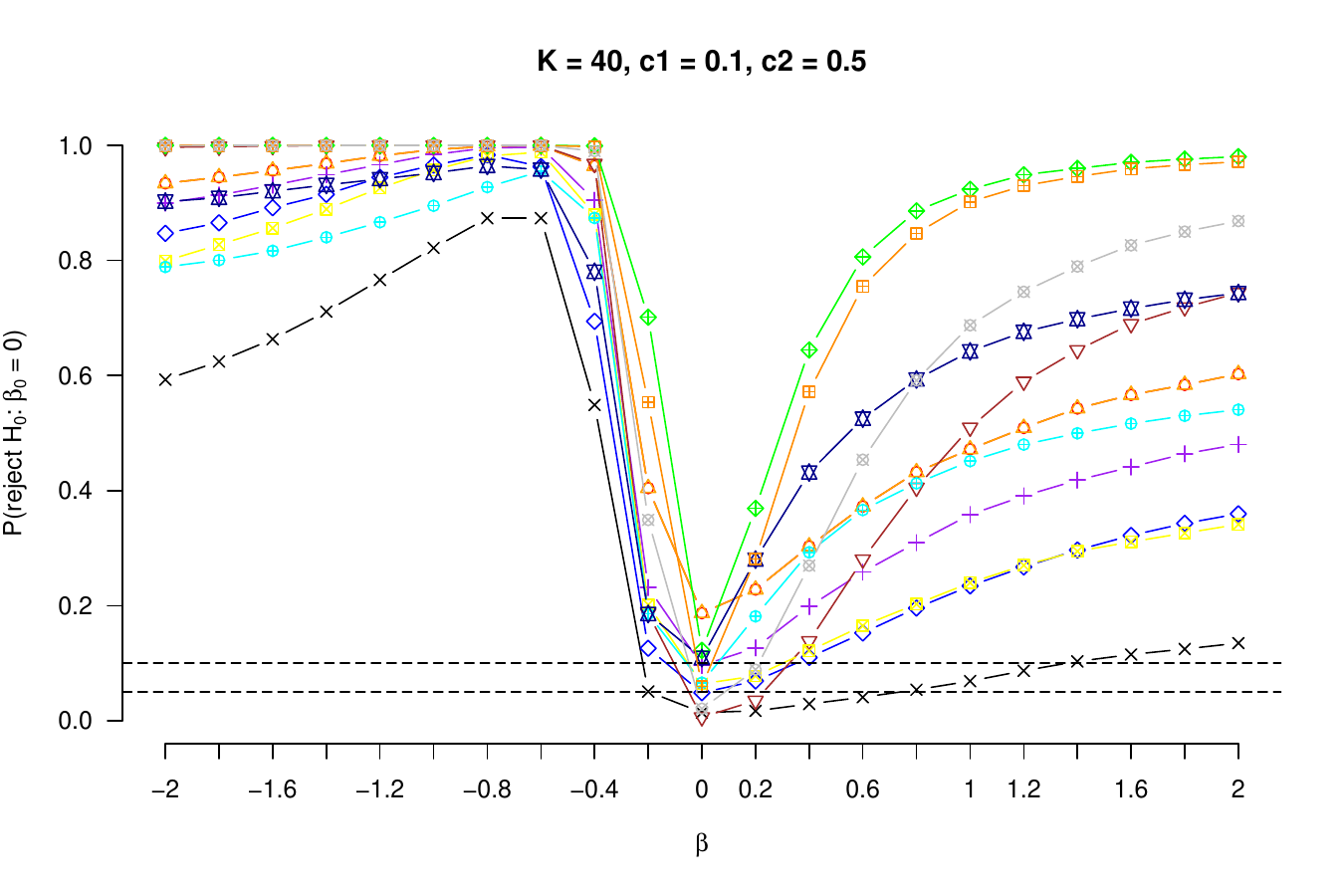}
    \caption{Plot with $(c_1,c_2) = (0.1,0.5)$ and $K = 40$}
     \textbf{Note:} We run 5,000 replications with 200 observations and 15 controls, using \citeauthor{Haus2012}{}'s DGP as in our main paper. 
\end{figure}

\begin{figure}[H]   \includegraphics[width=1\textwidth,height = 8cm]{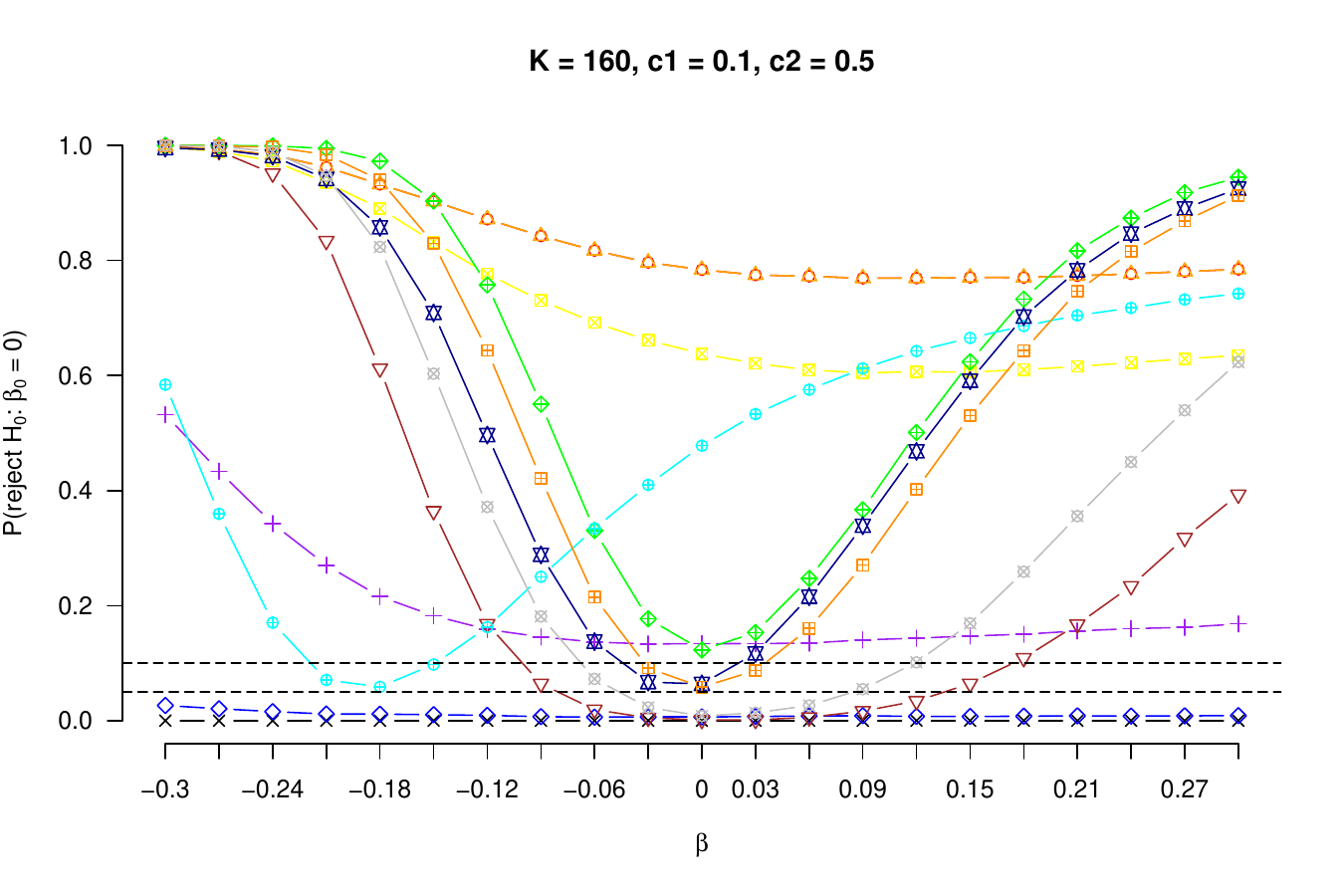}
    \caption{Plot with $(c_1,c_2) = (0.1,0.5)$ and $K = 160$}
     \textbf{Note:} We run 5,000 replications with 200 observations and 15 controls, using \citeauthor{Haus2012}{}'s DGP as in our main paper.  
\end{figure}

%%%
\begin{figure}[H]   \includegraphics[width=1\textwidth,height = 8cm]{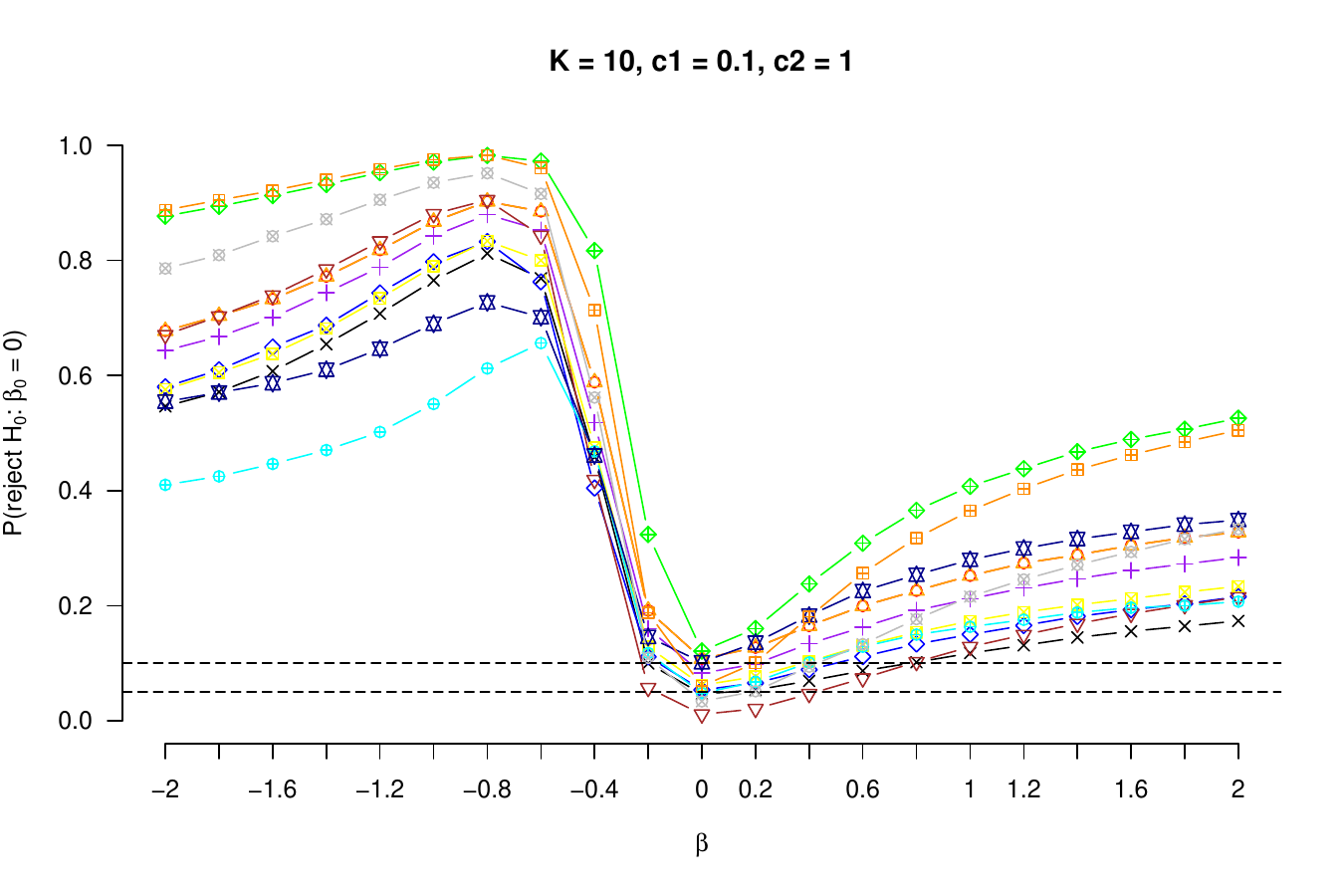}
    \caption{Plot with $(c_1,c_2) = (0.1,1)$ and $K = 10$}
     \textbf{Note:} We run 5,000 replications with 200 observations and 15 controls, using \citeauthor{Haus2012}{}'s DGP as in our main paper. 
\end{figure}

\begin{figure}[H]   \includegraphics[width=1\textwidth,height = 8cm]{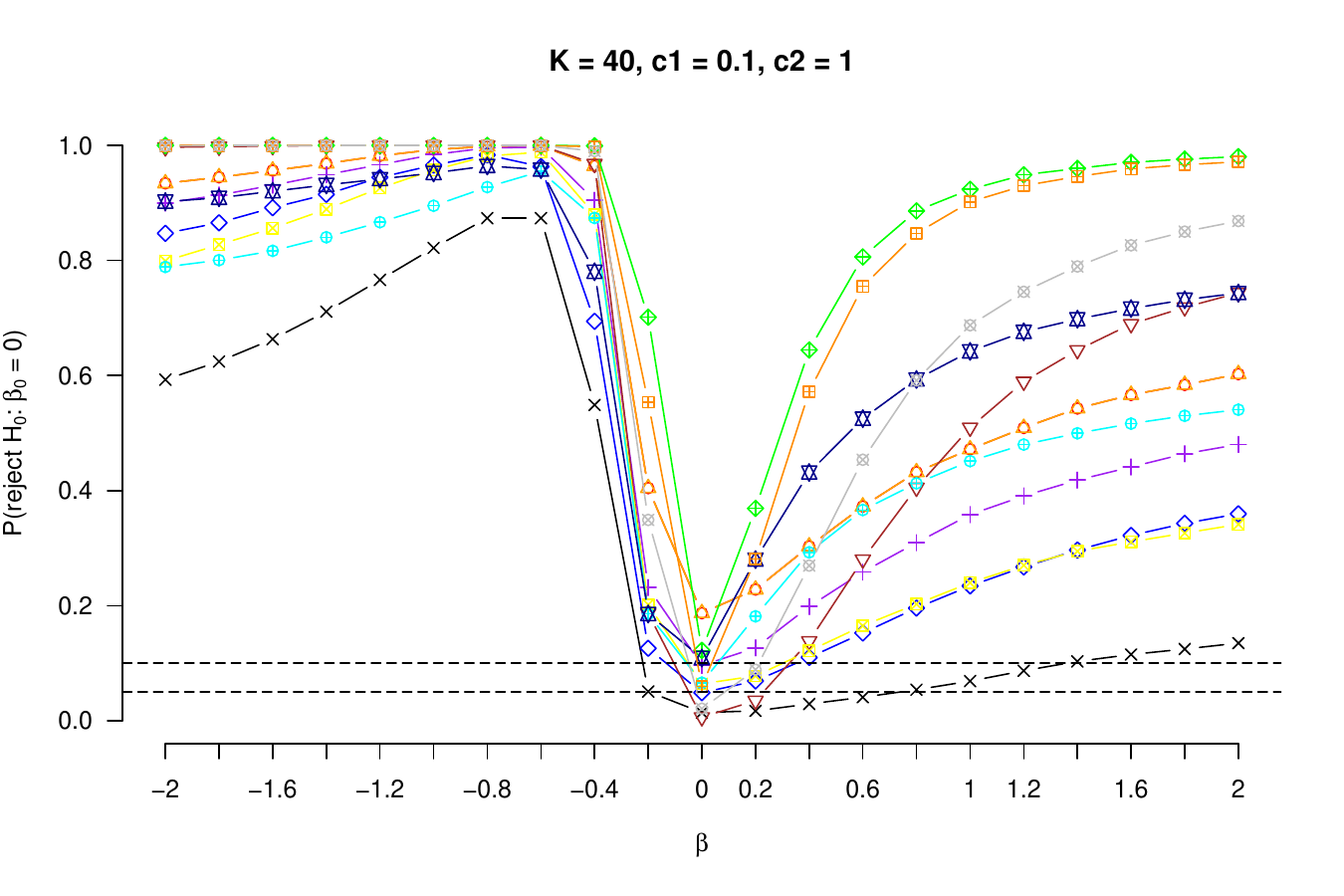}
    \caption{Plot with $(c_1,c_2) = (0.1,1)$ and $K = 40$}
     \textbf{Note:} We run 5,000 replications with 200 observations and 15 controls, using \citeauthor{Haus2012}{}'s DGP as in our main paper. 
\end{figure}

\begin{figure}[H]   \includegraphics[width=1\textwidth,height = 8cm]{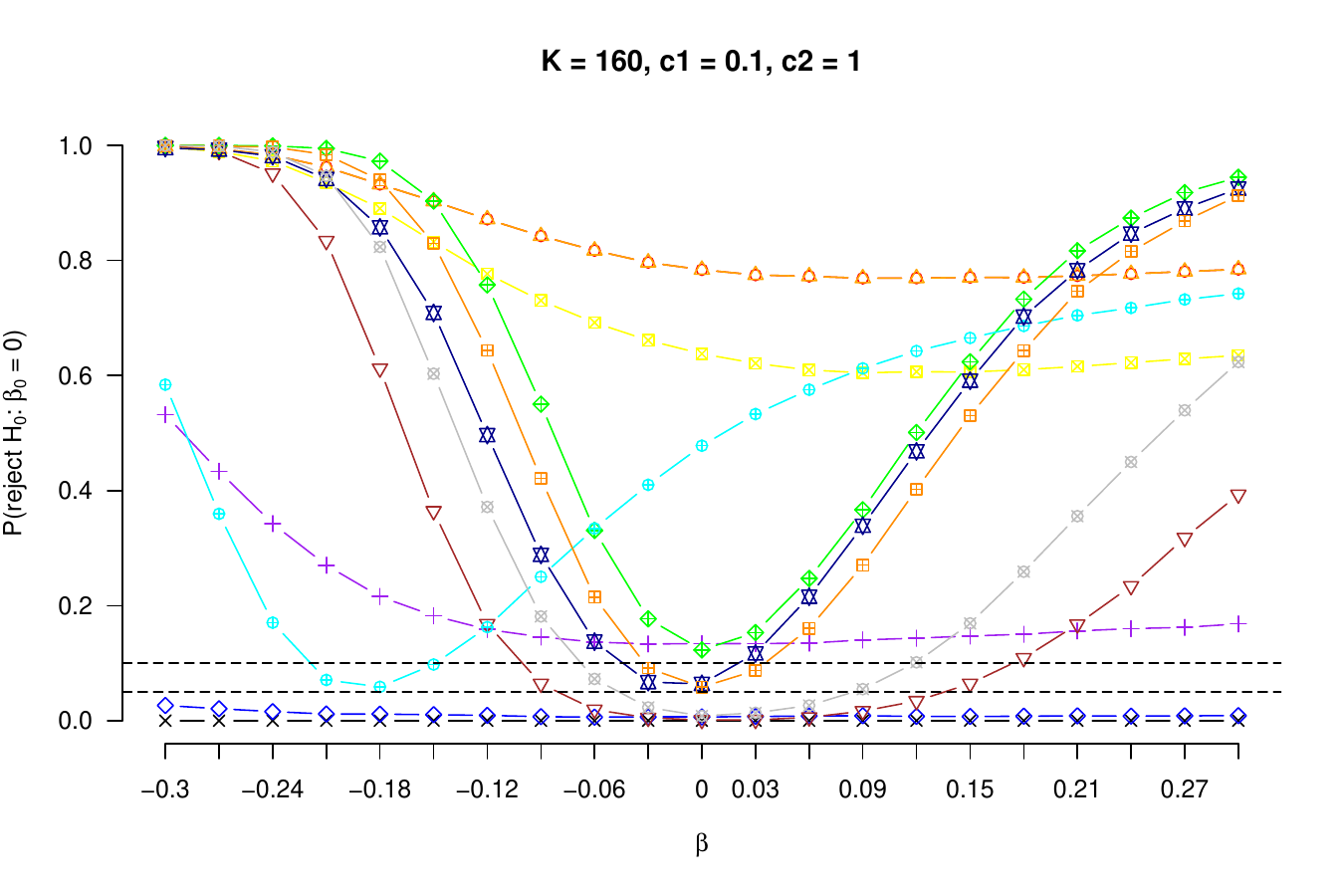}
    \caption{Plot with $(c_1,c_2) = (0.1,1)$ and $K = 160$}
     \textbf{Note:} We run 5,000 replications with 200 observations and 15 controls, using \citeauthor{Haus2012}{}'s DGP as in our main paper. 
\end{figure}

%%%
\begin{figure}[H]   \includegraphics[width=1\textwidth,height = 8cm]{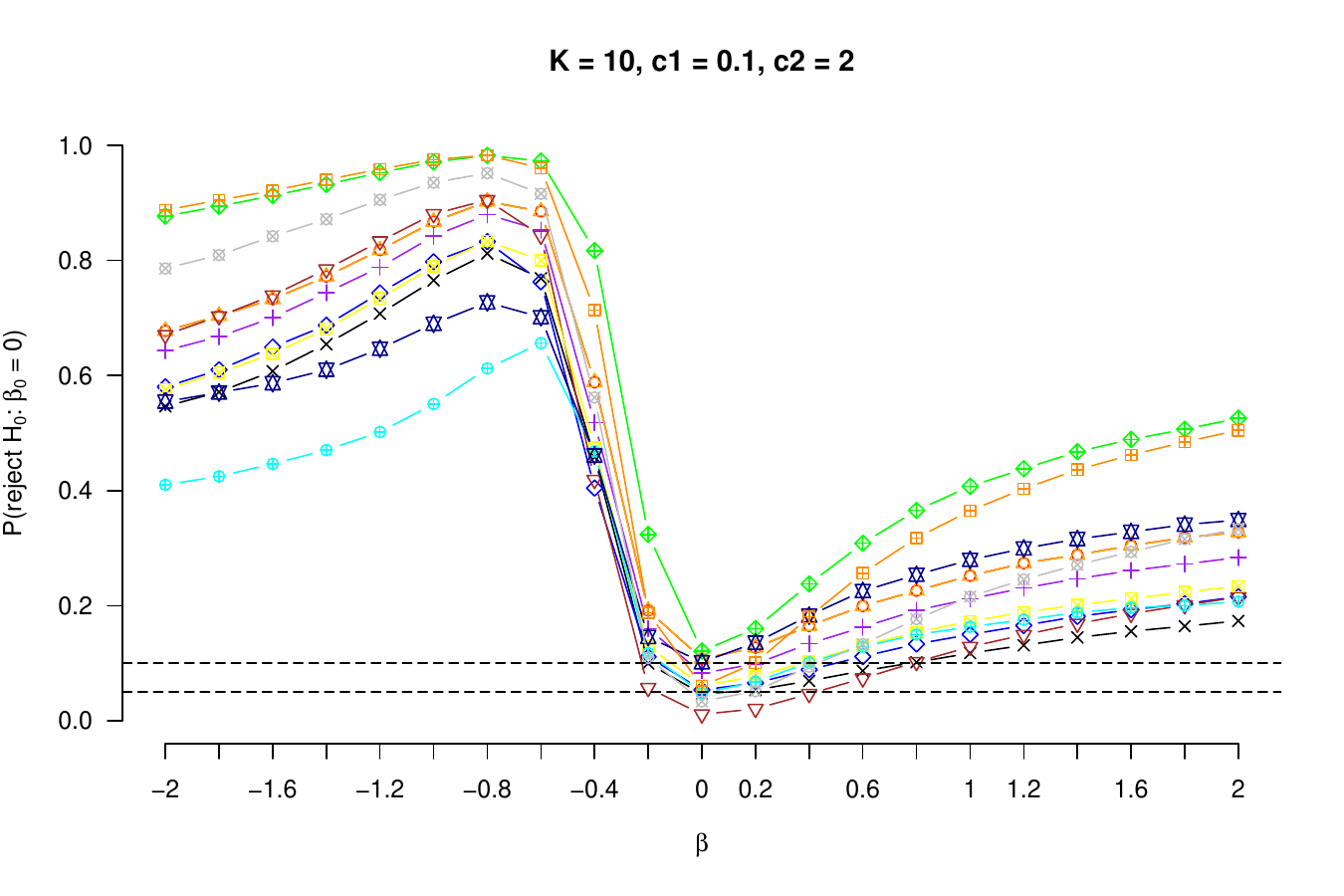}
    \caption{Plot with $(c_1,c_2) = (0.1,2)$ and $K = 10$}
     \textbf{Note:} We run 5,000 replications with 200 observations and 15 controls, using \citeauthor{Haus2012}{}'s DGP as in our main paper.  
\end{figure}

\begin{figure}[H]   \includegraphics[width=1\textwidth,height = 8cm]{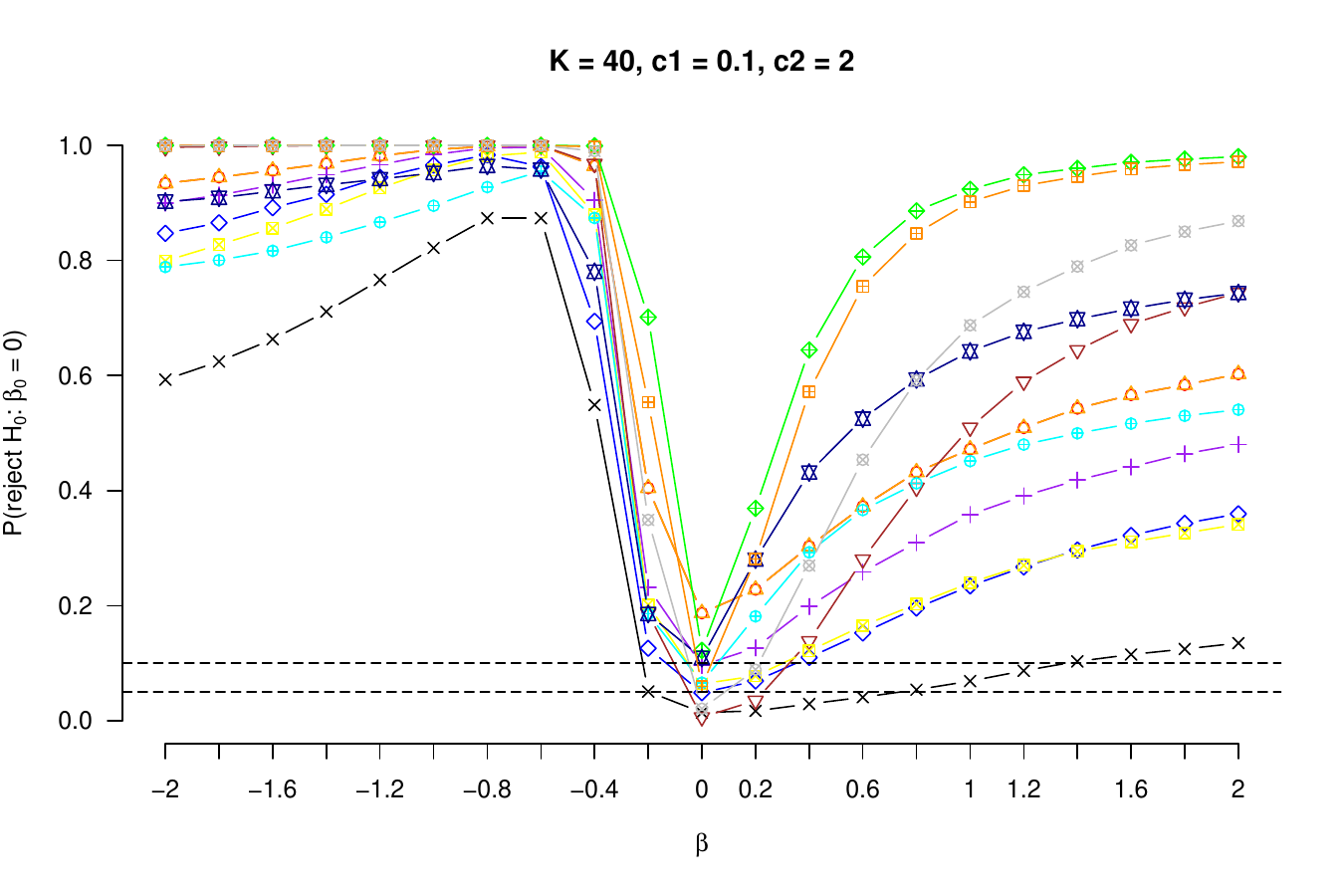}
    \caption{Plot with $(c_1,c_2) = (0.1,2)$ and $K = 40$}
     \textbf{Note:} We run 5,000 replications with 200 observations and 15 controls, using \citeauthor{Haus2012}{}'s DGP as in our main paper. 
\end{figure}

\begin{figure}[H]   \includegraphics[width=1\textwidth,height = 8cm]{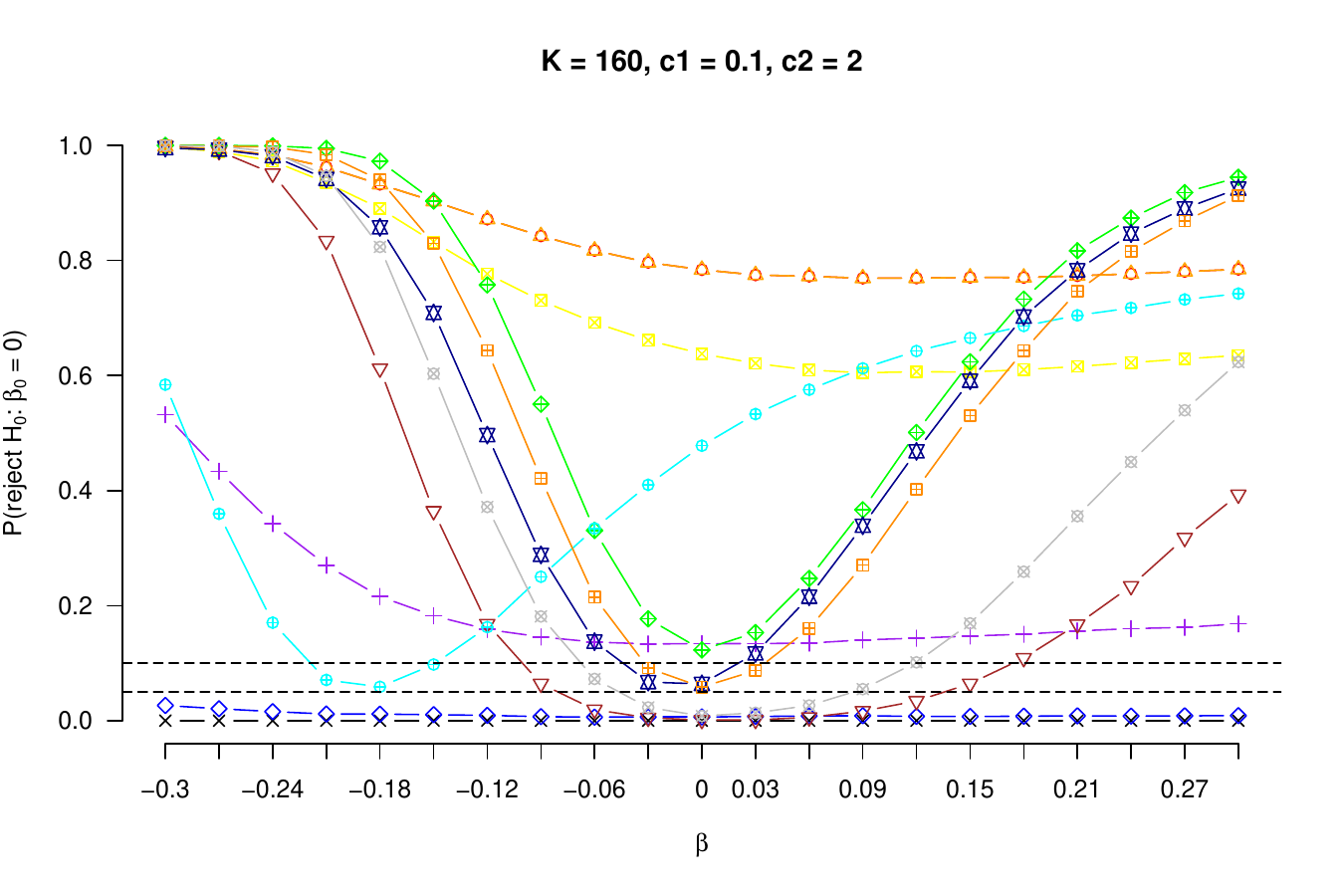}
    \caption{Plot with $(c_1,c_2) = (0.1,2)$ and $K = 160$}
     \textbf{Note:} We run 5,000 replications with 200 observations and 15 controls, using \citeauthor{Haus2012}{}'s DGP as in our main paper. 
\end{figure}

%%%
\begin{figure}[H]   \includegraphics[width=1\textwidth,height = 8cm]{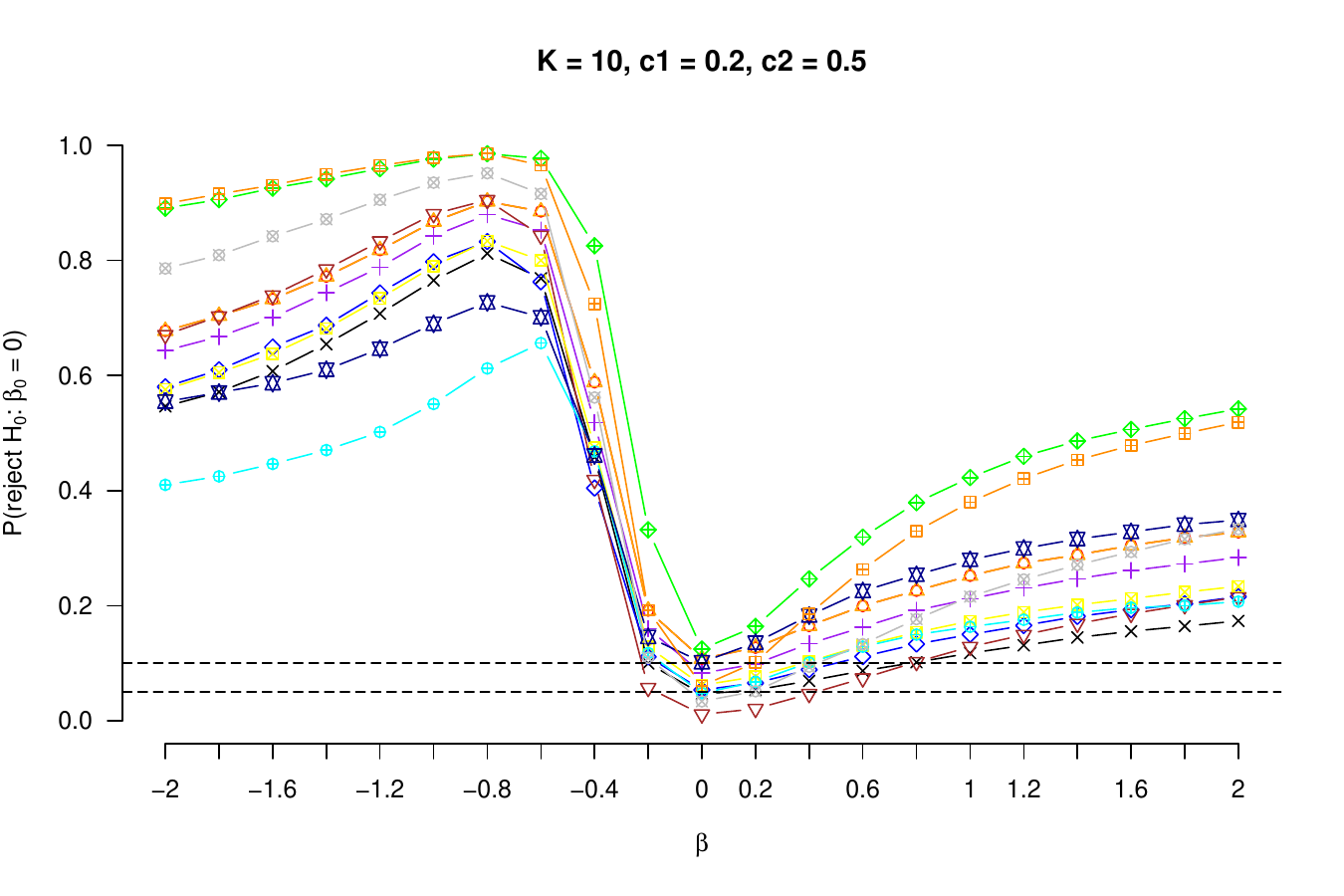}
    \caption{Plot with $(c_1,c_2) = (0.2,0.5)$ and $K = 10$}
    \textbf{Note:} We run 5,000 replications with 200 observations and 15 controls, using \citeauthor{Haus2012}{}'s DGP as in our main paper. 
\end{figure}

\begin{figure}[H]   \includegraphics[width=1\textwidth,height = 8cm]{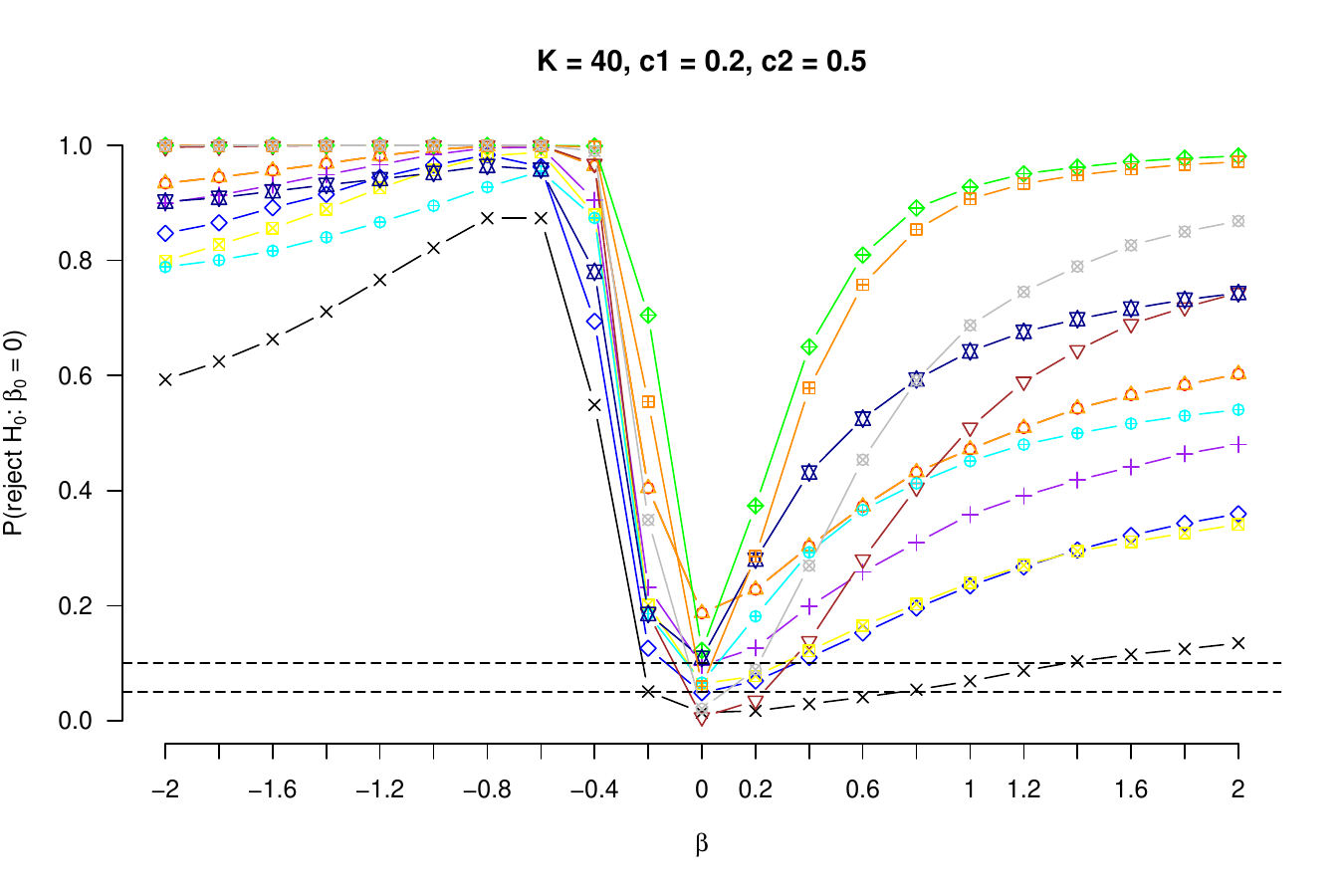}
    \caption{Plot with $(c_1,c_2) = (0.2,0.5)$ and $K = 40$}
     \textbf{Note:} We run 5,000 replications with 200 observations and 15 controls, using \citeauthor{Haus2012}{}'s DGP as in our main paper.  
\end{figure}

\begin{figure}[H]   \includegraphics[width=1\textwidth,height = 8cm]{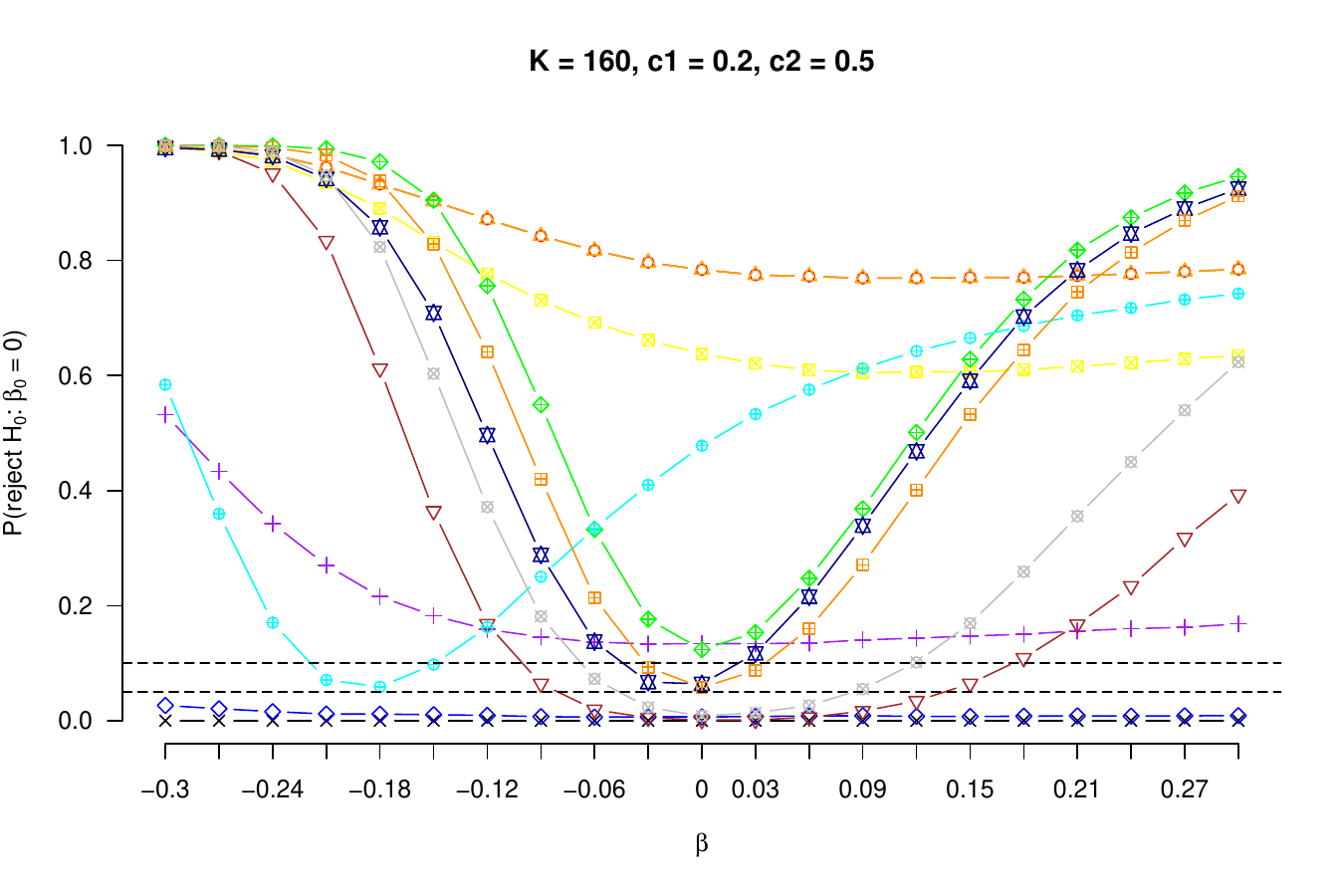}
    \caption{Plot with $(c_1,c_2) = (0.2,0.5)$ and $K = 160$}
     \textbf{Note:} We run 5,000 replications with 200 observations and 15 controls, using \citeauthor{Haus2012}{}'s DGP as in our main paper. 
\end{figure}

%%%
\begin{figure}[H]   \includegraphics[width=1\textwidth,height = 8cm]{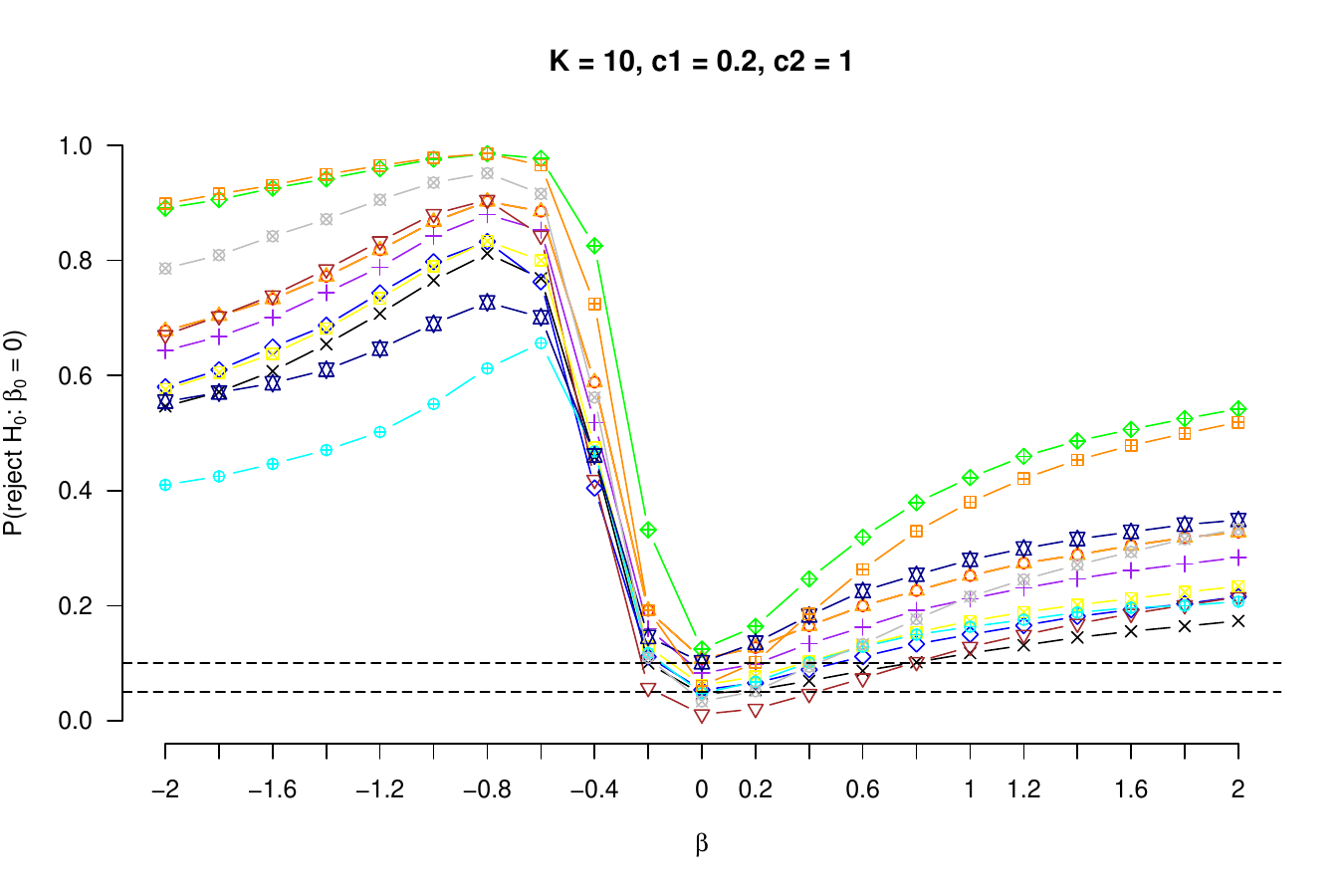}
    \caption{Plot with $(c_1,c_2) = (0.2,1)$ and $K = 10$}
     \textbf{Note:} We run 5,000 replications with 200 observations and 15 controls, using \citeauthor{Haus2012}{}'s DGP as in our main paper. 
\end{figure}

\begin{figure}[H]   \includegraphics[width=1\textwidth,height = 8cm]{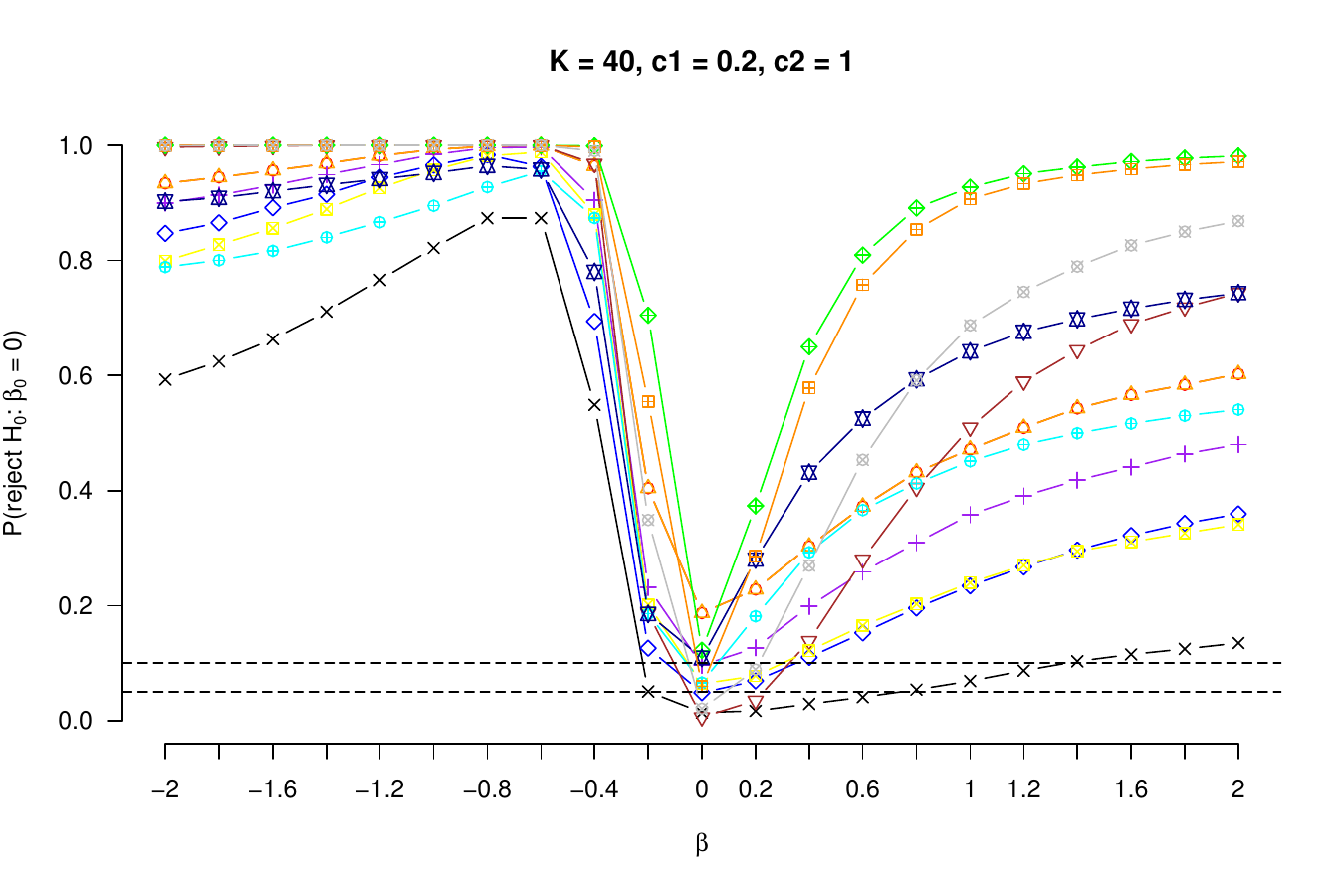}
    \caption{Plot with $(c_1,c_2) = (0.2,1)$ and $K = 40$}
     \textbf{Note:} We run 5,000 replications with 200 observations and 15 controls, using \citeauthor{Haus2012}{}'s DGP as in our main paper. 
\end{figure}

\begin{figure}[H]   \includegraphics[width=1\textwidth,height = 8cm]{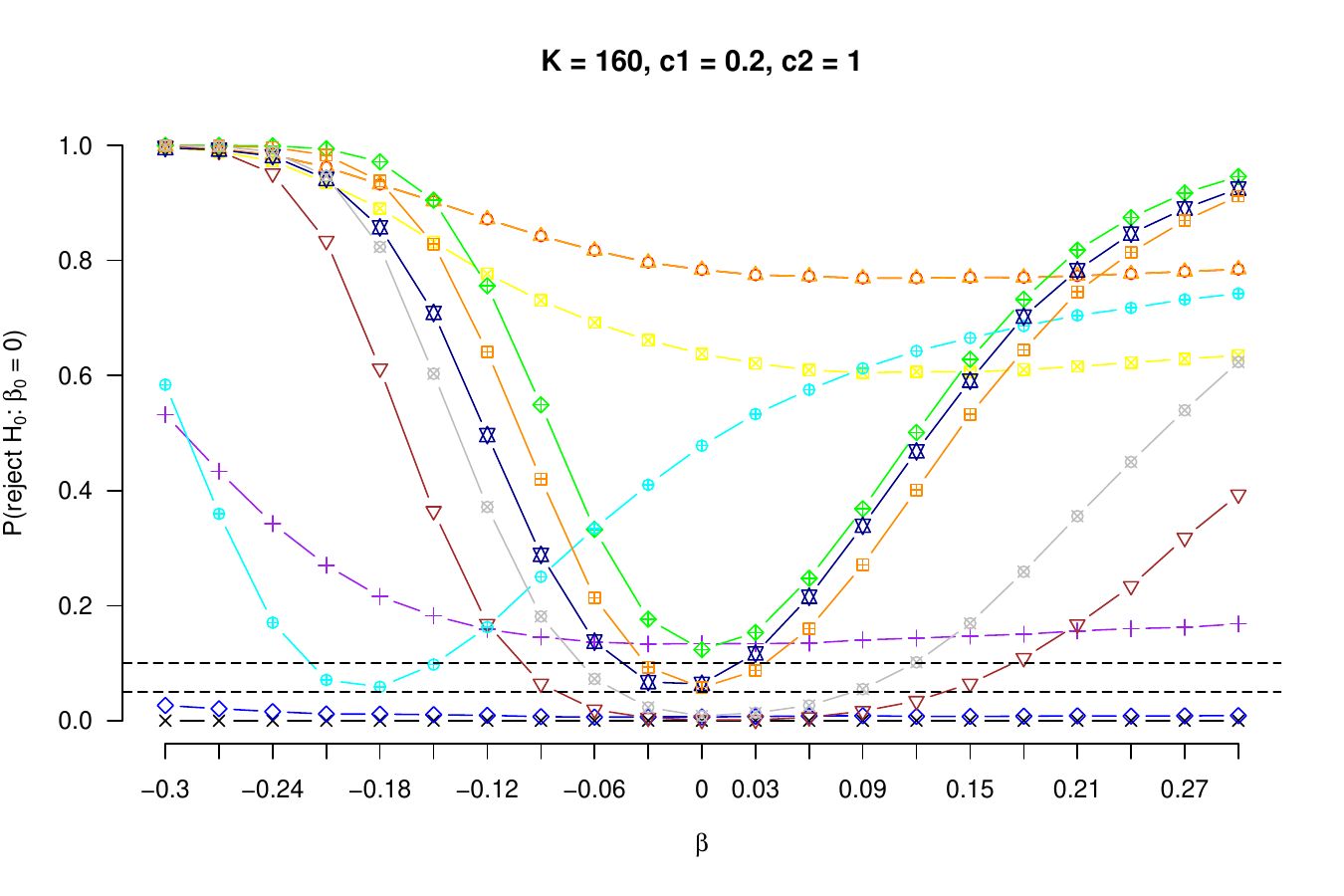}
    \caption{Plot with $(c_1,c_2) = (0.2,1)$ and $K = 160$}
     \textbf{Note:} We run 5,000 replications with 200 observations and 15 controls, using \citeauthor{Haus2012}{}'s DGP as in our main paper. 
\end{figure}

%%%
\begin{figure}[H]   \includegraphics[width=1\textwidth,height = 8cm]{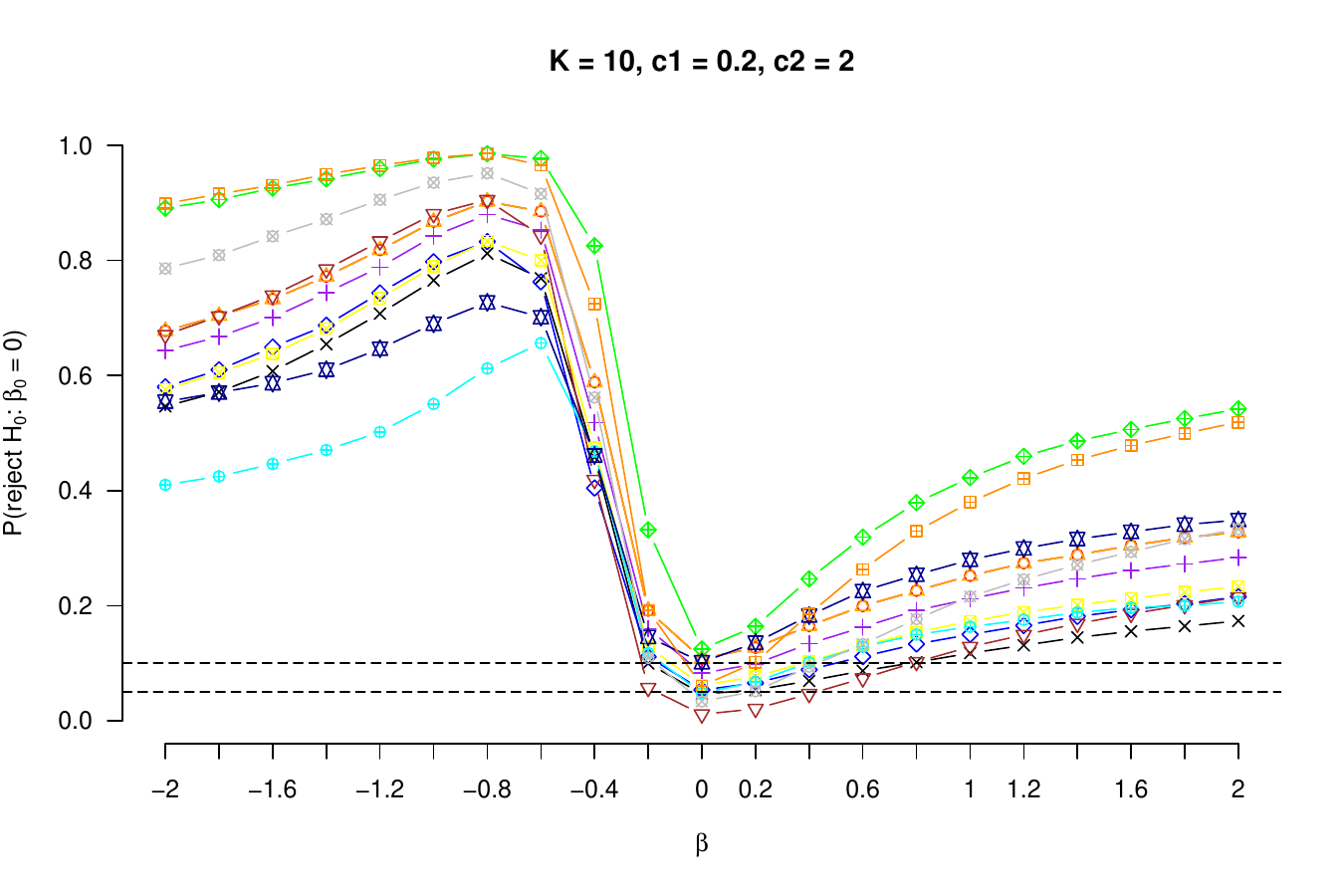}
    \caption{Plot with $(c_1,c_2) = (0.2,2)$ and $K = 10$}
     \textbf{Note:} We run 5,000 replications with 200 observations and 15 controls, using \citeauthor{Haus2012}{}'s DGP as in our main paper.  
\end{figure}

\begin{figure}[H]   \includegraphics[width=1\textwidth,height = 8cm]{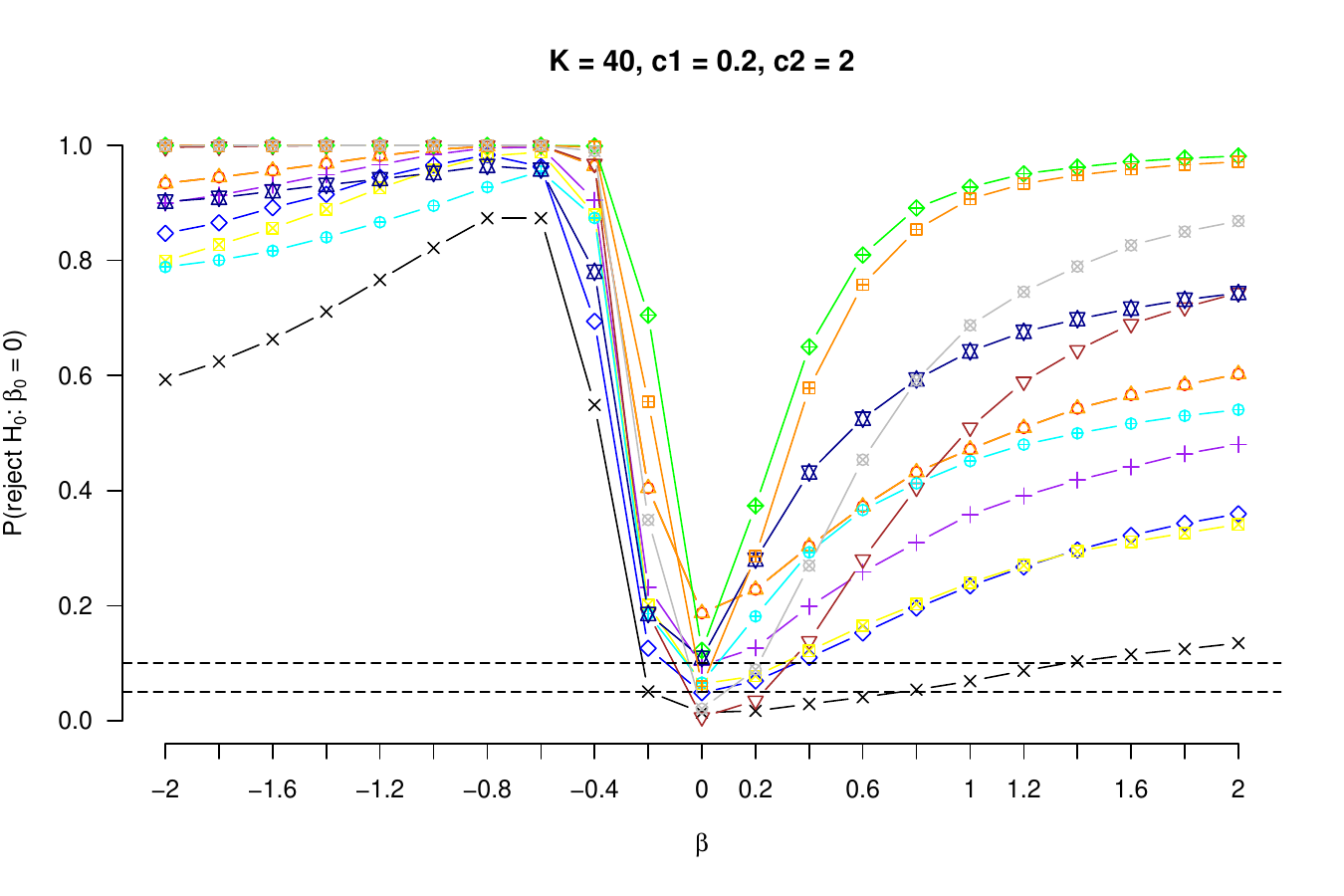}
    \caption{Plot with $(c_1,c_2) = (0.2,2)$ and $K = 40$}
     \textbf{Note:} We run 5,000 replications with 200 observations and 15 controls, using \citeauthor{Haus2012}{}'s DGP as in our main paper. 
\end{figure}

\begin{figure}[H]   \includegraphics[width=1\textwidth,height = 8cm]{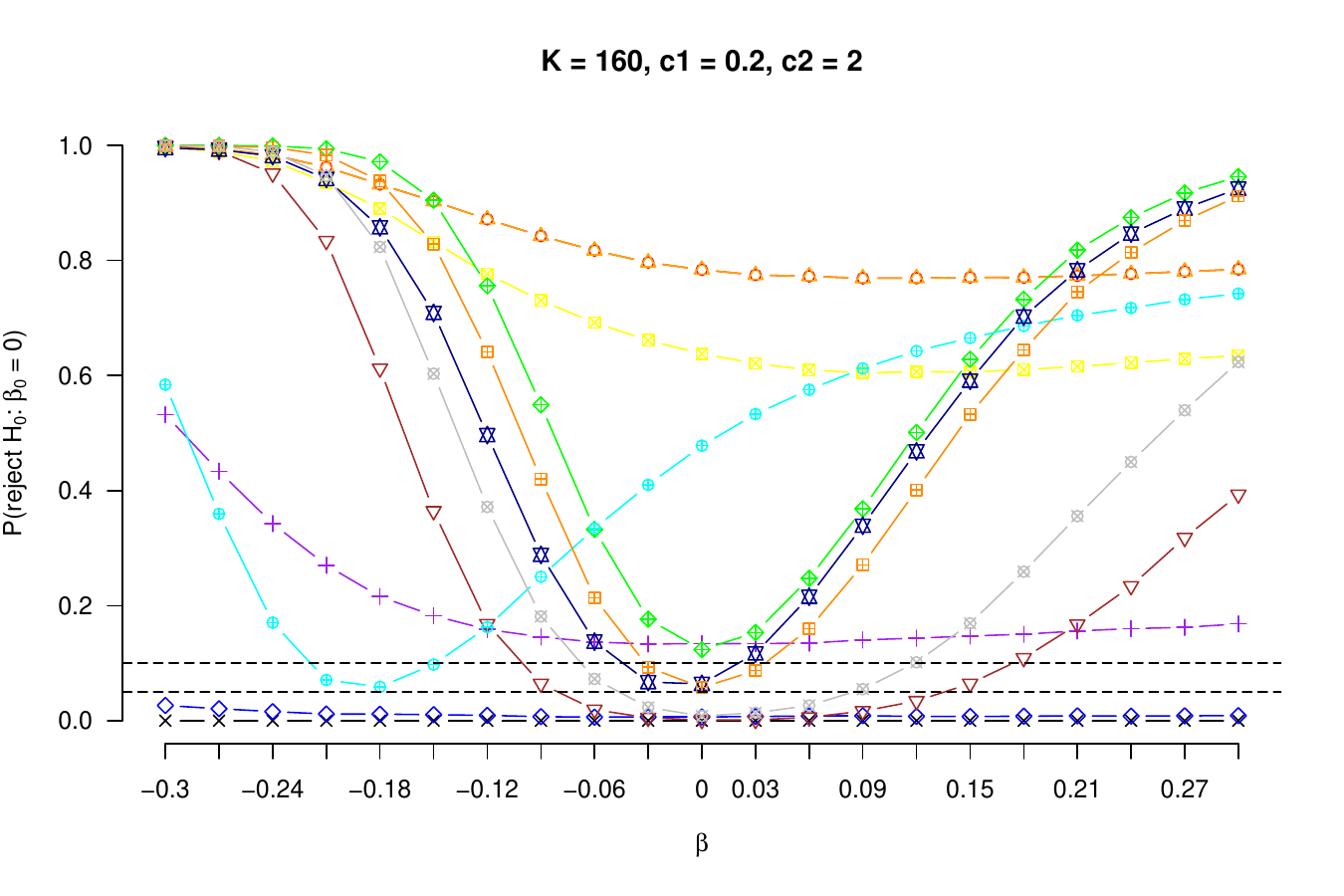}
    \caption{Plot with $(c_1,c_2) = (0.2,2)$ and $K = 160$}
   \textbf{Note:} We run 5,000 replications with 200 observations and 15 controls, using \citeauthor{Haus2012}{}'s DGP as in our main paper. 
    \label{Hausman_K_10_c1_0.2_c2_2}
\end{figure}

\singlespacing

\end{document}